\documentclass[11pt, letterpaper, reqno]{amsart}
\usepackage[utf8]{inputenc} 	
\usepackage{microtype} 			
\usepackage{geometry} 			
\usepackage{amsmath}			
\usepackage{amsthm}		 		
\usepackage{amssymb}	 		
\usepackage{bm}					
\usepackage{mathrsfs}			
\usepackage{thmtools}
\usepackage{xcolor}				
\definecolor{darkcandyapp}{rgb}{0.66, 0.0, 0.0}
\definecolor{darkblue}{rgb}{0.0, 0.0, 0.66}
\definecolor{mygreen}{rgb}{0.0, 0.66,0}
\usepackage{tikz}				
\usepackage{tikz-3dplot}
\usepackage{pgfplots}
\pgfplotsset{compat=1.13}
\usetikzlibrary{decorations.markings,decorations.pathmorphing,shapes}
\usepackage[all]{xy}			
\usepackage{graphicx}			
\usepackage{subcaption}
\usepackage{floatrow}
\usepackage{enumitem} 			
\usepackage{booktabs}
\usepackage{array}
\usepackage{lmodern}			
\usepackage[T1]{fontenc}		
\usepackage[colorlinks=true, citecolor=cyan, urlcolor=magenta]{hyperref}
\usepackage{cleveref}
\usepackage[backend=bibtex, style=alphabetic, date=year, backref=true, firstinits=true, isbn=false]{biblatex}
\makeatletter
\renewbibmacro{in:}{}
\DeclareFieldFormat{pages}{#1}
\renewcommand*{\bibnamedash}{%
	\leavevmode\raise +0.6ex\hbox to 5.5ex{\hrulefill}.\space\space}

\InitializeBibliographyStyle{\global\undef\bbx@lasthash}

\newbibmacro*{bbx:savehash}{\savefield{fullhash}{\bbx@lasthash}}

\renewbibmacro*{author}{%
	\ifboolexpr{
		test \ifuseauthor
		and
		not test {\ifnameundef{author}}
	}
	{%
		\iffieldequals{fullhash}{\bbx@lasthash}
		{\bibnamedash\addcomma\space}
		{\printnames{author}}%
		\usebibmacro{bbx:savehash}%
		\iffieldundef{authortype}
		{}
		{%
			\setunit{\addcomma\space}%
			\usebibmacro{authorstrg}%
		}%
	}
	{\global\undef\bbx@lasthash}%
}
\makeatother
\newtheorem{propositionx}{Proposition}[section]
\newenvironment{proposition}
{\pushQED{\qed}\propositionx}
{\popQED\endpropositionx}
\newenvironment{propositionp}
{\pushQED{\qed}\propositionx}
{\popQED\endpropositionx}
\newtheorem{theoremx}[propositionx]{Theorem}
\newenvironment{theorem}
{\pushQED{\qed}\theoremx}
{\popQED\endtheoremx}

\newtheorem{corollaryx}{Corollary}[propositionx]
\newenvironment{corollary}
{\pushQED{\qed}\corollaryx}
{\popQED\endcorollaryx}

\theoremstyle{remark}

\newtheorem*{examplex}{Example}
\newenvironment{example}
{\pushQED{\qed}\examplex}
{\popQED\endexamplex}

\numberwithin{equation}{section}

\newcommand{\dd}{\,\mathrm{d}}

\DeclareRobustCommand{\atled}{\text{\reflectbox{$\delta$}}}

\makeatletter
\@namedef{subjclassname@2020}{\textup{2020} Mathematics Subject Classification}
\makeatother

\newcommand{\bbC}{\mathbb{C}}

\newcommand{\bbN}{\mathbb{N}}

\newcommand{\bbR}{\mathbb{R}}

\newcommand{\bbZ}{\mathbb{Z}}

\newcommand{\calA}{\mathcal{A}}

\newcommand{\calD}{\mathcal{D}}
\newcommand{\calE}{\mathcal{E}}
\newcommand{\calF}{\mathcal{F}}

\newcommand{\calH}{\mathcal{H}}
\newcommand{\calI}{\mathcal{I}}
\newcommand{\calJ}{\mathcal{J}}

\newcommand{\calL}{\mathcal{L}}
\newcommand{\calM}{\mathcal{M}}

\newcommand{\calP}{\mathcal{P}}

\newcommand{\calS}{\mathcal{S}}

\newcommand{\calX}{\mathcal{X}}

\newcommand{\scrO}{\mathscr{O}}

\newcommand{\frakS}{\mathfrak{S}}

\newcommand{\bmalpha}{\bm{\alpha}}
\newcommand{\bmbeta}{\bm{\beta}}
\newcommand{\bmgamma}{\bm{\gamma}}

\newcommand{\bmdelta}{\bm{\delta}}

\newcommand{\bmlambda}{\bm{\lambda}}

\newcommand{\bmrho}{\bm{\rho}}
\newcommand{\bmvarrho}{\bm{\varrho}}

\title{The singularities of Selberg- and Dotsenko--Fateev-like integrals}
\author{Ethan Sussman}
\date{December 22nd, 2023 (Published version), January 8th, 2023 (Preprint).}
\email{ethanws@mit.edu}
\address{Department of Mathematics, Massachusetts Institute of Technology, Massachusetts 02139-4307, USA}

\begin{document}

\begin{abstract}
	We discuss the meromorphic continuation of certain hypergeometric integrals modeled on the Selberg integral, including the 3-point and 4-point functions of BPZ's minimal models of 2D CFT as described by Felder \& Silvotti and Dotsenko \& Fateev (the ``Coulomb gas formalism''). This is accomplished via a geometric analysis of the singularities of the integrands. In the case that the integrand is symmetric (as in the Selberg integral itself) or, more generally, what we call ``DF-symmetric,'' we show that a number of apparent singularities are removable, as required for the construction of the minimal models via these methods. 
\end{abstract}
	
\maketitle
\tableofcontents

\section{Introduction}
Let 
\begin{equation}
\triangle_N = \{(x_1,\ldots,x_N)\in [0,1]^N : x_1 \leq \cdots \leq x_N\} 
\end{equation}
denote the standard $N$-simplex, which we consider as a subset of $\bbC^N$. 
We study in this note \emph{Selberg-like} integrals, by which we mean definite integrals of the form 
\begin{equation} 
S_N[F](\bmalpha,\bmbeta,\bmgamma) = \int_{\triangle_N}  F(x_1,\ldots,x_N) \prod_{j=1}^N x_j^{\alpha_j}(1-x_j)^{\beta_j}   \prod_{1\leq j<k \leq N} (x_k-x_j)^{2\gamma_{j,k}}  \dd x_1\cdots \dd x_N,  
\label{eq:Selberg_def}
\end{equation} 
for $N\in \bbN^+$, $F\in C^\infty(\triangle_N)$, and $\bmalpha=\{\alpha_j\}_{j=1}^N,\bmbeta=\{\beta_j\}_{j=1}^N,\bmgamma=\{\gamma_{j,k}=\gamma_{k,j}\}_{1\leq j<k \leq N} \subset \bbC$ such that the integrand above is absolutely integrable on $\triangle_N$. 
Integrals of this form are relevant to an array of topics in mathematical physics \cite{FW}. 
However, it is often necessary to consider exponents $\bmalpha,\bmbeta,\bmgamma$ for which the integral above is \emph{not} absolutely convergent, in which case a meromorphic extension needs to be performed. 
In some applications, only the behavior of this extension at generic exponents is required. In others, such as the application -- discussed below -- to the construction of the minimal models of 2D CFT, it is necessary to consider particular values, e.g.\ $\gamma_{j,k}=-1$. Unfortunately, for these particular values, previous work on the subject is not sufficient.

We will identify indexed collections of complex numbers (and tuples thereof) with column vectors. For example, we identify $\bmgamma$ with an element of $\smash{\bbC^{N(N-1)/2}}$ and 
\begin{equation} 
(\bmalpha,\bmbeta,\bmgamma)\in \bbC^N\times \bbC^N \times \bbC^{N(N-1)/2}
\end{equation} 
with an element of $\bbC^{2N+N(N-1)/2}$. Similar identifications will be made throughout the rest of the paper without further comment.
Let 
\begin{equation} 
\Omega_{N} = \bigg\{(\bmalpha,\bmbeta,\bmgamma) \in \bbC^{2N+N(N-1)/2} :  \prod_{j=1}^N x_j^{\alpha_j}(1-x_j)^{\beta_j}  \prod_{1\leq j<k \leq N} (x_k-x_j)^{2\gamma_{j,k}} \in L^1(\triangle_N) \bigg\}
\end{equation} 
denote the (open, nonempty) subset of $\bbC^{2N+N(N-1)/2}$ consisting of the $(\bmalpha,\bmbeta,\bmgamma)\in \bbC^{2N+N(N-1)/2}$ for which the integrand in \cref{eq:Selberg_def} is absolutely integrable on $\triangle_N$. We begin with $S_N[F]$ defined as a function $S_N[F]:\Omega_N\to \bbC$.  
It can be checked -- see \S\ref{sec:geometry} -- that, letting 
\begin{align}
\begin{split} 
\alpha_{j,*}(\bmalpha,\bmbeta,\bmgamma) &=\sum_{j_0=1}^j \alpha_{j_0} + 2\sum_{\substack{ j_0,k \in \{1,\ldots,N\}\\ 1\leq j_0 < k \leq j}} \gamma_{j_0,k}, \\ \beta_{j,*}(\bmalpha,\bmbeta,\bmgamma) &= \sum_{j_0=N-j+1}^{N} \beta_{j_0}  + 2 \sum_{\substack{j_0,k\in \{1,\dots,N\}\\ N-j+1 \leq j_0 < k \leq N}} \gamma_{j_0,k} 
\end{split}
\label{eq:ab*}
\end{align}
for each $j\in \{1,\ldots,N\}$, and letting 
\begin{equation}
\gamma_{j,k,*}(\bmalpha,\bmbeta,\bmgamma) = 2 \sum_{\substack{j_0,k_0\in \{1,\ldots,N\} \\ j \leq j_0<k_0 \leq k}} \gamma_{j_0,k_0} 
\label{eq:c*}
\end{equation}
for each pair of $j,k\in \{1,\ldots,N\}$ with $j<k$,
\begin{equation}
\Omega_N = \Big[ \bigcap_{j=1}^N \{\Re \alpha_{j,*} > -j\}  \Big] \cap \Big[ \bigcap_{j=1}^N \{\Re \beta_{j,*} > -j\}  \Big] \cap \Big[ \bigcap_{1\leq j < k \leq N} \{\Re \gamma_{j,k,*} > -(k-j)\}  \Big].
\label{eq:Omega_N_definition}
\end{equation}
So, $\Omega_N$ is nonempty, open, and convex (in particular, connected) and contains all $(\bmalpha,\bmbeta,\bmgamma)\in \smash{\bbC^{2N+N(N-1)/2}}$ such that the real parts of the components of $\bmalpha,\bmbeta,\bmgamma$ are sufficiently large. 

To simplify the formula above, let $\gamma_{0,k,*}=\alpha_{k,*}$ and $\gamma_{N+1-j,N+1,*}=\beta_{j,*}$. Then 
\begin{equation}
\Omega_N = \bigcap_{0\leq j < k \leq N+1} \{(\bmalpha,\bmbeta,\bmgamma)\in \bbC^{2N+N(N-1)/2} : \Re \gamma_{j,k,*} > -(k-j) \}.
\end{equation}

Our first goal is to prove that $S_N[F]$ can be analytically continued to a subset 
\begin{equation} 
\dot{\Omega}_N\subseteq \bbC^{2N+N(N-1)/2}
\end{equation}  
having full measure in $\bbC^{2N+N(N-1)/2}$. 

In order to describe precisely the structure of the singularity at $\bbC^{2N+N(N-1)/2}\backslash \dot{\Omega}_N$, we introduce some terminology.
Let $\mathtt{T}(N)$ denote the collection of maximal families $\mathtt{I}$ of consecutive subsets $\calI\subsetneq \{0,\ldots,N+1\}$ such that 
\begin{itemize}
	\item $2\leq |\calI| \leq N+1$ for all $\calI\in \mathtt{I}$ and 
	\item if $\calI,\calI'\in \mathtt{I}$ satisfy $\calI\cap \calI'\neq \varnothing$, then either $\calI\subseteq \calI'$ or $\calI'\subseteq \calI$. 
\end{itemize}
``$\mathtt{T}$'' stands either for ``tree'' in ``full binary trees'' or ``Tamari'' in \textit{Tamari lattice} \cite{Tamari}\cite{Tamari2}, and the elements of $\mathtt{T}(N)$ can be thought of as specifying the valid ways of adding a maximal number of nonredundant parentheses to a string of $N+2$ identical characters. 
There are $\#\mathtt{T}(N)=C_{N+1}$ such ways, where $C_{N+1}$ is the $(N+1)$st Catalan number.
To each $\calI\in \mathtt{I}$, we associate the facet 
\begin{equation} 
\mathrm{f}_{\calI}=\{(x_1,\ldots,x_N) \in \triangle_N : x_j=x_k\text{ for all }j,k\in \calI \}
\end{equation} 
of $\triangle_N$, where $x_0=0$ and $x_{N+1}=1$. Let $o_\calI\in \bbN$ denote the order of vanishing of $F$ at $\mathrm{f}_{\calI}$. (So, $o_\calI=0$ unless $F$ is vanishing identically at $\mathrm{f}_\calI$.)
\begin{theorem}
	There exist entire functions $S_{N;\mathrm{reg},\mathtt{I}}[F]:\bbC^{2N+N(N-1)/2} \to \bbC$  associated to the $\mathtt{I}\in \mathtt{T}(N)$ such that 
	\begin{equation}
	S_N[F](\bmalpha,\bmbeta,\bmgamma) = \sum_{\mathtt{I}\in \mathtt{T}(N)} S_{N;\mathrm{reg},\mathtt{I}}[F](\bmalpha,\bmbeta,\bmgamma) \prod_{\calI\in \mathtt{I}} \Gamma( o_{\calI}+ |\calI|-1 + \gamma_{\min \calI,\max \calI,*} ) 
	\label{eq:misc_gen}
	\end{equation}
	for all $(\bmalpha,\bmbeta,\bmgamma)\in \Omega_N$.
	\label{thm:generic}
\end{theorem}
Here, $\Gamma: \bbC\backslash \{-n:n\in \bbN\}\to \bbC$ is the gamma function. As a consequence of the theorem, there exists an entire function $S_{N;\mathrm{reg}}[F]:\bbC^{2N+N(N-1)/2}\to \bbC$ such that

\begin{equation}
S_N[F](\bmalpha,\bmbeta,\bmgamma) =  S_{N;\mathrm{reg}}[F](\bmalpha,\bmbeta,\bmgamma) \prod_{0\leq j <k \leq N+1} \Gamma(k-j + \gamma_{j,k,*}) 
\end{equation}
for all $(\bmalpha,\bmbeta,\bmgamma)\in \Omega_N$.
\begin{corollary}
	The function $S_N[F]:\Omega_N\to \bbC$ admits an analytic continuation $\dot{S}_N[F]:\dot{\Omega}_N\to \bbC$
	to the domain
	\begin{multline}
	\dot{\Omega}_N = \bbC^{2N + N(N-1)/2}_{\bmalpha,\bmbeta,\bmgamma} \Big\backslash \Big[ 
	\Big( \bigcup_{j=1}^N \{\alpha_{j,*} \in \bbZ^{\leq -j - \delta_j}\}  \Big) \cup \Big( \bigcup_{j=1}^N \{\beta_{j,*} \in \bbZ^{\leq -j-\atled_j}\}  \Big) \\ \cup \Big( \bigcup_{1\leq j < k \leq N} \{\gamma_{j,k,*} \in \bbZ^{\leq -(k-j) - d_{j,k}} \}  \Big) \Big],
	\label{eq:dotOmegaN}
	\end{multline}
	where  $\bbZ^{\leq n} = \{m\in \bbZ: m\leq n\}$ and $\delta_j =\delta_j[F]= o_{\{0,\ldots,j\}}$, $\atled_j = \atled_j[F] = o_{\{N-j+1,\ldots,N+1\}}$, and $d_{j,k} =d_{j,k}[F]= o_{\{j,\ldots,k\}}$.
	\label{cor:main}
\end{corollary}
The set $\dot{\Omega}_N$ contains all elements of $\bbC^{2N+N(N-1)/2}$ lying outside of a locally finite arrangement of affine hyperplanes.

Consider $F\in \bbC[x_1,\ldots,x_N]$. Letting $[F]_{d_1,\ldots,d_N}$ denote the coefficient of $x_1^{d_1}\cdots x_N^{d_N}$ in $F$,
and letting $\mathsf{refl}:(x_1,\ldots,x_N)\mapsto (1-x_1,\ldots,1-x_N)$, we have 
\begin{align}
\delta_j[F] &= \min\{d_1+\cdots+d_j : [F]_{d_1,\ldots,d_j,d_{j+1},\ldots,d_N}\neq 0 \text{ for some }d_{j+1},\ldots,d_N\in \bbN\},\label{eq:misc_nu1}\\ 
\atled_j[F] &= \min\{d_N+\cdots+d_{N+1-j} : [F\circ \mathsf{refl}]_{d_1,\ldots,d_{N-j},d_{N+1-j},\ldots,d_N}\neq 0 \text{ for some }d_{1},\ldots,d_{N-j}\in \bbN\}. \label{eq:misc_nu2}
\end{align}

\begin{example}
	The simplest case is when $N=1$ and $F=1$ identically, when the integral is given by
	\begin{equation}
	S_1(\alpha,\beta,\gamma) = B(\alpha+1,\beta+1) = \int_0^1 x^\alpha(1-x)^\beta \dd x = \frac{\Gamma(1+\alpha)\Gamma(1+\beta)}{\Gamma(2+\alpha+\beta)},
	\end{equation}
	defined initially for $\Re\alpha,\Re\beta> -1$ via the definite integral and then extended meromorphically via the formula on the right-hand side above (or via another method). This is Euler's \emph{$\beta$-function}.
	One method of meromorphic continuation involves the Pochhammer contour (a.k.a.\ Pochhammer double loop)
	\begin{equation} 
	b^{-1}a^{-1} ba \in \pi_1(\bbC\backslash \{0,1\}),
	\end{equation} 
	where $a,b$ are the generators of $\pi_1(\bbC\backslash \{0,1\})$ corresponding to one (say, counterclockwise) circuit around each of $0,1$ respectively. 
	
	\begin{figure}[h]
		\begin{tikzpicture}[scale=3,  decoration={
			markings,
			mark=at position 0.5 with {\arrow[scale=1.5,>=latex]{>}}}]
		\coordinate (0) at (1.2,.2);
		\coordinate (1) at (1.2,-.35);
		\coordinate (2) at (-.2,-.35);
		\coordinate (3) at (-.2,+.2);
		\coordinate (4) at (.35,+.2);
		\coordinate (5) at (.35,-.25);
		\coordinate (6) at (1.1,-.25);
		\coordinate (7) at (1.1,.1);
		\coordinate (8) at (-.1,.1);
		\coordinate (9) at (-.1,-.2);
		\coordinate (10) at (.65,-.2);
		\coordinate (11) at (.65,.2);
		\draw[postaction={decorate}] (0) -- (1);
		\draw[postaction={decorate}] (1) -- (2);
		\draw[postaction={decorate}] (2) -- (3);
		\draw[postaction={decorate}] (3) -- (4);
		\draw[postaction={decorate}] (4) -- (5);
		\draw[postaction={decorate}] (5) -- (6);
		\draw[postaction={decorate}] (6) -- (7);
		\draw[postaction={decorate}] (7) -- (8);
		\draw[postaction={decorate}] (8) -- (9);
		\draw[postaction={decorate}] (9) -- (10);
		\draw[postaction={decorate}] (10) -- (11);
		\draw[postaction={decorate}] (11) -- (0);
		\filldraw[color=black] (0,0) circle (.5pt) node[below] {$0$};
		\filldraw[color=black] (1,0) circle (.5pt) node[below] {$1$};
		\end{tikzpicture}
		\caption{The Pochhammer contour in $\bbC\backslash \{0,1\}$, up to homotopy.}
	\end{figure}
	
	Then, $b^{-1}a^{-1} ba$ can be lifted to a closed contour $p$ in the cover $\calM$ of $\bbC\backslash \{0,1\}$ corresponding to the commutator subgroup of $\pi_1(\bbC\backslash \{0,1\})$. Then, choosing the basepoint of $p$ appropriately, 
	\begin{equation}
	B(\alpha+1,\beta+1) = \frac{1}{1-e^{-2\pi i \alpha}} \frac{1}{1-e^{-2\pi i \beta} } \int_{p} x^{\alpha} (1-x)^{\beta}  \dd x, 
	\end{equation}
	where we are now considering $x^\alpha(1-x)^\beta$ as an analytic function on $\calM$. 
	The theorem above tells us that there exist entire $S_{1;\mathrm{reg}, (\bullet\bullet)\bullet}$, $S_{1;\mathrm{reg}, \bullet(\bullet\bullet)}$ such that 
	\begin{equation}
	B(\alpha+1,\beta+1) = \Gamma(1+\alpha)S_{1;\mathrm{reg}, (\bullet\bullet)\bullet}(\alpha,\beta) + \Gamma(1+\beta)S_{1;\mathrm{reg}, \bullet(\bullet\bullet)}(\alpha,\beta).
	\end{equation}
	This splitting is not so obvious from the formula $B(\alpha+1,\beta+1)=\Gamma(1+\alpha)\Gamma(1+\beta)/\Gamma(2+\alpha+\beta)$. 
\end{example}
\begin{example}
	Now consider the case when $N=2$ and $F=1$.
	It can be computed that the Selberg-like integral is then 
	\begin{equation}
	S_2(\bmalpha,\bmbeta,\bmgamma) = \frac{\Gamma(1+\alpha_1)\Gamma(1+\beta_2)  \Gamma(2+2\gamma_{1,2} + \alpha_1+\alpha_2)\Gamma(1+2\gamma_{1,2})}{\Gamma(2+\alpha_1+2\gamma_{1,2})\Gamma(3+\alpha_1+\alpha_2+\beta_2+2\gamma_{1,2})} {}_3F_2(a,b;1),
	\label{eq:misc_vv1}
	\end{equation}
	where $a=(a_1,a_2,a_3)=(1+\alpha_1,-\beta_1,2+2\gamma_{1,2}+\alpha_1+\alpha_2)$ and $b=(b_1,b_2) = (2+\alpha_1+2\gamma_{1,2},3+\alpha_1+\alpha_2+\beta_2 + 2\gamma_{1,2})$, where ${}_p F_q$ denotes the generalized hypergeometric function. 
	For $N=2$, the theorem above reads 
	\begin{multline}
	S_2(\bmalpha,\bmbeta,\bmgamma) = \Gamma(1+\alpha_1)\Gamma(2+\alpha_1+\alpha_2+2\gamma_{1,2}) S_{2;\mathrm{reg},((\bullet\bullet)\bullet)\bullet}(\bmalpha,\bmbeta,\bmgamma) +\\  \Gamma(1+\alpha_1)\Gamma(1+\beta_2)S_{2;\mathrm{reg},(\bullet\bullet)(\bullet\bullet)}(\bmalpha,\bmbeta,\bmgamma)  + \Gamma(1+\beta_2)\Gamma(2+\beta_1+\beta_2+2\gamma_{1,2}) S_{2;\mathrm{reg},\bullet(\bullet(\bullet\bullet))}(\bmalpha,\bmbeta,\bmgamma) \\ + 
	\Gamma(1+2\gamma_{1,2}) \Gamma(2+\beta_1+\beta_2+2\gamma_{1,2}) S_{2;\mathrm{reg},\bullet((\bullet\bullet)\bullet)}(\bmalpha,\bmbeta,\bmgamma) \\ 
	+ 
	\Gamma(1+2\gamma_{1,2}) \Gamma(2+\alpha_1+\alpha_2+2\gamma_{1,2}) S_{2;\mathrm{reg},(\bullet(\bullet\bullet))\bullet}(\bmalpha,\bmbeta,\bmgamma),
	\end{multline}
	but once again this splitting is not so obvious from the exact formula \cref{eq:misc_vv1}.
	This example is explored more in the appendix. 
\end{example}

The proof below is lower-brow than the twisted homological constructions of \cite[\S5]{KT2}\cite{KT1}, Aomoto \cite{Ao}, and others  \cite{Varchenko2}\cite{Warnaar2}, as it is based on the method described in \cite[Chp.\ 10]{Varchenko}. This involves the geometric analysis of the singularities of the Selberg(-like) integrand. The key observation is that if the $N$-simplex is blown up to the $N$-dimensional \emph{associahedron} \cite{Stasheff}\cite[\S1.6]{Operads}\cite{Postnikov} (see \Cref{fig:geometricinterpretation}, \Cref{fig:K030}), then the Selberg integrand -- which is \emph{not} polyhomogeneous on $\triangle_N$ -- becomes one-step polyhomogeneous (a.k.a ``classical'') on the resolution.
See \S\ref{sec:geometry} for details.
This observation appears, in an essentially equivalent form (albeit with different terminology), already in \cite{KT2}\cite{KT1}\cite{Mimachi}, though the term ``associahedron'' does not appear there. Closely related observations have appeared in the physics literature \cite{Amplitudes17}\cite{Amplitudes18}\cite{Amplitudes19}\cite{Amplitudes20}.

The application of polyhomogeneity to the proof of the theorem above is given in \S\ref{sec:IR}. The classicality of the lift of the Selberg integrand on the associahedron allows us to reduce the problem to what is essentially a product of one-dimensional cases. 
The faces of the associahedron are in bijective correspondence with the quantities defined in \cref{eq:ab*}, \cref{eq:c*}. 
The correspondence is depicted in \Cref{fig:geometricinterpretation} in the case $N=3$.
\begin{figure}[h]
	\begin{center}
		\tdplotsetmaincoords{70}{115} 
		\begin{tikzpicture}[scale=2.5,tdplot_main_coords]
		\draw[opacity=0] (0,1,0) -- (1,-1.1,0); 
		\draw[dashed, color=mygreen, *-] (.05,.05,.05)  -- (-.5,-.1,-.1) -- (-.5,.8,-.1) node[right] {$\alpha_1+\alpha_2+\alpha_3 + 2\gamma_{1,2}+2\gamma_{1,3}+2\gamma_{2,3}$};
		\draw[dashed, color=darkgray, *-] (0.35,.35,0)  -- (0.35,.35,-.2) -- (0.75,.75,-.2) node[right] {$2\gamma_{2,3}$};
		\draw[dashed, color=darkgray, *-] (.1,.1,.5) -- (-.2,-.2,.5)  -- (-.2,-.2,.9) node[left] {$\alpha_1+\alpha_2+2\gamma_{1,2}$};
		\draw[dashed, color=darkgray, *-] (.55,.1,.1)  -- (.55,0,0) -- (.55,0,-.5) node[below] {$2\gamma_{1,2}+2\gamma_{1,3}+2\gamma_{2,3}$};
		\draw[dashed, color=darkgray, *-] (0,.35,.35) -- (-.2,.45,.35)  -- (-.2,0.45,.75) node[right] {$\alpha_1$};
		\draw[dashed, color=darkgray, *-] (.35,0,.35)  -- (.35,-.2,.35) -- (.35,-.2,.8) node[left] {$2\gamma_{1,2}$};
		\draw[dashed, color=darkgray, *-] (.9,.1,.1)  -- (.9,-.1,-.1) -- (.9,-.1,.2)  node[left] {$ \beta_1+\beta_2+\beta_3+2\gamma_{1,2}+2\gamma_{1,3}+2\gamma_{2,3}$};
		\draw[fill=mygreen,fill opacity = .3, dashed] (0,1/4,0) -- (0,1/8,1/4) -- (1/8,0,1/4) -- (1/4,0,1/8) -- (1/4,1/8,0) -- cycle;
		\draw[fill=gray,fill opacity = .1, dashed] (3/4,1/8,0) -- (3/4,0,1/8) -- (15/16,0,1/4) -- (1,0,1/8) -- (1,1/4,0) --	cycle;
		\draw[fill=gray,fill opacity = .1, draw=none] (3/4,1/8,0) -- (3/4,0,1/8) -- (1/4,0,1/8) -- (1/4,1/8,0) -- cycle;
		\draw[fill=gray,fill opacity = .1, draw=none] (0,1/8,1/4) -- (1/8,0,1/4) -- (1/8,0,7/8) -- (0,1/8,7/8) -- cycle;
		\draw[fill=gray,fill opacity = .1, dashed] (0,1/4,0) -- (0,1/8,1/4) -- (0,1/8,7/8) -- (-.15,7/8-.1,3/16) -- (0,3/4,0) -- cycle;
		\draw[fill=gray,fill opacity = .1] (1,0,1/8) -- (1,1/4,0) -- (0,3/4,0) -- (-.15,7/8-.1,3/16) -- cycle;
		\draw[fill=gray,fill opacity = .1, draw=none] (1,1/4,0) -- (0,3/4,0) -- (0,1/4,0) -- (1/4,1/8,0) -- (3/4,1/8,0) -- cycle;
		\draw[dashed]  (1/4,1/8,0) -- (3/4,1/8,0);
		\draw[fill=gray,fill opacity = .1, dashed] (1/4,0,1/8) -- (3/4,0,1/8) -- (15/16,0,1/4) -- (1/8,0,7/8) -- (1/8,0,1/4) -- cycle;
		\draw[fill=gray,fill opacity = .1] (0,1/8,7/8) -- (1/8,0,7/8) -- (15/16,0,1/4) -- (1,0,1/8) -- (-.15,7/8-.1,3/16) -- cycle;
		\draw[dashed, color=gray, *-] (.1,.2,.15) --   (.15,.8,.5) node[right] {$\beta_3$};
		\draw[dashed, color=gray, *-] (.5,.5,.15) --   (.85,1,.1) node[right] {$\beta_2+\beta_3+2\gamma_{2,3}$};
		\end{tikzpicture}
	\end{center} 
	\caption{The 3-dimensional associahedron, with its faces labeled by the associated functions in \cref{eq:ab*}, \cref{eq:c*}. The $C_4=14$ vertices are in correspondence with the 14 elements $\mathtt{T}(3)$.}
	\label{fig:geometricinterpretation}
\end{figure}
The quantities $\alpha_{j,*}, \beta_{j,*},\gamma_{j,k,*}$ are the orders of the Selberg integrand at the corresponding faces. 
Each $\mathtt{I}\in \mathtt{T}(N)$ is associated with a minimal facet of the associahedron, and the $\calI\in \mathtt{I}$ are associated with the faces containing that facet. 
Thus, we have a geometric interpretation of each of the terms in \cref{eq:misc_gen}.

The theorem cannot be sharpened while maintaining generality. Indeed, the proof of the theorem shows that if $F>0$ everywhere in $\triangle_N$ (including the boundary), then  
\begin{equation}
S_{N;\mathrm{reg}}[F](\bmalpha,\bmbeta,\bmgamma) \neq 0 
\end{equation}
for any $(\bmalpha,\bmbeta,\bmgamma)\in \bbR^{2N+N(N-1)/2}$ for which both of 
\begin{itemize}
	\item $\gamma_{j,k,*}\in \bbZ^{\leq -(k-j)}$ for precisely one pair of $j,k\in \{0,\ldots,N+1\}$ with $j<k$, 
	\item $\gamma_{j,k,*} > -(k-j)$ for all other $j,k$ 
\end{itemize}
hold, as for such $(\bmalpha,\bmbeta,\bmgamma)$ the quantity $S_{N;\mathrm{reg}}[F](\bmalpha,\bmbeta,\bmgamma)$ is proportional to a convergent integral of a positive integrand over the corresponding face of the associahedron. Consequently, $S_N[F]:\Omega_N\to \bbC$ cannot be analytically continued to the complement of any strictly smaller collection of hyperplanes than that in \cref{eq:dotOmegaN}. 

However, for the desired application, we do not need full generality. 
Of special importance is the case when $\bmalpha,\bmbeta,\bmgamma$ are each ``constant,'' meaning that, for some $\alpha,\beta,\gamma \in \bbC$, 
\begin{itemize}
	\item $\alpha_i=\alpha$ and $\beta_i=\beta$ for all $i\in \{1,\ldots,N\}$, and
	\item $\gamma_{j,k}=\gamma$ for all $j,k\in \{1,\ldots,N\}$ with $j<k$. 
\end{itemize}
In this case, we simply write
\begin{equation} 
S_N[F](\alpha,\beta,\gamma) = \int_{\triangle_N}  F(x_1,\ldots,x_N) \prod_{j=1}^N x_j^{\alpha}(1-x_j)^{\beta}   \prod_{1\leq j < k \leq N} (x_k-x_j)^{2\gamma}  \dd x_1\cdots \dd x_N. 
\label{eq:Selberg_def_symmetric}
\end{equation} 
We now consider $F \in \bbC[x_1,\ldots,x_N]^{\frakS_N}$, i.e.\ \textit{symmetric} polynomial $F$.
This case includes, of course, Selberg's original example, in which $F=1$, as well as the 3-point coefficients of the $(1,s)$- and $(r,1)$-primary fields and their descendants in the BPZ minimal models. It also includes certain Selberg-like integrals considered by Aomoto \cite{Ao}, Kadell \cite{Kadell, Kadell2}, and others \cite{AGT}.
The computation of such integrals is listed as an open problem in \cite{KT2}.

Below, we will introduce a more general notion of ``DF-symmetric'' Selberg-like integrals, this including the other 3-point coefficients. For the purposes of an introductory discussion we focus on the -- already interesting -- symmetric case. 

The integral \cref{eq:Selberg_def_symmetric} is defined initially on the subset $U_N[F]\subset \bbC^3_{\alpha,\beta,\gamma}$ given by 
\begin{multline}
U_N[F] = \Big\{(\alpha,\beta,\gamma)\in \bbC^3: \Re j(\alpha+(j-1)\gamma) >-1-\delta_j[F]\text{ and }  \\ \Re j(\beta+(j-1)\gamma) >-1-\atled_j[F]\text{ for all $j\in \{1,\ldots,N\}$, and } \Re \gamma > - \frac{1}{N-1} \Big\},
\end{multline}
which contains 
\begin{equation}
U_N = U_N[1] =  \Big\{(\alpha,\beta,\gamma)\in \bbC^3: \min\{\Re \alpha,\Re \beta\} + \min\{0,(N-1)\gamma\}>-1\text{ and }\Re\gamma> - \frac{1}{N-1} \Big\}  .
\end{equation}
An immediate corollary of the theorem above is that the function $S_N[F]:U_N[F]\to \bbC$ defined by \cref{eq:Selberg_def_symmetric} admits an analytic continuation $\dot{S}_N[F](\alpha,\beta,\gamma) : \dot{U}_N[F]\to \bbC$ to the domain $\dot{U}_N[F]\supsetneq U_N[F]$ given by 
\begin{multline}
\dot{U}_N[F] = \bbC^3_{\alpha,\beta,\gamma} \Big\backslash \Big[ 
\Big( \bigcup_{j=1}^N \{j(\alpha+(j-1)\gamma) \in \bbZ^{\leq -j-\delta_j}\}  \Big) \cup \Big( \bigcup_{j=1}^N \{ j(\beta+(j-1)\gamma)\in \bbZ^{\leq -j - \atled_j}\}  \Big) \\ \cup \Big( \bigcup_{j=1}^{N-1} \{ j(j+1) \gamma \in \bbZ^{-j} \}  \Big) \Big].
\label{eq:dotUN}
\end{multline}

\begin{example}
	Consider $F=1$, i.e.\ the Selberg integral. In this case, Selberg proved in \cite{Selberg} that $S_N(\alpha,\beta,\gamma) = S_N[1](\alpha,\beta,\gamma)$ is given by 
	\begin{equation}
	S_N(\alpha,\beta,\gamma) = \frac{1}{N!}\prod_{j=1}^{N} \frac{\Gamma(1+\alpha+(j-1)\gamma) \Gamma(1+\beta+(j-1)\gamma) \Gamma(1+j\gamma)}{\Gamma(2+\alpha+\beta+(N+j-2)\gamma ) \Gamma(1+\gamma)} .
	\label{eq:Selberg_formula}
	\end{equation}
	See \cite{FW} for a review of the history of this result. 
\end{example}

The example of the Selberg integral suggests that, in the symmetric case, \cref{eq:dotUN} is not the maximal domain of analyticity. Set 
\begin{equation} 
\deg_j[F] = \max\{d_1+\cdots+d_j : [F]_{d_1,\ldots,d_j,d_{j+1},\ldots,d_N} \neq 0\text{ for some }d_{j+1},\ldots,d_N\in \bbN\}. \label{eq:misc_nu3}
\end{equation} 
(Since $F$ is symmetric, $\deg_j[F]=\deg_j[F\circ \mathsf{refl}]$.)
Then:
\begin{theorem}
	For any $F \in \bbC[x_1,\ldots,x_N]^{\frakS_N}$, there exists an entire function $S_{N;\mathrm{Reg}}[F]:\bbC^3_{\alpha,\beta,\gamma}\to \bbC$ such that
	\begin{multline}
	S_N[F](\alpha,\beta,\gamma) = \Big[\prod_{j=1}^N \frac{\Gamma(1+\bar{\delta}_j +\alpha+(j-1)\gamma) \Gamma ( 1+\bar{\atled}_j  +\beta+(j-1)\gamma)\Gamma(1+j\gamma)}{\Gamma(2+\bar{d}_j+\alpha+\beta+(N+j-2)\gamma)\Gamma(1+\gamma)} \Big] \\ \times S_{N;\mathrm{Reg}}[F](\alpha,\beta,\gamma) \label{eq:SNF_goal_symmetric}
	\end{multline}
	for all $(\alpha,\beta,\gamma) \in U_N$, where $\bar{\delta}_j = \lceil j^{-1} \delta_j[F]\rceil$, $\bar{\atled}_j = \lceil j^{-1} \atled_j[F] \rceil$, and $\bar{d}_j =  \lfloor (N-j+1)^{-1} \deg_j [F] \rfloor$ for each $j\in \{1,\ldots,N\}$. 
	\label{thm:symmetric}
\end{theorem}
Thus, 
$S_N[F](\alpha,\beta,\gamma)$ admits an analytic continuation $\mathring{S}_N[F](\alpha,\beta,\gamma):\mathring{U}_N[F]\to \bbC$ to the domain $\mathring{U}_N[F]\supsetneq \dot{U}_N[F]$ defined by 
\begin{align}
\begin{split} 
\mathring{U}_N[F] = \bbC^3_{\alpha,\beta,\gamma} \Big\backslash \Big[ 
&\Big( \bigcup_{j=1}^N \{\alpha+\bar{\delta}_j+(j-1)\gamma \in \bbZ^{\leq -1}\}  \Big) \cup \\ &\Big( \bigcup_{j=1}^N \{ \beta + \bar{\atled}_j +(j-1)\gamma\in \bbZ^{\leq -1}\}  \Big) \cup \Big( \bigcup_{j=1}^{N-1} \{ (j+1) \gamma \in \bbZ^{\leq -1}, \gamma \notin \bbZ \}  \Big) \Big]. 
\end{split} 
\label{eq:ringUN}
\end{align}
Observe that \cref{eq:ringUN} allows $\gamma =-1$.

In the case of the original Selberg integral, \Cref{thm:symmetric} describes precisely the singularities and zeroes of the meromorphic continuation of the original integral, and $S_{N;\mathrm{Reg}} = S_{N;\mathrm{Reg}}[1]$ is just constant.  
The functions $S_{2}[F]$ and $S_{2;\mathrm{Reg}}[F]$ are explored in \S\ref{sec:example}.

The proof of the theorem above consists of several steps: 
\begin{enumerate}
	\item The first step is the removal of the fictitious singularities of $\dot{S}_N[F](\alpha,\beta,\gamma)$ only in $\gamma$ (as required e.g.\ in the Coulomb gas formalism with both kinds of screening charges). 
	
	The basic idea is to employ the relation -- which can be found in a heuristic form in \cite[Ap.\ A]{DF2} -- between the symmetrization of $S_N[F](\bmalpha,\bmbeta,\bmgamma)$ and the ``DF-like'' integral 
	\begin{multline} 
	\qquad\qquad I_N[F](\bmalpha,\bmbeta,\bmgamma) = \int_{\square_N} \Big[ \prod_{j=1}^N x_j^{\alpha_j}(1-x_j)^{\beta_j}  \Big] \\ \times\Big[ \prod_{1\leq j<k \leq N} (x_k-x_j+i0)^{2\gamma_{j,k}} \Big] F \dd x_1\cdots \dd x_N,  
	\label{eq:misc_iff}
	\end{multline} 
	where $\square_N = [0,1]^N$. We can analytically continue $I_N[F]$ via an argument similar to that used to prove \Cref{thm:generic}. Unlike that of $S_N[F](\bmalpha,\bmbeta,\bmgamma)$, this extension has no singularities associated with hyperplanes of constant $\gamma$. The true singularities of the extension of $S_N[F](\alpha,\beta,\gamma)$ associated with hyperplanes of constant $\gamma$ show up in the relation with the extension of $I_N[F](\alpha,\beta,\gamma)$.
	\item The second step removes the other unwanted singularities away from the loci of two or more unwanted singularities, via some identities proven via Aomoto \cite{Ao} in the $F=1$ case (and \cite[Ap.\ A]{DF2}, at a physicist's level of rigor). 
	The use of these identities for \emph{computing} the original Selberg integral is sketched in \cite{FW}. It seems there cannot be a similar computation of $S_N[F]$ in the $\deg F>1$ case, so a statement about the singularities is the best we can do.
	
	The simplex $\triangle_N\subset \bbR^N$ can be thought of as a subset of 
	\begin{equation} 
	(\bbC\backslash \{0,1\})^N=(\bbC P^1\backslash \{0,1,\infty\})^N 
	\end{equation}
	via the embedding $\bbR \hookrightarrow \bbC\hookrightarrow \bbC P^1$, and the rough idea of this step of the proof is to relate the integrals above to the result of replacing $\triangle_N$ with $L^{\boxtimes N}\triangle_N$ for $L$ one of the six linear  fractional transformations preserving $\bbC P^1\backslash \{0,1,\infty\}$. Only three of these are essentially different, and one of these three is just the identity and therefore uninteresting. The other two integrals each have meromorphic extensions with different manifest singularities. Using \Cref{prop:Aomoto}, these functions can be related to each other, and this can be used to remove most of the apparent singularities that are not present in all three functions.
	Some singularities are present in the relations between the integrals, and these cannot be removed.
	
	Once this has been done, the final step is the application of Hartog's theorem to remove the remaining removable singularities, which now lie on a codimension two subvariety of $\bbC^3$. 
\end{enumerate}
This argument is carried out in \S\ref{subsec:removing_singularities_1}. The version more relevant to \cite{DF2} (with the additional steps needed) is in \S\ref{subsec:removing_singularities_2}.  

We call $I_N[F]$ a ``DF-like'' integral because similar integrals appear, albeit at a somewhat formal level, in \cite{DF2}. A similar construction appears in \cite{Felder}.

Let $\Sigma\mathtt{T}(N)$ denote the collection of maximal collections $\mathtt{I}$ of pairs $(x_0,S)$ of $x_0\in \{0,1\}$ and nonempty subsets $S \subseteq \{1,\ldots,N\}$ such that, given $(x_0,S),(x_0,Q)\in \mathtt{I}$, either $S\subseteq Q$ or $Q\subseteq S$. 
\begin{theorem}
	There exist entire functions $I_{N;\mathrm{reg},\mathtt{I}}[F]:\bbC^{2N+N(N-1)/2}_{\bmalpha,\bmbeta,\bmgamma}\to \bbC$ associated to the $\mathtt{I}\in \Sigma \mathtt{T}(N)$ 
	such that 
	\begin{multline}
	I_N[F](\bmalpha,\bmbeta,\bmgamma) = \sum_{\mathtt{I}\in \Sigma\mathtt{T}(N)} \Big[ \prod_{(1,S) \in \mathtt{I}} \Gamma\Big(|S|+ \sum_{j\in S} \beta_j + 2 \sum_{j,k \in S, j<k} \gamma_{j,k}\Big) \Big] \\ \times  \Big[ \prod_{(0,S) \in \mathtt{I}} \Gamma\Big(|S|+ \sum_{j\in S} \alpha_j + 2 \sum_{j,k \in S, j<k} \gamma_{j,k}\Big) \Big]I_{N;\mathrm{reg},\mathtt{I}}[F](\bmalpha,\bmbeta,\bmgamma) 
	\end{multline}
	for all $(\bmalpha,\bmbeta,\bmgamma)$ for which the left-hand side is a well-defined  integral. 
	\label{thm:Imain}
\end{theorem}
In particular, $I_N[F](\bmalpha,\bmbeta,\bmgamma)$ admits an analytic extension $\dot{I}_N[F](\bmalpha,\bmbeta,\bmgamma)$ to an open, dense set 
\begin{equation}
\dot{V}_N = \bbC_{\bmalpha,\bmbeta,\bmgamma}^{2N+N(N-1)/2} \Big\backslash \Big[ \Big( \bigcup_{S\subseteq \{1,\ldots,N\}} \{\alpha_{S,*} \in \bbZ^{\leq - |S|} \} \Big) \cup \Big( \bigcup_{S\subseteq \{1,\ldots,N\}} \{\beta_{S,*} \in \bbZ^{\leq - |S|} \} \Big) \Big].
\label{eq:dotVn}
\end{equation}

\subsection{Some comments on the Coulomb gas formalism}
Here, we discuss a particular application to the Coulomb gas formalism (a.k.a.\ ``free field realization,'' ``Feigin--Fuchs representation,'' etcetera) of 2D CFT \cite{DF1}\cite{DF2}\cite{DF3}\cite{FS1}\cite[Chp.\ 9]{CFT}\cite{FW}. This approach of Dotsenko--Fateev to the construction of the ``minimal models'' of Belavin--Polyakov--Zamolodchikov (BPZ) \cite{BPZ} has been the subject of substantial interest, but it appears that it has not yet been placed on entirely rigorous mathematical footing. The construction in \cite{FS1}\cite{FS2} of the 3-point coefficients of the $(1,s)$- and $(r,1)$-primary fields and their descendants in the minimal models is satisfactorily rigorous, but it has remained an open problem to handle the rest of the primary fields at a similarly satisfactory degree of rigor. 
From our perspective, the issue is an insufficient treatment of the meromorphic continuation of Selberg-like integrals, which are instead treated somewhat formally in the original works. 

The issue is that Dotsenko \& Fateev (DF) take some of the $\gamma$'s to be $-1$ --- see e.g. \cite[Appendix A]{DF2}\cite[p. 27]{FS2}\cite[\S2]{FW} --- and then the integrand above is, say for $F=1$, no longer integrable over the integral's domain.
As a consequence, the integrals in \cite[Appendix A]{DF2} are formal. 
Dotsenko \& Fateev suggest making sense of them via meromorphic continuation in the exponents of the integrand, but they do not prove that a suitable meromorphic continuation exists, nor do they discuss the singularities of the extension in sufficient detail to justify their manipulations. 
Here, we have constructed a suitable extension and analyzed its singularities in detail. 

The reason why it is necessary to take some of the $\gamma$'s to be $-1$ is that, for fixed central charge, there are two sorts of vertex operators $V_{\alpha_\pm}$ used in screening operators. Both are necessary to produce all solutions of the BPZ equations. 
The relevant vertex operators are those of conformal weights $h_\pm =  1$. 
If the central charge is $c$, the two screening charges have conformal weights given in terms of $\alpha_\pm$ by
\begin{equation}
h_\pm = \alpha_\pm^2 - 2 \alpha_\pm \alpha_0, 
\end{equation}
according to the conventions in \cite[\S9.2.1]{CFT}, where $c=1-24\alpha_0^2$,
so, by Vieta, $\alpha_-\alpha_+=-1$. The correlation functions involving these screening charges are Dotsenko--Fateev integrals with $\gamma_{j,k} = \alpha_-\alpha_+ = -1$, as follows from the commutation properties of vertex operators. See \cite[\S9]{CFT} for further exposition.

A construction of Kanie--Tsuchiya \cite{KT2}\cite{KT1}, rediscovered by Mimachi--Yoshida \cite{Mimachi2,Mimachi}\cite{Masaaki}, yields the existence of some meromorphic continuation defined for almost all values of the exponents. This extension is not quite sufficient for our purposes: it has removable singularities that, while removable, are nontrivial to actually \textit{prove} removable.  
In particular, the Kanie--Tsuchiya construction has an apparent isolated singularity at $\gamma = -1$ (see \cite[\S5, above Thm. 5]{KT1}), along with at a few other problematic affine hyperplanes in the space of possible parameters. One of the advantageous features of the meromorphic continuation here is that it lacks these problematic apparent singularities and therefore applies to the cases considered in the physics literature.

Most of the rigorous work on the analysis of integrals of Dotsenko--Fateev type --- see e.g.\  \cite{Flores1}\cite{Flores2}\cite{Flores3}\cite{Flores4}\cite{Lenells} for some recent work --- focuses on screened multipoint correlation functions with at most one screening charge screening per insertion point. Such integrals are related to the $N=1$ case of $S_N(\bmalpha,\bmbeta,\bmgamma)$. Not much has been done about the $N>1$ case. 
Moreover, while a fair amount of work has gone into the study of general hypergeometric integrals associated to hyperplane arrangements --- the literature on this topic is large, so we just cite \cite{Varchenko}\cite{hypergeometric} --- it does not seem possible to deduce the specific, concrete results below from results in the current literature.

Note that $\dot{\Omega}_N=\dot{\Omega}_N[1]$, as defined in \cref{eq:dotOmegaN}, does \emph{not} contain $(\bmalpha,\bmbeta,\bmgamma)$ with $\gamma_{j,k} = -1$ for $|j-k|=1$, so \Cref{thm:generic} is insufficient for the construction of the BPZ minimal models. This is one of the motivations for proving the sharper theorems above.

The $\Gamma(2+\bar{d}_j+\alpha+\beta+(N+j-2)\gamma)$ term in the denominator of \cref{eq:SNF_goal_symmetric} implies that $\mathring{S}_N[F](\alpha,\beta,\gamma)=0$ for all 
\begin{equation}
(\alpha,\beta,\gamma) \in \mathring{U}_N[F] \cap  \{\alpha+\beta+(N+j-2)\gamma \in \bbZ^{\leq -2-\bar{d}_j}\text{ for some }j\in \{1,\ldots,N\}\}.
\label{eq:misc_kkk}
\end{equation}
When constructing the 3-point coefficients of the BPZ minimal models, this is one mechanism preventing the fusion of $(0,s)$-primary fields (which are not included in the model) with the primary fields that are included. In BPZ's terminology, this is the \emph{truncation} of the operator algebras, as originally argued for on the basis of the constraint of OPE associativity --- see \cite[\S6]{BPZ}\cite[Chp.\ 7.3.2]{CFT}.

For the full application to \cite{DF1, DF2}, we use the following notion of ``DF-symmetric'' polynomials.
Given $\lambda \in \bbC$ and $\mathtt{S}\subseteq \{1,\ldots,N\}$, let $\operatorname{DFSym}(N,\mathtt{S},\lambda)$ denote the set of $F\in \bbC[x_1,x_1^{-1},\ldots,x_N,x_N^{-1}]$ such that:
\begin{itemize}
	\item given any $\sigma \in \frakS_N$ such that $\sigma:\mathtt{S}\to \mathtt{S}$, i.e.\ in the Young subgroup associated to $\mathtt{S}$,  
	\begin{equation} 
	F=F\circ \sigma
	\end{equation} 
	where we are identifying $\sigma$ with the map $\bbC^N \ni \{x_i\}_{i=1}^N \to \{x_{\sigma(i)}\}_{i=1}^N \in \bbC^N$, and 
	\item for any $j\in \mathtt{S}$ and $k\in \{1,\ldots,N\}\backslash \mathtt{S}$, 
	\begin{equation}
	\lambda\Big(\frac{\partial}{\partial x_j} F\Big) \Big|_{x_j=x_k} =  \frac{\partial}{\partial x_k} \Big(F \Big|_{x_j=x_k} \Big) \in \bbC[x_1,x_1^{-1}\ldots,\hat{x}_j,\hat{x}_j^{-1}\ldots,x_N].
	\end{equation}
\end{itemize}
In particular, $\operatorname{DFSym}(N,\{1,\ldots,N\},\lambda) = \bbC[x_1,x_1^{-1},\ldots,x_N,x_N^{-1}]^{\frakS_N}$, so in this sense DF-symmetry is a generalization of ordinary symmetry. 
Our disallowal of Laurent polynomials $F$ in the symmetric case was without loss of generality, as, were $F$ Laurent, we could shuffle factors of $x_1\cdots x_N$ between the polynomial and the rest of the Selberg integrand. However, it is useful here to allow general Laurent polynomials.

For each $\lambda$ and $\mathtt{S}$, $\operatorname{DFSym}(N,\mathtt{S},\lambda)$ is a (unital) $\bbC$-subalgebra of $\bbC[x_1,x_1^{-1}\ldots,x_N,x_N^{-1}]$. It is nontrivial. If $\mathtt{S}$ is a proper subset of $\{1,\ldots,N\}$, then
\begin{align}
\begin{split} 
\lambda_- \sum_{j\in \mathtt{S} } x_j  + \lambda_+ \sum_{k\in \{1,\ldots,N\}\backslash \mathtt{S}} x_k &\in \operatorname{DFSym}(N,\mathtt{S},\lambda),  \\
\lambda_- \sum_{j\in \mathtt{S} } \frac{1}{x_j}  + \lambda_+ \sum_{k\in \{1,\ldots,N\}\backslash \mathtt{S}} \frac{1}{x_k} &\in \operatorname{DFSym}(N,\mathtt{S},\lambda) 
\end{split}
\end{align}
is a nonzero member for $\lambda_+$ defined by $\lambda_-^{-1} (\lambda_- + \lambda_+) = \lambda$, so $\operatorname{DFSym}(N,\mathtt{S},\lambda) $ contains polynomials of all degrees. 

The key method of constructing DF-symmetric Laurent polynomials is the following:
\begin{example}
	For any $M\in \bbN^+$ and matrix-valued polynomials $\varphi,\psi \in x\bbC^{M\times M}[x]$ such that the coefficients of $\varphi$ are strictly upper-triangular, the coefficients of $\psi$ are strictly lower-triangular.
	Suppose that the $A$'s all commute with each other and that the $B$'s all commute with each other. (We do not assume that the $A$'s commute with the $B$'s.)
	Then, the matrix elements of 
	\begin{equation}
	\exp \Big( 	\lambda_- \sum_{j\in \mathtt{S} } \psi(x_j^{-1})  + \lambda_+ \sum_{k\in \{1,\ldots,N\}\backslash \mathtt{S}} \psi(x_k^{-1}) \Big) \exp \Big( 	\lambda_- \sum_{j\in \mathtt{S} } \varphi(x_j)  + \lambda_+ \sum_{k\in \{1,\ldots,N\}\backslash \mathtt{S}} \varphi(x_k) \Big)  
	\label{eq:misc_b51}
	\end{equation}
	lie in $\operatorname{DFSym}(N,\mathtt{S},\lambda)$, where $\lambda=\lambda_-^{-1} (\lambda_- + \lambda_+)$.
	The vertex operators which Dotsenko and Fateev integrate to define the minimal model 3-point coefficients have this form up to some scalar factors which are part of the Selberg integrand. In this example, the coefficients of $\varphi$ are annihilation operations on some Fock space, and the coefficients of $\psi$ are creation operators, with all operators truncated to some finite dimensional subspace of the Fock space. That the creation operators in \cref{eq:misc_b51} are all to the left of the annihilation operators is normal ordering.
	
	See \cite{DF1}\cite{KT1}\cite{KT2}\cite{Felder}\cite[Chp.\ 9]{CFT}. 
\end{example}

For a set $\mathtt{S}\subseteq \{1,\ldots,N\}$ and $\alpha_\pm,\beta_\pm,\gamma_\pm,\gamma_0\in \bbC$, let $\bmalpha^{\mathrm{DF0},\mathtt{S}}$, $\bmbeta^{\mathrm{DF0},\mathtt{S}} \in \bbC^N$ be given by
\begin{align}
\begin{split} 
\alpha_j &= 
\begin{cases}
\alpha_+ & (j\in \mathtt{S}), \\
\alpha_- & (j\notin \mathtt{S}),  
\end{cases} \\
\beta_j &= 
\begin{cases}
\beta_+ & (j\in \mathtt{S}), \\
\beta_- & (j\notin \mathtt{S}),  
\end{cases}
\end{split}
\end{align}
and let $\bmgamma^{\mathrm{DF0},\mathtt{S}} \in \bbC^{N(N-1)/2}$ be given by 
\begin{align}
\gamma_{j,k} &= 
\begin{cases}
\gamma_+ &  (j,k\in \mathtt{S}), \\
\gamma_0 & (j\in \mathtt{S}, k \notin \mathtt{S} \text{ or vice versa}), \\
\gamma_- & (j,k\notin \mathtt{S}).
\end{cases}
\end{align} 
Let 
\begin{equation} 
\dot{W}_N^{\mathrm{DF0},\mathtt{S}}[F] = 
\{(\alpha_-,\alpha_+,\beta_-,\beta_+,\gamma_-,\gamma_0,\gamma_+)\in \bbC^7 : (\bmalpha^{\mathrm{DF0},\mathtt{S}},\bmbeta^{\mathrm{DF0},\mathtt{S}},\bmgamma^{\mathrm{DF0},\mathtt{S}} ) \in \dot{V}_N[F] \},
\end{equation} 
where $\dot{V}_N[F]$ is defined by \cref{eq:dotVn}.
Define, for $(\alpha_-,\alpha_+,\beta_-,\beta_+,\gamma_-,\gamma_0,\gamma_+) \in \dot{W}_N^{\mathrm{DF0},\mathtt{S}}[F]$,
\begin{equation}
\dot{I}_N^{\mathrm{DF0},\mathtt{S}}[F](\alpha_-,\alpha_+,\beta_-,\beta_+,\gamma_-,\gamma_0,\gamma_+) = \dot{I}_N[F](\bmalpha^{\mathrm{DF0},\mathtt{S}},\bmbeta^{\mathrm{DF0},\mathtt{S}},\bmgamma^{\mathrm{DF0},\mathtt{S}}). 
\end{equation}
Now let $\dot{W}_N^{\mathrm{DF},\mathtt{S}}[F]$ denote the set of $(\alpha_+,\beta_+,\gamma_+)\in \bbC^3$ such that $\gamma_+\neq 0$ and, setting 
\begin{equation}
\gamma_- = \gamma_+^{-1}, \quad \alpha_- = - \gamma_- \alpha_+, \quad \beta_- = - \gamma_- \beta_+ 
\label{eq:rels}
\end{equation}
--- cf.\ \cite[eq. A.2]{DF2} ---
it is the case that $(\alpha_-,\alpha_+,\beta_-,\beta_+,\gamma_-,-1,\gamma_+)\in \dot{W}_N^{\mathrm{DF0},\mathtt{S}}[F]$. This is an open and dense subset of $\bbC^3$. Let
\begin{equation}
\dot{I}_N^{\mathrm{DF},\mathtt{S}}[F](\alpha_+,\beta_+,\gamma_+) = \dot{I}_N^{\mathrm{DF0},\mathtt{S}}[F](-\gamma_+^{-1}\alpha_+,\alpha_+,-\gamma_+^{-1}\beta_+,\beta_+,\gamma_+^{-1},-1,\gamma_+).
\end{equation}
for $(\alpha_+,\beta_+,\gamma_+) \in \dot{W}_N^{\mathrm{DF},\mathtt{S}}[F]$. 
Set $N_+=|\mathtt{S}|$ and $N_- = N-N_+$.

\begin{theorem}
	Fix $\gamma_+\in \bbC\backslash \{0,1\}$ and $\mathtt{S}\subseteq \{1,\ldots,N\}$.  
	Suppose that 
	\begin{equation} 
	F\in \operatorname{DFSym}(N,\mathtt{S}, \gamma_+^{-1}(1-\gamma_+)).
	\end{equation} 
	Then, there exists an entire function $I_{N;\mathrm{Reg}}^{\mathrm{DF},\mathtt{S}}[F](\alpha_+,\beta_+,\gamma_+) : \bbC^2_{\alpha_+,\beta_+}\to \bbC$
	such that 
	\begin{multline}
	\dot{I}_N^{\mathrm{DF},\mathtt{S}}[F](\alpha_+,\beta_+,\gamma_+) =   \Big[\prod_\pm \prod_{j=1}^{N_\pm} \frac{\sin(\pi(\alpha_\pm+\beta_\pm+(N_\pm+j-2)\gamma_\pm))}{\sin(\pi(\alpha_\pm+(j-1)\gamma_\pm)) \sin(\pi(\beta_\pm+(j-1)\gamma_\pm))} \Big] \\ \times  I_{N;\mathrm{Reg}}^{\mathrm{DF},\mathtt{S}}[F](\alpha_+,\beta_+,\gamma_+)
	\end{multline}
	when $\alpha_-,\beta_-,\gamma_-$ are related to $\alpha_+,\beta_+,\gamma_+$ by \cref{eq:rels} and the left-hand side is well-defined.
	\label{thm:Ifinal}
\end{theorem}
If desired, it is possible to replace the sines with $\Gamma$-functions with appropriate integral shifts.
\begin{example}
	When $F=1$, Dotsenko and Fateev claim in \cite[Eqs.\ A.8, A.35]{DF2}\footnote{There seem to be a couple typos in \cite[Eq. A.35]{DF2}.  \Cref{eq:misc_x6x} has these fixed. The first few cases of \cref{eq:misc_x6x} have been numerically checked, so as to verify that the fixes are correct.}  that the integral above is given by 
	\begin{multline}
	\dot{I}_N^{\mathrm{DF},\mathtt{S}}[1](\alpha_+,\beta_+,\gamma_+) \propto \gamma_\mp^{2 N_- N_+} \Big[\prod_{j=1}^{N_\pm} e^{-i\pi(j-1)\gamma_\pm}\frac{\Gamma(j\gamma_\pm) \sin(\pi j \gamma_\pm)}{\Gamma(\gamma_\pm) \sin(\pi \gamma_\pm)} \Big] \\ \times \Big[\prod_{j=1}^{N_\mp} e^{-i\pi(j-1)\gamma_\mp} \frac{\Gamma(j\gamma_\mp - N_\pm) \sin(\pi j \gamma_\mp)}{\Gamma(\gamma_\mp)\sin(\pi \gamma_\mp)}\Big]  \Big[\prod_{j=1}^{N_\pm} \frac{\Gamma(1+\alpha_\pm + (j-1) \gamma_\pm)\Gamma(1+\beta_\pm + (j-1) \gamma_\pm)}{\Gamma(2-2N_\mp + \alpha_\pm + \beta_\pm + (N_\pm-2+j)\gamma_\pm)}\Big] \\ \times  \Big[\prod_{j=1}^{N_\mp} \frac{\Gamma(1 +\alpha_\mp + (j-1) \gamma_\mp- N_\pm)\Gamma(1+\beta_\mp + (j-1) \gamma_\mp - N_\pm)}{\Gamma(2-N_\pm + \alpha_\mp + \beta_\mp + (N_\mp-2+j)\gamma_\mp)}\Big] 
	\label{eq:misc_x6x}
	\end{multline}
	for each choice of sign. 
\end{example}

\section{Associahedra}
\label{sec:geometry}

We use the term `mwc' to mean \emph{manifold-with-corners} in the sense of Melrose -- e.g. \cite{MelroseCorners}\cite{MelroseMWC}, these possessing $C^\infty$-structure. Roughly, a mwc is locally diffeomorphic to an open neighborhood of $\smash{[0,\infty)^N}$, and there is an additional requirement that boundary hypersurfaces be embedded.
In this section, we define the mwcs that will be used to resolve the singularities of Selberg- and Dotsenko--Fateev-like integrands:
\begin{itemize}
	\item in \S\ref{subsec:K}, we define the associahedra $K_{\ell,m,n}$, used to meromorphically continue the Selberg-like integrals, and 
	\item in \S\ref{subsec:A} we define the mwcs $A_{\ell,m,n}$, used to meromorphically continue the DF-like integrals. 
\end{itemize}
Since $K_{0,N,0}$ is the usual $N$-dimensional associahedra, we refer to the mwcs defined below as associahedra as well, hence the title of this section. If $M$ is a mwc, we use $\calF(M)$ to denote the set of faces of $M$, where by \textit{faces} we mean only the boundary hypersurfaces. We use ``facet'' to refer to the higher codimension boundary components. 

It is worth comparing Melrose's notion of mwc to that of polyhedron. A mwc is locally a polyhedron, but the converse is not true, as the basic requirement of $M$ being locally diffeomorphic to a relatively open neighborhood of $\smash{[0,\infty)^N}$ means that every facet $\mathrm{f}\subsetneq M$ is the intersection of at most $N$ faces. While the (closed) ball, tetrahedron, cube, and dodecahedron are all mwcs, the octahedron and icosahedron are not. It is necessary for the argument in \S\ref{sec:IR} that the associahedra $A_{\ell,m,n}$ and $K_{\ell,m,n}$ are not just polyhedra, but rather mwcs. The reason is that, since $[0,\infty)^N$ is a product of half-closed intervals, any mwc is locally diffeomorphic to a product of open or half-closed intervals. This product structure is exploited in \S\ref{sec:IR}. In contrast, the octahedron is not, in any reasonable sense, a product of one-dimensional manifolds-with-boundary near its vertices.

To summarize, the notion of ``mwc'' used here plays a similar role in our analysis to that of ``polyhedra in general position'' in \cite[\S 10.7]{Varchenko}, but the notions are not equivalent.
For the purposes of this paper, we find it more natural (and technically simpler, as it avoids the need for polyhedral realizations) to use the language of mwcs. 

We keep track of the full $C^\infty$-structure of these mwcs below. Were it required, we could keep track of $C^\omega$- (i.e.\ real analytic) structure, but since this would require going somewhat beyond the existent literature on mwcs, and since this level of precision is not needed for the rest of the paper, we will restrict ourselves to the smooth category.  

If $\mathrm{f}$ is a facet of $M$, then the blowup $[M;\mathrm{f}]$ is a mwc, and the blowdown map 
\begin{equation}
\mathrm{bd}: [M;\mathrm{f}]\to M
\end{equation}
is smooth. For convenience, we can identify the interior $[M;\mathrm{f}]^\circ$ with $M^\circ$. (If $\mathrm{F}$ is a codimension $\leq 1$ facet of $M$, then we can identify $[M;\mathrm{F}]$ with $M$ itself.) Naturally, if $\mathrm{f}$ has codimension $\geq 2$, then 
\begin{equation} 
\calF([M;\mathrm{f}]) = \{[\mathrm{F};\mathrm{f}\cap \mathrm{F}] : \mathrm{F}\in \calF(M) \} \cup \{\mathrm{ff}\}, 
\end{equation}
where $\mathrm{ff} = \mathrm{bd}^{-1}(\mathrm{f})$ is the front face of the blowup. 
Then, given boundary-defining-functions (bdfs) $x_\mathrm{F} \in C^\infty(M;\bbR^+)$ of the faces $\mathrm{F}\in \calF(M)$, we can choose bdfs $x_{[\mathrm{F} ; \mathrm{f}\cap \mathrm{F}]}, x_{\mathrm{ff}}$ of the faces of $[M;\mathrm{f}]$ such that, for each $\mathrm{F}\in \calF(M)$, 
\begin{equation}
x_{\mathrm{F}}\circ \mathrm{bd} = 
\begin{cases}
x_{ [\mathrm{F};\mathrm{f}\cap \mathrm{F}]}x_{\mathrm{ff}} & (\mathrm{f}\subseteq \mathrm{F}),\\ 
x_{ [\mathrm{F};\mathrm{f}\cap \mathrm{F}]}  & (\text{otherwise}). 
\end{cases}
\end{equation}
(We identify polyhomogeneous -- in particular, smooth -- functions on $[M;\mathrm{f}]$ with their restrictions to the interior, so, going forwards, we can drop the ``$\circ\,\mathrm{bd}$.'')
Specifically, in addition to defining $x_{[\mathrm{F}; \mathrm{f} \cap \mathrm{F}]} = x_{\mathrm{F}}$ if $\mathrm{f} \not\subseteq \mathrm{F}$, we can take 
\begin{align}
x_{\mathrm{ff}} &= \sum_{\mathrm{F}\in \calF(M), \mathrm{f}\subseteq \mathrm{F}} x_{\mathrm{F}}, 
\intertext{and, if $\mathrm{f}\subseteq \mathrm{F}$, then}
x_{[\mathrm{F}; \mathrm{f}\cap \mathrm{F}]} &=  x_{\mathrm{F}} \Big( \sum_{\mathrm{F}\in \calF(M), \mathrm{f}\subseteq \mathrm{F}} x_{\mathrm{F}}\Big)^{-1}.
\end{align} 
This follows from the analogous result for blowing up a facet of $[0,\infty)^N$.
Note that because $M$ is a mwc and not just a polyhedron, if $\mathrm{F}_1,\ldots,\mathrm{F}_d\in \calF(M)$ are distinct faces with $\cap_{\delta}\mathrm{F}_\delta \neq \varnothing$, then the connected components of $\cap_{\delta}\mathrm{F}_\delta$ are codimension $d$ facets of $M$. (The 2D lens is an example of a mwc with two faces whose intersection is disconnected.)

If $U$ is an open subset of a mwc, then $U$ can be considered as a mwc in its own right. We will say that some function $x\in C^\infty(U;[0,\infty) )$ is a bdf in $U$ of $\mathrm{F}\in \calF(M)$ if it is a bdf of the face $\mathrm{F}\cap U$ of $U$, assuming that $\mathrm{F}\cap U\neq \varnothing$, in which case it is automatically a face of $U$. 
Let $\overline{\bbR}_t=\bbR_t\cup \{-\infty,+\infty\}$ denote the ``radial'' compactification of $\bbR$. This is a ($C^\infty$-)manifold-with-boundary, with $1/t$ serving as a bdf for $\{\infty\}$ in $\{t>0\}$ and $-1/t$ serving as a bdf for $\{-\infty\}$ in $\{t<0\}$. 

\subsection{The Associahedra $K_{\ell,m,n}$}
\label{subsec:K}

We now define the mwc $K_{\ell,m,n}$ for $\ell,m,n\in \bbN$ not all zero. The blowup procedure below is a generalization of that in \cite{KT2}.  We begin with the set 
\begin{multline}
\triangle_{\ell,m,n} =\{ (x_1,\ldots,x_N) \in \overline{\bbR}^N:  x_1 \leq \cdots \leq x_\ell\leq 0\\ \leq x_{\ell+1}\leq \cdots \leq x_{\ell+m}\leq 1 \leq x_{\ell+m+1}\leq \cdots \leq x_N\},
\label{eq:misc_tel}
\end{multline}
where $N=\ell+m+n$. 
This is a compact sub-mwc of $\overline{\bbR}^N$. Naturally, 
\begin{equation} 
\triangle_{\ell,m,n} \cong \triangle_{\ell,0,0}\times \triangle_{0,m,0}\times \triangle_{0,0,n}.
\end{equation} 
Also, $\triangle_{\ell,0,0}\cong \triangle_\ell$, $\triangle_{0,m,0}\cong \triangle_m$, and $\triangle_{0,0,n}\cong \triangle_n$. 

For example, in the case $N=2$,  we have six cases. These are $\triangle_{2,0,0},\triangle_{0,2,0},\triangle_{0,0,2}$, each of which is diffeomorphic to the triangle $\triangle_2$, and $\triangle_{1,1,0},\triangle_{1,0,1},\triangle_{0,1,1}$, each of which is diffeomorphic to the square $\square_2$. 

If $\ell,n=0$, in which case $m=N$, then $\triangle_{\ell,m,n}$ is just the standard $N$-simplex $\triangle_N$.

We call a subset $\calI\subseteq \bbZ/(N+3)\bbZ$ \emph{consecutive} if it is of the form $\{k \bmod (N+3),\cdots, k+\kappa \bmod (N+3)\bbZ\}$ for some $k\in \bbZ/(N+3)\bbZ$ and $\kappa\in \bbN$.  (Thus, the empty set will not be considered consecutive.)

We label the facets (of any codimension, possibly zero) of $\triangle_{\ell,m,n}$ using (unordered) partitions $\mathtt{I}$ of $\bbZ/(N+3)\bbZ$ into consecutive subsets $\calI$, with no two of $0,\ell+1,\ell+m+2 \in \bbZ/(N+3)\bbZ$ appearing together in any element $\calI\in \mathtt{I}$.
Specifically, 
\begin{equation}
\mathrm{f}_{0,\mathtt{I}} = \left\{(x_1,\ldots,x_N) \in \triangle_{\ell,m,n}: \left( \calI\in \mathtt{I}
\Rightarrow 
\begin{cases}
j\in \calI\Rightarrow t_j = \pm \infty & (0\in \calI) \\ 
j,k\in \calI \Rightarrow t_j = t_k & (0\notin \calI)
\end{cases}\right) 
\right\}, 
\end{equation}
where
\begin{itemize}
	\item $t_j = x_j$ for $j=1,\ldots,\ell$,
	\item $t_{\ell+1} = 0$, 
	\item $t_{\ell+1+j} = x_{\ell+j}$ for $j=1,\ldots,m$,
	\item $t_{\ell+m+2}=1$, and 
	\item $t_{\ell+m+2+j} = x_{\ell+m+j}$ for $j=1,\ldots,n$.
\end{itemize}
The dimension of $\mathrm{f}_{0,\mathtt{I}}$ is given by 
\begin{equation}
\operatorname{dim} \mathrm{f}_{0,\mathtt{I}} = |\mathtt{I}|-3.
\label{eq:misc_dim}
\end{equation}
For notational simplicity, if $\mathtt{I}_0 \subseteq \mathtt{I}$ is $\mathtt{I}$ with the singletons removed, then we define $\mathrm{f}_{\mathtt{I}_0} = \mathrm{f}_{0,\mathtt{I}}$. 
Thus, $\mathrm{f}_{\varnothing}$ denotes the ``bulk'' of $\triangle_{\ell,m,n}$, and the faces of $\triangle_{\ell,m,n}$ are of the form $\mathrm{f}_{\{\calI\}}$ for $\calI$ a consecutive pair. Rephrasing \cref{eq:misc_dim},
\begin{equation}
\operatorname{codim} \mathrm{f}_{\mathtt{I}} = \sum_{\calI\in \mathtt{I}} (|\calI| - 1).
\end{equation} 

As a bdf of $\mathrm{f}_{\{\calI\}}$ for $\calI=\{k\bmod \bbZ/(N+3)\bbZ,k+1\bmod \bbZ/(N+3) \}$ when $k\in \{1,\ldots,N+1\}$, we can take 
\begin{equation} 
x_{\mathrm{f}_{\{\calI\}}} = t_{k+1}-t_k. 
\end{equation} 
For the remaining two cases of $\mathrm{F}_{\{0,1\}}$ (which only exists if $\ell\geq 1$) and $\mathrm{f}_{\{N+2,N+3\}}$ (which only exists if $n\geq 1$), we can take $x_{\mathrm{f}_{\{0,1\}}} = -1/x_1$ and $x_{\mathrm{f}_{\{N+2,N+3\}}} = 1/x_N$. 

Let $\calF_{\ell,m,n} = \calF_{\ell,m,n}(\triangle)$ denote the family of facets $\mathrm{f}_{\mathtt{I}}$ of $\triangle_{\ell,m,n}$ such that $\mathtt{I} = \{\calI\}$ for some consecutive subset $\calI\subset \bbZ/(N+3)\bbZ$ of size $|\calI|\geq 2$ not containing any two of $0,\ell+1,\ell+m+2$. In other words, $\calF_{\ell,m,n}$ is the set of facets $\mathrm{f}_{\mathtt{I}}$ for
$\mathtt{I}$ defining a partition of $\bbZ/(N+3)\bbZ$ into a single interval of length at least two (not containing any two of $0,\ell+1,\ell+m+2$) and a number of singletons which are being omitted from the notation. 

For each $d\in \{0,\ldots,N\}$, let $\calF_{\ell,m,n;d}$ denote the set of elements of $\calF_{\ell,m,n}$ of dimension $d$. 
Then, the mwc $K_{\ell,m,n}$ is defined by the iterated blowup 
\begin{equation}
K_{\ell,m,n} = [\triangle_{\ell,m,n} ; \calF_{\ell,m,n,0};\cdots ; \calF_{\ell,m,n,N}] = [\cdots [ \triangle_{\ell,m,n}; \calF_{\ell,m,n;0}]\cdots ; \calF_{\ell,m,n;N}].
\label{eq:Kblowup}
\end{equation}
I.e.,\ we \textit{first} blow up the elements of the collection $\calF_{\ell,m,n;0}$ (which may be empty, namely if $\ell,m,n$ are all nonzero), and then,  proceeding from higher to lower codimension, iteratively blow up the lifts of the facets in $\calF_{\ell,m,n;d}$ (meaning the closures of the lifts of the interiors). 

We should check that the blowup \cref{eq:Kblowup} is well-defined, which concretely means that, for each $d$, the blow-ups in the step in which we blow up the lifts of the elements of $\calF_{\ell,m,n;d}$ commute. This can be done via a somewhat tedious inductive argument, which we only sketch. 

When the time has come to blow up the facets $\mathrm{f}\neq \mathrm{f}'$ in the lifted $\calF_{\ell,m,n;d}$, their intersection is -- if nonempty -- either a point (which we denote $K_{0,0,0}$) or else an associahedron $K_{\ell_\cap,m_\cap,n_\cap}$ (which will not change upon performing further blowups) of dimension $<N$, and a neighborhood thereof is diffeomorphic to 
\begin{equation}
[0,1)^{N-d}_{t} \times K_{\ell_\cap,m_\cap,n_\cap} \times [0,1)_{t'}^{N-d}, 
\label{eq:misc_n1t}
\end{equation} 
with $\mathrm{f}$ corresponding to $\{t=0\}$ and $\mathrm{f}'$ corresponding to $\{t'=0\}$; the blowups of these two faces in the product above commute, with the result being naturally diffeomorphic to 
\begin{equation}
[[0,1)^{N-d}_{t},\{0\}] \times K_{\ell_\cap,m_\cap,n_\cap} \times [[0,1)_{t'}^{N-d} ;\{0\}].
\end{equation}

In order to prove the claimed decomposition, \cref{eq:misc_n1t}, it is first useful to note when $\mathrm{f} \cap \mathrm{f}'=\varnothing$. If $\calI,\calI'$ satisfy $|\calI|=N-d+1=|\calI'|$ and $\calI\cap \calI'\neq \varnothing$, then the corresponding facets 
\begin{align}
\mathrm{f} &= \mathrm{cl}_{ [\triangle_{\ell,m,n} ; \calF_{\ell,m,n,0};\cdots ; \calF_{\ell,m,n,d-1}]} \mathrm{f}^\circ_{\{\calI\}} \\
\mathrm{f}' &= \mathrm{cl}_{ [\triangle_{\ell,m,n} ; \calF_{\ell,m,n,0};\cdots ; \calF_{\ell,m,n,d-1}]} \mathrm{f}^\circ_{\{\calI'\}} 
\end{align}
of $[\triangle_{\ell,m,n} ; \calF_{\ell,m,n,0};\cdots ; \calF_{\ell,m,n,d-1}]$ satisfy $\mathrm{f}\cap \mathrm{f}'=\varnothing$. Indeed, $\calI\cap \calI'\neq \varnothing$ implies \begin{equation} 
\mathrm{f}_{\{\calI\}}\cap \mathrm{f}_{\{\calI'\}} = \mathrm{f}_{\{\calI\cup \calI'\}} \in \calF_{\ell,m,n}(\triangle),
\end{equation} 
and since this is blown up in an earlier stage of the construction, $\mathrm{f}$ and $\mathrm{f}'$ cannot intersect. 

So, if our two facets $\mathrm{f},\mathrm{f}'$ to be blown up have nonempty intersection, then they must be the lifts of $\mathrm{f}_{\{\calI\}}$ and $\mathrm{f}_{\{\calI'\}}$ for $\calI,\calI'$ satisfying $\calI\cap \calI'=\varnothing$. The intersection $\mathrm{f} \cap \mathrm{f}'$ lies in the preimage of $\mathrm{f}_{\{\calI\}} \cap \mathrm{f}_{\{\calI'\}} = \mathrm{f}_{\{\calI, \calI'\}}$. This facet of $\triangle_{\ell,m,n}$ is of the form $\triangle_{\ell_\cap,m_\cap,n_\cap}$ for  $\ell_\cap+m_\cap+n_\cap = 2d-N \geq 0$. As seen inductively, the lift of this facet after performing the blow-ups so far is $K_{\ell_\cap,m_\cap,n_\cap}$, although this is not crucial for the proof that the construction is well-defined. Since this has dimension $2d-N$, a neighborhood of this facet in our partially blown-up manifold automatically has the form
\begin{equation}
L=[0,1)^{2N-2d} \times K_{\ell_\cap,m_\cap,n_{\cap}},
\end{equation}
so it just needs to be checked that $\mathrm{f},\mathrm{f}'$ sit inside of this in the expected way. The $d$-dimensional facets of $L$ containing $(0,\cdots,0) \times K_{\ell_\cap,m_\cap,n_{\cap}}$
all have the form $[0,\infty)^{N-d}\times  K_{\ell_\cap,m_\cap,n_{\cap}}$ for one of the $\binom{2N-2d}{N-d}$ divisors $[0,\infty)^{N-d}\subseteq [0,\infty)^{2N-2d}$. Thus, we can decompose
\begin{equation}
[0,1)^{2N-2d} = [0,1)^{\#}_t\times [0,1)^{\#'}_{t'} \times [0,\infty)^\delta_{t''}, 
\end{equation}
for some $\delta \in \bbN$,  such that $\mathrm{f}$ corresponds to $\{t=t''=0\}$ and $\mathrm{f}'$ corresponds to $\{t'=t''=0\}$.   But, if $\delta\neq 0$, then $\mathrm{f}\cap \mathrm{f}'$ is too big, so $\delta=0$. Thus, since $\mathrm{f},\mathrm{f}'$ both have dimension $d$, it must be the case that $\#=\#' = N-d$. This completes our sketch.

We now discuss the combinatorial structure of $K_{\ell,m,n}$. 
All of the faces of $\triangle_{\ell,m,n}$ are in $\calF_{\ell,m,n;N-1}$, so every face of $K_{\ell,m,n}$ is the front face of one of our blowups. 
So, the faces of $K_{\ell,m,n}$ are in bijection with the elements of $\calF_{\ell,m,n}$ and thus with $\calI$ as above. Such a subset is uniquely specified by its endpoints $j,k \in \bbZ/(N+3)\bbZ$, since only two consecutive subsets of $\bbZ/(N+3)\bbZ$ have the same endpoints as $\calI$, namely $\calI$ itself and $\smash{\calI^\complement} \cup \{j,k\}$, and the latter contains two of $0,\ell+1,\ell+m+2$. 
Let $\calJ_{\ell,m,n}$ denote the set of unordered pairs $\{j,k\}$ arising in this way. For $\{j,k\}\in \calJ_{\ell,m,n}$, let $\calI(j,k)=\calI(k,j)$ denote the unique consecutive subset of $\bbZ/(N+3)\bbZ$ having these endpoints and containing at most one member of $\{0,\ell+1,\ell+m+2\}$. 
For such $j,k$, let $\mathrm{F}_{j,k} = \mathrm{F}_{k,j}$ denote the corresponding face of $K_{\ell,m,n}$, and let $x_{\mathrm{F}_{j,k}} = x_{\mathrm{F}_{k,j}}$ denote a bdf of that face constructed inductively as in the introduction to this section. (Note that these bdfs may depend on the particular order in which the elements of the $\calF_{\ell,m,n;d}$ are blown up.)

There are $2^{-1} N(N+3)$ faces in $K_{\ell,m,n}$. 
\begin{example}
	Consider the case $N=2$. Then, up to essential equivalence, the cases to consider are $K_{1,1,0}$ and $K_{0,2,0}$. These are depicted in \Cref{fig:K2}. The mwc $K_{1,1,0}$ is identical to $A_{1,1,0}$; in \S\ref{subsec:A} we introduce notation for labeling the faces of the $A_{\ell,m,n}$,  and this notation appears in \Cref{fig:K111} alongside that used for the $K_{\ell,m,n}$. 
	\begin{figure}[t]
		\begin{center}
			\begin{tikzpicture}[scale=2.5]
			\draw[fill=gray,fill opacity = .2] (0,0) -- (.75,0) -- (1,.25) -- (1,1) -- (0,1) -- cycle;
			\draw[dashed, color=darkgray] (.5,1) node[above] {$ \mathrm{F}_{3,4}=\mathrm{F}_{\bullet\circ(\bullet\circ)}$};
			\draw[dashed, color=darkgray] (.35,0) node[below] {$ \mathrm{F}_{2,3}=\mathrm{F}_{\bullet(\circ\bullet)\circ}$};
			\draw[dashed, color=darkgray] (0,.5) node[left] {$ \mathrm{F}_{0,1}=\mathrm{F}_{(\circ\bullet)\circ\bullet\circ}$};
			\draw[dashed, color=darkgray] (1,.5) node[right] {$ \mathrm{F}_{1,2}=\mathrm{F}_{(\bullet\circ)\bullet\circ}$};
			\draw[dashed, color=darkgray] (.9,0) node[right] {$ \mathrm{F}_{1,3}=\mathrm{F}_{(\bullet\circ\bullet)\circ}$};
			\draw[dashed, color=darkgray, ->] (.05,.05) -- (.05,.4) node[right] {$x_2$};
			\draw[dashed, color=darkgray, ->] (.05,.05) -- (.4,.05) node[above] {$w_1$};
			\begin{scope}[shift={(3.25,-.2)}, scale=1.2]
			\draw[fill=gray,fill opacity = .2] (0, .25) -- (.125,.125) -- (.875,.875) -- (.75,1) -- (0,1) -- cycle;
			\draw[dashed, color=darkgray] (0,.65) node[left] {$ \mathrm{F}_{1,2}=\mathrm{F}_{(\circ\bullet)\bullet\circ}$};
			\draw[dashed, color=darkgray] (.35,1) node[above] {$ \mathrm{F}_{3,4}=\mathrm{F}_{\circ\bullet(\bullet\circ)}$};
			\draw[dashed, color=darkgray] (.5,.4) node[right] {$ \mathrm{F}_{2,3}=\mathrm{F}_{\circ(\bullet\bullet)\circ}$};
			\draw[dashed, color=darkgray] (.1,.15) node[left] {$ \mathrm{F}_{1,3}=\mathrm{F}_{(\circ\bullet\bullet)\circ}$};
			\draw[dashed, color=darkgray] (.85,.95) node[right] {$ \mathrm{F}_{2,4}=\mathrm{F}_{\circ(\bullet\bullet\circ)}$};
			\draw[dashed, color=darkgray, ->] (.05,.95) -- (.05,.6) node[right] {$1-x_2$};
			\draw[dashed, color=darkgray, ->] (.05,.95) -- (.4,.95) node[below] {$x_1$};
			\end{scope}
			\end{tikzpicture}
		\end{center}
		\caption{The associahedra $K_{1,1,0}$ (left) and $K_{0,2,0}$ (right), realized as polyhedra roughly in accordance with the blowup procedure. In the first figure, the horizontal axis is roughly $w_1 = 1/(1-x_1)$, increasing to the right. In the second figure, it is just (roughly) $x_1$. In both figures, the vertical axis is (roughly) $x_2$.}
		\label{fig:K2}
	\end{figure}
	We have introduced an additional notation for the faces of $K_{\ell,m,n}$, indicating $\calI$ in the subscript using the following conventions:
	\begin{itemize}
		\item The elements $0,\ell+1,\ell+m+2 \in \bbZ/5\bbZ$ are depicted using a `$\circ$,' and $0$ is omitted if not included in $\calI$. 
		\item The other elements of $\bbZ/5\bbZ$ are depicted using a `$\bullet$.'
		\item Except for $0$, the elements of $\bbZ/5\bbZ$ are depicted in order. If $0$ is to be depicted, it is listed either first or last.
	\end{itemize} 
	The elements included in $\calI$ are enclosed in parentheses.
\end{example}
\begin{example}
	Consider the case $N=3$. Then, up to essential equivalence, the cases to consider are $K_{1,1,1}$, $K_{1,2,0}$, and $K_{0,3,0}$. These are depicted in \Cref{fig:K111}, \Cref{fig:K121}, \Cref{fig:K030}. 
	The mwc $K_{1,1,1}$ is identical to $A_{1,1,1}$.
	
	We have modified the ``$\bullet$'' notation from the previous example and used it to label the faces in the figures, alongside the notation used in the rest of this section. 
	For instance, when considering $K_{0,3,0}$, ``$\circ (\bullet\bullet\bullet)\circ$'' denotes $\{2,3,4\}\subset \bbZ/6\bbZ$. When considering $K_{1,2,0}$, ``$\circ (\bullet \circ\bullet)\!\bullet\circ$'' denotes $\{1,2,3\}$. When considering $K_{1,1,1}$, ``$\bullet) \bullet\circ (\bullet \circ$'' denotes $\{0,1,5\}$.
\end{example}

\begin{figure}
	\begin{center}
		\tdplotsetmaincoords{70}{115}
		\begin{tikzpicture}[scale=2.5,tdplot_main_coords]
		\draw[opacity=0] (1.5,0,0) -- (0,1.5,0) -- (0,0,1.5); 
		\draw[dashed, color=darkgray, *-] (.4,.4,0) -- (.4,.4,-.3) -- (1,1,-.3) node[right] {$ \;\;\mathrm{F}_{4,5}=\mathrm{F}_{\bullet\circ\bullet (\circ\bullet)} = \mathrm{F}_{\{3\},\varnothing;1}$};
		\draw[dashed, color=darkgray, *-] (0,.45,.35) -- (-.3,.45,.35) -- (-.3,.45,1.05) node[right] {$ \mathrm{F}_{0,1}=\mathrm{F}_{(\circ \bullet)\circ\bullet\circ\bullet}=\mathrm{F}_{\{1\},\varnothing;\infty}$};
		\draw[dashed, color=darkgray, *-] (.1,.5,.875) --  (-.3,.5,.875)  --  (-.3,1.05,.875) -- (-.3,1.05,.7) node[right] {$ \mathrm{F}_{1,5}=\mathrm{F}_{\circ \bullet)\circ\bullet\circ(\bullet}=\mathrm{F}_{\{1\},\{3\};\infty}$};
		\draw[dashed, color=darkgray, *-] (.5,0,.5) -- (.5,-.25,.5) -- (.5,-.25,1) node[left] {$ \mathrm{F}_{2,3}=\mathrm{F}_{\bullet(\circ\bullet)\circ\bullet} = \mathrm{F}_{\{2\},\varnothing;0}$};
		\draw[fill=gray,fill opacity = .1] (1,.2,0) -- (1,.8,0) -- (1,1,.2) -- (1,1,1) -- (1,.2,1) -- cycle;
		\draw[fill=gray,fill opacity = .1] (1,1,.2) -- (1,1,1) -- (.2,1,1) -- (0,1,.8) -- (0,1,.2) -- cycle;
		\draw[fill=gray,fill opacity = .1] (1,1,1) -- (.2,1,1) -- (.2,0,1) -- (.8,0,1) -- (1,.2,1) -- cycle;
		\draw[fill=gray,fill opacity = .1] (1,1,.2) -- (0,1,.2) -- (0,.8,0) -- (1,.8,0) -- cycle;
		\draw[fill=gray,fill opacity = .1] (1,.2,0) -- (.8,0,0) -- (.8,0,1) -- (1,.2,1) -- cycle;
		\draw[fill=gray,fill opacity = .1, draw=none] (.2,1,1) -- (0,1,.8) -- (0,0,.8) -- (.2,0,1) -- cycle;
		\draw[fill=gray,fill opacity = .1, dashed] (0,0,.8) -- (.2,0,1) -- (.8,0,1) -- (.8,0,0) -- (0,0,0) -- cycle;
		\draw[fill=gray,fill opacity = .1, dashed] (0,0,0) -- (0,0,.8) -- (0,1,.8) -- (0,1,.2) -- (0,.8,0) -- cycle;
		\draw[fill=gray,fill opacity = .1, draw=none] (0,0,0) -- (.8,0,0) -- (1,.2,0) -- (1,.8,0) -- (0,.8,0) -- cycle;
		\draw[dotted, <-] (0,1.3,0) node[above] {$ x_2$} -- (0,0,0);
		\draw[dotted, <-] (0,0,1.1) node[left] {$y_3$} -- (0,0,0);
		\draw[dotted, <-] (1.5,0,0) node[left] {$w_1$} -- (0,0,0);
		\draw[dashed, color=gray, *-] (.875,.125,.5) -- (1.1,-.275,.5) node[left] {$\mathrm{F}_{1,3}=\mathrm{F}_{(\bullet\circ\bullet)\circ\bullet}=\mathrm{F}_{\{1\},\{2\};0}$};
		\draw[dashed, color=gray, *-] (.6,1,.6) -- (.5,1.5,.6) node[right] {$\mathrm{F}_{3,4}=\mathrm{F}_{\bullet\circ(\bullet\circ)\bullet}=\mathrm{F}_{\{2\},\{\varnothing\};1}$};
		\draw[dashed, color=gray, *-] (.6,.975,.2) -- (.6,1.2,0) -- (.6,1.3,0) node[right] {$\mathrm{F}_{3,5}=\mathrm{F}_{\bullet\circ(\bullet\circ\bullet)}=\mathrm{F}_{\{2\},\{3\};1}$};
		\draw[dashed, color=gray, *-] (1,.5,.5) -- (1.3,.5,.5) -- (1.3,.5,-.3) node[left] {$\mathrm{F}_{1,2}=\mathrm{F}_{(\bullet\circ)\bullet\circ\bullet}=\mathrm{F}_{\{1\},\{\varnothing\};0}$};
		\draw[dashed, color=gray, *-] (.75,.6,1) -- (.75,.6,1.55) node[above] {$\mathrm{F}_{0,5} = \mathrm{F}_{5,6}=\mathrm{F}_{\bullet\circ\bullet\circ(\bullet\circ)}=\mathrm{F}_{\varnothing,\{3\};\infty}$};
		\end{tikzpicture}
	\end{center}
	\caption{The mwc $K_{1,1,1}$, with labeled faces, realized as a polyhedron roughly in accordance with the blowup procedure. Here $w_1 = 1/(1-x_1)$ and $y_3 = (x_3-1)/x_3$. The faces in the line of sight are $\mathrm{F}_{1,2}=\mathrm{F}_{(\bullet\circ)\bullet\circ\bullet}$, $\mathrm{F}_{1,3} = \mathrm{F}_{(\bullet\circ\bullet)\circ\bullet}$, $\mathrm{F}_{3,4} = \mathrm{F}_{\bullet\circ(\bullet\circ)\bullet}$, $\mathrm{F}_{3,5} = \mathrm{F}_{\bullet\circ(\bullet\circ\bullet)}$, and $\mathrm{F}_{0,5} = \mathrm{F}_{\bullet\circ\bullet\circ(\bullet\circ)}$.}
	\label{fig:K111}
\end{figure}

\begin{figure}
	\floatbox[{\capbeside\thisfloatsetup{capbesideposition={left,top},capbesidewidth=8cm,capbesidesep=none}}]{figure}[\FBwidth]
	{\caption{ The mwc $K_{1,2,0}$, with labeled faces, realized as a polyhedron roughly in accordance with the blowup procedure. As above, $w_1 = 1/(1-x_1)$. The faces in the line of site are  $\mathrm{F}_{1,2} = F_{(\bullet\circ)\bullet\bullet\circ}$, $\mathrm{F}_{1,3} = \mathrm{F}_{(\bullet\circ\bullet)\bullet\circ}$, and $\mathrm{F}_{4,5} = \mathrm{F}_{\bullet\circ\bullet(\bullet\circ)}$. }
		\label{fig:K121}}
	{\tdplotsetmaincoords{70}{115} 
		\begin{tikzpicture}[scale=2.5,tdplot_main_coords]
		\draw[opacity=0] (0,1,0) -- (1,-1.1,0); 
		\coordinate (1) at (0,0,.1);
		\coordinate (2) at (0,.1,0);
		\coordinate (3) at (0,.7,0);
		\coordinate (3h) at (0,.9,.1);
		\coordinate (4) at (0,0,1);
		\coordinate (Alpha) at (.9,0,1);
		\coordinate (beta) at (.8,0,.1);
		\coordinate (Beta) at (.9,0,.2);
		\coordinate (gamma) at (.8,.1,0);
		\coordinate (z) at (1,.1,.9);
		\coordinate (y) at (1,.1,.1);
		\coordinate (x) at (1,.2,0);
		\coordinate (w) at (1,.7,0);
		\coordinate (wh) at (1,.9,.1);
		\draw[dashed, color=darkgray, *-] (.4,.4,0) -- (.4,.4,-.2) -- (1,1,-.2) node[right] {$ \;\;\mathrm{F}_{3,4}=\mathrm{F}_{\bullet\circ(\bullet\bullet)\circ}$};
		\draw[dashed, color=darkgray, *-] (.5,0,.5) -- (.5,-.25,.5) -- (.5,-.25,1) node[left] {$ \mathrm{F}_{2,3}=\mathrm{F}_{\bullet(\circ\bullet)\bullet\circ}$};
		\draw[dashed, color=darkgray, *-] (0,.25,.25) -- (-.2,.25,.25) -- (-.2,.25,.8) node[right] {$ \mathrm{F}_{0,1}=\mathrm{F}_{(\circ \bullet)\circ\bullet\bullet\circ}$};
		\draw[dashed, color=darkgray, *-] (.6,.875,.15) -- (.6,.875,-.1) node[right] {$ \mathrm{F}_{3,5}=\mathrm{F}_{ \bullet\circ(\bullet\bullet\circ)}$};
		\draw[dashed, color=darkgray, *-] (.9,.1,.1) -- (1.1,0,0) -- (1.1,0,.15) node[left] {$ \mathrm{F}_{1,4}=\mathrm{F}_{ (\bullet\circ\bullet\bullet)\circ}$};
		\draw[dashed, color=darkgray, *-] (.55,.1,.1)  -- (.55,0,0) -- (.55,-.2,.2)  -- (.55,-.4,.2) node[left] {$\mathrm{F}_{2,4} = \mathrm{F}_{\bullet(\circ\bullet\bullet)\circ}$};
		\draw[dashed] (1) -- (2) -- (3) -- (3h) -- (4) -- cycle;
		\draw[dashed] (1) -- (2) -- (gamma) -- (beta) -- cycle;
		\draw[dashed] (beta) -- (Beta) -- (y) -- (x) -- (gamma) -- cycle;
		\draw[dashed] (3) -- (3h) -- (wh) -- (w) -- cycle;
		\draw[draw=none, fill=gray, opacity=.1] (1) -- (2) -- (3) -- (3h) -- (4) -- cycle;
		\draw[draw=none, fill=gray, opacity=.1] (3) -- (3h) -- (wh) -- (w) -- cycle;
		\draw[draw=none, fill=gray, opacity=.1] (1) -- (2) -- (gamma) -- (beta) -- cycle;
		\draw[draw=none, fill=gray, opacity=.1] (1) -- (4)  -- (Alpha) -- (Beta) -- (beta) -- cycle;
		\draw[draw=none, fill=gray, opacity=.1] (beta) -- (Beta) -- (y) -- (x) -- (gamma) -- cycle;
		\draw[draw=none, fill=gray, opacity=.1] (2) -- (gamma) -- (x) -- (w) -- (3) -- cycle;
		\draw[draw=none, fill=gray, opacity=.1] (Alpha) -- (z) -- (y) -- (Beta) -- cycle;
		\draw[draw=none, fill=gray, opacity=.1] (4) -- (3h) -- (wh) -- (z) -- (Alpha) -- cycle;
		\draw[draw=none, fill=gray, opacity=.1] (z) -- (y) -- (x) -- (w) -- (wh) -- cycle;
		\draw (z) -- (y) -- (x) -- (w) -- (wh) -- cycle;
		\draw (4) -- (3h) -- (wh) -- (z) -- (Alpha) -- cycle;
		\draw (Alpha) -- (z) -- (y) -- (Beta) -- cycle;
		\draw[dotted, <-] (0,1.3,0) node[above] {$\quad\quad x_2$} -- (0,0,0);
		\draw[dotted, <-] (0,0,1.1) node[above] {$x_3-x_2$} -- (0,0,0);
		\draw[dotted, <-] (1.5,0,0) node[left] {$w_1$} -- (0,0,0);
		\draw[dashed, color=gray, *-] (.6,.5,.5) --  (.6,.9,.9)  node[right] {$\mathrm{F}_{4,5}=\mathrm{F}_{\bullet\circ\bullet(\bullet\circ)}$};
		\draw[dashed, color=gray, *-] (1,.35,.3) -- (1.2,.35,.3) -- (1.2,.35,-.2) node[below] {$\mathrm{F}_{1,2}=\mathrm{F}_{(\bullet\circ)\bullet\bullet\circ}$};
		\draw[dashed, color=gray, *-] (.925,.075,.5) -- (.9,-.2,.5) -- (.9,-.2,.75) node[left] {$\mathrm{F}_{1,3}=\mathrm{F}_{(\bullet\circ\bullet)\bullet\circ}$};
		\end{tikzpicture}}
\end{figure}

\begin{figure}[t]
	\floatbox[{\capbeside\thisfloatsetup{capbesideposition={left,bottom},capbesidewidth=8cm,capbesidesep=none}}]{figure}[\FBwidth]
	{\caption{ The mwc $K_{0,3,0}$, with labeled faces, realized as a polyhedron roughly in accordance with the blowup procedure. The faces in the line of sight are $\mathrm{F}_{4,5}=\mathrm{F}_{\circ\bullet\bullet(\bullet\circ)}$ and $\mathrm{F}_{3,5}=\mathrm{F}_{\circ\bullet(\bullet\bullet\circ)}$. Cf.\ \cite[Fig. 5.2]{KT2}, where the full blowup procedure is depicted.}
		\label{fig:K030}}
	{\tdplotsetmaincoords{70}{115} 
		\begin{tikzpicture}[scale=2.5,tdplot_main_coords]
		\draw[opacity=0] (0,1,0) -- (1,-1.1,0); 
		\draw[dashed, color=darkgray, *-] (.05,.05,.05)  -- (-.3,-.3,-.2) -- (.4,-.3,.31) node[left] {$ \mathrm{F}_{1,4}=\mathrm{F}_{(\circ\bullet\bullet\bullet)\circ}$};
		\draw[dashed, color=darkgray, *-] (0.35,.35,0)  -- (0.35,.35,-.2) -- (0.75,.75,-.2) node[right] {$ \mathrm{F}_{3,4}=\mathrm{F}_{\circ\bullet(\bullet\bullet)\circ}$};
		\draw[dashed, color=darkgray, *-] (.1,.1,.5) -- (-.2,-.2,.5)  -- (-.2,-.2,.9) node[left] {$\mathrm{F}_{1,3} = \mathrm{F}_{(\circ\bullet\bullet)\bullet\circ}$};
		\draw[dashed, color=darkgray, *-] (.55,.1,.1)  -- (.55,0,0) -- (.55,0,-.5) node[below] {$\mathrm{F}_{2,4} = \mathrm{F}_{\circ(\bullet\bullet\bullet)\circ}$};
		\draw[dashed, color=darkgray, *-] (0,.35,.35) -- (-.2,.45,.35)  -- (-.2,0.45,.75) node[right] {$\mathrm{F}_{1,2} = \mathrm{F}_{(\circ\bullet)\bullet\bullet\circ}$};
		\draw[dashed, color=darkgray, *-] (.35,0,.35)  -- (.35,-.2,.35) -- (.35,-.2,.8) node[left] {$\mathrm{F}_{2,3} = \mathrm{F}_{\circ(\bullet\bullet)\bullet\circ}$};
		\draw[dashed, color=darkgray, *-] (.9,.1,.1)  -- (.9,-.1,-.1) -- (.9,-.1,.2)  node[left] {$ \mathrm{F}_{2,5}=\mathrm{F}_{\circ(\bullet\bullet\bullet\circ)}$};
		\draw[fill=gray,fill opacity = .1, dashed] (0,1/4,0) -- (0,1/8,1/4) -- (1/8,0,1/4) -- (1/4,0,1/8) -- (1/4,1/8,0) -- cycle;
		\draw[fill=gray,fill opacity = .1, dashed] (3/4,1/8,0) -- (3/4,0,1/8) -- (15/16,0,1/4) -- (1,0,1/8) -- (1,1/4,0) --	cycle;
		\draw[fill=gray,fill opacity = .1, draw=none] (3/4,1/8,0) -- (3/4,0,1/8) -- (1/4,0,1/8) -- (1/4,1/8,0) -- cycle;
		\draw[fill=gray,fill opacity = .1, draw=none] (0,1/8,1/4) -- (1/8,0,1/4) -- (1/8,0,7/8) -- (0,1/8,7/8) -- cycle;
		\draw[fill=gray,fill opacity = .1, dashed] (0,1/4,0) -- (0,1/8,1/4) -- (0,1/8,7/8) -- (-.15,7/8-.1,3/16) -- (0,3/4,0) -- cycle;
		\draw[fill=gray,fill opacity = .1] (1,0,1/8) -- (1,1/4,0) -- (0,3/4,0) -- (-.15,7/8-.1,3/16) -- cycle;
		\draw[fill=gray,fill opacity = .1, draw=none] (1,1/4,0) -- (0,3/4,0) -- (0,1/4,0) -- (1/4,1/8,0) -- (3/4,1/8,0) -- cycle;
		\draw[dashed]  (1/4,1/8,0) -- (3/4,1/8,0);
		\draw[fill=gray,fill opacity = .1, dashed] (1/4,0,1/8) -- (3/4,0,1/8) -- (15/16,0,1/4) -- (1/8,0,7/8) -- (1/8,0,1/4) -- cycle;
		\draw[fill=gray,fill opacity = .1] (0,1/8,7/8) -- (1/8,0,7/8) -- (15/16,0,1/4) -- (1,0,1/8) -- (-.15,7/8-.1,3/16) -- cycle;
		\draw[dotted, <-] (0,1.2,0) node[above] {$\quad\quad x_2-x_1$} -- (0,0,0);
		\draw[dotted, <-] (0,0,1.1) node[above] {$x_3-x_2$} -- (0,0,0);
		\draw[dotted, <-] (1.5,0,0) node[left] {$x_1$} -- (0,0,0);
		\draw[dashed, color=gray, *-] (.1,.2,.15) --   (.15,.8,.5) node[right] {$ \mathrm{F}_{4,5}=\mathrm{F}_{\circ\bullet\bullet(\bullet\circ)}$};
		\draw[dashed, color=gray, *-] (.5,.5,.15) --   (.85,1,.1) node[right] {$\mathrm{F}_{3,5}=\mathrm{F}_{\circ\bullet(\bullet\bullet\circ)}$};
		\end{tikzpicture}}
\end{figure}

The $K_{\ell,m,n}$ satisfy the following ``universal property:'' 
\begin{itemize}
	\item For any subsets $S\subseteq \{1,\ldots,\ell\}$, $Q\subseteq \{\ell+1,\ldots,\ell+m\}$, $R\subseteq \{\ell+m+1,\ldots,N\}$ that are not all empty, let $\mathsf{forg}:\triangle_{\ell,m,n} \to \triangle_{|S|,|Q|,|R|}$ denote the forgetful map forgetting the variables $x_j$ for $j\notin S\cup Q\cup R$.  
	Then, $\mathsf{forg}$ lifts to a smooth $\mathrm{b}$-map \cite{MelroseCorners}
	\begin{equation} 
	\overline{\mathsf{forg}} : K_{\ell,m,n} \to K_{|S|,|Q|,|R|}.
	\end{equation}  
	Given any face $\mathrm{F}$ of $K_{|S|,|Q|,|R|}$, $\overline{\mathsf{forg}}^* x_{\mathrm{F}}$ vanishes to first order at each face in $\overline{\mathsf{forg}}^{-1}(\mathrm{F})$. 
\end{itemize}
This can be proven by inducting on the number of blowups.

\begin{proposition}
	Suppose that $\mu \in C^\infty(\triangle_{\ell,m,n};\Omega \triangle_{\ell,m,n})$ is a strictly positive smooth density on $\triangle_{\ell,m,n}$. Then, the lift of $\mu$ to $K_{\ell,m,n}$ has the form 
	\begin{equation}
	\Big[\prod_{\{j,k\}\in \calJ_{\ell,m,n}} x_{ \mathrm{F}_{j,k} }^{|j-k|-1} \Big]\overline{\mu} 
	\end{equation}
	for a strictly positive $\overline{\mu} \in C^\infty(K_{\ell,m,n};\Omega K_{\ell,m,n})$. Here, for $j,k \in \bbZ/(N+3)\bbZ$, we use the notation $|j-k| = \min\{ |j_0-k_0|, |k_0-j_0| :j_0,k_0\in \bbZ: j_0\equiv j\bmod (N+3), k_0 \equiv k \bmod (N+3)\}$. 
	\label{prop:density_lift_one}
\end{proposition}
In the product, each unordered pair is counted only once.
\begin{proof}
	We recall the following lemma: 
	\begin{itemize}
		\item Suppose that $M$ is a mwc and $\mu \in C^\infty(M;\Omega M)$ is a strictly positive smooth density on $M$. 
		Then, if $\mathrm{f}$ is a facet of $M$ of codimension $d\in \bbN^+$, the lift of $\mu$ to $[M;\mathrm{f}]$ has the form $x^{d-1}_{\mathrm{ff}} \nu$ and $\nu$ a strictly positive smooth density on $[M;\mathrm{f}]$. 
	\end{itemize}
	Working in local coordinates, this follows from the case of blowing up a facet in $[0,\infty)^N$. In this case, we can use cylindrical coordinates (that is, spherical coordinates if the facet we are blowing up is the corner). The result follows from the form of the Lebesgue measure in cylindrical coordinates. 
	
	The proposition follows from an inductive application of the lemma, once we note that $|j-k|$ is the codimension of $\mathrm{F}_{j,k}$.
\end{proof}

\begin{proposition}
	The Lebesgue measure on $\bbR^N$, which defines a strictly positive smooth density on $\triangle_{\ell,m,n}^\circ$, has the form
	\begin{equation}
	\Big[\prod_{j=1}^\ell (1-x_j)^2 \Big] \Big[ \prod_{j=\ell+m+1}^N  x_j^2 \Big]  \mu 
	\label{eq:misc_bbm}
	\end{equation}
	for $\mu \in C^\infty(\triangle_{\ell,m,n};\Omega \triangle_{\ell,m,n})$ a strictly positive smooth density on $\triangle_{\ell,m,n}$. 
	\label{lem:LebesgueLift}
\end{proposition}
\begin{proof}
	It is the case that the 1-form $\dd x_j \in \Omega^1 \triangle_{\ell,m,n}^\circ$ defines an extendable 1-form on $\triangle_{\ell,m,n}$ if $j \in \{\ell+1,\cdots,\ell+m\}$, and the extension is nonvanishing. The same holds for 
	\begin{itemize}
		\item  $\dd w_j = (1-x_j)^{-2} \dd x_j$ for $w_j = 1/(1-x_j)$ if $j\in \{1,\ldots,\ell\}$ and 
		\item $\dd y_j =  x_j^{-2} \dd x_j$ for $y_j = (x_j-1)/x_j$ if $j \in \{\ell+m+1,\cdots,N\}$,
	\end{itemize}
	since $\triangle_{\ell,m,n}$ is a submanifold of $\overline{\bbR}^N$. The $\mu$ in \cref{eq:misc_bbm} can therefore be taken to be $|\mathrm{d} w_1\wedge \cdots \wedge \dd w_\ell \wedge \dd x_{\ell+1}\wedge\cdots \wedge \dd x_{\ell+m}\wedge \dd y_{\ell+m+1}\wedge\cdots\wedge \dd y_N|$, which lies in $C^\infty(\triangle_{\ell,m,n};\Omega \triangle_{\ell,m,n})$ and is strictly positive. 
\end{proof}

We now record the results of lifting the factors $x_i,1-x_i$, and $x_j-x_k$ comprising the Selberg integrand to $K_{\ell,m,n}$. 
Beginning with the first two cases:
\begin{itemize}
	\item If $i\in \{1,\ldots,\ell\}$, then 
	\begin{align}
	-x_i &\in \Big[ \prod_{j=\ell+m+3}^{N+3} \prod_{k=i}^\ell x_{\mathrm{F}_{j,k}}^{-1}  \Big] \Big[ \prod_{j=1}^i \prod_{k=\ell+1}^{\ell+m+1} x_{\mathrm{F}_{j,k}} \Big] C^\infty(K_{\ell,m,n};\bbR^+), 
	\label{eq:misc_j2g} \\
	1-x_i &\in \Big[ \prod_{j=\ell+m+3}^{N+3} \prod_{k=i}^\ell x_{\mathrm{F}_{j,k}}^{-1}  \Big] C^\infty(K_{\ell,m,n};\bbR^+) .
	\end{align}
	\item If $i\in \{\ell+1,\ldots,\ell+m\}$, then 
	\begin{align}
	x_i &\in \Big[ \prod_{j=1}^{\ell+1} \prod_{k=i+1}^{\ell+m+1} x_{\mathrm{F}_{j,k}} \Big] C^\infty(K_{\ell,m,n};\bbR^+), \label{eq:misc_j3g} \\ 
	1-x_i &\in \Big[ \prod_{j=\ell+2}^{i+1} \prod_{k=\ell+m+2}^{N+2} x_{\mathrm{F}_{j,k}} \Big]C^\infty(K_{\ell,m,n};\bbR^+).
	\label{eq:misc_k11}
	\end{align}
	\item If $i\in \{\ell+m+1,\ldots,N\}$,   then 
	\begin{align}
	x_i &\in \Big[ \prod_{j=\ell+m+3}^{i+2} \prod_{k=0}^{\ell} x_{\mathrm{F}_{j,k}}^{-1} \Big]C^\infty(K_{\ell,m,n};\bbR^+), \label{eq:misc_j4g} \\
	-(1-x_i) &\in \Big[ \prod_{j=\ell+m+3}^{i+2} \prod_{k=0}^{\ell} x_{\mathrm{F}_{j,k}}^{-1} \Big] \Big[ \prod_{j=\ell+2}^{\ell+m+2} \prod_{k=i+2}^{N+2} x_{\mathrm{F}_{j,k}} \Big]C^\infty(K_{\ell,m,n};\bbR^+). 
	\end{align}
\end{itemize}
If $N=1$, then these are all trivial to prove. By applying the universal property of the associahedra, the $N\geq 2$ case follows from the $N=1$ case.

In a similar manner, by working out the case of $K_{0,2,0}$ in detail and applying the universal property, we get, for $k>j$: 
\begin{itemize}
	\item If $j,k \in \{\ell+1,\ldots,\ell+m\}$, then 
	\begin{equation}
	x_k-x_j \in \Big[ \prod_{j_0 = 1}^{\ell+1} \prod_{k_0=k+1}^{\ell+m+1} x_{\mathrm{F}_{j_0,k_0}} \Big]\Big[ \prod_{j_0=\ell+2}^{j+1} \prod_{k_0 =k+1}^{N+2} x_{\mathrm{F}_{j_0,k_0}}\Big]C^\infty(K_{\ell,m,n};\bbR^+).
	\end{equation}
\end{itemize} 
Indeed, in the case of $\ell,n=0$ and $m=2$, this says that $x_2-x_1 \in x_{\mathrm{F}_{1,3}} x_{\mathrm{F_{2,3}}}x_{\mathrm{F_{2,4}}}C^\infty(K_{0,2,0};\bbR^+)$. Indeed, if we construct $K_{0,2,0}$ by first blowing up $\mathrm{F}_{1,3}$ and then blowing up $\mathrm{F}_{\mathrm{2,4}}$, we get
\begin{equation}
x_{\mathrm{F}_{1,3}} = x_2, \qquad
x_{\mathrm{F}_{2,3}} = \frac{x_2-x_1}{2x_2-x_2^2-x_1}, \qquad
x_{\mathrm{F}_{2,4}} = \frac{2x_2-x_2^2-x_1}{x_2},
\end{equation}
so that $x_{\mathrm{F}_{1,3}} x_{\mathrm{F_{2,3}}}x_{\mathrm{F_{2,4}}} = x_2-x_1$, on the nose. On the other hand, if we reverse the order of the blowups, then we get 
\begin{equation}
x_{\mathrm{F}_{1,3}} = \frac{x_2-x_1^2}{1-x_1},\qquad 
x_{\mathrm{F}_{2,3},0} = \frac{x_2-x_1}{x_2-x_1^2},\qquad 
x_{\mathrm{F}_{2,4}} = 1-x_1,  
\end{equation}
so we still get $x_{\mathrm{F}_{1,3}} x_{\mathrm{F_{2,3}}}x_{\mathrm{F_{2,4}}} = x_2-x_1$.

From this, we can deduce the following. 
\begin{itemize} 
	\item If $j,k \in \{1,\ldots,\ell\}$, then, in terms of $w_i=-x_i/(1-x_i)$, $(x_k-x_j)=  (1-w_j)^{-1}(1-w_k)^{-1} (w_j-w_k)$, so,  
	\begin{equation}
	x_k-x_j \in \Big[ \prod_{j_0=j}^\ell \prod_{k_0=\ell+m+3}^{N+3} x_{\mathrm{F}_{j_0,k_0}}^{-1} \Big]\Big[\prod_{j_0=1}^j \prod_{k_0=k}^{\ell+m+1} x_{\mathrm{F}_{j_0,k_0}} \Big]C^\infty(K_{\ell,m,n};\bbR^+).
	\end{equation}
	\item If $j,k\in \{\ell+m+1,\ldots,N\}$, then, in terms of $y_i = 1/x_i$, $(x_k-x_j)= y_j^{-1}y_k^{-1}(y_j-y_k)$, so
	\begin{equation}
	x_k-x_j \in \Big[ \prod_{j_0=\ell+2}^{j+2} \prod_{k_0=k+2}^{N+2} x_{\mathrm{F}_{j_0,k_0}} \Big] \Big[\prod_{j_0=\ell+m+3}^{k+2} \prod_{k_0=0}^{\ell} x_{\mathrm{F}{j_0,k_0}}^{-1} \Big] C^\infty(K_{\ell,m,n};\bbR^+).
	\end{equation}
\end{itemize}

The next three follow from the $K_{1,1,0}$, $K_{1,0,1}$, and $K_{0,1,1}$ cases. We illustrate the $K_{1,1,0}$ case, and the others are similar. 
\begin{itemize}
	\item If $j\in \{1,\ldots,\ell\}$ and $k\in \{\ell+1,\ldots,\ell+m\}$, then $(x_k-x_j)=(1-w_j)^{-1}(w_j+x_k-x_kw_j)$, so 
	\begin{equation}
	x_k-x_j \in \Big[  \prod_{j_0=j}^\ell \prod_{k_0=\ell+m+3}^{N+3} x_{\mathrm{F}_{j_0,k_0}}^{-1}\Big] \Big[ \prod_{j_0=1}^j \prod_{k_0=k+1}^{\ell+m+1} x_{\mathrm{F}_{j_0,k_0}} \Big] C^\infty(K_{\ell,m,n};\bbR^+).
	\end{equation}
	
	In the case $\ell,m=1$, $n=0$, this says that $(x_2-x_1) \in x_{\mathrm{F}_{1,5}}^{-1} x_{\mathrm{F}_{1,3}}C^\infty(K_{1,1,0};\bbR^+)$.
	Indeed, the bdf $x_{\mathrm{F}_{1,3}}$ of $\mathrm{F}_{1,3}$ in $K_{1,1,0}$ is defined by 
	\begin{equation}
	x_{\mathrm{F}_{1,3}} = (1-w_1) + x_2 =  -\frac{x_1}{1-x_1} + x_2,
	\end{equation}
	and $x_{\mathrm{F}_{1,5}} = x_{\mathrm{F}_{0,1}} = w_1 = 1/(1-x_1)$. So, 
	\begin{equation} 
	x_{\mathrm{F}_{1,5}}^{-1} x_{\mathrm{F}_{1,3}} = x_2-x_1 - x_1x_2.
	\end{equation} 
	The supposed $C^\infty(K_{1,1,0};\bbR^+)$ term above is therefore $(x_2-x_1)(x_2-x_1 - x_1x_2)^{-1} = (1 - x_2x_1/(x_2-x_1))^{-1}$. One way (besides checking in a system of local coordinate charts) to see that this is smooth (and positive) on $K_{1,1,0}$ is the identity 
	\begin{equation}
	-\frac{x_2x_1}{x_2-x_1} = \frac{x_{\mathrm{F}_{1,2}}x_{\mathrm{F}_{1,3}}x_{\mathrm{F}_{2,3}}}{x_{\mathrm{F}_{1,2}} + x_{\mathrm{F}_{2,3}}x_{\mathrm{F}_{0,1}}}.
	\label{eq:misc_j66}
	\end{equation}
	The faces $\mathrm{F}_{0,1},\mathrm{F}_{2,3}$ are disjoint from $\mathrm{F}_{1,2}$ (see \Cref{fig:K2}), so the denominator on the right-hand side of \cref{eq:misc_j66} is nonvanishing, so the quotient is indeed smooth.
	
	Likewise: 
	\item If $j\in \{\ell+1,\ldots,\ell+m\}$ and $k\in\{\ell+m+1,\ldots,N\}$, then $(x_k-x_j) = y_k^{-1} (1-x_jy_k)$, so 
	\begin{equation}
	x_k-x_j \in \Big[ \prod_{j_0=\ell+2}^{j+1} \prod_{k_0=k+2}^{N+2} x_{\mathrm{F}_{j_0,k_0}} \Big] \Big[ \prod_{j_0=0}^\ell \prod_{k_0=\ell+m+3}^{k+2}  x_{\mathrm{F}_{j_0,k_0}}^{-1} \Big] C^\infty(K_{\ell,m,n};\bbR^+).
	\end{equation}
	\item If $j\in \{1,\ldots,\ell\}$ and $k\in \{\ell+m+1,\ldots,N\}$, then $(x_k-x_j) = y_k^{-1}(1-w_j)^{-1}(1-w_j+w_jy_k)$, so
	\begin{equation}
	x_k-x_j \in \Big[  \prod_{j_0=j}^{\ell}\prod_{k_0=k+3}^{N+3} x_{\mathrm{F}_{j_0,k_0}}^{-1} \Big] \Big[ \prod_{j_0=0}^{\ell}\prod_{k_0=\ell+m+3}^{k+2}  x_{\mathrm{F}_{j_0,k_0}}^{-1} \Big] C^\infty(K_{\ell,m,n};\bbR^+).
	\label{eq:misc_j2h}
	\end{equation}
\end{itemize}

We associate to each face $\mathrm{F}_\bullet\in \calF(K_{\ell,m,n})$ an affine functional 
\begin{equation} 
\rho_{\bullet}:\bbC^{2N+N(N-1)/2}\ni (\bmalpha,\bmbeta,\bmgamma)\mapsto \rho_{\bullet}(\bmalpha,\bmbeta,\bmgamma)\in \bbC.
\end{equation}
Suppose that we are given some $\bmalpha,\bmbeta \in \bbC^N$ and $\bmgamma = \{\gamma_{j,k} = \gamma_{k,j}\}_{1\leq j < k \leq N} \in \bbC^{N(N-1)/2}$.
If one of 
\begin{enumerate}[label=(\Roman*)]
	\item  $j,k\in \{1,\ldots,\ell\}$ 
	\item $j,k \in \{\ell+2,\ldots,\ell+m+1\}$, 
	\item  $j,k \in \{\ell+m+3,\ldots,N+2\}$
\end{enumerate}
holds, then, letting $k$ denote the larger of $\{j,k\}$, 
\begin{equation}
\rho_{j,k} = k-j+2\sum_{j'\leq j_0 < k_0 \leq k'}   \gamma_{j_0,k_0},
\label{eq:misc_rt1}
\end{equation}
where, for each $i\in \{j,k\}$, $i'=i$ if $i\in \{1,\ldots,\ell\}$, $i'=i-1$ if $i\in \{\ell+2,\ldots,\ell+m+1\}$, and $i'=i-2$ if $i\in \{\ell+m+3,\ldots,N+2\}$. 
The other cases are:
\begin{itemize}
	\item If $j\in \{1,\ldots,\ell+1\}$ and $k\in \{\ell+1,\ldots,\ell+m+1\}$ and $j\neq k$, then
	\begin{equation}
	\rho_{j,k} =k-j-1+ \sum_{i=j}^{\ell} \alpha_i + \sum_{i=\ell+2}^{k} \alpha_{i-1}  + 2\sum_{j \leq j_0 < k_0 \leq k-1} \gamma_{j_0,k_0}. 
	\label{eq:misc_rt2}
	\end{equation}
	\item If $j\in \{\ell+2,\ldots,\ell+m+2\}$ and $k\in\{\ell+m+2,\ldots,N+2\}$ and $j\neq k$, then 
	\begin{equation}
	\rho_{j,k} =k-j-1+ \sum_{i=j}^{\ell+m+1} \beta_{i-1} + \sum_{i=\ell+m+3}^{k} \beta_{i-2}  + 2\sum_{j-1 \leq j_0 < k_0 \leq k-2} \gamma_{j_0,k_0}. 
	\label{eq:misc_rt3}
	\end{equation}
	\item If $j\in \{0,\ldots,\ell\}$ and $k \in \{\ell+m+3,\ldots,N+3\}$ and at least one of $j\neq 0,k\neq N+3$ holds, then 
	\begin{multline}
	\rho_{j,k} = k  -j - N - 4  - \sum_{i=1}^j \alpha_{i}- \sum_{i=1}^j \beta_{i} - \sum_{i=k}^{N+2}\alpha_{i-2} - \sum_{i=k}^{N+2}\beta_{i-2} - 2 \sum_{j'=1}^j\sum_{i=1, i\neq j'}^N \gamma_{i,j'}\\ - 2\sum_{k'=k-2}^{N} \sum_{i=1, i\neq k'-2}^N \gamma_{i,k'} + 2 \sum_{1 \leq j' <k' \leq j} \gamma_{j',k'} + 2 \sum_{k-2 \leq j' <k' \leq N} \gamma_{j',k'} + 2\sum_{j'=1}^j \sum_{k'=k-2}^{N} \gamma_{j',k'}. 
	\label{eq:misc_rt4}
	\end{multline}
\end{itemize}

\begin{proposition}	
	Given any $\bmalpha,\bmbeta\in \bbC^N$ and $\bmgamma = \{\gamma_{j,k}=\gamma_{k,j}\}_{1\leq j < k \leq N} \in  \bbC^{N(N-1)/2}$, the Selberg-like integrand 
	\begin{equation} 
	\prod_{i=1}^N |x_i|^{\alpha_i}|1-x_i|^{\beta_i}\prod_{1\leq j < k \leq N} (x_k-x_j)^{2\gamma_{j,k}} |\mathrm{d} x_1\cdots \dd x_N| \in C^\infty(\triangle^\circ_{\ell,m,n};\Omega \triangle^\circ_{\ell,m,n} )
	\end{equation} 
	lifts, via the blowdown map $\mathrm{bd}:K_{\ell,m,n}\to \triangle_{\ell,m,n}$, to an extendable density of the form 
	\begin{equation}
	\Big[\prod_{\{j,k\}\in \calJ_{\ell,m,n} } x_{\mathrm{F}_{j,k}}^{\rho_{j,k}} \Big] \mu_{\ell,m,n}(\bmalpha,\bmbeta,\bmgamma),
	\end{equation}
	for some strictly positive smooth density $\mu_{\ell,m,n}(\bmalpha,\bmbeta,\bmgamma) \in  C^\infty(K_{\ell,m,n};\Omega K_{\ell,m,n})$, depending entirely on $\bmalpha,\bmbeta,\bmgamma$.
	\label{prop:Kres}
\end{proposition}

\begin{proof}
	Each $\rho_{j,k}$ is an affine function of $\bmalpha,\bmbeta,\bmgamma$, so it suffices to check $2N+N(N-1)/2+1$ cases, the case when all three of $\bmalpha,\bmbeta,\bmgamma$ are zero and $2N+N(N-1)/2$ cases where the triple $(\bmalpha,\bmbeta,\bmgamma)$ ranges over a basis of $\bbC^{2N+N(N-1)/2}$. Write
	\begin{equation}
	\rho_{j,k}(\bmalpha,\bmbeta,\bmgamma) = \rho_{j,k}^{(0)} + \rho_{j,k}^{(1)}(\bmalpha,\bmbeta,\bmgamma), 
	\end{equation}
	where $\rho_{j,k}^{(0)} = \rho_{j,k}(\bm0,\bm0,\bm0)$ and $\rho_{j,k}^{(1)}(\bmalpha,\bmbeta,\bmgamma)=\rho_{j,k}(\bmalpha,\bmbeta,\bmgamma) - \rho_{j,k}^{(0)}$ is the linear part of $\rho_{j,k}$. Thus, we want to show that, upon lifting to $K_{\ell,m,n}$,
	\begin{align}
	|\mathrm{d}x_1\cdots \dd x_N| &\in \Big[\prod_{\{j,k\}\in \calJ_{\ell,m,n}}  x_{\mathrm{F}_{j,k}}^{\rho_{j,k}^{(0)}} \Big]C^\infty(K_{\ell,m,n};\Omega K_{\ell,m,n}), \label{eq:misc_hhh} \\
	\prod_{i=1}^N |x_i|^{\alpha_i}|1-x_i|^{\beta_i}\prod_{1\leq j < k \leq N} (x_k-x_j)^{2\gamma_{j,k}}  &\in  \Big[\prod_{\{j,k\} \in \calJ_{\ell,m,n}  } x_{\mathrm{F}_{j,k}}^{\rho_{j,k}^{(1)}(\bmalpha,\bmbeta,\bmgamma)} \Big]  C^\infty(K_{\ell,m,n}; \bbR^+) \label{eq:misc_h41}, 
	\end{align}
	with it sufficing to check \cref{eq:misc_h41} on a basis of $\bbC^{2N+N(N-1)/2}$.
	\begin{itemize}
		\item  \Cref{eq:misc_hhh} is simply a restatement of \Cref{lem:LebesgueLift}. 
		\item For a basis of $\bbC^{2N+N(N-1)/2}$, we look at $(\bmalpha,\bmbeta,\bmgamma)$ such that all of the components $\alpha_1,\ldots,\alpha_N$, $\beta_1,\ldots,\beta_N$, $\gamma_{1,2},\cdots$ of $\bmalpha,\bmbeta,\bmgamma$ are all $0$ except for one, which we set to $1$. The result then follows, via a bit of algebra, from \cref{eq:misc_j2g} through \cref{eq:misc_j2h}.
	\end{itemize}
\end{proof} 

Let $\mathtt{T}(\ell,m,n)$ denote the collection of maximal families $\mathtt{I}$ of consecutive subsets $\calI\subsetneq \bbZ/(N+3)\bbZ$ such that 
\begin{itemize}
	\item $2\leq |\calI| \leq N+1$ for all $\calI\in \mathtt{I}$, 
	\item no two of $0,\ell+1,\ell+m+2$ are in any $\calI\in \mathtt{I}$ together, and 
	\item if $\calI,\calI'\in \mathtt{I}$ satisfy $\calI\cap \calI'\neq \varnothing$, then either $\calI\subseteq \calI'$ or $\calI'\subseteq \calI$. 
\end{itemize}
The elements of $\mathtt{T}(\ell,m,n)$ can be thought of as specifying valid ways of adding parentheses to group together the elements of $\bbZ/(N+3)\bbZ$ without grouping any of $0,\ell+1,\ell+m+2$ together. 
The minimal facets of $K_{\ell,m,n}$ are in bijective correspondence with the elements of $\mathtt{T}(\ell,m,n)$, with 
\begin{equation}
\mathrm{f}_{\mathtt{I}} = \bigcap_{\calI(j,k) \in \mathtt{I}}  \mathrm{F}_{j,k}
\label{eq:misc_rmf}
\end{equation}
the facet corresponding to $\mathtt{I}$. 

\subsection{The Associahedra $A_{\ell,m,n}$}
\label{subsec:A}

We now define the mwc $A_{\ell,m,n}$ for $\ell,m,n\in \bbN$ not all zero. We begin with the $N=\ell+m+n$ hypercube $\square_{N}=[0,1]^{N}$. Parametrizing $\square_N$ by $(t_1,\ldots,t_N)$, the hypercube is identified with
\begin{equation}
[-\infty,0]^\ell_{x_1,\ldots,x_\ell}\times [0,1]^m_{x_{\ell+1},\ldots,x_{\ell+m}} \times [1,\infty]^n_{x_{\ell+m+1},\ldots,x_{N}}
\end{equation} 
via the coordinate changes $t_i =  1/(1-x_i)$ for $x_i\in [-\infty,0]$ and $i\in \{1,\ldots,\ell\}$ and $t_i = (x_i-1)/x_i$ for $x_i\in [1,\infty]$ and $i\in \{\ell+m+1,\ldots,N\}$. 

The facets of $\square_N$ we label by sextuples $(S,Q,S',Q',S'',Q'')$ consisting of (possibly empty) subsets $S,Q\subseteq \{1,\ldots,\ell\}$, $S',Q'\subseteq \{\ell+1,\ldots,\ell+m\}$, and $S'',Q''\subseteq \{\ell+m+1,\ldots,N\}$ such that $S\cap Q = S'\cap Q' = S''\cap Q'' = \varnothing$. Let
\begin{equation}
\mathrm{F}_{S,Q,S',Q',S'',Q''} = 
\left\{ 
(t_1,\ldots,t_N) \in \square_N : 
\begin{array}{c}
j\in S\cup S'\cup S'' \Rightarrow t_j = 0 \\ 
\;j\in Q\cup Q'\cup Q'' \Rightarrow t_j = 1
\end{array}
\right\} . 
\end{equation}
For instance, $\square_N=\mathrm{F}_{\varnothing,\varnothing,\varnothing,\varnothing,\varnothing,\varnothing}$. 

Now let $\calF_{\ell,m,n}=\calF_{\ell,m,n}(\square)$ denote the family of facets defined by 
\begin{equation}
\calF_{\ell,m,n} =( \{\mathrm{F}_{S,\varnothing,\varnothing,\varnothing,\varnothing,Q''} \}_{S,Q''} \cup \{\mathrm{F}_{\varnothing,Q,S',\varnothing,\varnothing,\varnothing} \}_{Q,S'} \cup \{\mathrm{F}_{\varnothing,\varnothing,\varnothing,Q',S'',\varnothing}\}_{Q',S''} )\backslash \{\square_N\}
\end{equation}
where $S,S',S'Q,Q',Q''$ range over all subsets as above. 
For each $d\in \{0,\ldots,N-1\}$, let $\calF_{\ell,m,n;d}$ denote the set of elements of $\calF_{\ell,m,n}$ of dimension $d$. 
Then, $A_{\ell,m,n}$ is defined by the iterated blowup 
\begin{equation}
A_{\ell,m,n} = [\square_N ; \calF_{\ell,m,n}] = [\square_N ; \calF_{\ell,m,n;0} ; \cdots ; \calF_{\ell,m,n;N-1}].
\label{eq:Adef}
\end{equation}

As in the previous section, we should check that, for each $d=1,\ldots,N$, having already blown up $\calF_{\ell,m,n;d_0}$ for $d_0<d$, the blowups of the closures of the lifts of the interiors of all of the $\mathrm{F}\in \calF_{\ell,m,n;d}$ all commute. 
One way to see this is to split 
\begin{equation}
\square_N =  \bigcup_{S ,S',S'' } \square_{\ell,m,n}(S,S',S''), 
\label{eq:C_decomp}
\end{equation}
where $S$ varies over all subsets of $\{1,\ldots,\ell\}$, $S'$ varies over all subsets of $\{\ell+1,\ldots,\ell+m\}$, and $S''$ varies over all subsets of $\{\ell+m+1,\ldots,N\}$, and
\begin{equation}
\square_{\ell,m,n}(S,S',S'') = \left\{ (t_1,\ldots,t_N) \in \square_N : 
\begin{array}{c}
i\in S\cup S' \cup S'' \Rightarrow t_i \in [0,2/3) \\ 
i\notin S\cup S' \cup S'' \Rightarrow t_i \in (1/3,1] 
\end{array}
\right\}. 
\end{equation}
Once we have established that blowing up $\calF_{\ell,m,n;0},\cdots,\calF_{\ell,m,n;d-1}$ is fine, then 
\begin{equation}
[\square_N ; \calF_{\ell,m,n;0} ; \cdots ; \calF_{\ell,m,n;d-1}] = \bigcup_{S ,S',S'' }  [\square_{\ell,m,n}(S,S',S'') ; \calF_{\ell,m,n;0} ; \cdots ; \calF_{\ell,m,n;d-1}] 
\end{equation}
naturally, with the left-hand side being well-defined if the right-hand side is. Thus, it suffices to check that the blowups $[\square_{\ell,m,n}(S,S',S'') ; \calF_{\ell,m,n;0} ; \cdots ; \calF_{\ell,m,n;d}]$ are all well-defined. To see this, identify
\begin{multline}
\square_{\ell,m,n}(S,S',S'') = \Big(\Big[0,\frac{2}{3}\Big)^{S}_{\{t_i\}_{i\in S}}\times \Big(\frac{1}{3},1\Big]^{(S'')^\complement}_{\{t_i\}_{i\in (S'')^\complement}} \Big)\times \Big(\Big[0,\frac{2}{3}\Big)^{S'}_{\{t_i\}_{i\in S'}} \times \Big(\frac{1}{3},1\Big]^{S^\complement}_{\{t_i\}_{i\in S^\complement}} \Big) \\ \times \Big(\Big[0,\frac{2}{3}\Big)^{S''}_{\{t_i\}_{i\in S''}}\times \Big(\frac{1}{3},1\Big]^{(S')^\complement}_{\{t_i\}_{i\in (S')^\complement}} \Big)
\end{multline}
and note that the blowup prescription is just that of performing the total boundary (tb) blowup \cite{MelroseMWC} on each of the three factors. (Note that this is not the same as the total boundary blowup of the product of the factors.) Here,
\begin{itemize}
	\item $S^\complement = \{1,\ldots,\ell\}\backslash S$, 
	\item $(S')^\complement = \{\ell+1,\ldots,\ell+m\}\backslash S'$,
	\item  and $(S'')^\complement = \{\ell+m+1,\ldots,N\}\backslash S''$.
\end{itemize}
Thus, 
\begin{multline}
A_{\ell,m,n} = \bigcup_{S ,S',S'' } \Big[ \\ 
\Big(\Big[0,\frac{2}{3}\Big)^{S}\times \Big(\frac{1}{3},1\Big]^{(S'')^\complement} \Big)_{\mathrm{tb}}\times \Big(\Big[0,\frac{2}{3}\Big)^{S'}\times \Big(\frac{1}{3},1\Big]^{S^\complement} \Big)_{\mathrm{tb}} \times \Big(\Big[0,\frac{2}{3}\Big)^{S''}\times \Big(\frac{1}{3},1\Big]^{(S')^\complement} \Big)_{\mathrm{tb}} \Big].
\label{eq:adecomp}
\end{multline}

\begin{figure}[t!]
	\begin{center} 
		\tdplotsetmaincoords{75}{110}
		\begin{tikzpicture}[scale=1.75,tdplot_main_coords]
		\draw[fill=gray, fill opacity = .1] (1,0,0) -- (1,0,1) -- (1,1,1) -- (1,1,0) -- cycle;
		\draw[fill=gray, fill opacity = .1] (0,0,1) -- (1,0,1) -- (1,1,1) -- (0,1,1) -- cycle;
		\draw[fill=gray, fill opacity = .1] (0,1,1) -- (1,1,1) -- (1,1,0) -- (0,1,0) -- cycle;
		\draw[fill=gray, fill opacity = .1, dashed] (0,0,0) -- (0,0,1) -- (0,1,1) -- (0,1,0) -- cycle;
		\draw[fill=gray, fill opacity = .1, dashed] (0,0,0) -- (1,0,0) -- (1,1,0) -- (0,1,0) -- cycle;
		\draw[fill=gray, fill opacity = .1,draw=none] (0,0,1) -- (1,0,1) -- (1,0,0) -- (0,0,0) -- cycle;
		\draw[fill=red, fill opacity = .1, draw=red] 
		(1,1,.66) -- (.33,1,.66) -- (.33,.33,.66) -- (1,.33,.66) --cycle;
		\draw[fill=red, fill opacity = .1, draw=red] 
		(1,1,0) -- (.33,1,0) -- (.33,.33,0) -- (1,.33,0) --cycle;
		\draw[fill=red, fill opacity = .1, draw=red] 
		(1,1,.66) -- (1,1,0) -- (1,.33,0) -- (1,.33,.66) --cycle;
		\draw[fill=red, fill opacity = .1, draw=red] 
		(.33,1,.66) -- (.33,1,0) -- (.33,.33,0) -- (.33,.33,.66) --cycle;
		\draw[fill=red, fill opacity = .1, draw=red] 
		(1,1,.66) -- (1,1,0) -- (.33,1,0) -- (.33,1,.66) --cycle;
		\draw[fill=red, fill opacity = .1, draw=red] 
		(1,.33,.66) -- (1,.33,0) -- (.33,.33,0) -- (.33,.33,.66) --cycle;
		\draw[dotted, <-] (0,1.35,0) node[above] {$t_2$} -- (0,0,0);
		\draw[dotted, <-] (0,0,1.35) node[left] {$t_3$} -- (0,0,0);
		\draw[dotted, <-] (1.5,0,0) node[above] {$t_1$} -- (0,0,0);
		\end{tikzpicture}
		\qquad
		\begin{tikzpicture}[scale=1.75,tdplot_main_coords]
		\draw[opacity=0] (1.5,0,0) -- (0,1.5,0) -- (0,0,1.5); 
		\draw[fill=gray,fill opacity = .1] (1,.2,0) -- (1,.8,0) -- (1,1,.2) -- (1,1,1) -- (1,.2,1) -- cycle;
		\draw[fill=gray,fill opacity = .1] (1,1,.2) -- (1,1,1) -- (.2,1,1) -- (0,1,.8) -- (0,1,.2) -- cycle;
		\draw[fill=gray,fill opacity = .1] (1,1,1) -- (.2,1,1) -- (.2,0,1) -- (.8,0,1) -- (1,.2,1) -- cycle;
		\draw[fill=gray,fill opacity = .1] (1,1,.2) -- (0,1,.2) -- (0,.8,0) -- (1,.8,0) -- cycle;
		\draw[fill=gray,fill opacity = .1] (1,.2,0) -- (.8,0,0) -- (.8,0,1) -- (1,.2,1) -- cycle;
		\draw[fill=gray,fill opacity = .1, draw=none] (.2,1,1) -- (0,1,.8) -- (0,0,.8) -- (.2,0,1) -- cycle;
		\draw[fill=gray,fill opacity = .1, dashed] (0,0,.8) -- (.2,0,1) -- (.8,0,1) -- (.8,0,0) -- (0,0,0) -- cycle;
		\draw[fill=gray,fill opacity = .1, dashed] (0,0,0) -- (0,0,.8) -- (0,1,.8) -- (0,1,.2) -- (0,.8,0) -- cycle;
		\draw[fill=gray,fill opacity = .1, draw=none] (0,0,0) -- (.8,0,0) -- (1,.2,0) -- (1,.8,0) -- (0,.8,0) -- cycle;
		\node (lb) at (0,.2,1.2) {$A_{1,1,1}$};
		\end{tikzpicture}
		\begin{tikzpicture}[scale=1.75,tdplot_main_coords]
		\draw[opacity=0] (1.5,0,0) -- (0,1.5,0) -- (0,0,1.5); 
		\draw[fill=gray,fill opacity = .1] (1,.2,.1) -- (1,.1,.2) -- (1,.1,1) -- (1,.8,1) -- (1,1,.8) -- (1,1,.1) -- cycle;
		\draw[fill=gray,fill opacity = .1] (1,1,.1) -- (.9,1,0) -- (0,1,0) -- (0,1,.8) -- (1,1,.8) -- cycle;
		\draw[fill=gray,fill opacity = .1] (1,.8,1) -- (0,.8,1) -- (0,0,1) -- (.9,0,1) -- (1,.1,1) -- cycle;
		\draw[fill=gray,fill opacity = .1] (1,.8,1) -- (0,.8,1) -- (0,1,.8) -- (1,1,.8) -- cycle;
		\draw[fill=gray,fill opacity = .1] (1,1,.1) -- (.9,1,0) -- (.9,.2,0) -- (1,.2,.1) -- cycle;
		\draw[fill=gray,fill opacity = .1] (1,.1,1) -- (.9,0,1) -- (.9,0,.2) -- (1,.1,.2) -- cycle;
		\draw[fill=gray,fill opacity = .1] (1,.2,.1) -- (1,.1,.2) -- (.9,0,.2) -- (.8,0,.1) -- (.8,.1,0) -- (.9,.2,0) -- cycle;
		\draw[fill=gray,fill opacity = .1, draw=none] (0,.1,0) -- (0,0,.1) -- (.8,0,.1) -- (.8,.1,0) -- cycle;
		\draw[fill=gray,fill opacity = .1, dashed] (.9,.2,0) -- (.9,1,0) -- (0,1,0) -- (0,.1,0) -- (.8,.1,0) -- cycle;
		\draw[fill=gray,fill opacity = .1, dashed] (.8,0,.1) -- (.9,0,.2) -- (.9,0,1) -- (0,0,1) -- (0,0,.1) -- cycle;
		\draw[fill=gray,fill opacity = .1, dashed] (0,0,.1) -- (0,.1,0) -- (0,1,0) -- (0,1,.8) -- (0,.8,1) -- (0,0,1) -- cycle;
		\node (lb) at (0,.2,1.2) {$A_{1,2,0}$};
		\end{tikzpicture}
		\begin{tikzpicture}[scale=1.75,tdplot_main_coords]
		\draw[opacity=0] (1.5,0,0) -- (0,1.5,0) -- (0,0,1.5); 
		\draw[fill=gray,fill opacity = .1] (1,.9,.8) -- (1,.8,.9) -- (.9,.8,1) -- (.8,.9,1) -- (.8,1,.9) -- (.9,1,.8) -- cycle;
		\draw[fill=gray,fill opacity = .1] (1,.9,.8) -- (.9,1,.8) -- (.9,1,0) -- (1,.9,0) -- cycle;
		\draw[fill=gray,fill opacity = .1] (1,.8,.9) -- (.9,.8,1) -- (.9,0,1) -- (1,0,.9) -- cycle;
		\draw[fill=gray,fill opacity = .1] (.8,.9,1) -- (.8,1,.9) -- (0,1,.9) -- (0,.9,1) -- cycle;
		\draw[fill=gray,fill opacity = .1] (.8,1,.9) -- (.9,1,.8) -- (.9,1,0) -- (.1,1,0) -- (0,1,.1) -- (0,1,.9) -- cycle;
		\draw[fill=gray,fill opacity = .1] (1,.9,.8) -- (1,.8,.9) -- (1,0,.9) -- (1,0,.1) -- (1,.1,0) -- (1,.9,0) -- cycle;
		\draw[fill=gray,fill opacity = .1] (.9,.8,1) -- (.8,.9,1) -- (0,.9,1) -- (0,.1,1) -- (.1,0,1) -- (.9,0,1) -- cycle;
		\draw[fill=gray,fill opacity = .1,dashed] (0,.1,.2) -- (0,.2,.1) -- (.1,.2,0) -- (.2,.1,0) -- (.2,0,.1) -- (.1,0,.2) -- cycle;
		\draw[fill=gray,fill opacity = .1,draw=none] (0,.1,.2) -- (.1,0,.2) -- (.1,0,1) -- (0,.1,1) -- cycle;
		\draw[fill=gray,fill opacity = .1,draw=none] (0,.2,.1) -- (.1,.2,0) -- (.1,1,0) -- (0,1,.1) -- cycle;
		\draw[fill=gray,fill opacity = .1,draw=none] (.2,.1,0) -- (.2,0,.1) -- (1,0,.1) -- (1,.1,0) -- cycle;
		\draw[fill=gray,fill opacity = .1,dashed] (.2,0,.1) -- (.1,0,.2) -- (.1,0,1) -- (.9,0,1) -- (1,0,.9) -- (1,0,.1) -- cycle;
		\draw[fill=gray,fill opacity = .1,dashed] (0,.1,.2) -- (0,.2,.1) -- (0,1,.1) -- (0,1,.9) -- (0,.9,1) -- (0,.1,1) -- cycle;
		\draw[fill=gray,fill opacity = .1,dashed] (.1,.2,0) -- (.2,.1,0) -- (1,.1,0) -- (1,.9,0) -- (.9,1,0) -- (.1,1,0) -- cycle;
		\node (lb) at (0,.2,1.2) {$A_{0,3,0}$};
		\end{tikzpicture}
	\end{center} 
	\caption{The $3$-cube $\square_3$ and the three blowups $A_{1,1,1}=K_{1,1,1}$, $A_{1,2,0}$, $A_{0,3,0}$ thereof. The $\square_{\ell,m,n}(S,S',S'')$ are eight subcubes corresponding to the eight vertices of $\square_3$. One such cube is depicted in red.}
	\label{fig:A}
\end{figure}

\begin{figure}[t]
	\begin{center}
		\tdplotsetmaincoords{70}{115} 
		\begin{tikzpicture}[scale=2.5,tdplot_main_coords]
		\draw[opacity=0] (1.5,0,0) -- (0,1.5,0) -- (0,0,1.5); 
		\draw[dashed, color=darkgray, *-] (.4,.4,0) -- (.4,.4,-.3) -- (1,1,-.3) node[right] {$ \mathrm{F}_{\varnothing,\{3\};0}$};
		\draw[dashed, color=darkgray, *-] (0,.45,.4) -- (-.3,.45,.4) -- (-.3,.45,1.05) node[right] {$\mathrm{F}_{\varnothing,\{1\};\infty}$};
		\draw[dashed, color=darkgray, *-] (.5,0,.5) -- (.5,-.25,.5) -- (.5,-.25,1) node[left] {$ \mathrm{F}_{\varnothing,\{2\};0}$};
		\draw[dashed, color=darkgray, *-] (.4,.05,.075)  -- (.4,-.05,-.025) -- (.4,-.5,0) -- (.4,-.5,.2) node[left] {$ \mathrm{F}_{\varnothing,\{2,3\};0}$};
		\draw[fill=gray,fill opacity = .1] (1,.2,.1) -- (1,.1,.2) -- (1,.1,1) -- (1,.8,1) -- (1,1,.8) -- (1,1,.1) -- cycle;
		\draw[fill=gray,fill opacity = .1] (1,1,.1) -- (.9,1,0) -- (0,1,0) -- (0,1,.8) -- (1,1,.8) -- cycle;
		\draw[fill=gray,fill opacity = .1] (1,.8,1) -- (0,.8,1) -- (0,0,1) -- (.9,0,1) -- (1,.1,1) -- cycle;
		\draw[fill=gray,fill opacity = .1] (1,.8,1) -- (0,.8,1) -- (0,1,.8) -- (1,1,.8) -- cycle;
		\draw[fill=gray,fill opacity = .1] (1,1,.1) -- (.9,1,0) -- (.9,.2,0) -- (1,.2,.1) -- cycle;
		\draw[fill=gray,fill opacity = .1] (1,.1,1) -- (.9,0,1) -- (.9,0,.2) -- (1,.1,.2) -- cycle;
		\draw[fill=gray,fill opacity = .1] (1,.2,.1) -- (1,.1,.2) -- (.9,0,.2) -- (.8,0,.1) -- (.8,.1,0) -- (.9,.2,0) -- cycle;
		\draw[fill=gray,fill opacity = .1, draw=none] (0,.1,0) -- (0,0,.1) -- (.8,0,.1) -- (.8,.1,0) -- cycle;
		\draw[fill=gray,fill opacity = .1, dashed] (.9,.2,0) -- (.9,1,0) -- (0,1,0) -- (0,.1,0) -- (.8,.1,0) -- cycle;
		\draw[fill=gray,fill opacity = .1, dashed] (.8,0,.1) -- (.9,0,.2) -- (.9,0,1) -- (0,0,1) -- (0,0,.1) -- cycle;
		\draw[fill=gray,fill opacity = .1, dashed] (0,0,.1) -- (0,.1,0) -- (0,1,0) -- (0,1,.8) -- (0,.8,1) -- (0,0,1) -- cycle;
		\draw[dashed, color=gray, *-] (1,.5,.5) -- (1.3,.5,.5) -- (1.3,.5,-.1) node[left] {$\mathrm{F}_{\varnothing,\{1\};0}$};
		\draw[dashed, color=gray, *-] (.6,1,.4) -- (.5,1.5,.4) node[right] {$\mathrm{F}_{\varnothing,\{2\};1}$};
		\draw[dashed, color=gray, *-] (.5,.4,1) -- (.5,.4,1.35) node[above] {$\mathrm{F}_{\varnothing,\{3\};1}$};
		\draw[dashed, color=gray, *-] (.5,.9,.9) -- (.5,1.15,1.1) -- (.5,1.4,1.1) node[right] {$\mathrm{F}_{\varnothing,\{2,3\};1}$};
		\draw[dashed, color=gray, *-] (.9,.15,.1)  -- (.9,-.1,-.1) -- (.9,-.1,.1)  node[left] {$ \mathrm{F}_{\{1\},\{2,3\};0}$};
		\draw[dashed, color=gray, *-] (.95,.6,.075)  -- (1,.6,0) -- (1,.6,-.35)  node[below] {$ \mathrm{F}_{\{1\},\{3\};0}$};
		\draw[dashed, color=gray, *-] (.95,.075,.6)  -- (1,0,.6) -- (1,-.35,.6) -- (1,-.35,.8) node[left] {$ \mathrm{F}_{\{1\},\{2\};0}$};
		\end{tikzpicture}
		\quad
		\begin{tikzpicture}[scale=2.5,tdplot_main_coords]
		\draw[opacity=0] (0,0,0) -- (1.5,.5,-.5); 
		\draw[dashed, color=darkgray, *-] (.5,.5,0) -- (.5,.5,-.3) -- (1,1,-.3) node[right] {$ \mathrm{F}_{\varnothing,\{3\};0}$};
		\draw[dashed, color=darkgray, *-] (0,.55,.4) -- (-.3,.55,.4) -- (-.3,.55,1.05) node[right] {$\mathrm{F}_{\varnothing,\{1\};0}$};
		\draw[dashed, color=darkgray, *-] (.5,0,.5) -- (.5,-.3,.5) -- (.5,-.3,1) node[left] {$ \mathrm{F}_{\varnothing,\{2\};0}$};
		\draw[dashed, color=darkgray, *-] (.1,.1,.5) -- (-.2,-.2,.5)  -- (-.2,-.2,1.1) node[left] {$ \mathrm{F}_{\varnothing,\{1,2\};0}$};
		\draw[dashed, color=darkgray, *-] (.4,.05,.06)  -- (.4,-.05,-.025) -- (.4,-.5,0) -- (.4,-.5,.2) node[left] {$ \mathrm{F}_{\varnothing,\{2,3\};0}$};
		\draw[dashed, color=darkgray, *-] (.05,.6,.06)  -- (-.05,1.2,0) -- (-.05,1.2,.3) node[right] {$ \mathrm{F}_{\varnothing,\{1,3\};0}$};
		\draw[dashed, color=darkgray, *-] (.1,.1,.1)  -- (-.1,-.1,-.1) -- (1.8,.7,-.1) -- (1.8,.7,.05) node[left] {$\mathrm{F}_{\varnothing,\{1,2,3\};0}$};
		\draw[fill=gray,fill opacity = .1] (1,.9,.8) -- (1,.8,.9) -- (.9,.8,1) -- (.8,.9,1) -- (.8,1,.9) -- (.9,1,.8) -- cycle;
		\draw[fill=gray,fill opacity = .1] (1,.9,.8) -- (.9,1,.8) -- (.9,1,0) -- (1,.9,0) -- cycle;
		\draw[fill=gray,fill opacity = .1] (1,.8,.9) -- (.9,.8,1) -- (.9,0,1) -- (1,0,.9) -- cycle;
		\draw[fill=gray,fill opacity = .1] (.8,.9,1) -- (.8,1,.9) -- (0,1,.9) -- (0,.9,1) -- cycle;
		\draw[fill=gray,fill opacity = .1] (.8,1,.9) -- (.9,1,.8) -- (.9,1,0) -- (.1,1,0) -- (0,1,.1) -- (0,1,.9) -- cycle;
		\draw[fill=gray,fill opacity = .1] (1,.9,.8) -- (1,.8,.9) -- (1,0,.9) -- (1,0,.1) -- (1,.1,0) -- (1,.9,0) -- cycle;
		\draw[fill=gray,fill opacity = .1] (.9,.8,1) -- (.8,.9,1) -- (0,.9,1) -- (0,.1,1) -- (.1,0,1) -- (.9,0,1) -- cycle;
		\draw[fill=gray,fill opacity = .1,dashed] (0,.1,.2) -- (0,.2,.1) -- (.1,.2,0) -- (.2,.1,0) -- (.2,0,.1) -- (.1,0,.2) -- cycle;
		\draw[fill=gray,fill opacity = .1,draw=none] (0,.1,.2) -- (.1,0,.2) -- (.1,0,1) -- (0,.1,1) -- cycle;
		\draw[fill=gray,fill opacity = .1,draw=none] (0,.2,.1) -- (.1,.2,0) -- (.1,1,0) -- (0,1,.1) -- cycle;
		\draw[fill=gray,fill opacity = .1,draw=none] (.2,.1,0) -- (.2,0,.1) -- (1,0,.1) -- (1,.1,0) -- cycle;
		\draw[fill=gray,fill opacity = .1,dashed] (.2,0,.1) -- (.1,0,.2) -- (.1,0,1) -- (.9,0,1) -- (1,0,.9) -- (1,0,.1) -- cycle;
		\draw[fill=gray,fill opacity = .1,dashed] (0,.1,.2) -- (0,.2,.1) -- (0,1,.1) -- (0,1,.9) -- (0,.9,1) -- (0,.1,1) -- cycle;
		\draw[fill=gray,fill opacity = .1,dashed] (.1,.2,0) -- (.2,.1,0) -- (1,.1,0) -- (1,.9,0) -- (.9,1,0) -- (.1,1,0) -- cycle;
		\draw[dashed, color=gray, *-] (.5,.5,1) -- (.5,.5,1.35) node[above] {$\mathrm{F}_{\varnothing,\{3\};1}$};
		\draw[dashed, color=gray, *-] (.9,.9,.9)  -- (.9,1.5,.9) node[right] {$\mathrm{F}_{\varnothing,\{1,2,3\};1}$};
		\draw[dashed, color=gray, *-] (.9,.9,.3)  -- (.9,1.5,.3) -- (.9,1.5,.1) node[below] {$\mathrm{F}_{\varnothing,\{1,2\};1}$};
		\draw[dashed, color=gray, *-] (.5,1,.4)  -- (.5,1.5,.4) -- (.5,1.5,.1) node[right] {$\mathrm{F}_{\varnothing,\{2\};1}$};
		\draw[dashed, color=gray, *-] (1,.5,.5) -- (2.5,.5,.5) node[left] {$\mathrm{F}_{\varnothing,\{1\};1}$};
		\draw[dashed, color=gray, *-] (.95,.5,.95) -- (1.3,.5,.95) -- (1.3,-.1,.95) node[left] {$\mathrm{F}_{\varnothing,\{1,3\};1}$};
		\draw[dashed, color=gray, *-] (.5,.95,.95) -- (.5,1.05,1.05) -- (.5,1.4,1.05) node[right] {$\mathrm{F}_{\varnothing,\{2,3\};1}$};
		\end{tikzpicture}
	\end{center}
	\caption{The eleven faces of $A_{1,2,0}$ and the fourteen faces of $A_{0,3,0}$.}
	\label{fig:Alabels}
\end{figure}

The faces of $A_{\ell,m,n}$ are in bijection with the elements of $\calF_{\ell,m,n}$. We label the faces of $A_{\ell,m,n}$ as follows:
\begin{itemize}
	\item for $S\subseteq\{1,\ldots,\ell\}$ and $Q\subseteq \{\ell+m+1,\ldots,N\}$, the face corresponding to $\mathrm{F}_{S,\varnothing,\varnothing,\varnothing,\varnothing,Q}$ is labeled as $\mathrm{F}_{S,Q;\infty} = \mathrm{F}_{Q,S;\infty}$,
	\item for $Q\subseteq \{1,\ldots,\ell\}$ and $S\subseteq \{\ell+1,\ldots,\ell+m\}$, the face corresponding to $\mathrm{F}_{\varnothing,Q,S,\varnothing,\varnothing,\varnothing}$ is labeled as as $\mathrm{F}_{S,Q;0}=\mathrm{F}_{Q,S;0}$, and
	\item for $S\subseteq \{\ell+m+1,\ldots,N\}$ and $Q\subseteq \{\ell+1,\ldots,\ell+m\}$, the face corresponding to $\mathrm{F}_{\varnothing,\varnothing,\varnothing,Q,S,\varnothing}$ is labeled as $\mathrm{F}_{S,Q;1}=\mathrm{F}_{Q,S;1}$. 
\end{itemize}
Here, $S,Q$ are not allowed to both be empty.

For any subsets $S\subseteq \{1,\ldots,\ell\}$, $Q\subseteq \{\ell+1,\ldots,\ell+m\}$, $R\subseteq \{\ell+m+1,\ldots,N\}$ that are not all empty, let $\mathsf{forg}:\square_{\ell,m,n} \to \square_{|S|,|Q|,|R|}$ denote the forgetful map forgetting the coordinates $x_i$ for $i\notin S\cup Q\cup R$, Then, $\mathsf{forg}$ lifts to a smooth b-map 
\begin{equation} 
\overline{\mathsf{forg}}: A_{\ell,m,n} \to A_{|S|,|Q|,|R|},
\label{eq:forg}
\end{equation} 
and given any face $\mathrm{F}$ of $\smash{A_{|S|,|Q|,|R|}}$, the pullback $\overline{\mathsf{forg}}^* x_{\mathrm{F}}$ vanishes to first order at each face $\mathrm{F}_0$ satisfying 
\begin{equation}
F_0\subseteq \overline{\mathsf{forg}}^{-1}(\mathrm{F}).
\end{equation}  
This is the ``universal property'' of the $A_{\ell,m,n}$. 
Via the decomposition in \cref{eq:adecomp}, it follows from the corresponding universal property of the total boundary blowup of a product, which is essentially given by \Cref{prop:forgetful_explicit}. 
\begin{proposition}
	Suppose that $\mu$ is a strictly positive smooth density on $\square_{\ell,m,n}$. Then, the lift of $\mu$ to $A_{\ell,m,n}$ has the form 
	\begin{equation}
	\Big[ \prod_{\substack{S \subseteq \{1,\ldots,\ell\} \\ Q \subseteq \{\ell+m+1,\ldots,N\}}} x_{\mathrm{F}_{S,Q;\infty}}^{|S\cup Q|-1}\Big]\Big[ \prod_{\substack{S\subseteq \{1,\ldots,\ell\} \\ Q\subseteq \{\ell+1,\ldots,\ell+m\}}} x_{\mathrm{F}_{S,Q;0}}^{|S \cup Q|-1}\Big]\Big[ \prod_{\substack{S\subseteq \{\ell+1,\ldots,\ell+m\} \\ Q \subseteq \{\ell+m+1,\ldots,N\} }} x_{\mathrm{F}_{S,Q;1}}^{|S\cup Q|-1}\Big]\overline{\mu} 
	\end{equation}
	for a strictly positive smooth density $\overline{\mu} \in C^\infty(A_{\ell,m,n};\Omega A_{\ell,m,n})$ on $A_{\ell,m,n}$. 
	\label{prop:density_lift_two}
\end{proposition}
As a notational convenience, we are setting $x_{\mathrm{F}_{\varnothing,\varnothing;x_0}}=1$ for each $x_0\in \{0,1,\infty\}$. 
\begin{proof}
	Follows via induction on the number of blowups, as in the proof of \Cref{prop:density_lift_one}. 
\end{proof}

\begin{proposition}
	The Lebesgue measure on $\bbR^N$, which defines a strictly positive smooth density on $\square_{\ell,m,n}^\circ$, has the form
	\begin{equation}
	\Big[\prod_{j=1}^\ell (1-x_j)^2 \Big] \Big[ \prod_{j=\ell+m+1}^N  x_j^2 \Big]  \mu 
	\label{eq:misc_bzm}
	\end{equation}
	for some strictly positive smooth density $\mu \in C^\infty(\square_{\ell,m,n};\Omega\square_{\ell,m,n})$ on $\square_{\ell,m,n}$. 
	\label{lem:LebesgueLift_two}
\end{proposition}
\begin{proof}
	Follows from the same computation as in \Cref{lem:LebesgueLift}.
\end{proof}

\begin{proposition}
	For each pair of distinct $i,j\in \{1,\ldots,N\}$ such that either $i,j\in \{ 1,\ldots,\ell\}$, $i,j\in \{\ell+1,\ldots,\ell+m\}$, or $i,j\in \{\ell+m+1,\ldots,N\}$, the set $H_{j,k} = \mathrm{cl}_{A_{\ell,m,n}} \{p\in \square_{\ell,m,n}^\circ: x_j=x_k\}$ is a p-submanifold of $A_{\ell,m,n}$. 
	\label{prop:p-sub}
\end{proposition}
See \cite[\S1.2]{MelroseSCConfig} for the definition of ``p-submanifold.'' 
\begin{proof}
	Consider the neighborhood $\mathrm{bd}^{-1}(\square_{\ell,m,n}(S,S',S''))\subseteq A_{\ell,m,n}$. If one of $j,k$ is in $S\cup S'\cup S''$ and the other is not, then the intersection of $H_{j,k}$ with $\mathrm{bd}^{-1}(\square_{\ell,m,n}(S,S',S''))$ is a submanifold disjoint from the boundary and therefore a p-submanifold. It therefore suffices to consider the case when $j,k \in S\cup S'\cup S''$ (and the case when neither are in $S\cup S' \cup S''$ is similar). For notational simplicity, we only consider the case when $j,k\in S'$. Then, 
	\begin{multline}
	H_{j,k}\cap \mathrm{bd}^{-1}(\square_{\ell,m,n}(S,S',S'')) = \\  \Big(\Big[0,\frac{2}{3}\Big)^{S}\times \Big(\frac{1}{3},1\Big]^{(S'')^\complement} \Big)_{\mathrm{tb}}\times \Big(\tilde{H}_{j,k}\cap \Big(\Big[0,\frac{2}{3}\Big)^{S'} \times \Big(\frac{1}{3},1\Big]^{S^\complement} \Big)_{\mathrm{tb}}\Big) \times \Big(\Big[0,\frac{2}{3}\Big)^{S''}\times \Big(\frac{1}{3},1\Big]^{(S')^\complement} \Big)_{\mathrm{tb}},
	\end{multline}
	where $\tilde{H}_{j,k}$ is the closure of $\{x_j=x_k\}$ in $([0,2/3)^{S'}\times (1/3,1]^{S^\complement})_{\mathrm{tb}}$, which is a p-submanifold \cite{MelroseSCConfig} (this also follows from \Cref{prop:coordinates}).
	Thus, $H_{j,k}\cap \mathrm{bd}^{-1}(\square_{\ell,m,n}(S,S',S''))$ is a p-submanifold of $\mathrm{bd}^{-1}(\square_{\ell,m,n}(S,S',S''))$.
	As the neighborhoods $\mathrm{bd}^{-1}(\square_{\ell,m,n}(S,S',S''))$ cover $A_{\ell,m,n}$, the conclusion follows. 
\end{proof}

This result is illustrated in \Cref{fig:p}.

We now record the results of lifting $x_i$ and $1-x_i$ to $A_{\ell,m,n}$, these being derivable via the universal property. 
\begin{itemize}
	\item If $i\in \{1,\ldots,\ell\}$, then 
	\begin{align}
	-x_i &\in \Big[  \prod_{\substack{i \in Q \subseteq \{1,\ldots,\ell\} \\ S\subseteq \{\ell+m+1,\ldots,N\}}} x_{\mathrm{F}_{S,Q;\infty}}^{-1}  \Big] \Big[ \prod_{\substack{i \in S \subseteq \{1,\ldots,\ell\} \\ Q\subseteq \{\ell+1,\ldots,\ell+m\}}} x_{\mathrm{F}_{S,Q;0}} \Big] C^\infty(A_{\ell,m,n};\bbR^+), \\
	(1-x_i) &\in \Big[  \prod_{\substack{i \in Q \subseteq \{1,\ldots,\ell\} \\ S\subseteq \{\ell+m+1,\ldots,N\}}} x_{\mathrm{F}_{S,Q;\infty}}^{-1}  \Big]  C^\infty(A_{\ell,m,n};\bbR^+).
	\end{align}
	\item If $i \in \{\ell+1,\ldots,\ell+m\}$, then
	\begin{align}
	x_i &\in \Big[ \prod_{\substack{ S \subseteq \{1,\ldots,\ell\} \\ i\in Q\subseteq \{\ell+1,\ldots,\ell+m\}}} x_{\mathrm{F}_{S,Q;0}} \Big] C^\infty(A_{\ell,m,n};\bbR^+), \\
	(1-x_i) &\in \Big[ \prod_{\substack{i \in S \subseteq \{\ell+1,\ldots,\ell+m\}\\ Q\subseteq \{\ell+m+1,\ldots,N\}}} x_{\mathrm{F}_{S,Q;1}} \Big] C^\infty(A_{\ell,m,n};\bbR^+).
	\end{align}
	\item If $i\in \{\ell+m+1,\ldots,N\}$, then 
	\begin{align}
	x_i &\in \Big[ \prod_{\substack{ i \in S \subseteq \{\ell+m+1,\ldots,N\}\\Q\subseteq \{1,\ldots,\ell\} }} x_{\mathrm{F}_{S,Q;\infty}}^{-1} \Big] C^\infty(A_{\ell,m,n};\bbR^+), \\
	-(1-x_i) &\in \Big[ \prod_{\substack{S \subseteq \{\ell+1,\ldots,\ell+m\}\\ i\in Q\subseteq \{\ell+m+1,\ldots,N\}}} x_{\mathrm{F}_{S,Q;1}} \Big] \Big[ \prod_{\substack{ i \in S \subseteq \{\ell+m+1,\ldots,N\}\\Q\subseteq \{1,\ldots,\ell\} }} x_{\mathrm{F}_{S,Q;\infty}}^{-1} \Big]C^\infty(A_{\ell,m,n};\bbR^+).
	\end{align}
\end{itemize}

Let $\calI_1 = \{1,\ldots,\ell\}$, $\calI_2 = \{\ell+1,\ldots,\ell+m\}$, and $\calI_3 = \{\ell+m+1,\ldots,N\}$.
For $j,k \in \calI_\bullet$ for the same $\bullet \in \{1,2,3\}$, let $y_{j,k}$ denote a defining function of $H_{j,k}$, with the sign chosen so as to have the same sign as $x_j-x_k$. Then, for all distinct $j,k \in \{1,\ldots,N\}$, 
\begin{equation}
(x_j-x_k) \in Y_{j,k} X_{j,k} C^\infty(A_{\ell,m,n};\bbR^+),
\end{equation}
where $X_{j,k}=X_{k,j}$ is given by 
\begin{equation}
X_{j,k} = 
\begin{cases}
\Big[ \prod_{\substack{ S\subseteq \calI_3\\ j\in Q\subseteq \calI_1\text{ or }k\in Q\subseteq \calI_1}} x_{\mathrm{F}_{S,Q;\infty}}^{-1} \Big] \Big[ \prod_{\substack{j,k \in S\subseteq \calI_1 \\ Q \subseteq \calI_2}} x_{\mathrm{F}_{S,Q;0}}  \Big] & (j,k \in \calI_1), \\ 
\Big[ \prod_{\substack{S\subseteq \calI_1\\ j,k \in Q \subseteq \calI_2}} x_{\mathrm{F}_{S,Q;0}} \Big] \Big[  \prod_{\substack{j,k\in S\subseteq \calI_2\\ Q \subseteq \calI_3}} x_{\mathrm{F}_{S,Q;1}}\Big] & (j,k \in \calI_2 ), \\ 
\Big[ \prod_{\substack{ S\subseteq\calI_2\\ j,k\in Q\subseteq \calI_3}} x_{\mathrm{F}_{S,Q;1}} \Big] \Big[ \prod_{\substack{j \in S\subseteq\calI_3\text{ or }k \in S\subseteq\calI_3 \\ Q \subseteq\calI_1}} x_{\mathrm{F}_{S,Q;\infty}}^{-1}  \Big] & (j,k\in \calI_3), \\
\Big[ \prod_{\substack{j \in S \subseteq \calI_1 \\ Q\subseteq \calI_3}} x_{\mathrm{F}_{S,Q;\infty}}^{-1} \Big] \Big[ \prod_{\substack{j \in S \subseteq \calI_1\\ k\in Q \subseteq \calI_2}} x_{\mathrm{F}_{S,Q;0}} \Big] & (j\in \calI_1, k\in \calI_2), \\
\Big[ \prod_{\substack{ S \subseteq \calI_1 \\ k\in Q\subseteq \calI_3}} x_{\mathrm{F}_{S,Q;\infty}}^{-1} \Big]\Big[ \prod_{\substack{j \in S \subseteq \calI_2\\ k\in Q \subseteq \calI_3}} x_{\mathrm{F}_{S,Q;1}} \Big] & (j\in \calI_2, k\in \calI_3), \\
\Big[ \prod_{\substack{S \subseteq \calI_1, Q\subseteq \calI_3 \\ \{j,k\}\cap S\cup Q\neq \varnothing} } x_{\mathrm{F}_{S,Q;\infty}}^{-1} \Big] & (j\in \calI_3,k\in \calI_1),
\end{cases}
\end{equation}
and $Y_{j,k}=y_{j,k}$ if, for some $\bullet\in \{1,2,3\}$, we have $j,k \in \calI_\bullet$, and $Y_{j,k}=\pm 1$ otherwise. 

We associate to each face $\mathrm{F}_\bullet\in \calF(A_{\ell,m,n})$ an affine functional 
\begin{equation} 
\varrho_{\bullet}:\bbC^{2N+N(N-1)/2}\ni (\bmalpha,\bmbeta,\bmgamma)\mapsto \varrho_{\mathrm{F}_\bullet}(\bmalpha,\bmbeta,\bmgamma)\in \bbC.
\end{equation}
Suppose that we are given some $\bmalpha,\bmbeta \in \bbC^N$ and $\bmgamma = \{\gamma_{j,k} = \gamma_{k,j}\}_{1\leq j < k \leq N} \in \bbC^{N(N-1)/2}$. Then, $\varrho_{\bullet}(\bmalpha,\bmbeta,\bmgamma)$ is defined as follows:
\begin{itemize}
	\item For $S\subseteq \{1,\ldots,\ell\}$ and $Q\subseteq \{\ell+1,\ldots,\ell+m\}$, 
	\begin{equation}
	\varrho_{S,Q;0} = |S|+|Q|-1 + \sum_{j\in S\cup Q} \alpha_j + 2 \sum_{\substack{j,k\in S\cup Q \\ j> k}} \gamma_{j,k}. 
	\label{eq:varrho0}
	\end{equation}
	\item For $S\subseteq \{\ell+1,\ldots,\ell+m\}$ and $Q\subseteq \{\ell+m+1,\ldots,N\}$, 
	\begin{equation}
	\varrho_{S,Q;1} = |S|+|Q|-1 + \sum_{j\in S\cup Q} \beta_j + 2 \sum_{\substack{j,k\in S\cup Q \\ j> k}} \gamma_{j,k}. 
	\label{eq:varrho1}
	\end{equation}
	\item For $S\subseteq \{\ell+m+1,\ldots,N\}$ and $Q\subseteq \{1,\ldots,\ell\}$, 
	\begin{equation}
	\varrho_{S,Q;\infty} = -|S|-|Q|-1  - \sum_{j\in S\cup Q} \alpha_j- \sum_{j\in S\cup Q} \beta_j - 2 \sum_{\substack{j> k \\ j\in S\cup Q\text{ or }k\in S\cup Q}} \gamma_{j,k}.
	\label{eq:varrhoinfty} 
	\end{equation}
\end{itemize}
Then, letting $\Delta\subset \square_{\ell,m,n}$ be defined by $\Delta=\cup_{\bullet=1}^3\cup_{j\neq k,j,k\in \calI_\bullet}\{x_j=x_k\}$:
\begin{proposition}
	Given any $\bmalpha,\bmbeta\in \bbC^N$ and $\bmgamma = \{\gamma_{j,k} = \gamma_{k,j}\}_{1\leq j < k \leq N} \in  \bbC^{N(N-1)/2}$, 
	\begin{equation} 
	\prod_{i=1}^N |x_i|^{\alpha_i}|1-x_i|^{\beta_i}\prod_{1\leq j < k \leq N} (x_k-x_j+i0)^{2\gamma_{j,k}} |\mathrm{d}x_1\cdots \dd x_N| \in C^\infty(\square_{\ell,m,n}^\circ\backslash \Delta; \bbC\otimes \Omega (\square^\circ_{\ell,m,n}\backslash \Delta))
	\end{equation} 
	lifts, via the blowdown map $\mathrm{bd}:A_{\ell,m,n}\to \square_{\ell,m,n}$, to 
	\begin{multline}
	\Big[\prod_{\substack{S\subseteq \{1,\ldots,\ell\} \\ Q\subseteq \{\ell+1,\ldots,\ell+m\}}} x_{\mathrm{F}_{S,Q;0}}^{\varrho_{S,Q;0}} \Big]\Big[\prod_{\substack{S\subseteq \{\ell+1,\ldots,\ell+m\} \\ Q\subseteq \{\ell+m+1,\ldots,N\}}} x_{\mathrm{F}_{S,Q;1}}^{\varrho_{S,Q;1}} \Big]\Big[\prod_{\substack{S\subseteq \{\ell+m+1,\ldots,N\} \\ Q\subseteq \{1,\ldots,\ell\}}} x_{\mathrm{F}_{S,Q;\infty}}^{\varrho_{S,Q;\infty}} \Big] 
	\Big[\prod_{1\leq j<k \leq \ell } (y_{k,j}+i0)^{2\gamma_{j,k}} \Big]  \\\times \Big[\prod_{\ell+1\leq j<k \leq \ell+m } (y_{k,j}+i0)^{2\gamma_{j,k}} \Big] \Big[\prod_{\ell+m+1\leq j<k \leq N } (y_{k,j}+i0)^{2\gamma_{j,k}} \Big] \mu_{\ell,m,n} 
	\end{multline}
	for some strictly positive smooth density $\mu_{\ell,m,n} \in C^\infty(A_{\ell,m,n};\Omega A_{\ell,m,n})$ on $A_{\ell,m,n}$, depending entirely on $\bmalpha,\bmbeta,\bmgamma$. 
	\label{prop:Ares}
\end{proposition}
\begin{proof}
	Follows from the preceding computations, along with \Cref{lem:LebesgueLift_two}.
\end{proof}

If $M$ is an orientable mwc, we say that a collection  $\calP$ of interior p-submanifolds each of codimension one is \emph{consistently orientable} if we can choose an orientation on each such that, for any $p\in M$, the subset
\begin{equation}
\sum_{P\in \calP,p\in P} {}^{++}N^*_p P \subset T^*_p M
\label{eq:misc_ne1}
\end{equation}
does not contain zero, where  ${}^{++}N^* P \subset {}^{+}N^* P \subset T^* M$ is the induced positively oriented conormal bundle, sans the zero section, and $T^* M$ is the extendable cotangent bundle of $M$. Whether or not this holds does not depend on the choice of orientation of $M$.
Choosing defining functions $\{y_{P}\}_{P\in \calP}\subset C^\infty(M;\bbR)$ for the $P\in \calP$ such that 
\begin{equation} 
\mathrm{d}y_P(p) \in {}^+N^*_p P
\label{eq:misc_ne2}
\end{equation} 
for each $p\in P$, we say that the $\{y_P\}_{P\in \calP}$ are consistently oriented defining functions. 

\begin{example}
	In $\square_{0,3,0}^\circ = (0,1)^3$, consider $\calP = \{H_{1,2}^\circ,H_{2,3}^\circ,H_{3,1}^\circ\}$. 
	The functions $x_2-x_1$, $x_3-x_2$, $x_1-x_3$ are \textit{not} consistently oriented defining functions, as 
	\begin{equation}
	0 = \dd (x_2-x_1)+\dd (x_3-x_2) + \dd (x_1-x_3), 
	\end{equation}
	but $x_2-x_1$, $x_3-x_2$, and $x_3-x_1$ are.
\end{example}

Let 
\begin{equation} 
\calP = \{H_{j,k}\}_{j,k\in \calI_1,j\neq k}\cup \{H_{j,k}\}_{j,k\in \calI_2,j\neq k} \cup \{H_{j,k}\}_{j,k\in \calI_3,j\neq k}.
\label{eq:PH}
\end{equation} 

\begin{figure}[t]
	\begin{center} 
		\tdplotsetmaincoords{75}{115}
		\begin{tikzpicture}[scale=2.75,tdplot_main_coords]
		\draw[opacity=0] (1.5,0,0) -- (0,1.5,0) -- (0,0,1.5); 
		\draw[fill=mygreen, opacity=.3, dashed] (0,.05,.05) -- (0,.9,.9) -- (1,.9,.9) -- (1,.15,.15) -- (.8,.05,.05) -- cycle;
		\draw[fill=gray,fill opacity = .1] (1,.2,.1) -- (1,.1,.2) -- (1,.1,1) -- (1,.8,1) -- (1,1,.8) -- (1,1,.1) -- cycle;
		\draw[fill=gray,fill opacity = .1] (1,1,.1) -- (.9,1,0) -- (0,1,0) -- (0,1,.8) -- (1,1,.8) -- cycle;
		\draw[fill=gray,fill opacity = .1] (1,.8,1) -- (0,.8,1) -- (0,0,1) -- (.9,0,1) -- (1,.1,1) -- cycle;
		\draw[fill=gray,fill opacity = .1] (1,.8,1) -- (0,.8,1) -- (0,1,.8) -- (1,1,.8) -- cycle;
		\draw[fill=gray,fill opacity = .1] (1,1,.1) -- (.9,1,0) -- (.9,.2,0) -- (1,.2,.1) -- cycle;
		\draw[fill=gray,fill opacity = .1] (1,.1,1) -- (.9,0,1) -- (.9,0,.2) -- (1,.1,.2) -- cycle;
		\draw[fill=gray,fill opacity = .1] (1,.2,.1) -- (1,.1,.2) -- (.9,0,.2) -- (.8,0,.1) -- (.8,.1,0) -- (.9,.2,0) -- cycle;
		\draw[fill=gray,fill opacity = .1, draw=none] (0,.1,0) -- (0,0,.1) -- (.8,0,.1) -- (.8,.1,0) -- cycle;
		\draw[fill=gray,fill opacity = .1, dashed] (.9,.2,0) -- (.9,1,0) -- (0,1,0) -- (0,.1,0) -- (.8,.1,0) -- cycle;
		\draw[fill=gray,fill opacity = .1, dashed] (.8,0,.1) -- (.9,0,.2) -- (.9,0,1) -- (0,0,1) -- (0,0,.1) -- cycle;
		\draw[fill=gray,fill opacity = .1, dashed] (0,0,.1) -- (0,.1,0) -- (0,1,0) -- (0,1,.8) -- (0,.8,1) -- (0,0,1) -- cycle;
		\node (lb) at (0,.2,1.2) {$A_{1,2,0}$};
		\end{tikzpicture}
		\begin{tikzpicture}[scale=2.75,tdplot_main_coords]
		\draw[opacity=0] (1.5,0,0) -- (0,1.5,0) -- (0,0,1.5); 
		\draw[black, dashed] (.1,.1,.1) -- (.9,.9,.9);
		\filldraw[black]  (.1,.1,.1) circle (.5pt);
		\filldraw[black]  (.9,.9,.9) circle (.5pt);
		\draw[fill=mygreen, opacity=.3, dashed] (1,.85,.85) -- (.8,.95,.95) -- (0,.95,.95) -- (0,.15,.15) -- (.2,.05,.05)-- (1,.05,.05) -- cycle;
		\draw[fill=darkcandyapp, opacity=.3, dashed] (.85,.85,1) -- (.95,.95,.8) -- (.95,.95,0) -- (.15,.15,0) -- (.05,.05,.2)-- (.05,.05,1) -- cycle;
		\draw[fill=darkblue, opacity=.3, dashed] (.85,1,.85) -- (.95,.8,.95) -- (.95,0,.95) -- (.15,0,.15) -- (.05,.2,.05)-- (.05,1,.05) -- cycle;
		\draw[fill=gray,fill opacity = .1] (1,.9,.8) -- (1,.8,.9) -- (.9,.8,1) -- (.8,.9,1) -- (.8,1,.9) -- (.9,1,.8) -- cycle;
		\draw[fill=gray,fill opacity = .1] (1,.9,.8) -- (.9,1,.8) -- (.9,1,0) -- (1,.9,0) -- cycle;
		\draw[fill=gray,fill opacity = .1] (1,.8,.9) -- (.9,.8,1) -- (.9,0,1) -- (1,0,.9) -- cycle;
		\draw[fill=gray,fill opacity = .1] (.8,.9,1) -- (.8,1,.9) -- (0,1,.9) -- (0,.9,1) -- cycle;
		\draw[fill=gray,fill opacity = .1] (.8,1,.9) -- (.9,1,.8) -- (.9,1,0) -- (.1,1,0) -- (0,1,.1) -- (0,1,.9) -- cycle;
		\draw[fill=gray,fill opacity = .1] (1,.9,.8) -- (1,.8,.9) -- (1,0,.9) -- (1,0,.1) -- (1,.1,0) -- (1,.9,0) -- cycle;
		\draw[fill=gray,fill opacity = .1] (.9,.8,1) -- (.8,.9,1) -- (0,.9,1) -- (0,.1,1) -- (.1,0,1) -- (.9,0,1) -- cycle;
		\draw[fill=gray,fill opacity = .1,dashed] (0,.1,.2) -- (0,.2,.1) -- (.1,.2,0) -- (.2,.1,0) -- (.2,0,.1) -- (.1,0,.2) -- cycle;
		\draw[fill=gray,fill opacity = .1,draw=none] (0,.1,.2) -- (.1,0,.2) -- (.1,0,1) -- (0,.1,1) -- cycle;
		\draw[fill=gray,fill opacity = .1,draw=none] (0,.2,.1) -- (.1,.2,0) -- (.1,1,0) -- (0,1,.1) -- cycle;
		\draw[fill=gray,fill opacity = .1,draw=none] (.2,.1,0) -- (.2,0,.1) -- (1,0,.1) -- (1,.1,0) -- cycle;
		\draw[fill=gray,fill opacity = .1,dashed] (.2,0,.1) -- (.1,0,.2) -- (.1,0,1) -- (.9,0,1) -- (1,0,.9) -- (1,0,.1) -- cycle;
		\draw[fill=gray,fill opacity = .1,dashed] (0,.1,.2) -- (0,.2,.1) -- (0,1,.1) -- (0,1,.9) -- (0,.9,1) -- (0,.1,1) -- cycle;
		\draw[fill=gray,fill opacity = .1,dashed] (.1,.2,0) -- (.2,.1,0) -- (1,.1,0) -- (1,.9,0) -- (.9,1,0) -- (.1,1,0) -- cycle;
		\node (lb) at (0,.2,1.2) {$A_{0,3,0}$};
		\end{tikzpicture}
	\end{center}
	\caption{The sets $H_{j,k} \in \calP$ in $A_{1,2,0}$ and $A_{0,3,0}$. Pictured are $\color{darkcandyapp} H_{1,2}$, $\color{darkblue} H_{1,3}$, and $\color{mygreen} H_{2,3}$, where the axes are oriented as in \Cref{fig:A}. In $A_{0,3,0}$, the intersection $H_{1,2}\cap H_{1,3} \cap H_{2,3}$ has been indicated with an extra dashed line.}
	\label{fig:p}
\end{figure}

\begin{proposition}
	The collection $\calP$  defined by \cref{eq:PH} is consistently orientable, and $\{y_{k,j}\}_{j<k}$ is a set of consistently oriented defining functions.
\end{proposition}
\begin{proof}
	We will show that, for any $p\in A_{\ell,m,n}$ and $\{\lambda_{j,k}\}_{p\in H_{j,k} \in \calP} \in [0,\infty)$, if the 1-form 
	\begin{equation} 
	\sum_{H_{j,k}\in \calP \text{ s.t. } p\in H_{j,k}} \lambda_{j,k} \mathrm{d} y_{k,j} \in \Omega^1(A_{\ell,m,n})
	\end{equation} 
	vanishes at $p$, then $\lambda_{j,k} = 0$ for all $H_{j,k} \in \calP$ such that $p\in H_{j,k}$.
	Put differently, we want to show that if $\mathtt{P}$ is any partition of $\{1,\ldots,N\}$ into nonempty subsets $S \subset \calI_1,\calI_2,\calI_3$, then, given any $\{\lambda_{j,k}\}_{j,k\in S\in \mathtt{P}, j<k} \subset \bbR^{\geq 0}$ not all zero, then 
	\begin{equation}
	\sum_{\substack{j,k\in S\in \mathtt{P} \\ j<k} } \lambda_{j,k} \dd y_{k,j}  
	\end{equation}
	is nonvanishing on $\cap_{j,k\in S\in \mathtt{P}, j<k} H_{j,k}$. If $\mathtt{P}$ consists only of singletons, then this is vacuously true, so it suffices to consider the case when at least one member of $\mathtt{P}$ has cardinality $>1$. 
	
	This is certainly true for $p\in \square_{\ell,m,n}^\circ$, as $\mathrm{d} y_{k,j} \propto \dd x_k - \dd x_j$ on $\square_{\ell,m,n}^\circ\cap H_{j,k}$, where the coefficient of proportionality is positive. Indeed, by the results above,
	\begin{equation} 
	x_k-x_j=f_{j,k} y_{k,j}
	\end{equation} 
	for some $f_{j,k}\in C^\infty(A_{\ell,m,n};\bbR^{\geq 0})$ that is nonvanishing in the interior, so 
	\begin{equation}
	\dd y_{k,j} = f_{j,k}^{-1} (\mathrm{d} x_k - \dd x_j) - f_{j,k}^{-1} y_{k,j} \dd f_{j,k}
	\label{eq:misc_v6v}
	\end{equation}
	in $\square_{\ell,m,n}^\circ$, which is equal to $f_{j,k}^{-1}(\mathrm{d} x_k - \dd x_j)$ on $H_{j,k} \cap  \square_{\ell,m,n}^\circ$, as claimed. This argument does not work for $p\in \partial A_{\ell,m,n}$, as $f_{j,k}$ may vanish there.

	A homogeneity argument can be used to show that, for any $p\in \smash{\partial A_{\ell,m,n}}$, there exists a tubular neighborhood $T:U\to U_0$ of a neighborhood  $U_0\subset F_0$ of $p$ in $F_0$, where $F_0$ is the smallest facet containing $p$, such that the intersections $U\cap P$ of this neighborhood with the $P\in \calP$ are all vertical subsets, meaning of the form $T^{-1}(B)$ for some $B\subset U_0$. This implies that if the 1-form above vanishes at $p$, then it also vanishes on the fiber of the tubular neighborhood over $p$ and hence somewhere in $\square^\circ_{\ell,m,n} \cap(\bigcap_{H_{j,k}\ni p} H_{j,k})$. 
	
\end{proof}

We illustrate the preceding argument with an example. Consider the case when the only one of $\ell,m,n$ that is nonzero is $m$, and consider $p \in \cap_{H_{j,k}\in \calP} H_{j,k}$. The set $\cap_{H_{j,k}\in \calP} H_{j,k} \subset A_{0,N,0}$ (the ``small diagonal'') is a p-submanifold located away from all but the very first two blowups involved in the construction of $A_{0,N,0}$. Near this p-submanifold, $A_{0,N,0}$ is canonically diffeomorphic to $[\square_{0,N,0},\{\bf0\},\{\bf1\}]$, the result of blowing up two opposite corners of the $N$-cube. We consider the situation near the blowup of 
\begin{equation} 
\{{\bf0}\}=\mathrm{F}_{\varnothing,\varnothing,\{1,\ldots,N\},\varnothing,\varnothing,\varnothing},
\end{equation} 
and the situation near the opposite corner is similar. In the interior of the front face of that blowup, we can use $\varrho = x_1$ as a bdf and coordinates $\hat{x}_j = x_j/x_1$ for $j=2,\ldots,N$ as parametrizing the face itself. In terms of these coordinates, 
\begin{equation} 
\cap_{H_{j,k}\in \calP} H_{j,k} = \{\hat{x}_2,\cdots,\hat{x}_N=1\}
\end{equation} 
locally, and, for $1\leq j<k \leq N$, we can write $y_{k,j} = \tilde{y}_{k,j} C^\infty(A_{0,N,0};\bbR^+)$ for $\tilde{y}_{k,j}$ given locally by $\tilde{y}_{k,j} = \varrho^{-1}(x_k-x_j) = \hat{x}_k-\hat{x}_j$, where $\hat{x}_1=1$. This satisfies
\begin{equation}
\dd \tilde{y}_{k,j} = 
\begin{cases}
\mathrm{d} \hat{x}_k &(j=1), \\ 
\mathrm{d} \hat{x}_k - \dd \hat{x}_j  &(j\neq 1). 
\end{cases}
\end{equation}
So, if $\lambda_{k,j}\geq 0$, then $\sum_{1\leq j<k\leq N} \lambda_{j,k}\mathrm{d}\tilde{y}_{k,j} = 0 \Rightarrow \lambda_{j,k}=0$ for all $k,j$. Since the $y_{k,j}$ differ from the $\tilde{y}_{k,j}$ by a (smooth) positive factor, the $y_{k,j}$ have the same property on $\cap_{H_{j,k}\in \calP} H_{j,k} $. 

There is a more direct argument using the coordinates in \Cref{prop:coordinates} (with the decomposition \cref{eq:adecomp}). Namely, using \cref{eq:adecomp}, the result follows from the analogous result for $[0,1)_{\mathrm{tb}}^N$. Given any $\sigma\in \frakS_N$, consider the coordinates $\varrho,\hat{x}_{\sigma(2)},\cdots,\hat{x}_{\sigma(N)}$ defined in \Cref{prop:coordinates}, these giving a $C^\infty$-atlas as $\sigma$ varies over all permutations. In these coordinate systems, the relevant p-submanifolds are, locally,  
\begin{equation}
H_{j,k} =\{ \hat{x}_{j+1}\cdots \hat{x}_k=1\} \subset [0,1)_{\mathrm{tb}}^N,
\end{equation}
so have defining functions $y_{k,j}=-1+\hat{x}_{j+1}\cdots \hat{x}_k$. This satisfies 
\begin{equation}
\dd y_{k,j} =  \sum_{i=j+1}^k \frac{\dd \hat{x}_i}{\hat{x}_i}
\label{eq:misc_885}
\end{equation}
on $H_{j,k}$. The 1-forms, $\omega_{j,k} = \sum_{i=j+1}^k \hat{x}_i^{-1} \dd \hat{x}_i$, defined by the right-hand side of \cref{eq:misc_885} satisfy 
\begin{multline} 
\{\lambda_{j,k}\}_{1\leq j<k \leq N\text{ and }j,k\in S\in \mathtt{P}}\subset [0,\infty)\text{ and } \sum_{1\leq j<k\leq N \text{ and }j,k\in S\in \mathtt{P}} \lambda_{j,k} \omega_{j,k} = 0 \\ \Rightarrow \lambda_{j,k}=0\text{ for all $j<k$ in $S\in \mathtt{P}$},
\end{multline} 
from which the result follows.

Let $\Sigma\mathtt{T}(\ell,m,n)$ denote the collection of maximal families $\mathtt{I}$ of pairs $(x_0,S)$ of $x_0\in \{0,1,\infty\}$ and nonempty $S\subseteq \{1,\ldots,N\}$ such that 
\begin{itemize}
	\item if $(x_0,S),(x_0,Q)\in \mathtt{I}$, either $S\subseteq Q$ or $Q\subseteq S$, 
	\item 
	\begin{equation}
	(x_0,S)\in \mathtt{I} \Rightarrow 
	\begin{cases}
	S\cap \calI_3 = \varnothing & (x_0=0), \\
	S\cap \calI_1 = \varnothing & (x_0=1), \\
	S\cap \calI_2 = \varnothing & (x_0=\infty). 
	\end{cases}
	\end{equation}
\end{itemize}
The minimal facets of $A_{\ell,m,n}$ are in bijective correspondence with the elements of $\Sigma\mathtt{T}(\ell,m,n)$, with 
\begin{equation}
\mathrm{f}_{\mathtt{I}} = \bigcap_{(x_0,\calS) \in \mathtt{I},\; S\cup Q \subseteq \calS} \mathrm{F}_{S,Q;x_0}
\end{equation}
the facet corresponding to $\mathtt{I}$. 

\section{Meromorphic continuation}
\label{sec:IR}

We now turn to the analytic extension of Selberg-like integrals to dense, open subsets of the space of possible exponents. As discussed in the introduction, the results in this section are apparently sharp for generic Selberg-like integrals, but for e.g.\ symmetric Selberg-like integrals they are only preliminary.
Nevertheless, the results we prove here will be useful in establishing the sharp results later.
For our discussion of the symmetric and DF-symmetric cases, it is useful to consider somewhat more general integrals than \cref{eq:Selberg_def}. Let $\ell,m,n\in \bbN$ satisfy $\ell+m+n=N\in \bbN^+$. Fix a \emph{finite} collection $\calD$ of indexed sets 
\begin{equation}
\{d_{\mathrm{F}} \}_{\mathrm{F} \in \calF(K_{\ell,m,n}) }  \subseteq \bbR. 
\end{equation}
Define   
\begin{equation}
S_{\ell,m,n}[F](\bmalpha,\bmbeta,\bmgamma) = \int_{\triangle_{\ell,m,n}}\Big[ \prod_{i=1}^N |x_i|^{\alpha_i}|1-x_i|^{\beta_i}  \Big] \Big[ \prod_{1\leq j<k \leq N} (x_k-x_j)^{2\gamma_{j,k}} \Big] F \dd x_1\cdots \dd x_N,
\label{eq:misc_snf}
\end{equation}
for $(\bmalpha,\bmbeta,\bmgamma)\in \Omega_{\ell,m,n}[\calD]$, where 
\begin{itemize}
	\item $\Omega_{\ell,m,n}[\calD]$ denotes the set of $(\bmalpha,\bmbeta,\bmgamma)\in \bbC^N_{\bmalpha}\times \bbC^N_{\bmbeta}\times \bbC^{N(N-1)/2}_{\bmgamma}$ such that 
	\begin{equation}
	\Big[ \prod_{i=1}^N |x_i|^{\alpha_i}|1-x_i|^{\beta_i}  \Big] \Big[ \prod_{1\leq j<k \leq N} (x_k-x_j)^{2\gamma_{j,k}} \Big] \Big[\prod_{\mathrm{F} \in \calF (K_{\ell,m,n}) } x_{\mathrm{F}}^{d_{\mathrm{F}}}\Big]  \in L^1(\triangle_{\ell,m,n},\dd x_1\cdots \dd x_N )
	\label{eq:misc_b31}
	\end{equation}
	for all $\{d_{\mathrm{F}} \}_{\mathrm{F} \in \calF(K_{\ell,m,n}) } \in \calD$, and  
	\item  $F$ has the form 
	\begin{equation} 
	F= \sum_{\{d_{\mathrm{F}} \}_{\mathrm{F} \in \calF(K_{\ell,m,n}) }} \Big[ \prod_{\mathrm{F} \in \calF (K_{\ell,m,n}) } x_{\mathrm{F}}^{d_{\mathrm{F}}} \Big] F_{\{d_{\mathrm{F}} \}_{\mathrm{F}\in \calF(K_{\ell,m,n})} }
	\label{eq:misc_uyu}
	\end{equation} 
	for some $F_{\{d_{\mathrm{F}} \}_{\mathrm{F}\in \calF(K_{\ell,m,n})} } \in C^\infty(K_{\ell,m,n})$. 
\end{itemize}
We denote the set of such $F$ by $\calA^{\calD}(K_{\ell,m,n})$. From the definition \cref{eq:misc_tel} of $\triangle_{\ell,m,n}$, the integrand is nonvanishing there, so the the absolute values in \cref{eq:misc_snf} amount to a choice of branch.

Observe that $\Omega_{\ell,m,n}[\calD]$ is a nonempty, open, and connected subset of $\bbC^{2N+N(N-1)/2}$.  In the case $N=1$, we consider $\Omega_{\ell,m,n}[\calD]$ as a subset of $\bbC^2_{\alpha,\beta}$. 

We write $\Omega[F]$ to denote $\Omega_{\ell,m,n}[\calD]$ for arbitrary $\calD$ such that $F\in \calA^{\calD}(K_{\ell,m,n})$. 
Let 
\begin{equation} 
\Omega_{\ell,m,n}=\Omega_{\ell,m,n}[\{\{0 \}_{\mathrm{F}  \in \calF(K_{\ell,m,n})}\} ].
\end{equation}  
If $\varphi\in C^\infty(\triangle_{\ell,m,n})$, then we can consider $\varphi$ as an element of $C^\infty(K_{\ell,m,n})$, so $S_{\ell,m,n}[\varphi]:\Omega_{\ell,m,n}\to \bbC$ is well-defined, and $\Omega_{\ell,m,n}[\varphi]\supseteq \Omega_{\ell,m,n}$. 
If $f\in \bbC[x_1,\ldots,x_N]$, then the lift of $f\varphi$ is also a classical symbol on $K_{\ell,m,n}$ (it is smooth if $\ell,n=0$, but not necessarily otherwise), so 
\begin{equation} 
S_{\ell,m,n}[f\varphi]:\Omega_{\ell,m,n}[f]\to \bbC
\end{equation} 
is well-defined, except now we may have $\Omega_{\ell,m,n}[f\varphi] \not\subseteq \Omega_{\ell,m,n}$ if $\ell\neq 0$ or $n\neq 0$.

In the special case when $\ell,n=0$ and $m=N$, we use the abbreviations $\Omega_{0,N,0}=\Omega_N$, $\Omega_{0,N,0}[\bullet] = \Omega_N[\bullet]$, and 
\begin{equation} 
S_{0,N,0}[F](\bmalpha,\bmbeta,\bmgamma) = S_{N}[F](\bmalpha,\bmbeta,\bmgamma),
\end{equation} 
this being consistent with our earlier notation. 

As in the introduction, when $\bmalpha,\bmbeta,\bmgamma$ are constant, we just write `$\alpha$' in place of `$\bmalpha$,' `$\beta$' in place of `$\bmbeta$,' and `$\gamma$' in place of `$\bmgamma$.' 
Let $U_{\ell,m,n}[\bullet]$ denote the set of $(\alpha,\beta,\gamma)\in \bbC^3$ such that $(\bmalpha,\bmbeta,\bmgamma) \in \Omega_{\ell,m,n}[\bullet]$ holds when $\bmalpha=\alpha$, $\bmbeta=\beta$, and $\bmgamma=\gamma$.

Similar abbreviations will be used throughout the rest of this paper. 

In addition to the general Selberg-like integral above, we have the following general integral of Dotsenko--Fateev type:  
\begin{equation}
I_{\ell,m,n}[F](\bmalpha,\bmbeta,\bmgamma) = \int_{\square_{\ell,m,n}} \Big[ \prod_{i=1}^N |x_i|^{\alpha_i}|1-x_i|^{\beta_i}  \Big] \Big[ \prod_{1\leq j<k \leq N} (x_k-x_j+i0)^{2\gamma_{j,k} } \Big] F  \dd x_1\cdots \dd x_N
\label{eq:misc_inf}
\end{equation}
for $(\bmalpha,\bmbeta,\bmgamma) \in V_{\ell,m,n}[\calD]$, 
where now $\calD$ denotes a finite collection of indexed sets $\{d_{\mathrm{F}}\}_{\mathrm{F} \in \calF(A_{\ell,m,n})}\subseteq \bbC$, 
\begin{itemize}
	\item $V_{\ell,m,n}[\calD]$ denotes the set of $(\bmalpha,\bmbeta,\bmgamma) \in \bbC^{2N+N(N-1)/2}$ for which the integrand in \cref{eq:misc_inf}  lies in $L^1(\square_{\ell,m,n}, \dd x_1\cdots \dd x_N)$ -- that is the set of $(\bmalpha,\bmbeta,\bmgamma)$ such that 
	\begin{equation}
	\Big[ \prod_{i=1}^N |x_i|^{\alpha_i}|1-x_i|^{\beta_i}  \Big] \Big[ \prod_{1\leq j<k \leq N} |x_k-x_j|^{2\gamma_{j,k}} \Big] \Big[\prod_{\mathrm{F} \in \calF (K_{\ell,m,n}) } x_{\mathrm{F}}^{d_{\mathrm{F}}}\Big]  \in L^1(\square_{\ell,m,n},\dd x_1\cdots \dd x_N )
	\end{equation}
	for all $\{d_{\mathrm{F}} \}_{\mathrm{F} \in \calF(A_{\ell,m,n}) } \in \calD$, and 
	\item $F$ has the form \cref{eq:misc_uyu} for $F_{\{d_{\mathrm{F}}\}_{\mathrm{F} \in \calF(A_{\ell,m,n}) }}\in C^\infty(A_{\ell,m,n})$. 
\end{itemize}
In \cref{eq:misc_inf}, $x^\gamma = e^{\pi i \gamma} e^{\gamma \log |x|}$ if $x<0$ and $x^{\gamma} = e^{\gamma \log x}$ if $x>0$. 
We apply abbreviations for Dotsenko--Fateev-like integrals that are analogous to those used for Selberg-like integrals.  

Let $W_{\ell,m,n}[\bullet]$ denote the set of $(\alpha,\beta,\gamma)\in \bbC^3$ such that $(\bmalpha,\bmbeta,\bmgamma) \in V_{\ell,m,n}[\bullet]$ holds when $\bmalpha=\alpha$, $\bmbeta=\beta$, and $\bmgamma=\gamma$. 
Let $W_{\ell,m,n}^{\mathrm{DF0}}[F]$ denote the set of $(\alpha_-,\alpha_+,\beta_-,\beta_+,\gamma_-,\gamma_0,\gamma_+) \in \bbC^7$ such that $(\bmalpha^{\mathrm{DF0}},\bmbeta^{\mathrm{DF0}},\bmgamma^{\mathrm{DF0}}) \in V_{\ell,m,n}[F]$. 
Let 
\begin{equation}
I^{\mathrm{DF0};\mathtt{S}}_{\ell,m,n}[F](\alpha_-,\alpha_+,\beta_-,\beta_+,\gamma_-,\gamma_0,\gamma_+) = I_{\ell,m,n}[F](\bmalpha^{\mathrm{DF0}},\bmbeta^{\mathrm{DF0}},\bmgamma^{\mathrm{DF0}}).
\label{eq:misc_xb1}
\end{equation}

This section is split into many short subsections. 
The general analytic framework in which the extension is performed is discussed in \S\ref{subsec:generalities}, and the specific application to Selberg-like integrals is contained in \S\ref{subsec:specifics}. We prove a family of identities relating $I_{\ell,m,n},I_{\ell,n,m},I_{n,\ell,m},\cdots$ in \S\ref{subsec:simple}.
As preparation for our discussion of singularity removal in the DF-symmetric case, we discuss in \S\ref{subsec:alternative} an alternative regularization procedure suggested by Dotsenko--Fateev that works for some suboptimal range of parameters (in particular allowing $\gamma_0=-1$, but not allowing the real parts of $\alpha_-,\alpha_+,\beta_-,\beta_+$ to be too negative). 
It should be remarked that this regularization technique can be combined with that in \S\ref{subsec:generalities} to yield proofs of the main theorems without the technicalities associated with needing to understand the analyticity of products of distributions like $(y \pm i0)^\lambda$ in $\lambda$. As this lacks the purely analytic flavor of the proof in \S\ref{subsec:generalities}, it is not the approach we follow here.
The $I_{\ell,m,n}$ are related to the Selberg-like integrals $S_{\ell,m,n}$ in \S\ref{subsec:symmetrization}.
A key lemma used in the removal of singularities is in \S\ref{subsec:relations}. This lemma is a generalization of a result proven by Aomoto \cite{Ao} and discussed heuristically by Dotsenko--Fateev \cite{DF2}. For completeness and later convenience, we record in \S\ref{subsec:continuation_symmetric} the symmetric and DF-symmetric cases of the results in \S\ref{subsec:specifics} regarding the Dotsenko--Fateev integrals.

Let $\frakS_{\ell,m,n}=\frakS_\ell\times \frakS_m\times \frakS_n$, which we consider as the subgroup of $\frakS_N$ leaving each of $\calI_1,\calI_2,\calI_3$ invariant, where $\calI_1,\calI_2,\calI_3$ are as in the previous section, a.k.a.\ the Young subgroup associated with the partition $\{1,\ldots,n\}= \calI_1\sqcup \calI_2 \sqcup \calI_3$. 
Given a permutation $\sigma \in \frakS_{\ell,m,n}$, let 
\begin{equation}
I_{\ell,m,n}[F](\bmalpha,\bmbeta,\bmgamma)^\sigma = \int_{\square_{\ell,m,n}} \Big[ \prod_{i=1}^N |x_i|^{\alpha_i}|1-x_i|^{\beta_i}  \Big] \Big[ \prod_{1\leq j<k \leq N} (x_{\sigma(k)}-x_{\sigma(j)}+i0)^{2\gamma_{\sigma(j),\sigma(k)} } \Big] F  \dd x_1\cdots \dd x_N, 
\end{equation}
defined for $(\bmalpha,\bmbeta,\bmgamma) \in V_{\ell,m,n}[F]$.
If we define $\bmalpha^\sigma,\bmbeta^\sigma,\bmgamma^\sigma$ by $\alpha_j^{\sigma} = \alpha_{\sigma(j)}$, $\beta_j^\sigma = \beta_{\sigma(j)}$, and $\gamma_{j,k}^{\sigma} = \gamma_{\sigma(j),\sigma(k)}$, and 
\begin{equation}
F^\sigma(y_1,\ldots,y_N) = F(y_{\sigma^{-1}(1)},\ldots,y_{\sigma^{-1}(N)} ),
\end{equation}
then $I_{\ell,m,n}[F](\bmalpha,\bmbeta,\bmgamma)^\sigma = I_{\ell,m,n}[F^\sigma](\bmalpha^\sigma,\bmbeta^\sigma,\bmgamma^\sigma)$. This relation will be very useful below. More generally, for any $\sigma \in \frakS_N$, let  
\begin{align}
I_{\ell,m,n}[F](\bmalpha,\bmbeta,\bmgamma)^\sigma &=
I_{\ell,m,n}[F^\sigma](\bmalpha^\sigma,\bmbeta^\sigma,\bmgamma^\sigma) \\
S_{\ell,m,n}[F](\bmalpha,\bmbeta,\bmgamma)^\sigma &=
S_{\ell,m,n}[F^\sigma](\bmalpha^\sigma,\bmbeta^\sigma,\bmgamma^\sigma),
\end{align}
defined for $(\bmalpha,\bmbeta,\bmgamma) \in V_{\ell,m,n}[F]$ in the former case or for 
\begin{equation} 
(\bmalpha,\bmbeta,\bmgamma)\in \Omega_{\ell,m,n}[F]^\sigma = \{(\bmalpha,\bmbeta,\bmgamma)\in \bbC^{2N+N(N-1)/2} : (\bmalpha^\sigma,\bmbeta^\sigma,\bmgamma^\sigma) \in \Omega_{\ell,m,n}[F^\sigma] \}
\end{equation}  
in the latter case. We will use similar notation for other subsets of $\bbC^{2N+N(N-1)/2}$ below, as well as for the meromorphic extensions of $S_{\ell,m,n}[F]$ and $I_{\ell,m,n}[F]$.

\subsection{Some generalities}
\label{subsec:generalities}
Let $N\in \bbN$ be arbitrary. 
For a Fr\'echet space $\calX$, let $ \scrO(\bbC^{N};\calX)$ denote the Fr\'echet space  of entire $\calX$-valued functions on $\bbC^N$, where the topology is that of uniform convergence in compact subsets, as measured with respect to each Fr\'echet seminorm on $\calX$, and similarly for $\calX$ an LF-space. 
Let $\calE'(\bbR^N)$ denote the LCTVS of compactly supported distributions on $\bbR^N$. By the Schwartz representation theorem, 
\begin{equation}
\calE'(\bbR^N)=\cup_{m\in \bbR} H_{\mathrm{c}}^{m,s}(\bbR^N),
\end{equation} 
where $H_{\mathrm{c}}^{m}(\bbR^N)$ is the set of compactly supported elements of $H^{m}(\bbR^N)$.

Let $N\in \bbN^+$, $k\in \{0,\ldots,N\}$, and $\kappa \in \bbN$. 
For any 
\begin{equation} 
\psi \in C_{\mathrm{c}}^\infty(\bbR^k_{t_1,\cdots,t_k} ; \calE'(\bbR^{N-k}_{t_{k+1},\cdots,t_N})) = \bigcup_{m,s\in \bbR} C_{\mathrm{c}}^\infty(\bbR^k_{t_1,\cdots,t_k} ; H_{\mathrm{sc,c}}^{m,s}(\bbR^{N-k}) )
\label{eq:misc_psi}
\end{equation} 
let, for $\bmrho = (\rho_1,\ldots,\rho_k)$,  
\begin{equation}
I_{N,k,\kappa}[\psi](\bmrho) = \int_0^\infty \cdots \int_0^\infty t_1^{\rho_1}\cdots t_k^{\rho_k} \langle 1,\psi(t_1,\ldots,t_k,-) \rangle \dd t_1\cdots \dd t_k, 
\label{eq:misc_ink}
\end{equation}
which we abbreviate as 
\begin{equation}
I_{N,k,\kappa}[\psi](\bmrho) = \int_{\bbR^N_k} t_1^{\rho_1}\cdots t_k^{\rho_k}  \psi(t) \dd^N t.
\end{equation} 
Here, $\bbR^N_k = [0,\infty)^k_{t_1,\cdots,t_k} \times \bbR^{n-k}_{t_{k+1},\cdots,t_N}$, and $I_{N,k,\kappa}[\psi](\bmrho)$ is defined initially for $\Re \rho_1, \cdots ,\Re \rho_k> -1$, for which the right-hand side of \cref{eq:misc_ink} is a well-defined integral. 

Let $H_{\mathrm{sc,c}}^{m,s}(\bbR^N)$ denote the set of compactly supported elements of $H_{\mathrm{sc}}^{m,s}(\bbR^N) = \langle r \rangle^{-s} H^m(\bbR^N)$.  
Let 
\begin{equation}
\scrO(\bbC^k\times \bbC^\kappa; C_{\mathrm{c}}^\infty(\bbR^k_{t_1,\cdots,t_k} ; \calE'(\bbR^{N-k}_{t_{k+1},\cdots,t_N}))) = \bigcap_{\Omega} \bigcup_{m,s\in \bbR} 	\scrO(\Omega;C_{\mathrm{c}}^\infty(\bbR^k_{t_1,\cdots,t_k} ; H_{\mathrm{sc,c}}^{m,s}(\bbR^{N-k}) )),
\label{eq:misc_162} 
\end{equation}
endowed with the strongest topology such that the inclusions 
\begin{equation} 
\bigcap_{\Omega} 
\scrO(\Omega;C_{\mathrm{c}}^\infty(\bbR^k_{t_1,\cdots,t_k} ; H_{\mathrm{sc,c}}^{m,s}(\bbR^{N-k}) )) \hookrightarrow 	\scrO(\bbC^k\times \bbC^\kappa; C_{\mathrm{c}}^\infty(\bbR^k_{t_1,\cdots,t_k} ; \calE'(\bbR^{N-k}_{t_{k+1},\cdots,t_N})))
\end{equation} 
are all continuous, where the left-hand side is an LF space. 
Here, $\Omega$ is varying over bounded domains in $\bbC^k\times \bbC^\kappa$. We are identifying functions on $\bbC^{k}\times \bbC^\kappa$ with their restrictions to subdomains. 
In other words, an element of the space defined by \cref{eq:misc_162} is locally an analytic family of elements of  $C_{\mathrm{c}}^\infty(\bbR^k_{t_1,\cdots,t_k} ; H_{\mathrm{sc,c}}^{m,s}(\bbR^{N-k}) )$ for some $m,s\in \bbR$ which are allowed to depend on $\Omega$.

\begin{proposition}
	Suppose that, for each $\bmrho \in \bbC^k$ and $\bmdelta \in \bbC^\kappa$, we are given some $\psi(-;\bmrho,\bmdelta)$ as in \cref{eq:misc_psi}, 
	depending entirely on $\bmrho,\bmdelta$ in the sense that the map 
	\begin{equation} 
	\bbC^k\times \bbC^\kappa \ni (\bmrho,\bmdelta) \mapsto \psi \in C_{\mathrm{c}}^\infty(\bbR^k;\calE'(\bbR^{N-k}))
	\end{equation} 
	is entire, i.e.\ lies in $\scrO(\bbC^k\times \bbC^\kappa; C_{\mathrm{c}}^\infty(\bbR^k_{t_1,\cdots,t_k} ; \calE'(\bbR^{N-k}_{t_{k+1},\cdots,t_N})))$. 
	Define 
	\begin{equation} 
	I_{N,k,\kappa}[\psi](\bmrho,\bmdelta) = I_{N,k,\kappa}[\psi(\bmrho,\bmdelta)](\bmrho).
	\end{equation}  
	Then, the function $J_{N,k,\kappa}[\psi]$ defined by 
	\begin{equation}
	I_{N,k,\kappa}[\psi](\bmrho,\bmdelta) = \Big[ \prod_{j=1}^k\Gamma(\rho_j+1) \Big] J_{N,k,\kappa}[\psi](\bmrho,\bmdelta)
	\label{eq:misc_n31}
	\end{equation}
	extends to an entire function on $\bbC_{\bmrho}^k\times \bbC_{\bmdelta}^\kappa$. 
	Moreover, the function
	\begin{equation} 
	J_{N,k,\kappa}[-]:\scrO(\bbC^{k}\times \bbC^\kappa ;C^\infty_{\mathrm{c}}( \bbR^k_{t_1,\ldots,t_k}; \calE'(\bbR^{N-k}_{t_{k+1},\ldots,t_N} ))) \ni 
	\psi
	\mapsto J_{N,k,\kappa}[\psi ]  \in \scrO(\bbC^{k}\times \bbC^\kappa) 
	\end{equation}
	is continuous. 
	\label{prop:basic_local}
\end{proposition}

Cf.\ \cite{GS}\cite[Lemma 10.7.9]{Varchenko}.

\begin{proof}
	
	The $k=0$ case is essentially tautologous.
	
	We now proceed inductively on $k$. Let $k\geq 1$, and assume that we have proven the result for smaller $k$. Expanding $\psi$ in Taylor series around $t_1=0$,  there exist
	\begin{align}
	\psi^{(j)} &\in \scrO\left(\bbC^{k}\times \bbC^\kappa ;C^\infty_{\mathrm{c}}\big( \bbR^{k-1}_{t_2,\ldots,t_k}; \calE'(\bbR^{N-k}_{t_{k+1},\ldots,t_N} )\big)\right) \\
	E^{(j)} &\in \scrO\left(\bbC^{k}\times \bbC^\kappa ;C^\infty\big(\bbR_{t_1};C^\infty_{\mathrm{c}}( \bbR^{k-1}_{t_2,\ldots,t_k}; \calE'(\bbR^{N-k}_{t_{k+1},\ldots,t_N}))\big)\right),
	\end{align}	
	which can be regarded as smooth functions (or generalized functions) of $t_1,\ldots,t_N$, depending analytically on parameters $\bmrho \in \bbC^k$ and $\bmdelta \in \bbC^\kappa$. 
	such that 
	\begin{equation}
	\psi(t_1,\cdots,t_N;\bmrho,\bmdelta) = \sum_{j=0}^J  t_1^j  \psi^{(j)}(t_2,\cdots,t_{N};\bmrho,\bmdelta) + t_1^{J+1} E^{(J+1)}(t_1,\cdots,t_N;\bmrho,\bmdelta)
	\label{eq:misc_taz}
	\end{equation}
	for all $J\in \bbN$. 
	Let $K \subset \bbC^{k+\kappa}$ be an arbitrary nonempty compact set. There exists some $T>0$ such that $\operatorname{supp} \psi(-;\bmrho,\bmdelta)\subseteq \{-T\leq t_1 \leq T\}$ for all $(\bmrho,\bmdelta)\in K$. Then, if $\Re \rho_1,\cdots,\Re \rho_k>-1$ and $(\bmrho,\bmdelta)\in K$, 
	\begin{equation}
	I_{N,k,\kappa}[\psi](\bmrho,\bmdelta) = \sum_{j=0}^J \frac{ I_{N-1,k-1,\kappa}[\psi^{(j)}](\hat{\bmrho}, \bmdelta)}{\rho_1+j+1} T^{\rho_1+j+1} 
	+  \int_0^T t_1^{\rho_1+J+1} I_{N-1,k-1}[E^{(J+1)}(t_1,-)](\hat{\bmrho},\bmdelta) \dd t_1,
	\label{eq:misc_861}
	\end{equation}
	where $\hat{\bmrho} = (\rho_2,\cdots,\rho_k)$.
	We now \emph{define} $J_{N,k,\kappa}[\psi](\bmrho,\bmdelta): \{\Re \rho_1 >-2-J\}\times \bbC^\kappa_{\bmdelta}\to \bbC $ by 
	\begin{multline}
	J_{N,k,\kappa}[\psi](\bmrho,\bmdelta) = \frac{1}{\Gamma(\rho_1+1)}\sum_{j=0}^J \frac{ J_{N-1,k-1,\kappa}[\psi^{(j)}](\hat{\bmrho}, \bmdelta)}{\rho_1+j+1} T^{\rho_1+j+1} 
	\\ +  \frac{1}{\Gamma(\rho_1+1)}\int_0^T t_1^{\rho_1+J+1} J_{N-1,k-1}[E^{(J+1)}(t_1,-)](\hat{\bmrho},\bmdelta) \dd t_1.
	\label{eq:misc_862}
	\end{multline}
	By construction, \cref{eq:misc_n31} holds when $\Re \rho_1,\cdots,\Re \rho_k>-1$. By the continuity clause of the inductive hypothesis, the integral in \cref{eq:misc_861} is a well-defined Bochner integral, for each individual $(\bmrho,\bmdelta) \in  \{\Re \rho_1 >-2-J\}\times \bbC^\kappa$. 
	Moreover, the right-hand side of \cref{eq:misc_862} depends analytically on $(\bmrho,\bmdelta) \in \{\Re \rho_1 >-1-J\}\times \bbC^\kappa$. By the inductive hypothesis, this is true for the sum on the first line (multiplied by $\Gamma(\rho_1+1)^{-1}$), as the simple poles due to the factors of $1/(\rho_1+j+1)$ cancel with those of $\Gamma(\rho_1+1)$. So, in order to show that the whole right-hand side of \cref{eq:misc_862} depends analytically on $(\bmrho,\bmdelta)$ in this domain, we can show it for 
	\begin{equation}
	\int_0^T t_1^{\rho_1+J+1} J_{N-1,k-1}[E^{(J+1)}(t_1,-)](\hat{\bmrho},\bmdelta) \dd t_1.
	\label{eq:misc_543}
	\end{equation}
	Justifying differentiation under the integral sign, this is a $C^1$-function of $(\Re \rho_1,\Im \rho_1) \in \{(u,v)\in \bbR^2,u> -1-J\}$, and it satisfies the Cauchy--Riemann equations, so it follows that the integral in \cref{eq:misc_543} is analytic as a function of $\rho_1 \in \{\Re \rho_1 >-1-J\}$, for each fixed $\hat{\bmrho}\in \bbC^{k-1}$ and $\bmdelta \in \bbC^\kappa$. Adding $\hat{\bmrho},\bmdelta$-dependence does not change the argument.
	
	So, the formula \cref{eq:misc_861} yields an analytic extension of $I_{N,k,\kappa}$, and we can take a union over all $J \in \bbN$, the various partial extensions agreeing with each other via analyticity. 
	The continuity clause is evident from the formula \cref{eq:misc_862} and the inductive hypothesis.
\end{proof}

Consequently, $I_{N,k,\kappa}[\psi]$ admits an analytic continuation $\dot{I}_{N,k,\kappa}[\psi]:\Omega\to \bbC$ to the set $\Omega=(\bbC^k_{\bmrho} \backslash \bigcup_{j \in \{1,\ldots,k\}} \{\rho_j \in \bbZ^{\leq -1}\})\times \bbC^\kappa_{\bmdelta}$, and the map
\begin{equation} 
\dot{I}_{N,k,\kappa}[-]:\scrO(\bbC^{k}\times \bbC^\kappa ;C^\infty_{\mathrm{c}}( \bbR^k_{t_1,\ldots,t_k}; \calE'(\bbR^{N-k}_{t_{k+1},\ldots,t_N} ))) \ni 
\psi
\mapsto \dot{I}_{N,k,\kappa}[\psi ]  \in \scrO(\Omega) 
\end{equation} 
is continuous.

If $\calP$ is a consistently orientable collection of codimension-1 interior p-submanifolds on a mwc $M$, then, letting $x_{\mathrm{F}}$ for $\mathrm{F}\in \calF(M)$ denote a bdf of the face $\mathrm{F}$, it is the case that, for any $\bmdelta \in \bbC^\calP$ and $\bmrho\in \bbC^{\calF(M)}$, the product 
\begin{equation}
\omega(\bmrho,\bmdelta)=	\prod_{\mathrm{F}\in \calF(M)} x_{\mathrm{F}}^{\rho_{\mathrm{F}}}\prod_{P\in \calP} (y_P + i0)^{\delta_P}:  \dot{C}^\infty_{\mathrm{c}}(M;\Omega M) \ni \mu \mapsto  \lim_{\varepsilon \to 0^+} \int_M  \prod_{\mathrm{F}\in \calF(M)}\prod_{P\in \calP} x_{\mathrm{F}}^{\rho_{\mathrm{F}}}  (y_P + i\varepsilon)^{\delta_P}  \mu 
\label{eq:misc_j78}
\end{equation}
is a well-defined classical distribution on $M$, where $\{y_P\}_{P\in \calP}$ are consistently oriented defining functions. 
(Here, $\dot{C}^\infty_{\mathrm{c}}(M;\Omega M)$ is the set of compactly supported smooth densities on $M$ that are Schwartz at each boundary hypersurface.)
That is, $\omega$ is an extendable distribution on $M$ and defines, for small $\epsilon>0$, an element of $C^\infty([0,\epsilon)_{x_{\mathrm{F}}} ; \calD'(\mathrm{F}))$ for each face $\mathrm{F}$.
We write the right-hand side of \cref{eq:misc_j78} as $\int_M \omega(\bmrho,\bmdelta) \mu$. 
More generally, if $\mu \in C_{\mathrm{c}}^\infty(M;\Omega M)$, then 
\begin{equation}
\lim_{\varepsilon \to 0^+ } \int_M  \prod_{\mathrm{F}\in \calF(M)}\prod_{P\in \calP} x_{\mathrm{F}}^{\rho_{\mathrm{F}}}  (y_P + i\varepsilon)^{\delta_P}  \mu = \int_M \omega(\bmrho,\bmdelta) \mu 
\end{equation}
exists whenever $\rho_{\mathrm{F}}>-1$ for all $\mathrm{F}\in \calF(M)$.  

Let $\varkappa\in \bbN$.
Suppose that we are given some entire family 
\begin{equation} 
\mu: \bbC^{\calF(M)}\times \bbC^\calP \times \bbC^{\varkappa}\to  C^\infty_{\mathrm{c}}(M;\Omega M)
\end{equation} 
of compactly supported smooth densities $\mu(\bmrho,\bmdelta,\bmlambda) \in C^\infty_{\mathrm{c}}(M;\Omega M)$ on $M$.  Consider the function 
\begin{equation} 
I[M,\mu](\bmrho,\bmdelta,\bmlambda):\{ (\bmrho,\bmdelta,\bmlambda) \in \bbC^{\calF(M)}\times \bbC^\calP\times \bbC^\varkappa : \rho_{\mathrm{F}}>-1 \text{  for all } \mathrm{F}\in \calF(M)\} \to \bbC
\end{equation} 
defined by 
\begin{equation} 
I[M,\mu](\bmrho,\bmdelta,\bmlambda)  = \int_M \omega(\bmrho,\bmdelta) \mu(\bmrho,\bmdelta,\bmlambda). \label{eq:misc_k99}
\end{equation}

\begin{proposition}
	Suppose that, for some $N_0\in \bbN^+$, we are given an affine map $L = (L_1,L_2,L_3):\smash{\bbC^{N_0}_{\bmvarrho}}\to \smash{\bbC^{\calF(M)}_{\bmrho}}\times\bbC^{\calP}_{\bmdelta}\times \bbC^\varkappa_{\bmlambda}$ such that, for each $\mathrm{F}\in \calF(M)$, the affine functional
	\begin{equation} 
	(L\bullet)_{\mathrm{F}}: \bbC^{N_0}\ni \bmvarrho \mapsto (L_1\bmvarrho)_{\mathrm{F}} \in \bbC
	\end{equation} 
	is nonconstant. 
	Then, 
	there exist entire functions $I_{\mathrm{reg},\mathrm{f}}[M,\mu](L\bullet): \bbC^{N_0}_{\bmvarrho} \to \bbC$ associated to the minimal facets $\mathrm{f}$  of $M$ such that 
	\begin{equation}
	I[M,\mu](L\bmvarrho) = \sum_{\mathrm{f}} \Big[ \prod_{\mathrm{F}\in \calF(M), \mathrm{F}\supseteq \mathrm{f}} \Gamma(1+ (L\varrho)_{\mathrm{F}} ) \Big] I_{\mathrm{reg},\mathrm{f}}[M,\mu](L\bmvarrho) 
	\label{eq:misc_imu}
	\end{equation}
	for all $\bmvarrho\in \bbC^{N_0}$ for which the left-hand side is defined by \cref{eq:misc_k99}.
	\label{prop:optimized_mwc}
\end{proposition}
\begin{proof}
	Pass to a partition of unity subordinate to a system of coordinate charts on $M$ and apply \Cref{prop:basic_local} locally. 
\end{proof}

Then, letting $\calL = \{ (L\bullet)_{\mathrm{F}}  :\mathrm{F}\in \calF(M)\}$, 
\begin{equation}
\Big[ \prod_{\Lambda \in \calL} \frac{1}{\Gamma(1+\Lambda (\bmvarrho))^{\#_\Lambda}}   \Big]	I[M,\mu](L\bmvarrho)
\end{equation}
extends to an entire function $\bbC_{\bmvarrho}^{N_0}\to \bbC$, where $\#_\Lambda\in \bbN^+$ is the maximum size of any set $S\subseteq \calF(M)$ of faces such that $\cap_{\mathrm{F}\in S} \mathrm{F} \neq \varnothing$ and  $(L\bullet)_{\mathrm{F}}=\Lambda$ for all $\mathrm{F}\in S$.
Indeed, this follows from the proposition above since, for each facet $\mathrm{f}$, 
\begin{equation} 
\Big[\prod_{\Lambda \in \calL}  \frac{1}{\Gamma(1+\Lambda (\bmvarrho))^{\#_\Lambda}} \Big] \prod_{\mathrm{F}\in \calF(M), \mathrm{F}\supseteq \mathrm{f}} \Gamma(1+ (L\bmvarrho)_{\mathrm{F}} )
\end{equation} 
is entire.

\subsection{Specialization to generic Selberg- and DF-like integrals}
\label{subsec:specifics}

We now apply the results of the previous section to the specific case of the integrals \cref{eq:misc_snf} and \cref{eq:misc_inf}. Fix $\ell,m,n \in \bbN$ satisfying $\ell+m+n=N$, $N\in \bbN^+$.

\subsubsection{The Selberg case}

Fix $F\in \calA^\calD(K_{\ell,m,n})$. 
Let $\rho_{j,k} = \rho_{j,k}(\bmalpha,\bmbeta,\bmgamma)$ be defined by \cref{eq:misc_rt1}, \cref{eq:misc_rt2}, \cref{eq:misc_rt3}, and \cref{eq:misc_rt4}. Recalling the definition of $\mathtt{T}(\ell,m,n)$ given in \S\ref{subsec:K}:
\begin{proposition} 
	There exist entire functions  
	\begin{equation} 
	S_{\ell,m,n;\mathrm{reg}, \mathtt{I}, \{d_{\mathrm{F}} \}_{\mathrm{F} \in \calF(K_{\ell,m,n}) } }[F]: \bbC^{2N+N(N-1)/2}_{\bmalpha,\bmbeta,\bmgamma} \to \bbC,
	\end{equation}  
	associated to pairs of minimal facets $\mathrm{f}$ of $K_{\ell,m,n}$ and collections $\{d_{\mathrm{F}} \}_{\mathrm{F} \in \calF(_{\ell,m,n}) }\in \calD$ of weights 
	such that 
	\begin{multline}
	S_{\ell,m,n}[F](\bmalpha,\bmbeta,\bmgamma)  = \sum_{\mathtt{I}\in \mathtt{T}(\ell,m,n)} \sum_{\{d_{\mathrm{F}} \}_{\mathrm{F}} \in \calF(K_{\ell,m,n})  \in \calD}  \Big[\prod_{ \calI(j,k)\in \mathtt{I}}\Gamma(1+\rho_{j,k}+d_{\mathrm{F}_{j,k}}) \Big] \\ \times S_{\ell,m,n;\mathrm{reg}, \mathtt{I}, \{d_{\mathrm{F}} \}_{\mathrm{F}} \in \calF(K_{\ell,m,n}) }[F](\bmalpha,\bmbeta,\bmgamma)
	\end{multline}
	for all $(\bmalpha,\bmbeta,\bmgamma) \in \Omega_{\ell,m,n}[\calD]$.
	\label{prop:genext}
\end{proposition}
\begin{proof}
	This is a corollary of \Cref{prop:Kres} and \Cref{prop:optimized_mwc}, using the fact that the minimal facets of $K_{\ell,m,n}$ are in correspondence with the elements of $\mathtt{T}(\ell,m,n)$ via \cref{eq:misc_rmf}.
\end{proof}

Consequently, there exists an analytic extension $\dot{S}_{\ell,m,n}[F]:\dot{\Omega}_{\ell,m,n}[\calD]\to \bbC$ of $S_{\ell,m,n}[F]:\Omega_{\ell,m,n}[\calD]\to \bbC$, where 
\begin{equation}
\dot{\Omega}_{\ell,m,n}[\calD] =   \bbC^{2N+N(N-1)/2}_{\bmalpha,\bmbeta,\bmgamma} \Big\backslash  \Big[\bigcup_{\{d_{\mathrm{F}} \}_{\mathrm{F}  \in \calF(K_{\ell,m,n})} \in \calD} \Big(\bigcup_{\{j,k\} \in \calJ_{\ell,m,n}} \\ \{ \rho_{j,k} +d_{\mathrm{F}_{j,k}} \in \bbZ^{\leq -1} \}\Big)\Big].
\end{equation}
This is an open and connected subset of full measure; namely, it is the complement of a locally finite collection of complex (affine) hyperplanes in $\bbC^{2N+N(N-1)/2}$. In the case $m=N$, this agrees with \cref{eq:dotOmegaN}.

As a corollary of the previous proposition, there exists an entire function 
\begin{equation} 
S_{\ell,m,n;\mathrm{reg}}[F]:\bbC^{2N+N(N-1)/2}_{\bmalpha,\bmbeta,\bmgamma}\to \bbC
\end{equation} 
such that 
\begin{equation}
S_{\ell,m,n}[F](\bmalpha,\bmbeta,\bmgamma)  = \Big[\prod_{\{j,k\}\in \calJ_{\ell,m,n}  } \Gamma(1+\rho_{j,k}+d_{\mathrm{F}_{j,k}}^{\mathrm{min}} )\Big] S_{\ell,m,n;\mathrm{reg}}[F](\bmalpha,\bmbeta,\bmgamma)
\end{equation}
holds for all $(\bmalpha,\bmbeta,\bmgamma) \in \Omega_{\ell,m,n}[\calD]$, where $d^{\mathrm{min}}_{\mathrm{F}} = \min\{ d_{\mathrm{F}} : \{ d_{\mathrm{F}_0} \}_{\mathrm{F}_0\in \calF(K_{\ell,m,n}) } \in \calD  \}$.

The case of the proposition above where $m=N$ gives \Cref{thm:generic}. Indeed, if $F \in C^\infty(\triangle_{N})$, $F$ lifts to an element of $C^\infty(K_{0,N,0})$, and the orders of vanishing of $F$ at the relevant facets of $\triangle_N$ imply the same order of vanishing at the lift in $K_{0,N,0}$. 

\subsubsection{The Dotsenko--Fateev case}

Fix $F\in \calA^{\calD}(A_{\ell,m,n})$, where $\calD$ is now a collection of orders for the faces of $A_{\ell,m,n}$.  Recalling the definition of $\Sigma\mathtt{T}(\ell,m,n)$ given in \S\ref{subsec:A}:
\begin{proposition}
	There exist entire functions  
	\begin{equation} 
	I_{\ell,m,n;\mathrm{reg}, \mathtt{I}, \{d_{\mathrm{F}} \}_{\mathrm{F} \in \calF(A_{\ell,m,n}) } }[F]: \bbC^{2N+N(N-1)/2}_{\bmalpha,\bmbeta,\bmgamma} \to \bbC
	\end{equation} 
	associated to the $\mathtt{I}\in \Sigma\mathtt{T}(\ell,m,n)$ 
	such that 
	\begin{multline}
	I_{\ell,m,n}[F](\bmalpha,\bmbeta,\bmgamma)  = \sum_{\mathtt{I}\in \Sigma \mathtt{T}(\ell,m,n)} \sum_{\{d_{\mathrm{F}} \}_{\mathrm{F} \in \calF(A_{\ell,m,n})} \in \calD} \Big( \Big[ \prod_{(x_0,\calS) \in \mathtt{I}} \Gamma(1+\varrho_{S,Q;x_0} + d_{\mathrm{F}_{S,Q;x_0}} ) \Big]  \\ \times I_{\ell,m,n;\mathrm{reg},\mathtt{I} \{d_{\mathrm{F}} \}_{\mathrm{F} \in \calF(A_{\ell,m,n})}}[F](\bmalpha,\bmbeta,\bmgamma)\Big)
	\end{multline}
	for all $(\bmalpha,\bmbeta,\bmgamma) \in V_{\ell,m,n}[\calD]$, where we have abbreviated $\calI_1 \cap \calS$, $\calI_2 \cap \calS$, and $\calI_3 \cap \calS$ as $S$ or $Q$ as appropriate. 
	\label{prop:Icontinuation}
\end{proposition}
\begin{proof}
	Follows from \Cref{prop:Ares} and \Cref{prop:optimized_mwc}. 
\end{proof}

Consequently, $I_{\ell,m,n}[F]:V_{\ell,m,n}[\calD]\to \bbC$ admits an analytic continuation $\dot{I}_{\ell,m,n}[F]:\dot{V}_{\ell,m,n}[\calD]\to \bbC$, where 
\begin{equation}
\dot{V}_{\ell,m,n}[\calD] = \bbC^{2N+N(N-1)/2}_{\bmalpha,\bmbeta,\bmgamma} \Big\backslash \bigcup_{\{d_{\mathrm{F}}\}_{\mathrm{F} \in \calF(A_{\ell,m,n})} } \bigcup_{x_0 \in \{0,1,\infty\}}\bigcup_{S,Q} \{  \varrho_{S,Q;x_0} + d_{\mathrm{F}_{S,Q;x_0}} \in \bbZ^{\leq -1}  \}.
\end{equation}
Note that $\dot{V}_{\ell,m,n}[F] \supseteq \cap_{\sigma \in \frakS_{\ell,m,n}} \dot{\Omega}_{\ell,m,n}[F]^\sigma$, as every functional $(\bmalpha,\bmbeta,\bmgamma)\mapsto \varrho_{S,Q;x_0}(\bmalpha,\bmbeta,\bmgamma)$ has the form $\rho_{j,k}(\bmalpha^\sigma,\bmbeta^\sigma,\bmgamma^\sigma)$ for some $\sigma \in \frakS_{\ell,m,n}$ and $\{j,k\}\in \calJ_{\ell,m,n}$. 

As a corollary of the previous proposition, there exists a function 
\begin{equation} 
I_{\ell,m,n;\mathrm{reg}}[F]:\bbC^{2N+N(N-1)/2}_{\bmalpha,\bmbeta,\bmgamma} \to \bbC
\end{equation} 
such that, for all $(\bmalpha,\bmbeta,\bmgamma) \in V_{\ell,m,n}[\calD]$, 
\begin{equation}
I_{\ell,m,n}[F](\bmalpha,\bmbeta,\bmgamma) = \Big[ \prod_{x_0 \in \{0,1,\infty\}} \prod_{S,Q} \Gamma(1+\varrho_{S,Q;x_0} + d_{\mathrm{F}_{S,Q;x_0}}^{\mathrm{min}} ) \Big] I_{\ell,m,n;\mathrm{reg}}[F](\bmalpha,\bmbeta,\bmgamma), 
\end{equation}
where $S,Q$ vary over subsets of $\calI_1=\{1,\ldots,\ell\}$, $\calI_2=\{\ell+1,\ldots,\ell+m\}$, and $\calI_3=\{\ell+m+1,\ldots,N\}$, depending on $x_0$.  

The $m=N$ case of the previous proposition is \Cref{thm:Imain}.

\subsection{A simple identity}
\label{subsec:simple}

For each permutation $\sigma$ of $\{0,1,\infty\}$. Let 
\begin{equation}
(\ell',m',n') = 
\begin{cases}
(\ell,m,n) & (\sigma = 1), \\ 
(n,m,\ell) & (\sigma = (0\; 1) ), \\ 
(\ell,n,m) & (\sigma = (0\; \infty) ), \\ 
(m,\ell,n) & (\sigma = (1\; \infty) ), \\
(n,\ell,m) & (\sigma = (0\;1\; \infty) ), \\
(m,n,\ell) & (\sigma = (1\;0\; \infty) ).
\end{cases}
\end{equation}
In other words, if the elements of $\{0,1,\infty\}$ label the vertices of a triangle and the edges are labeled accordingly -- that is, `$\ell$' labels the edge between $0$ and $\infty$, `$m$' labels the edge between $0$ and $1$, and `$n$' labels the edge between $1$ and $\infty$ -- then $(\ell',m',n')$ is the permutation of $(\ell,m,n)$ resulting from applying $\sigma$ to the triangle and reading off the new labels.

Let $\mathsf{T}_\sigma:\bbC P^1\to \bbC P^1$ denote the unique automorphism acting on $\{0,1,\infty\}$ via $\sigma$.
These are 
\begin{equation}
\mathsf{T}_1(z) =z, \quad \mathsf{T}_{(0\; 1)}(z) = 1-z, \quad \mathsf{T}_{(0\;\infty)}(z) = \frac{1}{z}, \quad \mathsf{T}_{(1\;\infty)}(z) = -\frac{z}{1-z},
\end{equation}
\begin{equation}
\mathsf{T}_{(0\;1\;\infty)}(z) = \frac{1}{1-z},  \qquad \mathsf{T}_{(0\;\infty\;1)}(z) = \frac{z-1}{z}. 
\end{equation}
Let $\sigma^{\mathrm{param}}:\bbC^{2N+N(N-1)/2}\to \bbC^{2N+N(N-1)/2}$ denote the affine map 
\begin{equation}
\sigma^{\mathrm{param}}(\bmalpha,\bmbeta,\bmgamma) = 
\begin{cases}
(\bmalpha,\bmbeta,\bmgamma) & (\sigma = 1), \\ 
(\bmbeta,\bmalpha,\bmgamma) & (\sigma = (0\; 1) ), \\ 
(-2-\bmalpha-\bmbeta-2 \bmgamma\lrcorner\bf1,\bmbeta,\bmgamma) & (\sigma = (0\; \infty) ), \\ 
(\bmalpha, -2- \bmalpha - \bmbeta -2\bmgamma\lrcorner\bf1, \bmgamma ) & (\sigma = (1\; \infty) ), \\
(-2-\bmalpha-\bmbeta-2 \bmgamma\lrcorner\bf1,\bmalpha,\bmgamma) & (\sigma = (0\;1\; \infty) ), \\
(\bmbeta,-2-\bmalpha-\bmbeta-2 \bmgamma\lrcorner\bf1,\bmgamma) & (\sigma = (1\;0\; \infty) ),
\end{cases}
\end{equation}
where $\bmgamma\lrcorner{\bf1}\in \bbC^N$ has $j$th component $\sum_{k\neq j} \gamma_{j,k}$. 
Let $\mathrm{rev}\in \frakS_{\ell',m',n'}$ denote the permutation that reverses the order of the elements in each of the sets $\{1,\ldots,\ell'\}$, $\{\ell'+1,\ldots,\ell'+m'\}$, and $\{\ell'+m'+1,\ldots,N\}$. Let $|\sigma|$ denote the order of $\sigma$. 
\begin{proposition}
	If $(\bmalpha,\bmbeta,\bmgamma)\in \dot{V}_{\ell,m,n}$, then $\sigma^{\mathrm{param}}(\bmalpha,\bmbeta,\bmgamma) \in \dot{V}_{\ell',m',n'}$, and if $(\bmalpha,\bmbeta,\bmgamma)\in \dot{\Omega}_{\ell,m,n}$, then $\sigma^{\mathrm{param}}(\bmalpha,\bmbeta,\bmgamma) \in \dot{\Omega}_{\ell',m',n'}^{\mathrm{rev}^{|\sigma|}}$, and 
	\begin{align}
	\begin{split} 
	\dot{I}_{\ell,m,n}[1](\bmalpha,\bmbeta,\bmgamma) &= \dot{I}_{\ell',m',n'}[1]( \sigma^{\mathrm{param}} (\bmalpha,\bmbeta,\bmgamma))^{\mathrm{rev}^{|\sigma|}}, \\
	\dot{S}_{\ell,m,n}[1](\bmalpha,\bmbeta,\bmgamma) &= \dot{S}_{\ell',m',n'}[1]( \sigma^{\mathrm{param}} (\bmalpha,\bmbeta,\bmgamma))^{\mathrm{rev}^{|\sigma|}}
	\end{split}
	\label{eq:misc_2z1}
	\end{align}
	for all $(\bmalpha,\bmbeta,\bmgamma)\in \dot{\Omega}_{\ell,m,n}$.
	\label{prop:simple} 
\end{proposition}
\begin{proof}
	It can be checked case-by-case that 
	\begin{equation} 
	\{\varrho_{S,Q;\bullet} \circ \sigma^{\mathrm{param}} : \bullet \in \{0,1,\infty\} ,S,Q\text{ as above}\} = \{\varrho_{S,Q;\bullet} : \bullet \in \{0,1,\infty\} ,S,Q\text{ as above} \},
	\label{eq:misc_75n}
	\end{equation} 
	where on the left-hand side $(S,Q)$ varies over appropriate pairs of subsets of $\{1,\ldots,\ell'\}$, $\{\ell'+1,\ldots,\ell'+m'\}$, and $\{\ell'+m'+1,\ldots,N\}$ and on the right-hand side $(S,Q)$ varies over appropriate pairs of subsets $\{1,\ldots,\ell\}$, $\{\ell+1,\ldots,\ell+m\}$, and $\{\ell+m+1,\ldots,N\}$, depending on $\bullet$. It can be seen from \cref{eq:misc_75n} that 
	\begin{equation} 
	\dot{V}_{\ell,m,n} = (\sigma^{\mathrm{param}})^{-1}( \dot{V}_{\ell',m',n'}).
	\end{equation} 
	The case of $\dot{\Omega}_{\ell,m,n}$ is similar but more complicated. 
	
	\Cref{eq:misc_2z1} can be proven for $(\bmalpha,\bmbeta,\bmgamma)\in \Omega_{\ell,m,n}$ by way of a change-of-variables by substituting $x=\mathsf{T}_{\sigma^{-1}}(y)$. The full result follows via analytic continuation.
\end{proof}

\subsection{An imperfect alternative}
\label{subsec:alternative}

For $\mathtt{I} \in \{(-\infty,0],[0,1],[1,\infty)\}$ and $r>0$, let $\Gamma_{\mathtt{I},\pm,r}:(0,1)\to \bbC$ be defined by 
\begin{equation}
\Gamma_{[0,1],\pm,r}(t) = 
\begin{cases}
t\pm irt & (t\in (0,1/3)), \\
t\pm ir/3 & (t\in [1/3,2/3]), \\
t\pm ir/3 \mp ir(t-2/3) & (t\in (2/3,1)),
\end{cases}
\end{equation}
$\Gamma_{[1,\infty),\pm,r}(t) = \Gamma_{[0,1],\mp,r}(1-t)^{-1}$, and $\Gamma_{(-\infty,0],\pm,r}(t) = 1- \Gamma_{[1,\infty),\mp,r}(1-t)$. Note that the images of these contours are permuted amongst themselves by the transformations $\mathsf{T}_\sigma$ above.

\begin{figure}[h]
	\begin{tikzpicture}[scale=3, decoration={
		markings,
		mark=at position 0.6 with {\arrow[scale=1.5,>=latex]{>}}}]
	\draw[step=.25,gray,thin, dotted] (-1.1,0) grid (2.1,1.1);
	\draw[->, thick, dashed] (0,0) -- (0,1.1) node[above] {$\Im z$};
	\draw[->, thick, dashed] (-1.1,0) -- (2.1,0) node[below] {$\Re z$};
	\draw (0,0) -- (1/3,1/6) -- (2/3,1/6)  -- (1,0);
	\node (a) at (1/2,.27) {$\Gamma_{[0,1],+,1}$};
	\draw (0,0) -- (1/3,2/3);
	\draw (2/3,2/3) -- (1,0) ;
	\draw[postaction={decorate}] (1/3,2/3) -- (2/3,2/3) node[right] {$\Gamma_{[0,1],+,4}$};
	\filldraw[color=black] (0,0) circle (.5pt) node[below] {$0$};
	\filldraw[color=black] (1,0) circle (.5pt) node[below] {$1$};
	\draw[-latex] (3/2,3/4) -- (2,1);
	\draw[latex reversed-] (-1,1) -- (-1/2,3/4);
	\draw plot[smooth,tension=1] coordinates{(1,0)  (15/13,3/26)  (6/5,3/10)};
	\draw plot[smooth,tension=1] coordinates{(0,0)  (-2/13,3/26)  (-1/5,3/10)};
	\draw plot[smooth,tension=1] coordinates{(6/5,3/10)   (18/13,6/13)  (3/2,3/4)};
	\draw plot[smooth,tension=1] coordinates{(-1/5,3/10)   (-5/13,6/13)  (-1/2,3/4)};
	\node (b) at (1.75,1/2) {$\Gamma_{[1,\infty),+,1}$};
	\node (c) at (-.6,1/3) {$\Gamma_{(-\infty,0],+,1}$};
	\end{tikzpicture}
	\caption{The contours $\Gamma_{(-\infty,0],+,1}$, $\Gamma_{[0,1],+,1}$,  $\Gamma_{[0,1],+,4}$, $\Gamma_{[1,\infty),+,1}$. Cf.\ \cite[Figure 16]{DF2}. (For our purposes, the contours drawn by Dotsenko \& Fateev approach $\pm \infty$ with imaginary part too small. This is why our $\Gamma_{\mathtt{I},\pm,r}$ look different for $\mathtt{I}\neq [0,1]$.}
\end{figure}

Suppose that $F\in \bbC[x_1,x_1^{-1},\ldots,x_N,x_N^{-1}]$. 
For any compact $\mathsf{K}\Subset \bbC$ with nonempty interior, let $O=O[F,\mathsf{K}]$ denote the set, which depends on $\ell,m,n\in \bbN$, though we suppress this dependence notationally, of $(\bmalpha,\bmbeta)\in \bbC^{2N}$ such that
\begin{multline}
\int_{\Gamma_{(-\infty,0],+,0}} \cdots \int_{\Gamma_{(-\infty,0],+,\ell-1}} \Big[ \int_{\Gamma_{[0,1],+,0}} \cdots \int_{\Gamma_{[0,1],+,m-1}} \Big[ \int_{\Gamma_{[1,\infty),+,0}} \cdots \int_{\Gamma_{[1,\infty),+,n-1}}  \\
\Big( \prod_{j=1}^N z_j^{\alpha_j}(1-z_j)^{\beta_j}  \Big)  \prod_{1\leq j < k \leq N } (z_k-z_j)^{2\gamma_{j,k}} F_0 \dd z_N\cdots \dd z_{\ell+m+1} \Big] \dd z_{\ell+m}\cdots \dd z_{\ell+1} \Big] \dd z_\ell \cdots \dd z_1 
\label{eq:misc_mib}
\end{multline}
is an absolutely convergent Lebesgue integral whenever $\gamma_{j,k} \in \mathsf{K}$ for all $j,k \in \{1,\ldots,N\}$ with $j<k$, for every monomial $F_0$ in $F$. 
In the definition of the integral above we are defining the integrand such that the branch cuts are not encountered.
For such $(\bmalpha,\bmbeta,\bmgamma)$, 
\begin{equation} 
(\bmalpha,\bmbeta,\bmgamma)\in \dot{V}_{\ell,m,n}[F],
\end{equation} 
and the integral in \cref{eq:misc_mib} is equal to $\dot{I}_{\ell,m,n}(\bmalpha,\bmbeta,\bmgamma)[F]$, assuming that we choose our branches appropriately. The latter part of this statement can be proven by checking that the integral defined above depends analytically on its parameters and agrees with $I_{\ell,m,n}(\bmalpha,\bmbeta,\bmgamma)[F]$ for $(\bmalpha,\bmbeta,\bmgamma)\in V_{\ell,m,n}[F]$, which in turn is proven via a contour deformation argument. 

The set $O$ is nonempty, open, and contains an affine cone. 
If 
\begin{itemize}
	\item $\alpha_j$ has sufficiently large real part for $j\in \calI_1 \cup \calI_2$ and sufficiently negative real part for $j\in \calI_3$, and 
	\item $\beta_j$ has sufficiently large real part for $j\in \calI_2\cup \calI_3$ and sufficiently negative real part for $j\in \calI_1$, 
\end{itemize}
then $(\bmalpha,\bmbeta)\in O[F,\mathsf{K}]$, where what ``sufficiently large'' means depends on $\mathsf{K}$. 
Consequently, given any subset $S\subseteq \frakS_\ell\times \frakS_m\times \frakS_n$, the set $O^{S\cap}$ defined by 
\begin{equation}
O^{S\cap} = \{(\bmalpha,\bmbeta) \in \bbC^{2N} : (\bmalpha^\sigma,\bmbeta^\sigma)\in O[F^\sigma , \mathsf{K}^\sigma] \text{ for all }\sigma \in S\}
\label{eq:misc_osc} 
\end{equation}
is open and nonempty. 
If $\mathsf{K}$ contains e.g.\ $-1$, then $O[F,\mathsf{K}]$ contains some $(\bmalpha,\bmbeta)$ such that $(\bmalpha,\bmbeta,\bmgamma)\notin V_{\ell,m,n}[F]$. 
So, \cref{eq:misc_mib} gives us an alternative definition of $\dot{I}_{\ell,m,n}(\bmalpha,\bmbeta,\bmgamma)[F]$ for some range of parameters.

\begin{proposition}
	Consider $\bullet \in \{1,2,3\}$ and $j,k\in \calI_\bullet$ with $j<k$ and $|j-k|=1$. Suppose that $\gamma_{j,k}\in \bbZ$. 
	Let $\tau\in \frakS_{\ell,m,n}$ denote the transposition swapping $j,k$. 
	Then, 
	\begin{multline}
	I_{\ell,m,n}[F](\bmalpha,\bmbeta,\bmgamma) - I_{\ell,m,n}[F](\bmalpha,\bmbeta, \bmgamma)^\tau = \int_{\Gamma_{(-\infty,0],+,0;1}} \cdots \int_{\Gamma_{(-\infty,0],+,\ell-1;\ell}} \Big[ \\  \int_{\Gamma_{[0,1],+,0;\ell+1}} \cdots \int_{\Gamma_{[0,1],+,m-1;\ell+m}} \Big[ \int_{\Gamma_{[1,\infty),+,0;\ell+m+1}} \cdots \int_{\Gamma_{[1,\infty),+,n-1;N}}  
	\Big( \prod_{j_0=1}^N z_{j_0}^{\alpha_{j_0}}(1-z_{j_0})^{\beta_{j_0}}  \Big) \\ \times  \Big( \prod_{1\leq j_0 < k_0 \leq N } (z_{k_0}-z_{j_0})^{2\gamma_{j_0,k_0}} \Big) F \dd z_N\cdots \dd z_{\ell+m+1} \Big]   \dd z_{\ell+m}\cdots \dd z_{\ell+1} \Big] \dd z_\ell \cdots \dd z_1, 
	\label{eq:misc_dff}
	\end{multline}
	whenever  $(\bmalpha,\bmbeta,\bmgamma) \in O^{\cap\{1,\tau\}}$, 
	where $\Gamma_{\mathtt{I},+,r;i}=\Gamma_{\mathtt{I},+,r;i}$ unless $i=j$, in which case $\Gamma_{\mathtt{I},+,r;i} = \Gamma_{\mathtt{I},+,r;i}(\{z_{i_0}\}_{i_0\neq j})$ is a small counterclockwise circle around $z_k$ not winding around any of the other $z$'s or $0,1$. 
	\label{prop:DF_transposition}
\end{proposition}
\begin{proof}
	It suffices to consider the case $F=1$. Indeed, if $F$ is a monomial, then we can simply absorb it into a redefinition of $\bmalpha$. The set $O^{\cap\{1,\tau\}}$ is decreasing with the set of monomials in $F$, so once the result has been proven for monomials, it follows for all Laurent polynomials.
	
	For $\bullet=2$, the proposition follows via a straightforward countour deformation argument.  The case $\bullet \in \{1,3\}$ can be reduced to $\bullet=3$ via \Cref{prop:simple}.
\end{proof}

\subsection{Symmetrization}
\label{subsec:symmetrization} 
Let $F\in \calA^\calD(A_{\ell,m,n})$.
\begin{proposition}
	For any $(\bmalpha,\bmbeta,\bmgamma) \in \cap_{\sigma \in \frakS_{\ell,m,n}} \dot{\Omega}_{\ell,m,n}[F]^\sigma$,  
	\begin{equation}
	\dot{I}_{\ell,m,n}[F](\bmalpha,\bmbeta,\bmgamma) = \sum_{\sigma \in \frakS_{\ell,m,n}} e^{ i \Theta(\sigma^{-1})} \dot{S}_{\ell,m,n}[F](\bmalpha,\bmbeta,\bmgamma)^\sigma ,
	\label{eq:misc_n66}
	\end{equation}
	where $\Theta(\sigma)= 2\pi \sum_{1\leq j < k \leq N} 1_{\sigma(j)>\sigma(k)} \gamma_{j,k}$. 
	\label{prop:symmetrization}
\end{proposition}
\begin{proof}
	By analyticity, it suffices to prove the result when the quantities above are well-defined Lebesgue integrals. 
	Decomposing $\square_{\ell,m,n}$ into $\ell!m!n!$ copies of $\triangle_{\ell,m,n}$, 
	\begin{multline}
	I_{\ell,m,n}[F](\bmalpha,\bmbeta,\bmgamma) = \sum_{\sigma \in \frakS_{\ell,m,n}} \int_{\triangle_{\ell,m,n}} \prod_{j=1}^N |x_j|^{\alpha_{\sigma(j)}}|1-x_j|^{\beta_{\sigma(j)}}  \prod_{1\leq j<k \leq N} (x_{\sigma^{-1}(k)}-x_{\sigma^{-1}(j)}+i0)^{2\gamma_{j,k} } \\ 
	\times F(x_{\sigma^{-1}(1)},\cdots,x_{\sigma^{-1}(N)} ) \dd^N x.
	\end{multline}
	The right-hand side is 
	\begin{equation}
	\sum_{\sigma \in \frakS_{\ell,m,n}} e^{i\Theta(\sigma^{-1})} \int_{\triangle_{\ell,m,n}} \prod_{j=1}^N |x_j|^{\alpha_{\sigma(j)}}|1-x_j|^{\beta_{\sigma(j)}}  \prod_{1\leq j<k \leq N} (x_{k}-x_{j})^{2\gamma_{\sigma(j),\sigma(k)} } (F^\sigma) \dd^N x, 
	\end{equation}
	which is the right-hand side of \cref{eq:misc_n66}. 
\end{proof}

\begin{proposition}
	Suppose that $\bmalpha,\bmbeta,\bmgamma$ are invariant under all $\sigma \in \frakS_{\ell,m,n}$, and suppose now that $F\in \bbC[x_1,\ldots,x_N]^{\frakS_{\ell,m,n}}$. 
	Then, for all $(\bmalpha,\bmbeta,\bmgamma)\in \dot{\Omega}_{\ell,m,n}[F]$,
	\begin{equation}
	\dot{I}_{\ell,m,n}[F](\bmalpha,\bmbeta,\bmgamma) = \Big[  \prod_{k=1}^\ell \frac{1-e^{2\pi i k \gamma_1}}{1-e^{2\pi i \gamma_1}} \Big]\Big[  \prod_{k=1}^m \frac{1-e^{2\pi i k \gamma_2}}{1-e^{2\pi i \gamma_2}} \Big]\Big[  \prod_{k=1}^n \frac{1-e^{2\pi i k \gamma_3}}{1-e^{2\pi i \gamma_3}} \Big] \dot{S}_{\ell,m,n}[F](\bmalpha,\bmbeta,\bmgamma),
	\label{eq:misc_ju3}
	\end{equation}
	where, for each $\bullet \in \{1,2,3\}$, $\gamma_\bullet=\gamma_{j,k}$ for all distinct $j,k\in \calI_\bullet$.  
	\label{prop:permutation_relation}
\end{proposition}
Here, we are treating $(1-e^{2\pi i \gamma})^{-1}(1-e^{2\pi i k\gamma})$ as an entire function.
\begin{proof}
	Applying the previous proposition,
	\begin{equation}
	I_{\ell,m,n}[F](\bmalpha,\bmbeta,\bmgamma) = \Big[\sum_{\sigma \in \frakS_\ell} e^{\pi i \operatorname{o.o.}(\sigma) \gamma }\Big]\Big[\sum_{\sigma \in \frakS_m} e^{\pi i \operatorname{o.o.}(\sigma) \gamma }\Big]\Big[\sum_{\sigma \in \frakS_n} e^{\pi i \operatorname{o.o.}(\sigma) \gamma }\Big] S_{\ell,m,n}[F](\bmalpha,\bmbeta,\bmgamma) , 
	\label{eq:misc_k21}
	\end{equation}
	where $\operatorname{o.o.}(\sigma)$ is the number of out-of-order pairs in $\sigma$. We appeal to the algebraic identity 
	\begin{equation}
	\bbZ[\zeta]\ni \sum_{\sigma \in \frakS_N} \zeta^{\operatorname{o.o.}(\sigma)} = \prod_{n=0}^{N-1} \sum_{m=0}^n \zeta^m = \prod_{n=1}^{N} \frac{1-\zeta^{n}}{1-\zeta},
	\end{equation}
	which holds for all $N\in \bbN$ and encodes the bijection between $\frakS_N$ and the set of possible runs of the bubble sort algorithm. Plugging in $\zeta=e^{\pi i \gamma}$, \cref{eq:misc_k21} becomes \cref{eq:misc_ju3}. 
\end{proof}

\subsection{The Aomoto-Dotsenko--Fateev relations}
\label{subsec:relations}

Fix $N\in \bbN^+$ and $F\in \bbC[x_1,x_1^{-1},\ldots,x_N,x_N^{-1}]$. For each $j\in \{1,\ldots,N\}$, let $\sigma_j \in \frakS_N$ be the permutation that takes $1$ and inserts it in the $j$th position while maintaining the relative order of the other terms. That is, $\sigma_j = (1\;j\; j-1\;\cdots \;2)$. 

For any $\ell\in \bbN^+$ and $m,n \in \bbN$ with $\ell+m+n=N$, let 
\begin{align}
\begin{split} 
\dot{\Lambda}_{\ell,m,n}[F] &= \dot{V}_{\ell,m,n}[F] \cap \dot{V}_{\ell-1,m+1,n}[F]^{\sigma_\ell} \cap \dot{V}_{\ell-1,m,n+1}[F]^{\sigma_{\ell+m}} \\
&= \dot{V}_{\ell,m,n}[F] \cap \dot{V}_{\ell-1,m+1,n}[F]^{\sigma_{\ell+m}} \cap \dot{V}_{\ell-1,m,n+1}[F]^{\sigma_{N}}, 
\end{split} 
\end{align}
\begin{equation} 
\dot{\mho}_{\ell,m,n}[F] = (\cap_{j=1}^\ell \dot{\Omega}_{\ell,m,n}[F]^{\sigma_j}) \cap (\cap_{j=\ell}^{\ell+m} \dot{\Omega}_{\ell-1,m+1,n}[F]^{\sigma_j}) \cap (\cap_{j=\ell+m}^{N} \dot{\Omega}_{\ell-1,m,n+1}[F]^{\sigma_j}).
\end{equation} 
Note that $\dot{\mho}_{\ell,m,n}[F],\dot{\Lambda}_{\ell,m,n}[F]$  are open, dense, and connected subsets of $\bbC^{2N+N(N-1)/2}$, being the complements of locally finite unions of complex affine hyperplanes.
\begin{proposition}
	For any $(\bmalpha,\bmbeta,\bmgamma)\in \dot{\mho}_{\ell,m,n}[F]$, 
	\begin{multline}
	0=\sum_{j=1}^\ell e^{\pm  i  \theta_j } \dot{S}_{\ell,m,n}[F]( \bmalpha,\bmbeta,\bmgamma)^{\sigma_j} + \sum_{j=\ell}^{\ell+m} e^{\pm i\vartheta_j } \dot{S}_{\ell-1,m+1,n}[F](\bmalpha,\bmbeta,\bmgamma)^{\sigma_j} \\ +  \sum_{j=\ell+m}^{N} e^{\pm i\varphi_j} \dot{S}_{\ell-1,m,n+1}[F](\bmalpha,\bmbeta,\bmgamma)^{\sigma_j}
	\label{eq:misc_z3z}
	\end{multline}
	holds for each choice of sign, 
	where $\theta_j = 2\pi \sum_{2\leq j_0 \leq j} \gamma_{1,j_0}$, $\vartheta_j = \pi \alpha + 2\pi  \sum_{2\leq j_0 \leq j} \gamma_{1,j_0}$, and $\varphi_j = \pi \alpha+\pi \beta+2 \pi \sum_{2\leq j_0 \leq j} \gamma_{1,j_0}$.
	\label{lem:main_lemma}
\end{proposition}
\begin{proof}
	Without loss of generality, we may assume $F=1$. 
	Let $\mho_{\ell,m,n}$ denote the subset of $(\bmalpha,\bmbeta,\bmgamma)\in \bbC^{2N+N(N-1)/2}$ defined by
	\begin{equation}
	\mho_{\ell,m,n} = \left\{(\bmalpha,\bmbeta,\bmgamma)  : (\bmalpha^{\sigma_j},\bmbeta^{\sigma_j}, \bmgamma^{\sigma_j}) \in 
	\begin{cases}
	\Omega_{\ell,m,n} & (j\in \{1,\ldots,\ell-1\}) \\ 
	\Omega_{\ell-1,m+1,n} & (j\in \{\ell+1,\ldots,\ell+m-1\}) \\ 
	\Omega_{\ell-1,m,n+1} & (j\in \{\ell+m+1,\ldots,N\}) \\
	\Omega_{\ell,m,n} \cap \Omega_{\ell-1,m+1,n} & (j=\ell) \\
	\Omega_{\ell-1,m+1,n}\cap \Omega_{\ell-1,m,n+1} & (j=\ell+m)
	\end{cases}
	\right\}.
	\label{eq:misc_0s0}
	\end{equation}
	
	Let $\varepsilon>0$. 
	For each $\digamma_1,\digamma_2,\digamma_3>0$ and $\underline{\gamma},\overline{\gamma}\in (-(N-1)^{-1},0)$ with $\underline{\gamma}<\overline{\gamma}$, let $\mho_{0,\digamma,\underline{\gamma},\overline{\gamma}}$ (suppressing the $\ell,m,n$ dependence for brevity) denote the set of $(\bmalpha,\bmbeta,\bmgamma)\in \bbC^{2N+N(N-1)/2}$ such that 
	\begin{itemize}
		\item $\underline{\gamma}<\Re \gamma_{j,k}<\overline{\gamma}$ for all $j,k\in \{1,\ldots,N\}$ with $j\neq k$, 
		\item $\digamma_1<\Re \alpha_{j} <\digamma_2$ for each $j\in \{2,\ldots,\ell\}$, $\Re \alpha_j >\digamma_1$ for each $j\in \{\ell+1,\ldots,\ell+m\}$, and $\Re \alpha_{j} < - \digamma_3$ for each $j\in \{\ell+m+1,\ldots,N\}$, 
		\item $\digamma_1<\Re \beta_{j} < \digamma_2$ for each $j\in \{\ell+m+1,\ldots,N\}$, $\Re \beta_j > \digamma_1$ for each $j\in \{\ell+1,\ldots,\ell+m\}$, and $\Re \beta_{j} <- \digamma_3$ for $j\in \{2,\ldots,\ell\}$,
	\end{itemize}
	where $\digamma = (\digamma_1,\digamma_2,\digamma_3)$. 
	The set $\mho_{0,\digamma,\underline{\gamma},\overline{\gamma}}$ is open and nonempty. By \cref{eq:Omega_N_definition} and the analogue of \cref{eq:Omega_N_definition} for the $m<N$ case, there exist $\digamma_{00},\digamma_0,\digamma_{01}>0$ (depending on $\ell,m,n,\underline{\gamma},\overline{\gamma}$) such that 
	\begin{equation}
	\mho_{\digamma,\underline{\gamma},\overline{\gamma}} \overset{\mathrm{def}}{=}	\{ (\bmalpha,\bmbeta,\bmgamma) \in  \mho_{0,\digamma,\underline{\gamma},\overline{\gamma}} \text{ and }(\alpha_1,\beta_1)\in \Omega_{1,0,0}\cap \Omega_{0,1,0}\cap \Omega_{0,0,1}  \}\subset \mho_{\ell,m,n} 
	\end{equation}
	whenever $\digamma_2>\digamma_1>\digamma_0$ and $\digamma_3>\digamma_{01}\digamma_2+\digamma_{00}$. 
	Observe that $\Omega_{1,0,0}\cap \Omega_{0,1,0}\cap \Omega_{0,0,1}$ is the subset of $\bbC^2_{\alpha,\beta}$ defined by the inequalities $-1 < \Re \alpha, \Re \beta $ and $\Re \alpha + \Re \beta < -1$. The set 
	\begin{equation} 
	\{(r_1,r_2) \in \bbR^2 : -1 < r_1,r_2 \text{ and }r_1+r_2 < -1 \}
	\end{equation} 
	is a nonempty triangle. So, $\mho_{\digamma,\underline{\gamma},\overline{\gamma}}$ is an open and nonempty subset of $\bbC^{2N+N(N-1)/2}$ and moreover of $\dot{\mho}_{\ell,m,n}$. 
	
	For such $\digamma$ and $(\bmalpha,\bmbeta,\bmgamma) \in \mho_{\digamma,\underline{\gamma},\overline{\gamma}}$, \cref{eq:misc_z3z} (with $F=1$) 
	just reads 
	\begin{multline}
	0=\sum_{j=1}^\ell e^{\pm  i  \theta_j } S_{\ell,m,n}[1]( \bmalpha,\bmbeta,\bmgamma)^{\sigma_j} + \sum_{j=\ell}^{\ell+m} e^{\pm i\vartheta_j } S_{\ell-1,m+1,n}[1](\bmalpha,\bmbeta,\bmgamma)^{\sigma_j} \\ +  \sum_{j=\ell+m}^{N} e^{\pm i\varphi_j} S_{\ell-1,m,n+1}[1](\bmalpha,\bmbeta,\bmgamma)^{\sigma_j}
	\label{eq:misc_z4z}
	\end{multline}
	(note the absence of the dots over the $S$'s). 
	By the analyticity of all of the functions in \cref{eq:misc_z3z} on $\dot{\mho}_{\ell,m,n}$, it suffices to prove that \cref{eq:misc_z4z} holds for such $(\bmalpha,\bmbeta,\bmgamma)$. 
	
	By Fubini's theorem, the right-hand side of \cref{eq:misc_z4z} is 
	\begin{equation}
	\int_{\triangle_{\ell-1,m,n}} \omega(x_2,\ldots,x_N) \Big[ \int_{-\infty}^{+\infty} (-x_1 \pm i0)^{\alpha_1} (1-x_1\pm  i0)^{\beta_1} \Big(\prod_{j=2}^{N} (x_j-x \pm  i0)^{2\gamma_{1,j}}\Big)
	\dd x \Big] \dd x_2\cdots \dd x_{N},
	\end{equation}
	where $\omega(x_2,\ldots,x_N) = [\prod_{j=2}^N|x_j|^{\alpha_j} |1-x_j|^{\beta_j}] \prod_{2\leq j < k \leq N} (x_k-x_j)^{2\gamma_{j,k}}$. 
	The claim then follows from 
	\begin{equation}
	0=\int_{-\infty}^{+\infty} (-x \pm i0)^{\alpha} (1-x \pm i0)^{\beta} \Big( \prod_{j=2}^{N} (x_j-x \pm i0)^{2\gamma_{j}}
	\Big) \dd x, 
	\label{eq:misc_h12}
	\end{equation}
	which holds for every $(x_2,\ldots,x_N)\in (\bbR\backslash \{0,1\})^{N-1}$ such that $x_2,\ldots,x_N$ are pairwise distinct and all $\alpha,\beta,\gamma_2,\ldots,\gamma_N\in \bbC$ for which 
	\begin{itemize}
		\item the integrand of \cref{eq:misc_h12} lies in $L^1(\bbR)$ and 
		\item $\Re \gamma_{j}\in (-1,0)$ for all $j\in \{2,\ldots,N\}$. 
	\end{itemize}
	Denote the right-hand side of \cref{eq:misc_h12} by $\calI_\pm =\calI_\pm(x_2,\ldots,x_{N};\alpha,\beta,\gamma_2,\ldots,\gamma_N)$. 
	For $R > \max\{|x_1|,\ldots,|x_{N-1}|\}$, 
	\begin{equation}
	0=\int_{\Gamma_\mp(R)} (-z\pm i0)^{\alpha} (1-z \pm i0)^{\beta} \prod_{j=2}^{N} (x_j-z\pm i0)^{2\gamma_j} \dd z, 
	\label{eq:misc_112}
	\end{equation}
	where $\Gamma_\pm(R) =\Gamma_\pm(R)(x_2,\ldots,x_{N})\subset \bbC$ is the semicircular contour (with $N+1$ semicircular insets placed so that the contour avoids $x_2,\ldots,x_N$) connecting $-R$ and $+R$, with the arc and semicircular insets in the half-plane $\{z\in \bbC: \pm \Im z\geq 0\}$. See \Cref{fig:Aomoto_contour}. In \cref{eq:misc_112}, the integrand is defined taking the branch cut along the negative real axis, so 
	\begin{equation}
	(x-z\pm i0)^{2\gamma_j} = 
	\begin{cases}
	\exp (2\gamma_j (\log |x-z|+ i \operatorname{arg}(x-z) ) )  & (\text{$+$ case}, \Im z \leq 0), \\
	\exp (2\gamma_j (\log |x-z| -2\pi i+ i \operatorname{arg}(x-z) ) ) & (\text{$-$ case}, \Im z \geq 0),
	\end{cases}
	\end{equation}
	for any $x\in \bbR$, where $\operatorname{arg}(x-z)\in [0,2\pi)$. We orient $\Gamma_+$ counter-clockwise and $\Gamma_-$ clockwise.
	
	\begin{figure}
		\begin{center}
			\begin{tikzpicture}[scale=2.5,  decoration={
				markings,
				mark=at position 0.6 with {\arrow[scale=1.5,>=latex]{>}}}]
			\draw[step=.25,gray,thin, dotted] (-1.5,0) grid (1.5,1.3);
			\draw[->, thick, dashed] (0,0) -- (0,1.3) node[above] {$\Im z$};
			\draw[->, thick, dashed] (-1.5,0) -- (1.5,0) node[right] {$\Re z$};
			\filldraw[color=black] (0,0) circle (.5pt) node[below] {$0$};
			\filldraw[color=black] (1,0) circle (.5pt) node[below] {$1$};
			\filldraw[color=black] (-.33,0) circle (.5pt) node[below] {$x_2$};
			\filldraw[color=black] (-.8,0) circle (.5pt) node[below] {$x_1$};
			\filldraw[color=black] (.5,0) circle (.5pt) node[below] {$x_3$};
			\filldraw[color=black] (-1.2,0) circle (.5pt) node[below] {$-R$};
			\filldraw[color=black] (+1.2,0) circle (.5pt) node[below] {$R$};
			\draw[postaction={decorate}] (1.2,0) arc(0:180:1.2);
			\draw[postaction={decorate}] (-1.2,0)  -- (-.9,0);
			\draw (-.9,0) arc(0:-180:-.1);
			\draw[postaction={decorate}] (-.7,0) -- (-.4,0);
			\draw (-.4,0) arc(0:-180:-.066);
			\draw[postaction={decorate}] (-.266,0) -- (-.1,0);
			\draw (-.1,0) arc(0:-180:-.1);
			\draw[postaction={decorate}] (.1,0) -- (.4,0);
			\draw (.4,0) arc(0:-180:-.1);
			\draw[postaction={decorate}] (.6,0) -- (.9,0);
			\draw (.9,0) arc(0:-180:-.1);
			\draw (1.1,0) -- (1.2,0);
			\end{tikzpicture} 
		\end{center}
		\caption{The contour $\Gamma_+(R)$ in the case $\ell=2,m=1,n=0$.}
		\label{fig:Aomoto_contour}
	\end{figure}
	
	Let $\Gamma_{++}(R)$ denote the large arc of $\Gamma_+(R)$ and $\Gamma_{+0}(R)$ denote the rest, and likewise let $\Gamma_{--}(R)$ denote the large arc of $\Gamma_-(R)$ and $\Gamma_{-0}(R)$ denote the rest. Then, 
	\begin{equation}
	\calI_\pm = \lim_{R\to\infty} \int_{\Gamma_{\mp 0}(R)} (-z\pm i0)^{\alpha} (1-z \pm i0)^{\beta}  \prod_{j=2}^{N} (x_j-z\pm i0)^{2\gamma_j} \dd z. 
	\label{eq:misc_113}
	\end{equation}
	On the other hand, for $R$ sufficiently large, 
	\begin{equation} 
	\Big|\int_{\Gamma_{\mp\mp}(R)} (-z\pm i0)^\alpha (1-z \pm i0)^\beta  \prod_{j=2}^{N} (x_j-z \pm i0)^{2\gamma_j} \dd x \Big| \leq \pi (2R)^{1+\Re \alpha+\Re \beta} = O(R^{-\varepsilon}) 
	\label{eq:misc_114}
	\end{equation} 
	for some $\varepsilon>0$ depending on $(\alpha,\beta)\in \Omega_{1,0,0}\cap \Omega_{0,1,0}\cap \Omega_{0,0,1}$. Combining \cref{eq:misc_112}, \cref{eq:misc_113}, and \cref{eq:misc_114}, we get $\calI_\pm = 0$.
\end{proof}

\begin{proposition}
	For any $F\in \bbC[x_1,x_1^{-1},\ldots,x_N,x_N^{-1}]$, 
	\begin{multline}
	0= \dot{I}_{\ell,m,n}[F]( \bmalpha,\bmbeta, \bmgamma ) + e^{+\pi i(\alpha+2 \sum_{j=2}^\ell \gamma_{1,j} ) } \dot{I}_{\ell-1,m+1,n}[F](\bmalpha,\bmbeta,\bmgamma)^{\sigma_\ell} \\ +   e^{+ \pi i(\alpha+\beta+2 \sum_{j=2}^{\ell+m} \gamma_{1,j} )} \dot{I}_{\ell-1,m,n+1}[F](\bmalpha,\bmbeta,\bmgamma)^{\sigma_{\ell+m}}
	\label{eq:misc_z5z}
	\end{multline}
	\begin{multline}
	0= \dot{I}_{\ell,m,n}[F]( \bmalpha,\bmbeta, \bmgamma )^{\sigma_\ell} + e^{- \pi i(\alpha+2 \sum_{j=2}^\ell \gamma_{1,j} ) } \dot{I}_{\ell-1,m+1,n}[F](\bmalpha,\bmbeta,\bmgamma)^{\sigma_{\ell+m}} \\ +   e^{-\pi i(\alpha+\beta+2 \sum_{j=2}^{\ell+m} \gamma_{1,j} )} \dot{I}_{\ell-1,m,n+1}[F](\bmalpha,\bmbeta,\bmgamma)^{\sigma_{N}}
	\label{eq:misc_z6z}
	\end{multline}
	both hold, for all $(\bmalpha,\bmbeta,\bmgamma) \in \dot{\Lambda}_{\ell,m,n}[F]$. 
	\label{prop:main_lemma2}
\end{proposition}
\begin{proof}
	Let $\frakS'_{\ell,m,n}$ denote the Young subgroup of $\frakS_{\ell,m,n}$ consisting of permutations which fix $1$, i.e.\ 
	\begin{equation} 
	\frakS'_{\ell,m,n} =  \{\sigma \in \frakS_{\ell,m,n} \text{ s.t. } \sigma(1)=1\}.
	\end{equation}
	Via analyticity, it suffices to prove this for all $(\bmalpha,\bmbeta,\bmgamma)\in \cap_{\sigma \in \frakS'_{\ell,m,n}} \dot{\mho}_{\ell,m,n}[F]^\sigma$. 
	
	For such $(\bmalpha,\bmbeta,\bmgamma)$, we can cite the previous proposition to get 
	\begin{multline}
	0=\sum_{\sigma \in \frakS_{\ell,m,n}'}e^{\pi i \Theta(\sigma^{-1})} \Big[ \sum_{j=1}^\ell e^{\pm  i  \theta_j^\sigma } \dot{S}_{\ell,m,n}[F]( \bmalpha,\bmbeta,\bmgamma)^{\sigma_j\sigma} \\ + \sum_{j=\ell}^{\ell+m} e^{\pm i\vartheta_j^\sigma } \dot{S}_{\ell-1,m+1,n}[F](\bmalpha,\bmbeta,\bmgamma)^{\sigma_j\sigma} +  \sum_{j=\ell+m}^{N} e^{\pm i\varphi_j^\sigma} \dot{S}_{\ell-1,m,n+1}[F](\bmalpha,\bmbeta,\bmgamma)^{\sigma_j\sigma }\Big], 
	\label{eq:misc_9h9}
	\end{multline}
	where $\theta_j^\sigma = 2\pi \sum_{2\leq j_0 \leq j} \gamma_{1,\sigma(j_0)}$, $\vartheta_j^\sigma = \pi \alpha + 2\pi  \sum_{2\leq j_0 \leq j} \gamma_{1,\sigma(j_0)}$, and $\varphi_j^\sigma = \pi \alpha+\pi \beta+2 \pi \sum_{2\leq j_0 \leq j} \gamma_{1,\sigma(j_0)}$. The order of multiplication is such that $\sigma_j \sigma$ is a permutation satisfying $(\sigma_j \sigma)(1)=j$. In \cref{eq:misc_9h9}, $\Theta$ is defined as in \Cref{prop:symmetrization}.
	
	Every $\sigma_0 \in \frakS_{\ell,m,n}$ has the form $\sigma_0 = \sigma_j \sigma$ for some $j\in \{1,\ldots,N\}$ and $\sigma\in \frakS_{\ell,m,n}$ satisfying $\sigma(1) = 1$. It can be seen that 
	\begin{equation} 
	\Theta(\sigma^{-1}_0) = \Theta(\sigma^{-1})+ \theta_j^\sigma.
	\end{equation} 
	Using \Cref{prop:symmetrization}, we check that the two cases of \cref{eq:misc_9h9} yield the two results, \cref{eq:misc_z5z} and \cref{eq:misc_z6z}. For instance, 
	\begin{align}
	\begin{split} 
	\sum_{\sigma \in \frakS_{\ell,m,n}'} e^{\pi i \Theta(\sigma^{-1})}\sum_{j=1}^\ell e^{+  i  \theta_j^\sigma } \dot{S}_{\ell,m,n}[F]( \bmalpha,\bmbeta,\bmgamma)^{\sigma_j\sigma} &= \sum_{\sigma \in \frakS_{\ell,m,n}} e^{\pi i \Theta(\sigma^{-1})} \dot{S}_{\ell,m,n}[F]( \bmalpha,\bmbeta,\bmgamma)^{\sigma} \\ 
	\label{eq:misc_8j1}
	&= \dot{I}_{\ell,m,n}[F](\bmalpha,\bmbeta,\bmgamma).
	\end{split} 
	\end{align}
	Similar statements apply to the other two sums in \cref{eq:misc_9h9} in the `$+$' case, thus yielding \cref{eq:misc_z5z}. Similar computations apply to the `$-$' case. 
\end{proof}

\subsection{The symmetric and DF-symmetric cases}
\label{subsec:continuation_symmetric}
Fix $F \in \calA^{\calD}(A_{\ell,m,n})$, not necessarily symmetric. We assume that $d_{\mathrm{F}_{S,Q;\bullet}}\in \bbZ$ for all $\mathrm{F}_{S,Q;\bullet}\in \calF(A_{\ell,m,n})$.
Let 
\begin{align}
\delta_{k} &= \min\{d_{\mathrm{F}_{S,Q;0}}  : S\subseteq \calI_1,Q\subseteq \calI_2, |S\cup Q|=k\}
\intertext{for each $k\in \{1,\ldots,\ell+m\}$,}
\atled_{k} &= \min\{d_{\mathrm{F}_{S,Q;1}}  : S\subseteq \calI_2,Q\subseteq \calI_3,|S\cup Q|=k\} 
\intertext{for each $k\in \{1,\ldots,m+n\}$, and} 
d_{k} &= -\min\{d_{\mathrm{F}_{S,Q;\infty}}  : S\subseteq \calI_3,Q\subseteq \calI_1,|S\cup Q|=k\}
\end{align}
for each $k\in \{1,\ldots,\ell+n\}$. 
Here, we are ranging over all $\{d_{\mathrm{F}}\}_{\mathrm{F}\in \calF(A_{\ell,m,n})} \in \calD$.  

Let $\dot{W}_{\ell,m,n}[\calD]$ denote the set of $(\alpha,\beta,\gamma)\in \bbC^3$ such that $(\bmalpha,\bmbeta,\bmgamma)\in \smash{\dot{V}_{\ell,m,n}[\calD]}$ whenever $\bmalpha,\bmbeta,\bmgamma$ have components given by $\alpha_j = \alpha$ and $\beta_j = \beta$ for all indices $j\in \{1,\ldots,N\}$ and $\gamma_{j,k} = \gamma$ for all $j,k\in \{1,\ldots,N\}$ with $j<k$.

\begin{proposition}
	There exists an entire function $I_{\ell,m,n;\mathrm{Reg}}[F]:\bbC^3\to \bbC$ such that 
	\begin{multline}
	\dot{I}_{\ell,m,n}[F](\alpha,\beta,\gamma) = \Big[ \prod_{k=1}^{\ell+m} \Gamma(\delta_k +k(1+\alpha+(k-1)\gamma)) \Big] \Big[ \prod_{k=1}^{m+n} \Gamma(\atled_k+k(1+\beta+(k-1)\gamma))  \Big] \\ \times \Big[ \prod_{k=1}^{\ell+n} \Gamma(-d_k-k(1+\alpha + \beta +  (2N - k - 1)\gamma )) \Big] I_{\ell,m,n;\mathrm{Reg}}[F](\alpha,\beta,\gamma)
	\label{eq:misc_bzb}
	\end{multline}
	for all $(\alpha,\beta,\gamma) \in \dot{W}_{\ell,m,n}[\calD]$.
	\label{prop:mero_sym}
\end{proposition}
\begin{proof}
	Follows from \Cref{prop:Icontinuation}.
\end{proof}

For later reference, consider the special case $F\in \bbC[x_1,\ldots,x_N]^{\frakS_N}$. Referring to \cref{eq:misc_nu1}, \cref{eq:misc_nu2}, and \cref{eq:misc_nu3}, set $d_{\mathrm{F}_{S,Q;0}} = \delta_j[F]$, $d_{\mathrm{F}_{S,Q;1}} = \atled_j[F]$, and $d_{\mathrm{F}_{S,Q;\infty}} = \deg_j[F]$,  
for $S,Q\subseteq \{1,\ldots,N\}$ as usual, where, for each $S$ and $Q$, $j=|S\cup Q|$. Then, as follows straightforwardly from \cref{eq:misc_j2g}, \cref{eq:misc_j3g}, \cref{eq:misc_j4g}, 
\begin{equation} 
F\in \prod_{\mathrm{F}\in \calF(A_{\ell,m,n})} x_{\mathrm{F}}^{d_{\mathrm{F}}} C^\infty(A_{\ell,m,n}).
\end{equation}
Thus, letting $\calD$ denote the collection of the integers above, $F\in \calA^{\calD}(A_{\ell,m,n})$. 
We can therefore apply the results above, with $\delta_j = \delta_j[F]$, $\atled_j=\atled_j[F]$, and $d_j=-\deg_j[F]$. 

We now turn to the ``DF0-symmetric'' case. For any $\mathtt{S}\subseteq \{1,\ldots,N\}$, let 
\begin{equation}
\dot{W}_{\ell,m,n}^{\mathrm{DF0},\mathtt{S}}[F] = \{(\alpha_-,\alpha_+,\beta_-,\beta_+,\gamma_-,\gamma_0,\gamma_+)\in \bbC^7 :(\bmalpha^{\mathrm{DF0}},\bmbeta^{\mathrm{DF0}},\bmgamma^{\mathrm{DF0}}) \in \dot{V}_{\ell,m,n}[F]\}.
\end{equation} 
This is a dense, open, and connected subset of $\bbC^7$ and depends on $\mathtt{S}$ only through the numbers $|\mathtt{S} \cap \calI_j|$. Actually, we need a slightly refined version of this later; let 
\begin{multline}
\dot{W}_{\ell,m,n}^{\mathrm{DF1},\mathtt{S}}[F] = \{(\alpha_{-,1}, \alpha_{-,2},\alpha_{-,3},\alpha_{+,1}, \alpha_{+,2},\alpha_{+,3},\beta_{-,1}, \beta_{-,2},\beta_{-,3},\beta_{+,1}, \beta_{+,2},\beta_{+,3},\gamma_-,\gamma_0,\gamma_+)\in \bbC^9 \\ :(\bmalpha^{\mathrm{DF1}},\bmbeta^{\mathrm{DF1}},\bmgamma^{\mathrm{DF0}}) \in \dot{V}_{\ell,m,n}[F]\},
\end{multline}
where $\bmalpha^{\mathrm{DF1}}, \bmbeta^{\mathrm{DF1}}$ are defined as their DF0-counterparts, but defining the $j$th component using $\alpha_{+,\nu}$ in place of $\alpha_+$ and $\beta_{+,\nu}$ in place of $\beta_+$ for $\nu\in \calI_\nu$.

For $(\alpha_-,\alpha_+,\beta_-,\beta_+,\gamma_-,\gamma_0,\gamma_+)\in \dot{W}^{\mathrm{DF0},\mathtt{S}}_{\ell,m,n}[F]$, let 
\begin{equation}
\dot{I}^{\mathrm{DF0};\mathtt{S}}_{\ell,m,n}[F](\alpha_-,\alpha_+,\beta_-,\beta_+,\gamma_-,\gamma_0,\gamma_+) = \dot{I}_{\ell,m,n}[F](\bmalpha^{\mathrm{DF}0},\bmbeta^{\mathrm{DF}0},\bmgamma^{\mathrm{DF}0}). 
\end{equation}
Let $\ell_+ = \mathtt{S} \cap \calI_1$, $\ell_- = \ell-\ell_+$, $m_+ = \mathtt{S}\cap \calI_2$, $m_- = m-m_+$, $n_+ = \mathtt{S} \cap \calI_3$, and $n_- =n-n_+$. Set $N_+ = |\mathtt{S}|$ and $N_-=N-N_+$. 

Suppose now that $F \in \calA^{\calD}(A_{\ell,m,n})$ is symmetric in the variables $\{x_i\}_{i\in \mathtt{S}}$ and $\{x_i\}_{i\notin \mathtt{S}}$ separately. Let
\begin{align}
\delta_{j_-,j_+} &= \min\{d_{\mathrm{F}_{S,Q;0}}  : S\subseteq \calI_1 ,Q\subseteq \calI_2, |(S\cup Q) \backslash \mathtt{S}|=j_-, |(S\cup Q)\cap \mathtt{S}|=j_+\} \label{eq:misc_78z}
\intertext{for $j_- \in \{1,\ldots,\ell_-+m_-\}$ and $j_+ \in \{1,\ldots,\ell_++m_+\}$,} 
\atled_{j_-,j_+} &= \min\{d_{\mathrm{F}_{S,Q;1}}  : S\subseteq \calI_2,Q\subseteq \calI_3,|(S\cup Q) \backslash \mathtt{S}|=j_-, |(S\cup Q)\cap \mathtt{S}|=j_+\} \label{eq:misc_78x} 
\intertext{for $j_- \in \{1,\ldots,m_-+n_-\}$ and $j_+\in\{1,\ldots,m_++n_+\}$, and}
d_{j_-,j_+}&= -\min\{d_{\mathrm{F}_{S,Q;\infty}}  : S\subseteq \calI_3,Q\subseteq \calI_1,|(S\cup Q) \backslash \mathtt{S}|=j_-, |(S\cup Q)\cap \mathtt{S}|=j_+\} \label{eq:misc_78c}
\end{align}
for $j_- \in \{1,\ldots,\ell_-+n_-\}$ and $j_+\in\{1,\ldots,\ell_++n_+\}$.
A similar argument to that above yields:
\begin{propositionp}
	There exists an entire function $I_{\ell,m,n;\mathrm{Reg}}^{\mathrm{DF0};\mathtt{S}}[F]:\bbC^7\to \bbC$ such that 
	\begin{multline}
	\dot{I}_{\ell,m,n}^{\mathrm{DF0}}[F](\alpha_-,\alpha_+,\beta_-,\beta_+,\gamma_-,\gamma_0,\gamma_+) =
	I_{\ell,m,n;\mathrm{Reg}}^{\mathrm{DF0}}[F](\alpha_-,\alpha_+,\beta_-,\beta_+,\gamma_-,\gamma_0,\gamma_+) \\ \times  \Big[ \prod_{j_-=1}^{\ell_-+m_-} \prod_{j_+=1}^{\ell_++m_+} \Gamma(\delta_{j_-,j_+} + j_-(1+\alpha_-+(j_--1) \gamma_-)+j_+(1+\alpha_++(j_+-1)\gamma_+)+2\gamma_0 j_-j_+ ) \Big] \\  \times \Big[ \prod_{j_-=1}^{m_-+n_-} \prod_{j_+=1}^{m_++n_+} \Gamma(\atled_{j_-,j_+} + j_-(1+\beta_-+(j_--1) \gamma_-)+j_+(1+\beta_++(j_+-1)\gamma_+)+2\gamma_0 j_-j_+ ) \Big]  \\ \times \Big[ \prod_{j_-=1}^{\ell_-+n_-} \prod_{j_+=1}^{\ell_++n_+} \Gamma(-d_{j_-,j_+}-  j_-(1+\alpha_-+\beta_-+(2N_--j_--1)\gamma_- ) \\ - j_+(1+\alpha_++\beta_++(2N_+-j_+-1)\gamma_+ )  - 2\gamma_0 j_-j_+ ) \Big]
	\end{multline}
	holds whenever $(\alpha_-,\alpha_+,\beta_-,\beta_+,\gamma_-,\gamma_0,\gamma_+)\in \dot{W}_{\ell,m,n}^{\mathrm{DF0}}[F]$. 
	\label{prop:DF0_result}
\end{propositionp}

\section{Removing singularities}
\label{sec:removing_singularities}

As in previous sections, fix $\ell,m,n\in \bbN$ not all zero,  and let $N=\ell+m+n$ and $\calI_1 = \{1,\ldots,\ell\}$, $\calI_2 = \{\ell+1,\ldots,\ell+m\}$, and $\calI_3 = \{\ell+m+1,\ldots,N\}$. 
For $k\in  \bbN$, let 
\begin{equation} 
\digamma_k:\bbC_\gamma\backslash \{k \gamma \in \bbZ^{\leq -1} \text{ and }\gamma\notin \bbZ\}\to \bbC
\end{equation} 
denote the analytic function given by $\digamma_k(\gamma) = \Gamma(1+\gamma)^{-1} \Gamma(1+k\gamma )$ for $k\gamma \notin \bbZ^{\leq -1}$. We can consider $\digamma_k^{-1}$ as an entire function. 

\subsection{The symmetric case}
\label{subsec:removing_singularities_1}
Fix $F\in \bbC[x_1,\ldots,x_N]^{\frakS_N}$, and let $\delta_j,\atled_j,d_j\in \bbN$ be as above.

Let $\dot{U}_{\ell,m,n}[F]$ denote the set of $(\alpha,\beta,\gamma)\in \bbC^3$ such that $(\bmalpha,\bmbeta,\bmgamma) \in \dot{\Omega}_{\ell,m,n}[F]$ whenever $\bmalpha,\bmbeta,\bmgamma$ have components given by $\alpha_j = \alpha$ and $\beta_j = \beta$ for all indices $j\in \{1,\ldots,N\}$ and $\gamma_{j,k} = \gamma$ for all $j<k$. Thus, we can define 
\begin{equation}
\dot{S}_{\ell,m,n}[F](\alpha,\beta,\gamma) = \dot{S}_{\ell,m,n}[F](\bmalpha,\bmbeta,\bmgamma)
\end{equation}
for any $(\alpha,\beta,\gamma) \in \dot{U}_{\ell,m,n}[F]$.

\begin{proposition}
	The function $S_{\ell,m,n}^{\mathrm{reg}}[F]:\dot{U}_{\ell,m,n}[F]\to \bbC$ defined by 
	\begin{multline}
	S_{\ell,m,n}^{\mathrm{reg}}[F](\alpha,\beta,\gamma) = \Big[\prod_{k=1}^{\ell+m}  \Gamma(\delta_k+k(1+\alpha+(k-1)\gamma)) \Big]^{-1} \\  \times  \Big[ \prod_{k=1}^{m+n}\Gamma(\atled_k+k(1+\beta+(k-1)\gamma)) \Big]^{-1}\Big[ \prod_{k=1}^{\ell+n} \Gamma(-d_k-k(1+\alpha+\beta+(N+k-2)\gamma))  \Big]^{-1} \\ \times \Big[ \prod_{k=1}^\ell \frac{1}{\digamma_k(\gamma)} \Big]\Big[ \prod_{k=1}^m \frac{1}{\digamma_k(\gamma)} \Big]  \Big[ \prod_{k=1}^n \frac{1}{\digamma_k(\gamma)} \Big] \dot{S}_{\ell,m,n}[F](\alpha,\beta,\gamma) 
	\label{eq:misc_nb5}
	\end{multline}
	extends to an entire function $\bbC^3_{\alpha,\beta,\gamma}\to \bbC$.
	\label{prop:three_final}
\end{proposition}
\begin{proof}
	\begin{itemize}
		\item Since the prefactor on the right-hand side of \cref{eq:misc_nb5} consisting of all of the $\Gamma$-function reciprocals is entire, $\smash{S^{\mathrm{reg}}_{\ell,m,n}[F]}$ extends to an analytic function on $\dot{U}_{\ell,m,n}[F]$, the domain of $\dot{S}_{\ell,m,n}[F](\alpha,\beta,\gamma)$. 
		\item For all $(\alpha,\beta,\gamma)\in \dot{U}_{\ell,m,n}[F]$, we have 
		\begin{multline}
		\Big[ \prod_{k=1}^{\ell} \frac{1-e^{2\pi i k\gamma}}{1-e^{2\pi  i\gamma}} \digamma_k(\gamma)\Big]\Big[ \prod_{k=1}^{m} \frac{1-e^{2\pi i k\gamma}}{1-e^{2\pi  i\gamma}} \digamma_k(\gamma)\Big]\Big[ \prod_{k=1}^{n} \frac{1-e^{2\pi i k\gamma}}{1-e^{2\pi  i\gamma}} \digamma_k(\gamma)\Big] S_{\ell,m,n}^{\mathrm{reg}}[F](\alpha,\beta,\gamma) \\ = I_{\ell,m,n;\mathrm{Reg}}[F](\alpha,\beta,\gamma) 
		\end{multline} 
		by \Cref{prop:permutation_relation}. By \Cref{prop:mero_sym}, this
		extends to an entire function $\bbC^3_{\alpha,\beta,\gamma}\to \bbC$. 
		
		The product $\digamma_k(\gamma) (1-e^{2\pi i k \gamma})(1-e^{2\pi i \gamma})^{-1}$, with its removable singularities removed, vanishes if and only if $k\gamma \in \bbN$ and $\gamma\notin \bbN$.  
		Thus, $S_{\ell,m,n}^{\mathrm{reg}}[F]$ extends to an analytic function on 
		\begin{equation} 
		\bbC^3_{\alpha,\beta,\gamma}\backslash \cup_{k=2}^{M} \{k\gamma \in \bbN, \gamma\notin \bbN\},
		\end{equation} 
		where $M=\max\{\ell,m,n\}$. 
	\end{itemize}
	
	Combining these two observations, $S_{\ell,m,n}^{\mathrm{reg}}[F]$ extends to an analytic function on  $\dot{U}_{\ell,m,n}[F] \cup (\bbC^3_{\alpha,\beta,\gamma}\backslash \cup_{k=2}^{M} \{k\gamma \in \bbN, \gamma\notin \bbN\})$. 
	
	The set $\cup_{k=2}^{M} \{k\gamma \in \bbN, \gamma\notin \bbN\}$
	is a union of hyperplanes, and it is disjoint from 
	\begin{equation} 
	\bigcup_{k=1}^N \{k(k+1)\gamma \in \bbZ^{\leq -k}\},
	\end{equation} 
	so $\dot{U}_{\ell,m,n}[F] \cup (\bbC^3_{\alpha,\beta,\gamma}\backslash \cup_{k=2}^{M} \{k\gamma \in \bbN, \gamma\notin \bbN\})$ is the complement in $\bbC^3_{\alpha,\beta,\gamma}$ of a locally finite collection of complex codimension-2 affine subspaces of $\bbC^3$. 
	The result therefore follows from Hartog's extension theorem.
\end{proof}

For any $\ell\in \bbN^+$ and $m,n\in \bbN$,  
\begin{equation} 
\{(\alpha,\beta,\gamma) \in \bbC^3: (\bmalpha,\bmbeta,\bmgamma) \in \dot{\mho}_{\ell,m,n}[F]\}= \dot{U}_{\ell,m,n}[F] \cap \dot{U}_{\ell-1,m+1,n}[F] \cap \dot{U}_{\ell-1,m,n+1}[F].
\label{eq:misc_olm}
\end{equation} 
The symmetric case of \Cref{lem:main_lemma} reads, after multiplying through by $1-e^{\pm 2i \gamma}$, 
\begin{multline} 
0= (1-e^{\pm 2\pi i \ell \gamma} )\dot{S}_{\ell,m,n}[F]( \alpha,\beta,\gamma ) + e^{\pm \pi i (\alpha+2(\ell-1) \gamma )} (1-e^{\pm 2 \pi i (m+1)\gamma} ) \dot{S}_{\ell-1,m+1,n}[F]( \alpha,\beta, \gamma ) \\ +  e^{\pm  \pi i (\alpha+\beta+ 2(\ell-1+m)\gamma)} (1-e^{\pm2  \pi i (n+1)\gamma})\dot{S}_{\ell-1,m,n+1}[F](\alpha,\beta, \gamma)
\label{eq:Aomoto_relation}
\end{multline} 
for all $(\alpha,\beta,\gamma)$ in the set defined by \cref{eq:misc_olm}. Define 
\begin{align}
O_{N;0} &= \{(\alpha,\beta,\gamma)\in \bbC^3: \alpha+j\gamma \notin \bbZ\text{ for any }j\in \{0,\ldots,N-1\}\}, \\
O_{N;1} &= \{(\alpha,\beta,\gamma)\in \bbC^3: \beta+j\gamma \notin \bbZ\text{ for any }j\in \{0,\ldots,N-1\}\}. 
\end{align} 
\begin{proposition}[Cf.\ \cite{DF2}\cite{Ao}\cite{FW}]
	\hfill 
	\begin{itemize}
		\item For all  $(\alpha,\beta,\gamma) \in \dot{U}_{N,0,0}[F] \cap \dot{U}_{0,N,0}[F] \cap O_{N;1}$,
		\begin{equation}
		\dot{S}_{0,N,0}[F](\alpha,\beta,\gamma) = (-1)^N \Big[ \prod_{m=0}^{N-1} \frac{\sin(\pi(\alpha+\beta+(N+m-1)\gamma))}{\sin(\pi(\beta+m\gamma ))} \Big] \dot{S}_{N,0,0}[F](\alpha,\beta,\gamma).
		\end{equation}
		\item For all $(\alpha,\beta,\gamma) \in \dot{U}_{0,N,0}[F] \cap \dot{U}_{0,0,N}[F] \cap O_{N;0}$, 
		\begin{equation}
		\dot{S}_{0,N,0}[F](\alpha,\beta,\gamma) = (-1)^N \Big[ \prod_{m=0}^{N-1}  \frac{\sin(\pi(\alpha+\beta+(N+m-1)\gamma))}{\sin(\pi(\alpha+m\gamma)) } \Big] \dot{S}_{0,0,N}[F](\alpha,\beta,\gamma).
		\label{eq:misc_ju6}
		\end{equation}
	\end{itemize}
	\label{prop:Aomoto}
\end{proposition}
\begin{proof}
	We prove the second claim, and the proof of the first is similar. Suppose that 
	\begin{equation} 
	(\alpha,\beta,\gamma) \in \bigcap_{n=0}^N \dot{U}_{0,N-n,n}[F] \cap \bigcap_{n=0}^{N-1} \dot{U}_{1,N-1-n,n}[F].
	\end{equation} 
	We can apply \cref{eq:Aomoto_relation} for $\ell=1$ and all pairs of $m,n \in \{0,\ldots,N-1\}$ such that $m+n = N-1$.
	Combining the plus and minus cases of \cref{eq:Aomoto_relation} to eliminate the $\dot{S}_{1,N-n-1,n}[F]$ term, 
	\begin{multline}
	\frac{1}{2 i} \Big[ e^{+\pi i \alpha} \frac{1-e^{+2 \pi i (m+1)\gamma} }{1-e^{+2 \pi  i \gamma}} -e^{-\pi i \alpha} \frac{1-e^{-2 \pi i (m+1)\gamma} }{1-e^{-2 \pi  i \gamma}}\Big] \dot{S}_{0,N-n,n}[F](\alpha,\beta,\gamma) \\ = - \frac{1}{2 i} \Big[ e^{ +\pi i (\alpha+\beta + 2 (N-n-1) \gamma )} \frac{1-e^{+2 \pi i (n+1)\gamma} }{1-e^{+2 \pi  i \gamma}} - e^{ -\pi i (\alpha+\beta + 2 (N-n-1) \gamma )} \frac{1-e^{-2 \pi i (n+1)\gamma} }{1-e^{-2 \pi  i \gamma}}\Big] \\ \times  \dot{S}_{0,N-n-1,n+1}[F](\alpha,\beta,\gamma)
	\end{multline}
	if $\gamma \notin \bbZ$. 
	We calculate: 
	\begin{equation}
	\frac{1}{2 i}  \Big[ e^{+\pi i \alpha} \frac{1-e^{+2 \pi i (N-n)\gamma} }{1-e^{+2 \pi  i \gamma}} -  e^{-\pi i \alpha} \frac{1-e^{-2 \pi i (N-n)\gamma} }{1-e^{-2 \pi  i \gamma}}\Big] =  \frac{2s(\gamma)s(\alpha+(N-n-1)\gamma)s((N-n)\gamma)}{1-\cos(2\pi \gamma)} 
	\label{eq:misc_n41}
	\end{equation}
	and 
	\begin{multline}
	\frac{1}{2 i} \Big[ e^{ +\pi i (\alpha+\beta + 2 (N-n-1) \gamma )} \frac{1-e^{+2 \pi i (n+1)\gamma} }{1-e^{+2 \pi  i \gamma}} - e^{ -\pi i (\alpha+\beta + 2 (N-n-1) \gamma )} \frac{1-e^{-2 \pi i (n+1)\gamma} }{1-e^{-2 \pi  i \gamma}}\Big]  \\ =  \frac{2s(\gamma)}{1-\cos(2\pi \gamma)} s(\alpha + \beta + (2N-n-2)\gamma)s((n+1)\gamma), 
	\end{multline}
	where $s(t) = \sin(\pi t)$. 
	So, for $(\alpha,\beta,\gamma)$ as above such that none of the trigonometric factors on the right-hand side of \cref{eq:misc_n41} vanish, 
	\begin{equation}
	\dot{S}_{0,N-n,n}[F]( \alpha,\beta, \gamma ) =-  \frac{s(\alpha+\beta+(2N-n-2)\gamma) s((n+1)\gamma)}{s(\alpha+(N-n-1)\gamma) s((N-n)\gamma)} \dot{S}_{0,N-1-n,n+1}[F](\alpha,\beta, \gamma).
	\label{eq:misc_zy1}
	\end{equation}
	Applying this recursively for $n=0,\ldots,N-1$, we end up with \cref{eq:misc_ju6}. 
	
	In summary, \cref{eq:misc_ju6} holds for a nonempty, open subset of $(\alpha,\beta,\gamma) \in \dot{U}_{0,N,0}[F]\cap \dot{U}_{0,0,N}[F] \cap O_{N;0}$. 
	By analyticity, the result follows.
\end{proof}

\begin{proposition}
	The function $S_{N;\mathrm{Reg}}[F](\alpha,\beta,\gamma)$ defined by 
	\begin{equation}
	S_{N;\mathrm{Reg}}[F](\alpha,\beta,\gamma) =\Big[\prod_{j=1}^N \frac{\Gamma(2+\bar{d}_j+\alpha+\beta+(N+j-2)\gamma)} {\Gamma(1+\bar{\delta}_j +\alpha+(j-1)\gamma) \Gamma ( 1+\bar{\atled}_j  +\beta+(j-1)\gamma)\digamma_j(\gamma)} \Big]  S_N[F](\alpha,\beta,\gamma)
	\label{eq:misc_b1y}
	\end{equation}
	extends to an entire function $S_{N;\mathrm{Reg}}[F]:\bbC^3_{\alpha,\beta,\gamma}\to \bbC$. 
	\label{prop:symmetric_final}
\end{proposition}
\begin{proof}
	We begin by defining the following open (and dense) subsets of $\bbC^3$: 
	\begin{align}
	\begin{split} 
	Q_{N;0} &= \{(\alpha,\beta,\gamma)\in \bbC^3: \bar{\delta}_j+ \alpha+(j-1)\gamma \notin \bbN\text{ for any }j\in \{1,\ldots,N\}\}, \\
	Q_{N;1} &= \{(\alpha,\beta,\gamma)\in \bbC^3: \bar{\atled}_j+ \beta+(j-1)\gamma \notin \bbN\text{ for any }j\in \{1,\ldots,N\}\},\\
	Q_{N;\infty} &= \{(\alpha,\beta,\gamma)\in \bbC^3: \bar{d}_j+ \alpha+\beta+(N+j-2)\gamma \notin \bbZ^{\leq -2}\text{ for any }j\in \{1,\ldots,N\}\}, \\
	U_{N;0} &= \{(\alpha,\beta,\gamma)\in \bbC^3: \delta_j+ j(\alpha+(j-1)\gamma) \notin \bbZ^{\leq -j}\text{ for any }j\in \{1,\ldots,N\}\}, \\
	U_{N;1} &= \{(\alpha,\beta,\gamma)\in \bbC^3: \atled_j+ j(\beta+(j-1)\gamma) \notin \bbZ^{\leq -j}\text{ for any }j\in \{1,\ldots,N\}\},\\
	U_{N;\infty} &= \{(\alpha,\beta,\gamma)\in \bbC^3:-d_j- j(1 + \alpha+\beta+(N+j-2)\gamma) \notin \bbZ^{\leq 0}\text{ for any }j\in \{1,\ldots,N\}\} \\ 
	&= \{(\alpha,\beta,\gamma)\in \bbC^3:d_j+ j(1 + \alpha+\beta+(N+j-2)\gamma) \notin \bbN\text{ for any }j\in \{1,\ldots,N\}\}. 
	\end{split} 
	\end{align}
	\begin{figure}
		\begin{center}
			\includegraphics[scale=.35]{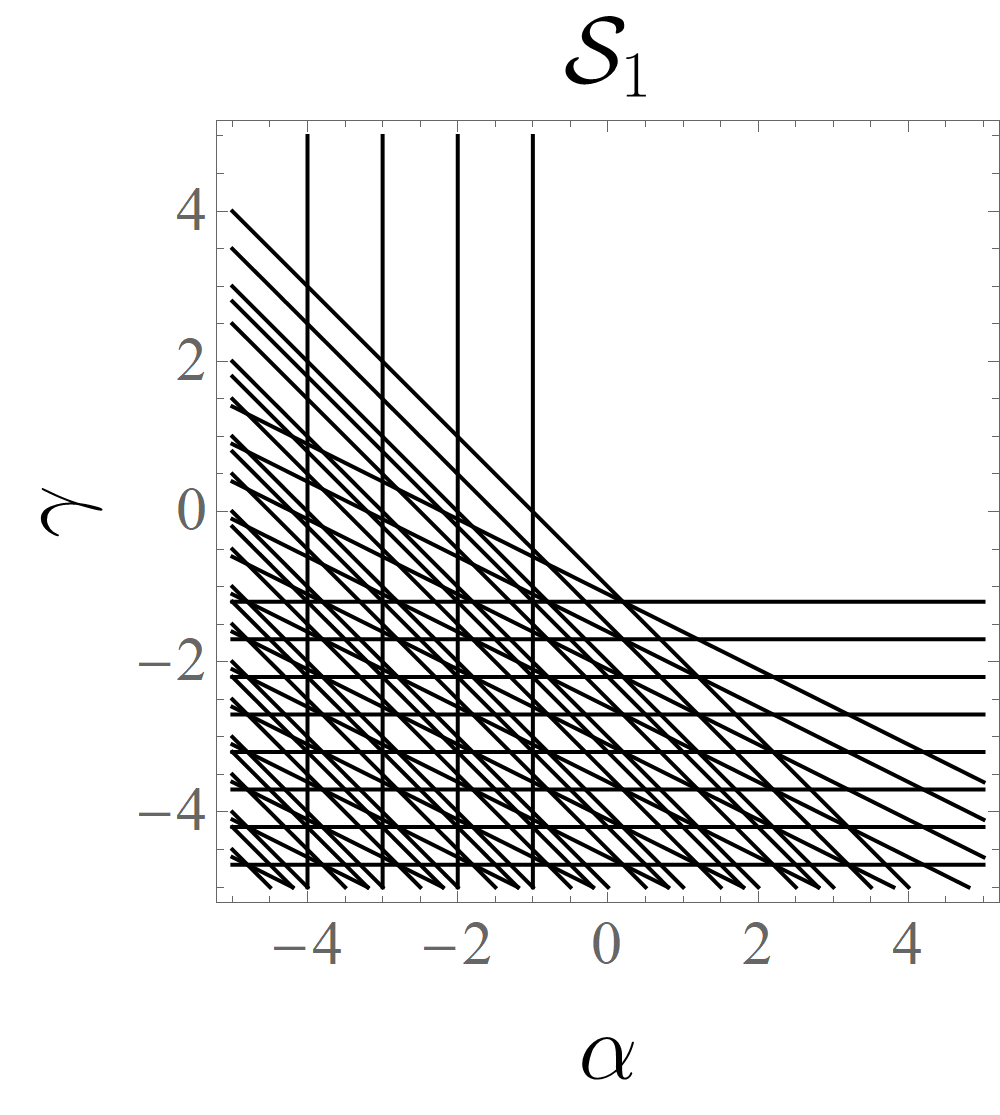}
			\includegraphics[scale=.35]{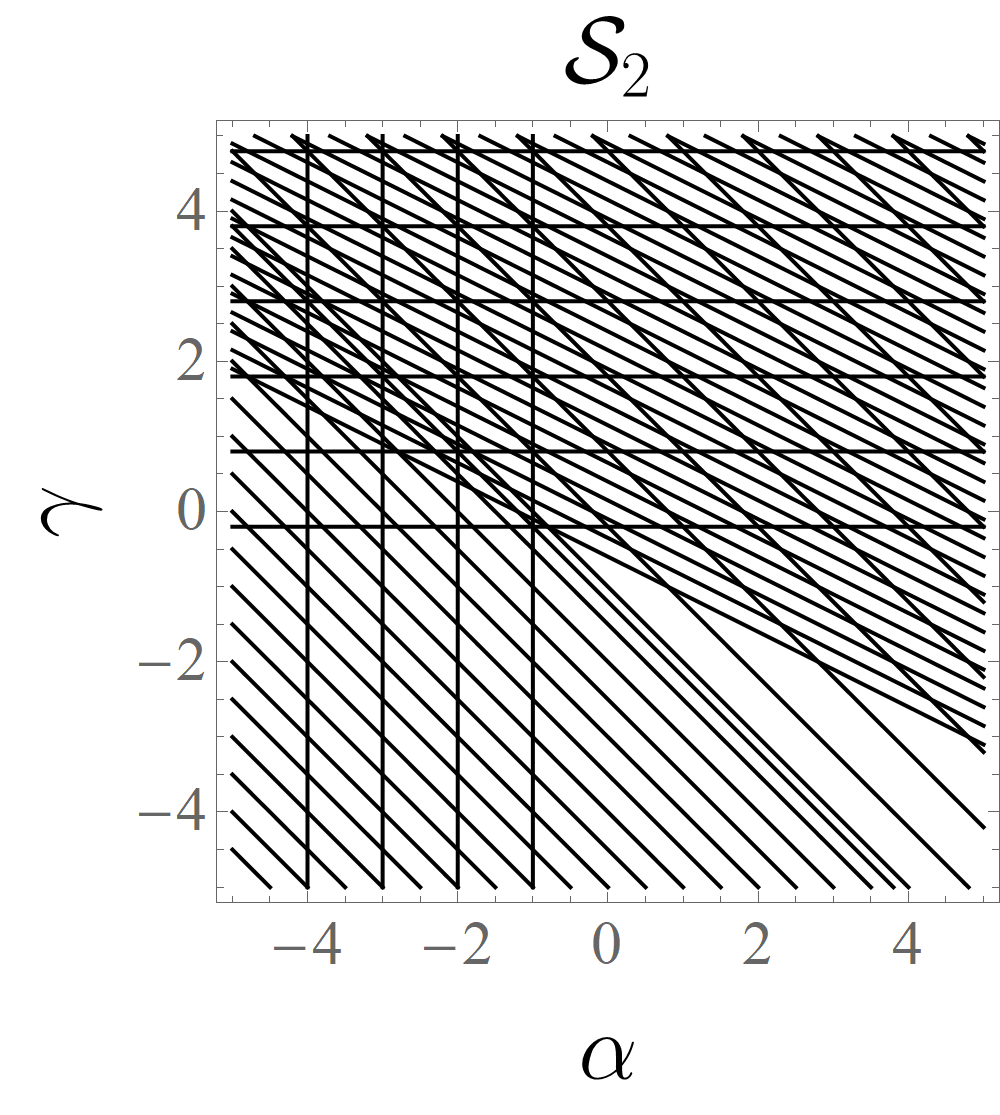}
			\includegraphics[scale=.35]{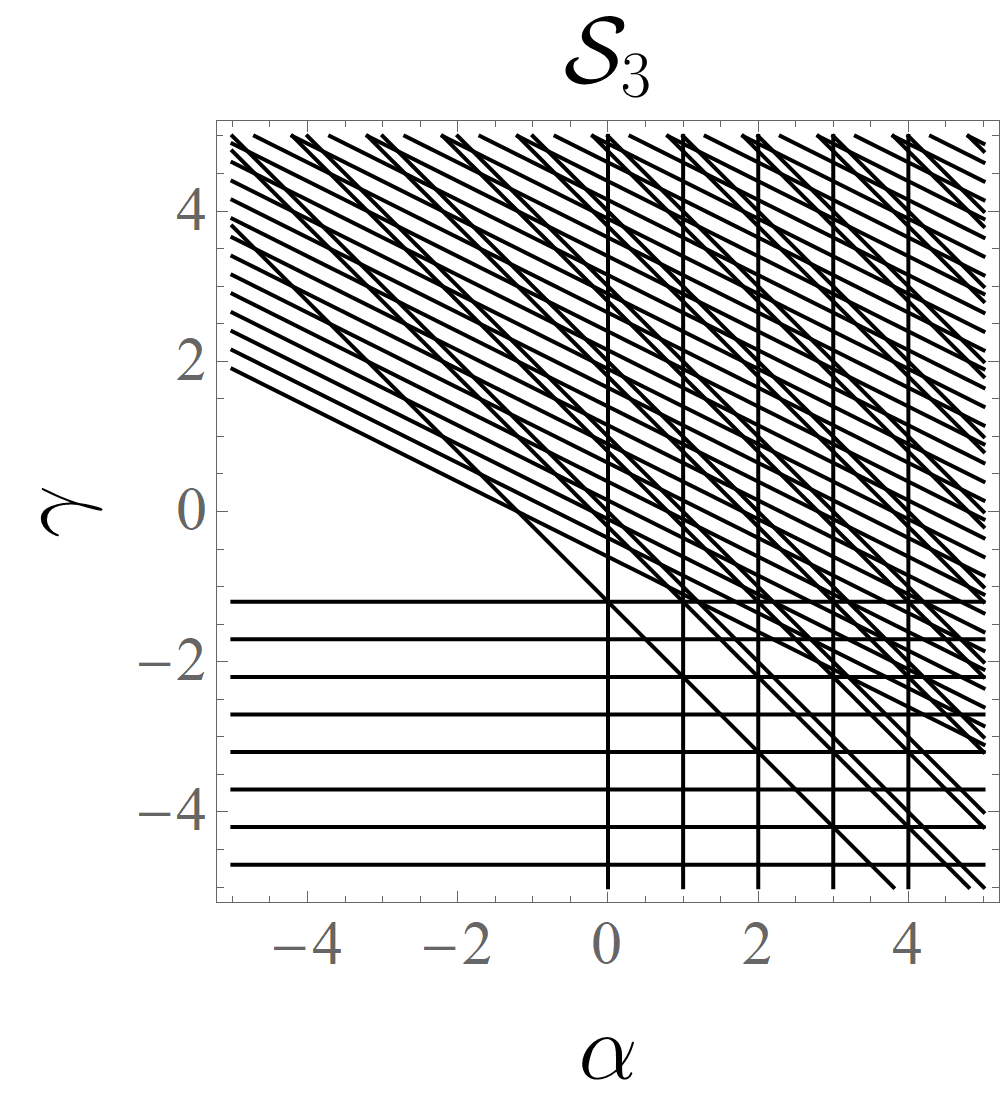}
		\end{center}
		\caption{The sets in $\calS_1,\calS_2,\calS_3$ in $\bbR^3_{\alpha,\beta,\gamma}\cap \{\beta=1/5\}$ in the case $N=2$.}
	\end{figure}
	We write 
	\begin{multline}
	S_{N;\mathrm{Reg}}[F](\alpha,\beta,\gamma) = \Upsilon_0(\alpha,\beta,\gamma) \Upsilon_1(\alpha,\beta,\gamma) \\ \times \Big[\prod_{j=1}^N \Gamma(\delta_j+j(1+\alpha+(j-1)\gamma))\Gamma(\atled_j+j(1+\beta+(j-1)\gamma)) \digamma_j(\gamma) \Big]^{-1} S_N[F](\alpha,\beta,\gamma)
	\label{eq:misc_hhj}
	\end{multline}
	for 
	\begin{align}
	\Upsilon_0(\alpha,\beta,\gamma) &=  \prod_{j=1}^N \frac{\Gamma(\delta_j+j(1+\alpha+(j-1)\gamma))\Gamma(\atled_j+j(1+\beta+(j-1)\gamma))}{\Gamma(1+ \bar{\delta}_j+ \alpha+(j-1)\gamma) \Gamma(1+\bar{\atled}_j+ \beta+(j-1)\gamma)}   , \\ 
	\Upsilon_1(\alpha,\beta,\gamma)  &= \prod_{j=1}^N \Gamma(2+\bar{d}_j+\alpha+\beta+(N+j-2)\gamma) . 
	\end{align}
	By \Cref{prop:three_final}, the second line on the right-hand side of \cref{eq:misc_hhj} defines an entire function.  
	Since $\Upsilon_0$ extends to an analytic function on $U_{N;0}\cap U_{N;1}$ and $\Upsilon_1$ extends to an analytic function on $Q_{N;\infty}$, $S_{N;\mathrm{Reg}}[F]$ extends to an analytic function on $U_{N;0}\cap U_{N;1} \cap Q_{N;\infty}$. 
	
	In $O_{N;0}  \cap \dot{U}_{0,N,0} \cap \dot{U}_{0,0,N}$, \Cref{prop:Aomoto} gives 
	\begin{multline}
	S_{N;\mathrm{Reg}}[F](\alpha,\beta,\gamma) =  (-1)^N  \Upsilon_2(\alpha,\beta,\gamma)\Upsilon_3(\alpha,\beta,\gamma)  \\ \times  \Big[ \prod_{j=1}^N  \Gamma(-d_j-j(1+\alpha+\beta+(N+j-2)\gamma))\Gamma(\atled_j+j( 1+\beta+(j-1)\gamma)) \digamma_j(\gamma)   \Big]^{-1} \dot{S}_{0,0,N}[F](\alpha,\beta,\gamma),
	\label{eq:misc_mmm}
	\end{multline}
	where
	\begin{equation}
	\Upsilon_2 = \prod_{j=1}^N \frac{ \Gamma(\atled_j+j(1+\beta+(j-1)\gamma))}{s(\alpha+(j-1)\gamma ) \Gamma(1+ \bar{\delta}_j+ \alpha+(j-1)\gamma) \Gamma(1+\bar{\atled}_j+\beta+(j-1)\gamma)}  
	\end{equation}
	\begin{multline}
	\Upsilon_3 = \prod_{j=1}^Ns(\alpha+\beta+(N+j-2) \gamma) \Gamma(2+\bar{d}_j+\alpha+\beta+(N+j-2)\gamma)  \\  \times \Gamma(-d_j-j(1+\alpha+\beta+(N+j-2)\gamma)). 
	\end{multline}
	By \Cref{prop:three_final}, the function on the second line of \cref{eq:misc_mmm} extends to an entire function of $\alpha,\beta,\gamma$. 
	On the other hand, $\Upsilon_2$ extends to an analytic function on $Q_{N;0}\cap U_{N;1}$, and $\Upsilon_3$ extends to an analytic function on $U_{N;\infty}$. Combining these observations, $S_{N;\mathrm{Reg}}[F]$ analytically continues to $Q_{N;0}\cap U_{N;1}\cap U_{N;\infty}$. 
	
	Likewise, $S_{N;\mathrm{Reg}}[F]$ extends analytically to $U_{N;0}\cap Q_{N;1} \cap U_{N;\infty}$, using $O_{N;1}$ in place of $O_{N;0}$ and the other part of \Cref{prop:Aomoto}. 
	
	So, $S_{N;\mathrm{Reg}}[F](\alpha,\beta,\gamma)$ analytically continues to 
	\begin{equation}
	U = (U_{N;0} \cap U_{N;1} \cap Q_{N;\infty} )\cup (U_{N;0} \cap Q_{N;1} \cap U_{N;\infty} ) \cup (Q_{N;0} \cap U_{N;1} \cap U_{N;\infty} ).
	\end{equation}
	This is 
	\begin{equation}
	U= \bbC^3 \Big\backslash \Big[ \Big(\bigcup_{H_1 \in \calS_1,H_2\in \calS_2,H_3 \in \calS_3} H_1\cap H_2 \cap H_3) \Big],
	\end{equation}
	where 
	\begin{itemize}
		\item 
		$\calS_1$ is the set of hyperplanes that are contained in the complement of one of $U_{N;0},U_{N;1},Q_{N;\infty}$, 
		\item  $\calS_2$ is the set of hyperplanes that are contained in the complement of one of $U_{N;0},Q_{N;1},U_{N;\infty}$, and 
		\item  $\calS_3$ is the set of hyperplanes that are contained in the complement of one of $Q_{N;0},U_{N;1},U_{N;\infty}$. 
	\end{itemize} 
	Let 
	\begin{equation} 
	\calH = \{H_1 \cap H_2 \cap H_3\neq \varnothing: H_1\in \calS_1, H_2 \in \calS_2, H_3\in \calS_3\},
	\end{equation}
	so that $S_{N;\mathrm{Reg}}[F]$ defines an analytic function on $U=\bbC^3\backslash \smash{\cup_{H\in \calH}} H$. 
	Observe that every $H\in \calH$ is an affine subspace of $\bbC^3$ of complex codimension two or three (since $\calS_1 \cap \calS_2\cap \calS_3=\varnothing$), and the collection $\calH$ is locally finite.

	Hartog's theorem therefore implies that $S_{N;\mathrm{Reg}}[F]$ analytically continues to the entirety of $\bbC^3$. 
\end{proof}

This completes the proof of \Cref{thm:symmetric}.

\subsection{The DF-symmetric case}
\label{subsec:removing_singularities_2}

Given $\gamma_+ \in \bbC\backslash \{0,1 \}$ and $\alpha_+,\beta_+\in \bbC $, let $\gamma_- = \gamma_+^{-1}$, $\alpha_- = - \gamma_- \alpha_+$, and $\beta_- = - \gamma_- \beta_+$ as in the introduction. Fix $\mathtt{S}\subseteq \{1,\ldots,N\}$. 

Given $\gamma_+ \neq 0,1$ and $F\in \operatorname{DFSym}(N;\mathtt{S},\lambda)$ for $\lambda = \gamma_+^{-1}(\gamma_+-1)$, let $\dot{W}_{\ell,m,n}^{\mathrm{DF},\mathtt{S}}[F;\gamma_+]$ denote the set of $(\alpha_+,\beta_+) \in \bbC^2$ such that 
\begin{equation}
(\alpha_-,\alpha_+,\beta_-,\beta_+,\gamma_-,-1,\gamma_+) \in \dot{W}_{\ell,m,n}^{\mathrm{DF0},\mathtt{S}}[F].
\end{equation}
For $(\alpha_+,\beta_+)\in \dot{W}_{\ell,m,n}^{\mathrm{DF},\mathtt{S}}[F;\gamma_+]$, let 
\begin{equation}
\dot{I}_{\ell,m,n}^{\mathrm{DF};\mathtt{S}}[F](\alpha_+,\beta_+,\gamma_+) = I_{\ell,m,n}^{\mathrm{DF0};\mathtt{S}}[F](\alpha_-,\alpha_+,\beta_-,\beta_+,\gamma_-,-1,\gamma_+).
\end{equation}

Then, as adumbrated by Dotsenko and Fateev:
\begin{proposition} 
	For any $\sigma \in \frakS_{\ell,m,n}$, 
	\begin{equation} 
	\dot{I}_{\ell,m,n}^{\mathrm{DF};\mathtt{S}}[F](\alpha_+,\beta_+,\gamma_+)=\dot{I}_{\ell,m,n}^{\mathrm{DF};\mathtt{S}}[F](\alpha_+,\beta_+,\gamma_+)^\sigma 
	\label{eq:misc_nv1}
	\end{equation}
	for all $(\alpha_+,\beta_+) \in \dot{W}_{\ell,m,n}^{\mathrm{DF},\mathtt{S}}[F;\gamma_+]$.
	\label{prop:switching} 
\end{proposition}

Since $\dot{W}_{\ell,m,n}^{\mathrm{DF},\mathtt{S}}[F]$ depends only on $\mathtt{S}$ through $|\mathtt{S} \cap \calI_1|,|\mathtt{S} \cap \calI_2|,|\mathtt{S} \cap \calI_3|$, 
\begin{equation}
\dot{W}_{\ell,m,n}^{\mathrm{DF},\mathtt{S}}[F;\gamma_+] = \dot{W}_{\ell,m,n}^{\mathrm{DF},\mathtt{S}}[F;\gamma_+]^\sigma, 
\end{equation}
so the right-hand side of \cref{eq:misc_nv1} is defined for any $(\alpha_+,\beta_+) \in \dot{W}_{\ell,m,n}^{\mathrm{DF},\mathtt{S}}[F;\gamma_+]$.
\begin{proof}
	Since $\frakS_{\ell,m,n}$ is generated by transpositions $\tau$ of adjacent elements of $\calI_1,\calI_2,\calI_3$, it suffices to consider the case when $\sigma$ is such a transposition, $\tau$. For notational simplicity, we consider the case when $\tau$ is a transposition of some $j,j+1\in \calI_2$ and $j\in \mathtt{S}$. The other cases are similar but involve some notational changes.

	Let $\dot{W}_{\ell,m,n}^{\mathrm{DF},1,\mathtt{S}}[F;\gamma_+] \subseteq \bbC^6$ denote the set of $(\alpha_{1,+},\alpha_{2,+},\alpha_{3,+},\beta_{1,+},\beta_{2,+},\beta_{3,+})\in \bbC^6$ such that 
	\begin{equation}
	(\alpha_{-,1}, \alpha_{-,2},\alpha_{-,3},\alpha_{+,1}, \alpha_{+,2},\alpha_{+,3},\beta_{-,1}, \beta_{-,2},\beta_{-,3},\beta_{+,1}, \beta_{+,2},\beta_{+,3},\gamma_-,-1,\gamma_+) \in \dot{W}_{\ell,m,n}^{\mathrm{DF1},\mathtt{S}}[F], 
	\end{equation}
	where $\alpha_{-,\nu} = - \gamma_- \alpha_{+,\nu}$ and $\beta_{-,\nu} = - \gamma_- \beta_{+,\nu}$. It suffices to prove that, for any $\sigma \in \frakS_{\ell,m,n}$, 
	\begin{multline} 
	\dot{I}_{\ell,m,n}^{\mathrm{DF},1,;\mathtt{S}}[F](\alpha_{-,1}, \alpha_{-,2},\alpha_{-,3},\alpha_{+,1}, \alpha_{+,2},\alpha_{+,3},\beta_{-,1}, \beta_{-,2},\beta_{-,3},\beta_{+,1}, \beta_{+,2},\beta_{+,3},\gamma_+) \\ =\dot{I}_{\ell,m,n}^{\mathrm{DF},1;\mathtt{S}}[F](\alpha_{-,1}, \alpha_{-,2},\alpha_{-,3},\alpha_{+,1}, \alpha_{+,2},\alpha_{+,3},\beta_{-,1}, \beta_{-,2},\beta_{-,3},\beta_{+,1}, \beta_{+,2},\beta_{+,3},\gamma_+)^\sigma 
	\label{eq:misc_nv2}
	\end{multline}
	for all $(\alpha_{1,+},\alpha_{2,+},\alpha_{3,+},\beta_{1,+},\beta_{2,+},\beta_{3,+}) \in \dot{W}_{\ell,m,n}^{\mathrm{DF},1,\mathtt{S}}[F;\gamma_+]$, where 
	\begin{multline}
	\dot{I}_{\ell,m,n}^{\mathrm{DF},1;\mathtt{S}}[F] (\alpha_{-,1}, \alpha_{-,2},\alpha_{-,3},\alpha_{+,1}, \alpha_{+,2},\alpha_{+,3},\beta_{-,1}, \beta_{-,2},\beta_{-,3},\beta_{+,1}, \beta_{+,2},\beta_{+,3},\gamma_+) = \\ 
	\dot{I}_{\ell,m,n}[F](\bmalpha^{\mathrm{DF},1},\bmbeta^{\mathrm{DF},1},\bmgamma^{\mathrm{DF}}), 
	\end{multline}
	where $\bmalpha^{\mathrm{DF},1},\bmbeta^{\mathrm{DF},1}$ are defined as $\bmalpha^{\mathrm{DF}1},\bmbeta^{\mathrm{DF}1}$, using 
	$\alpha_{-,\nu} = - \gamma_- \alpha_{+,\nu}$ and $\beta_{-,\nu} = - \gamma_- \beta_{+,\nu}$. 
	
	First observe that there exists a nonempty, open subset 
	\begin{equation} 
	O\subset  \dot{W}_{\ell,m,n}^{\mathrm{DF},1,\mathtt{S}}[F;\gamma_+]
	\end{equation} 
	(containing an affine cone)
	such that $(\bmalpha^{\mathrm{DF},1},\bmbeta^{\mathrm{DF},1},\bmgamma^{\mathrm{DF}}) \in O^{\{1,\tau\}}$ whenever $(\alpha_{+,1},\ldots,\beta_{+,3})\in O$, where $O^{\{1,\tau\}}$ is defined as in \S\ref{subsec:alternative}. We can choose $O$ such that $\Re \alpha_{\pm,2},\Re \beta_{\pm,2}>0$ everywhere in $O$. 
	
	Since $\dot{W}_{\ell,m,n}^{\mathrm{DF},\mathtt{S}}[F;\gamma_+]$ is connected, it suffices via analyticity to prove the result for \begin{equation} 
	(\alpha_{1,+},\alpha_{2,+},\alpha_{3,+},\beta_{1,+},\beta_{2,+},\beta_{3,+})\in O.
	\end{equation} 
	We write $\alpha_\pm$ in place of $\alpha_{\pm,2}$ and $\beta_\pm$ in place of $\beta_{\pm,2}$ below. 
	
	We can apply \Cref{prop:DF_transposition} for $(\alpha_{+,1},\ldots,\beta_{+,3}) \in O$. 
	By \Cref{prop:DF_transposition}, it suffices to check that, whenever all of the $z_k$'s  besides $z_j$ and $z_{j+1}$ are somewhere in the interior of the corresponding contour in \cref{eq:misc_dff}, 
	\begin{equation}
	\int_{\Gamma_{[0,1],+,j-\ell}} \oint  z_j^{\alpha_+}(1-z_j)^{\beta_+} z_{j+1}^{\alpha_-} (1-z_{j+1})^{\beta_-} 
	\Big( \prod_{\substack{ 1\leq j_0 < k_0 \leq N \\ \{j_0,k_0\} \cap \{j,j+1\}\neq \varnothing } } (z_{k_0}-z_{j_0})^{2\gamma_{j_0,k_0}} \Big) F \dd z_j \dd z_{j+1} =0,
	\end{equation} 
	where the inner integral is taken over a small circle around $z_{j+1}$, for each $z_{j+1}\in \Gamma_{[0,1],+,j-\ell} \backslash \{0,1\}$. 
	
	Since the integrand is holomorphic in $z_j$ in a punctured neighborhood of $z_{j+1}$, we apply the Cauchy residue theorem to deduce that the left-hand side is proportional to 
	\begin{equation}
	\int_{\Gamma_{[0,1],+,j-\ell}}G_0 \frac{\partial G}{\partial z_j}\Big|_{z_j = z_{j+1}} \dd z_{j+1}, 
	\label{eq:misc_11c}
	\end{equation}
	where 
	\begin{align}
	G_0(z_1,\ldots,z_N) &=  z_{j+1}^{\alpha_-}  (1-z_{j+1})^{\beta_-}  \Big[\prod_{j_0 \in ([N]\backslash \mathtt{S}) \backslash \{j+1\}} (z_{j+1}-z_{j_0})^{2\gamma_-}\Big] \Big[\prod_{j_0 \in \mathtt{S}\backslash \{j\}} (z_{j+1}-z_{j_0})^{-2}\Big], \\
	G(z_1,\ldots,z_N) &=  z_j^{\alpha_+}(1-z_j)^{\beta_+}   \Big[\prod_{j_0 \in ([N]\backslash \mathtt{S}) \backslash \{j+1\}} (z_{j}-z_{j_0})^{-2}\Big]
	\Big[\prod_{j_0 \in \mathtt{S}\backslash \{j\}} (z_{j}-z_{j_0})^{2\gamma_+}\Big] 
	F.
	\end{align}
	We are choosing branch cuts such that we do not encounter any as $z_j,z_{j+1}$ are integrated along $\Gamma_{[0,1],+,j-\ell}$ (except at the endpoints). Other than that, it is not important what the precise choice of branch cuts are.

	The integrand in \cref{eq:misc_11c} is computed to be 
	\begin{multline}
	G_0 \frac{\partial G}{\partial z_j}\Big|_{z_j=z_{j+1} } = \Big[\frac{\alpha_+}{z_{j+1}} - \frac{\beta_+}{1-z_{j+1}} + 2\gamma_+\sum_{j_0\in \mathtt{S}\backslash \{j\}} \frac{1}{z_{j+1}-z_{j_0}} - 2 \sum_{j_0 \in ([N]\backslash \mathtt{S}) \backslash \{j+1\}} \frac{1}{z_{j+1}-z_{j_0}} \\ +\frac{\partial_{z_j} F}{F}\Big|_{z_j=z_{j+1}} \Big] H, 
	\end{multline}
	where $H = G_0 G|_{z_j = z_{j+1}}$. On the other hand, 
	\begin{multline}
	\frac{\partial H}{\partial z_{j+1}} =   \Big[\frac{\alpha_-+\alpha_+}{z_{j+1}} - \frac{\beta_-+\beta_+}{1-z_{j+1}} + (2\gamma_+-2)\sum_{j_0\in \mathtt{S}\backslash \{j\}} \frac{1}{z_{j+1}-z_{j_0}} + (2\gamma_--2) \sum_{j_0 \in ([N]\backslash \mathtt{S}) \backslash \{j+1\}} \frac{1}{z_{j+1}-z_{j_0}} \\ +\frac{\partial_{z_{j+1}}( F|_{z_j=z_{j+1}})}{F|_{z_j=z_{j+1}}} \Big] H.
	\end{multline}
	Since $1-\gamma_-=\alpha_+^{-1} (\alpha_- + \alpha_+) = \beta_+^{-1}(\beta_-+\beta_+) = \gamma_+^{-1}(\gamma_+-1)=\lambda$, and since $F\in \operatorname{DFSym}(N,\mathtt{S},\lambda)$, 
	\begin{equation}
	G_0 \frac{\partial G}{\partial z_j}\Big|_{z_j=z_{j+1} }  = \frac{1}{\lambda} \frac{\partial H}{\partial z_{j+1}}.
	\end{equation}
	Consequently, 
	\begin{equation}
	\int_{\Gamma_{[0,1],+,j-\ell}} G_0 \frac{\partial G}{\partial z_j}\Big|_{z_j=z_{j+1} } \dd z_{j+1} = \frac{1}{\lambda} \int_{\Gamma_{[0,1],+,j-\ell}} \frac{\partial H}{\partial z_{j+1}} \dd z_{j+1}.
	\label{eq:misc_nh1}
	\end{equation}
	The right-hand side is proportional to 
	\begin{equation} 
	\int_p  \frac{\partial H}{\partial z_{j+1}} \dd z_{j+1}
	\label{eq:misc_j41}
	\end{equation} 
	if $\alpha_+,\beta_+\notin \bbZ$, where $p$ is a Pochhammer contour in $\bbC^2\backslash \{0,1\}$ staying sufficiently close to $\Gamma_{[0,1],+,j-\ell}$.
	Lifting to a cover of a neighborhood of $\Gamma_{[0,1],+,j-\ell}$ on which $H$ lifts to a single-valued analytic function, we can conclude (using analyticity) that the integral in \cref{eq:misc_j41} is zero. By analyticity, we can remove the nonintegrality constraint on $\alpha_+,\beta_+$ to conclude that 
	\begin{equation}
	\int_{\Gamma_{[0,1],+,j-\ell}} \frac{\partial H}{\partial z_{j+1}} \dd z_{j+1} = 0
	\end{equation}
	for all $(\alpha_{+,1},\ldots,\beta_{+,3}) \in O$.

\end{proof}

\begin{proposition}[Cf.\ \cite{DF2}]
	Given the setup above, for arbitrary $\mathtt{S}$:
	\begin{itemize}
		\item For all $(\alpha_+,\beta_+) \in \dot{W}_{N,0,0}^{\mathrm{DF},\mathtt{S}}[F;\gamma_+] \cap \dot{W}_{0,N,0}^{\mathrm{DF},\mathtt{S}}[F;\gamma_+]$ such that $\beta_\pm+m_\pm\gamma_\pm \notin \bbZ$ for any $m_\pm\in \{0,\ldots,N_\pm-1\}$
		\begin{multline}
		\dot{I}_{0,N,0}^{\mathrm{DF},\mathtt{S}}[F](\alpha_+,\beta_+,\gamma_+) = (-1)^N \dot{I}_{N,0,0}^{\mathrm{DF},\mathtt{S}}[F](\alpha_+,\beta_+,\gamma_+) \\ \times  \Big[ \prod_{m_+=0}^{N_+-1} \frac{\sin(\pi(\alpha_++\beta_++(N_++m_+-1)\gamma_+))}{\sin(\pi(\beta_++m_+\gamma_+ ))} \Big] \Big[ \prod_{m_-=0}^{N_--1} \frac{\sin(\pi(\alpha_-+\beta_-+(N_-+m_--1)\gamma_-))}{\sin(\pi(\beta_-+m_-\gamma_- ))} \Big]. 
		\label{eq:misc_ff1}
		\end{multline}
		\item For all $(\alpha_+,\beta_+) \in \dot{W}_{0,N,0}^{\mathrm{DF},\mathtt{S}}[F;\gamma_+] \cap \dot{W}_{0,0,N}^{\mathrm{DF},\mathtt{S}}[F;\gamma_+]$ such that $\alpha_\pm+m_\pm\gamma_\pm \notin \bbZ$ for any $m_\pm\in \{0,\ldots,N_\pm-1\}$, 
		\begin{multline}
		\dot{I}_{0,N,0}^{\mathrm{DF},\mathtt{S}}[F](\alpha_+,\beta_+,\gamma_+) = (-1)^N \dot{I}_{0,0,N}^{\mathrm{DF},\mathtt{S}}[F](\alpha_+,\beta_+,\gamma_+)
		\\ \times\Big[ \prod_{m_+=0}^{N_+-1}  \frac{\sin(\pi(\alpha_++\beta_++(N_++m_+-1)\gamma_+))}{\sin(\pi(\alpha_++m_+\gamma_+)) } \Big] \Big[ \prod_{m_-=0}^{N_--1}  \frac{\sin(\pi(\alpha_-+\beta_-+(N_-+m_--1)\gamma_-))}{\sin(\pi(\alpha_-+m_-\gamma_-)) } \Big] .
		\label{eq:misc_ff2}
		\end{multline}
	\end{itemize}
	For $\mathtt{S} = \{1,\ldots,N_+\}$, we also have:
	\begin{itemize}
		\item For all $(\alpha_+,\beta_+) \in \dot{W}_{N_+,N_-,0}^{\mathrm{DF},\mathtt{S}}[F;\gamma_+] \cap \dot{W}_{0,N,0}^{\mathrm{DF},\mathtt{S}}[F;\gamma_+]$ such that $\beta_++m_+\gamma_+ \notin \bbZ$ for any $m_+\in \{0,\ldots,N_+-1\}$
		\begin{multline}
		\dot{I}_{0,N,0}^{\mathrm{DF},\mathtt{S}}[F](\alpha_+,\beta_+,\gamma_+) = (-1)^{N_+} \dot{I}_{N_+,N_-,0}^{\mathrm{DF},\mathtt{S}}[F](\alpha_+,\beta_+,\gamma_+) \\ \times \Big[ \prod_{m_+=0}^{N_+-1} \frac{\sin(\pi(\alpha_++\beta_++(N_++m_+-1)\gamma_+))}{\sin(\pi(\beta_++m_+\gamma_+ ))} \Big]. 
		\label{eq:misc_ff3}
		\end{multline}
		\item For all $(\alpha_+,\beta_+) \in \dot{W}_{0,N,0}^{\mathrm{DF},\mathtt{S}}[F;\gamma_+] \cap \dot{W}_{0,N_+,N_-}^{\mathrm{DF},\mathtt{S}}[F;\gamma_+]$ such that $\alpha_-+m_-\gamma_- \notin \bbZ$ for any $m_-\in \{0,\ldots,N_--1\}$, 
		\begin{multline}
		\dot{I}_{0,N,0}^{\mathrm{DF},\mathtt{S}}[F](\alpha_+,\beta_+,\gamma_+) = (-1)^{N_-} \dot{I}_{0,N_+,N_-}^{\mathrm{DF},\mathtt{S}}[F](\alpha_+,\beta_+,\gamma_+)
		\\ \times \Big[ \prod_{m_-=0}^{N_--1}  \frac{\sin(\pi(\alpha_-+\beta_-+(N_-+m_--1)\gamma_-))}{\sin(\pi(\alpha_-+m_-\gamma_-)) } \Big] .
		\label{eq:misc_ff4}
		\end{multline}
	\end{itemize}
	Similarly, for $\mathtt{S} = \{N-N_++1,\ldots,N\}$, we have:
	\begin{itemize}
		\item For all $(\alpha_+,\beta_+) \in \dot{W}_{N_-,N_+,0}^{\mathrm{DF},\mathtt{S}}[F;\gamma_+] \cap \dot{W}_{0,N,0}^{\mathrm{DF},\mathtt{S}}[F;\gamma_+]$ such that $\beta_-+m_-\gamma_- \notin \bbZ$ for any $m_-\in \{0,\ldots,N_--1\}$
		\begin{multline}
		\dot{I}_{0,N,0}^{\mathrm{DF},\mathtt{S}}[F](\alpha_+,\beta_+,\gamma_+) = (-1)^{N_-} \dot{I}_{N_-,N_+,0}^{\mathrm{DF},\mathtt{S}}[F](\alpha_+,\beta_+,\gamma_+) \\\times  \Big[ \prod_{m_-=0}^{N_--1} \frac{\sin(\pi(\alpha_-+\beta_-+(N_-+m_--1)\gamma_-))}{\sin(\pi(\beta_-+m_-\gamma_- ))} \Big]. \label{eq:misc_ff5}
		\end{multline}
		\item For all $(\alpha_+,\beta_+) \in \dot{W}_{0,N,0}^{\mathrm{DF},\mathtt{S}}[F;\gamma_+] \cap \dot{W}_{0,N_-,N_+}^{\mathrm{DF},\mathtt{S}}[F;\gamma_+]$ such that $\alpha_++m_+\gamma_+ \notin \bbZ$ for any $m_+\in \{0,\ldots,N_+-1\}$, 
		\begin{multline}
		\dot{I}_{0,N,0}^{\mathrm{DF},\mathtt{S}}[F](\alpha_+,\beta_+,\gamma_+) = (-1)^{N_+} \dot{I}_{0,N_-,N_+}^{\mathrm{DF},\mathtt{S}}[F](\alpha_+,\beta_+,\gamma_+)
		\\ \times  \Big[ \prod_{m_+=0}^{N_+-1}  \frac{\sin(\pi(\alpha_++\beta_++(N_++m_+-1)\gamma_+))}{\sin(\pi(\alpha_++m_+\gamma_+)) } \Big] . \label{eq:misc_ff6}
		\end{multline}
	\end{itemize}
	\label{prop:DFAomoto}
\end{proposition}
\begin{proof}
	Follows from a repeated application of \Cref{prop:main_lemma2}, as in the proof of \Cref{prop:Aomoto}. The only difference with the proof of \Cref{prop:Aomoto} is that we appeal to \Cref{prop:switching} to show (instead of it being an automatic consequence of symmetry) that 
	\begin{align}
	\dot{I}_{\ell,m,n}[F](\bmalpha^{\mathrm{DF},\mathtt{S}},\bmbeta^{\mathrm{DF},\mathtt{S}},\bmgamma^{\mathrm{DF},\mathtt{S}})  &= \dot{I}_{\ell,m,n}[F](\bmalpha^{\mathrm{DF},\mathtt{S}},\bmbeta^{\mathrm{DF},\mathtt{S}},\bmgamma^{\mathrm{DF},\mathtt{S}})^{\sigma_\ell} \\ 
	\dot{I}_{\ell-1,m+1,n}[F](\bmalpha^{\mathrm{DF},\mathtt{S}},\bmbeta^{\mathrm{DF},\mathtt{S}},\bmgamma^{\mathrm{DF},\mathtt{S}})^{\sigma_\ell}  &= \dot{I}_{\ell-1,m+1,n}[F](\bmalpha^{\mathrm{DF},\mathtt{S}},\bmbeta^{\mathrm{DF},\mathtt{S}},\bmgamma^{\mathrm{DF},\mathtt{S}})^{\sigma_{\ell+m}} \\
	\dot{I}_{\ell-1,m,n+1}[F](\bmalpha^{\mathrm{DF},\mathtt{S}},\bmbeta^{\mathrm{DF},\mathtt{S}},\bmgamma^{\mathrm{DF},\mathtt{S}})^{\sigma_{\ell+m}}  &= \dot{I}_{\ell-1,m,n+1}[F](\bmalpha^{\mathrm{DF},\mathtt{S}},\bmbeta^{\mathrm{DF},\mathtt{S}},\bmgamma^{\mathrm{DF},\mathtt{S}})^{\sigma_{N}}.
	\label{eq:misc_g67}
	\end{align}
\end{proof}

\begin{propositionp}
	For $\gamma_+ \in \bbC\backslash \{0,1\}$, the functions $I_{N;\mathrm{Reg}}^{\mathrm{DF},\mathtt{S}}[F](\alpha_+,\beta_+,\gamma_+)$ defined by 
	\begin{multline}
	\dot{I}_N^{\mathrm{DF},\mathtt{S}}[F](\alpha_+,\beta_+,\gamma_+) =   \Big[\prod_\pm \prod_{j=1}^{N_\pm} \frac{\sin(\pi(\alpha_\pm+\beta_\pm+(N_\pm+j-2)\gamma_\pm))}{\sin(\pi(\alpha_\pm+(j-1)\gamma_\pm)) \sin(\pi(\beta_\pm+(j-1)\gamma_\pm))} \Big] \\ \times  I_{N;\mathrm{Reg}}^{\mathrm{DF},\mathtt{S}}[F](\alpha_+,\beta_+,\gamma_+)
	\label{eq:misc_ggh}
	\end{multline}
	extend to  entire functions $I_{N;\mathrm{Reg}}^{\mathrm{DF},\mathtt{S}}[F] : \bbC^2_{\alpha_+,\beta_+}\to \bbC$.
\end{propositionp}
\begin{proof}
	The proof is very similar to that used to prove \Cref{prop:symmetric_final}. Using the previous proposition with \Cref{prop:switching}, it suffices to note that the union of all nine of the sets 
	\begin{multline} 
	\dot{W}_{N,0,0}^{\mathrm{DF},\mathtt{S}_-}[F;\gamma_+],\dot{W}_{0,N,0}^{\mathrm{DF},\mathtt{S}_-}[F;\gamma_+],\dot{W}_{0,0,N}^{\mathrm{DF},\mathtt{S}_-}[F;\gamma_+], \dot{W}_{N_-,N_+,0}^{\mathrm{DF},\mathtt{S}_-}[F;\gamma_+],\dot{W}_{N_-,0,N_+}^{\mathrm{DF},\mathtt{S}_-}[F;\gamma_+],\\ 
	\dot{W}_{0,N_-,N_+}^{\mathrm{DF},\mathtt{S}_-}[F;\gamma_+], \dot{W}_{N_+,N_-,0}^{\mathrm{DF},\mathtt{S}_+}[F;\gamma_+],\dot{W}_{N_+,0,N_-}^{\mathrm{DF},\mathtt{S}_+}[F;\gamma_+],\dot{W}_{0,N_+,N_-}^{\mathrm{DF},\mathtt{S}_+}[F;\gamma_+]  \subset \bbC^2_{\alpha_+,\beta_+},
	\end{multline} 
	where $\mathtt{S}_+ =\{1,\ldots,N_+\}$ and $\mathtt{S}_-=\mathtt{S}^\complement$, 
	is the complement of locally finite set of points, and then the result follows via Hartog's theorem.

\end{proof}



\appendix

\section{The $N=2$ case}
\label{sec:example}

We now consider the $N=2$ case in some detail, beginning with the formula 
\begin{equation}
S_2(\bmalpha,\bmbeta,\bmgamma) = \frac{\Gamma(1+\alpha_1)\Gamma(1+\beta_2)  \Gamma(2+2\gamma_{1,2} + \alpha_1+\alpha_2)\Gamma(1+2\gamma_{1,2})}{\Gamma(2+\alpha_1+2\gamma_{1,2})\Gamma(3+\alpha_1+\alpha_2+\beta_2+2\gamma_{1,2})} \cdot {}_3F_2(a,b;1),
\label{eq:misc_ooo}
\end{equation}
$a=(a_1,a_2,a_3)=(1+\alpha_1,-\beta_1,2+2\gamma_{1,2}+\alpha_1+\alpha_2)$ and $b=(b_1,b_2) = (2+\alpha_1+2\gamma_{1,2},3+\alpha_1+\alpha_2+\beta_2 + 2\gamma_{1,2})$.
This is asymmetric in the role of the $\alpha$'s and $\beta$'s, but there is an analogous formula with the $\alpha$'s and $\beta$'s on the right-hand side switched.
Some of the singularities of $S_2$ are manifest in this formula, but others are hidden in the ${}_3 F_2$ factor. 

Consider now the Dotsenko--Fateev-like integral 
\begin{equation}
I_2^{\mathrm{DF}0}(\alpha_1,\alpha_2,\beta_1,\beta_2,\gamma) = \int_0^1 \int_0^1  x_1^{\alpha_1} x_2^{\alpha_2}  (1-x_1)^{\beta_1} (1-x_2)^{\beta_2} (x_2-x_1+i0)^{2\gamma} \dd x_1 \dd x_2. 
\end{equation}
By the previous proposition:
\begin{corollary}
	\begin{multline}
	\dot{I}_2^{\mathrm{DF}0}(\alpha_1,\alpha_2,\beta_1,\beta_2,\gamma) = \Gamma(2+2\gamma+\alpha_1+\alpha_2)\Gamma(1+2\gamma) \Big[ \frac{\Gamma(1+\alpha_1)\Gamma(1+\beta_2)  {}_3F_2(a,b;1)}{\Gamma(2+\alpha_1+2\gamma)\Gamma(3+\alpha_1+\alpha_2+\beta_2+2\gamma)} \\ + e^{2\pi i \gamma}\frac{\Gamma(1+\alpha_2)\Gamma(1+\beta_1) {}_3F_2(a',b';1)}{\Gamma(2+\alpha_2+2\gamma)\Gamma(3+\alpha_1+\alpha_2+\beta_1+2\gamma)} \Big], 
	\label{eq:misc_b66}
	\end{multline}
	where $a'=(a_1',a_2',a_3')=(1+\alpha_2,-\beta_2,2+2\gamma+\alpha_1+\alpha_2)$ and $b'=(b_1',b_2') = (2+\alpha_2+2\gamma,3+\alpha_1+\alpha_2+\beta_1 + 2\gamma)$.
\end{corollary}
The formula \cref{eq:misc_b66} is not suitable for analytic continuation to $\gamma= -1$, for which we instead use the method described in \S\ref{subsec:alternative}. That yields 
\begin{multline}
\dot{I}_2(\alpha_1,\alpha_2,\beta_1,\beta_2,\gamma) = \frac{\Gamma(1+\alpha_1)\Gamma(1+\beta_1)}{\Gamma(2+\alpha_1+\beta_1)}\int_\Gamma \Big[ z^{\alpha_2+2\gamma} (1-z)^{\beta_2} \\ 
\times {}_2 F_1\Big(-2\gamma,1+\alpha_1,2+\alpha_1+\beta_1; \frac{1}{z}\Big)\Big] \dd z,
\label{eq:misc_b67}
\end{multline}
where the $\Gamma$ is a trapezoidal contour in the upper-half of the complex plane. This formula can be used to numerically compute $\dot{I}_2^{\mathrm{DF}0}(\alpha_1,\alpha_2,\beta_1,\beta_2,\gamma)$ for $\gamma$ with large negative real part, as long as $\alpha_1,\alpha_2,\beta_1,\beta_2$ have sufficiently large positive real part relative to $\gamma$.

We illustrate the method of proof of \Cref{thm:generic} with the computation of the residues associated with $\alpha_-+\alpha_++2\gamma \in  \bbZ^{-2-d}$. Introducing coordinates $\varrho = x_2$ and $\lambda = x_1/x_2$, 
\begin{align}
\begin{split} 
S_2(\alpha_1,\alpha_2,\beta_1,\beta_2,\gamma)&= \int_0^1\int_0^{x_2} x_1^{\alpha_1} x_2^{\alpha_2}  (1-x_1)^{\beta_1} (1-x_2)^{\beta_2} (x_2-x_1)^{2\gamma} \dd x_1 \dd x_2 \label{eq:misc_bn3} \\
&= \int_0^1 \int_0^1 \varrho^{1+\alpha_1+\alpha_2+2\gamma} \lambda^{\alpha_1}  (1-\lambda\varrho)^{\beta_1}(1-\varrho)^{\beta_2}  (1-\lambda)^{2\gamma} \dd \lambda \dd \varrho. 
\end{split} 
\end{align}
Expanding $(1-\lambda\varrho)^{\beta_1}(1-\varrho)^{\beta_2}$ in Taylor series around $\varrho=0$, 
we have 
\begin{equation}
(1-\lambda\varrho)^{\beta_1}(1-\varrho)^{\beta_2} = \sum_{k=0}^\infty \sum_{\kappa=0}^k \binom{\beta_1}{\kappa} \binom{\beta_2}{k-\kappa} \lambda^\kappa(- \varrho)^k.
\end{equation}
Then, computing the outer integral term by term and using the formula for the $\beta$-function for the inner integral, 
\begin{equation}
S_2(\alpha_1,\alpha_2,\beta_1,\beta_2,\gamma) \sim \sum_{k=0}^\infty \frac{(-1)^k}{2+\alpha_1+\alpha_2+2\gamma+k} \sum_{\kappa=0}^k \binom{\beta_1}{\kappa} \binom{\beta_2}{k-\kappa}  \frac{\Gamma(1+\alpha_1+\kappa)\Gamma(1+2\gamma)}{\Gamma(2+\alpha_1+2\gamma+\kappa)}
\label{eq:misc_nn3}
\end{equation}
where the `$\sim$' means modulo an error which is not singular at (all but a positive codimension subset of) the hyperplane under investigation.
The right-hand side of this has an apparent pole whenever $\alpha_1+\alpha_2+2\gamma \in \bbZ^{\leq -2}$. 

We now examine some special cases. Fix $d\in \bbN$. First consider 
\begin{equation}
S_2[x^d](\alpha,\beta,\gamma) = S_2(\alpha,\alpha+d,\beta,\beta,\gamma) =  \int_0^1 \int_0^{x_2} x_1^\alpha x_2^{d+\alpha}(1-x_1)^\beta(1-x_2)^\beta (x_2-x_1)^{2\gamma} \dd x_1 \dd x_2.
\end{equation}
By \cref{eq:misc_ooo}, 
\begin{equation}
S_2[x^d](\alpha,\beta,\gamma) = \frac{\Gamma(1+\alpha)\Gamma(1+\beta)\Gamma(1+2\gamma)\Gamma(2+2\alpha+2\gamma+d)}{\Gamma(2+\alpha+2\gamma)\Gamma(3+2\alpha+\beta+2\gamma+d)} \cdot {}_3F_2(a,b;1),
\label{eq:misc_aaa}
\end{equation}
where now $a = (a_1,a_2,a_3) = (1+\alpha,-\beta,2+2\alpha+2\gamma+d)$ and $b = (b_1,b_2)=(2+\alpha+2\gamma,3+2\alpha+\beta+2\gamma+d)$. A numerically generated plot of the absolute value of the right-hand side is given in the case $d=2$ in \Cref{fig:example}.
Applying \Cref{thm:generic}, in which $\alpha_{1,*} =  \alpha$, $\alpha_{2,*} = d  + 2 \alpha + 2  \gamma$, $\beta_{1,*} =  \beta$, $\beta_{2,*} = 2 \beta + 2  \gamma$, and $\gamma_{1,2,*}=2 \gamma$, we deduce that $\smash{S_2[x^d](\alpha,\beta,\gamma)}$ extends to an analytic function on 
\begin{equation}
\bbC^{3}_{\alpha,\beta,\gamma} \Big\backslash \Big[ \{ \alpha \in \bbZ^{\leq -1}\} \cup \{\alpha+\gamma \in 2^{-1}\bbZ^{\leq -2-d} \} \cup \{\beta \in \bbZ^{\leq -1}\} \cup \{\beta+\gamma \in 2^{-1} \bbZ^{\leq -2}\}  \cup \{\gamma \in 2^{-1}\bbZ^{\leq -1}\} \Big]. 
\label{eq:misc_aa1}
\end{equation} 
In the $d=2$ case, this can be seen in \Cref{fig:example}.

\begin{figure}[t!]
	\begin{center}
		\includegraphics[scale=.54]{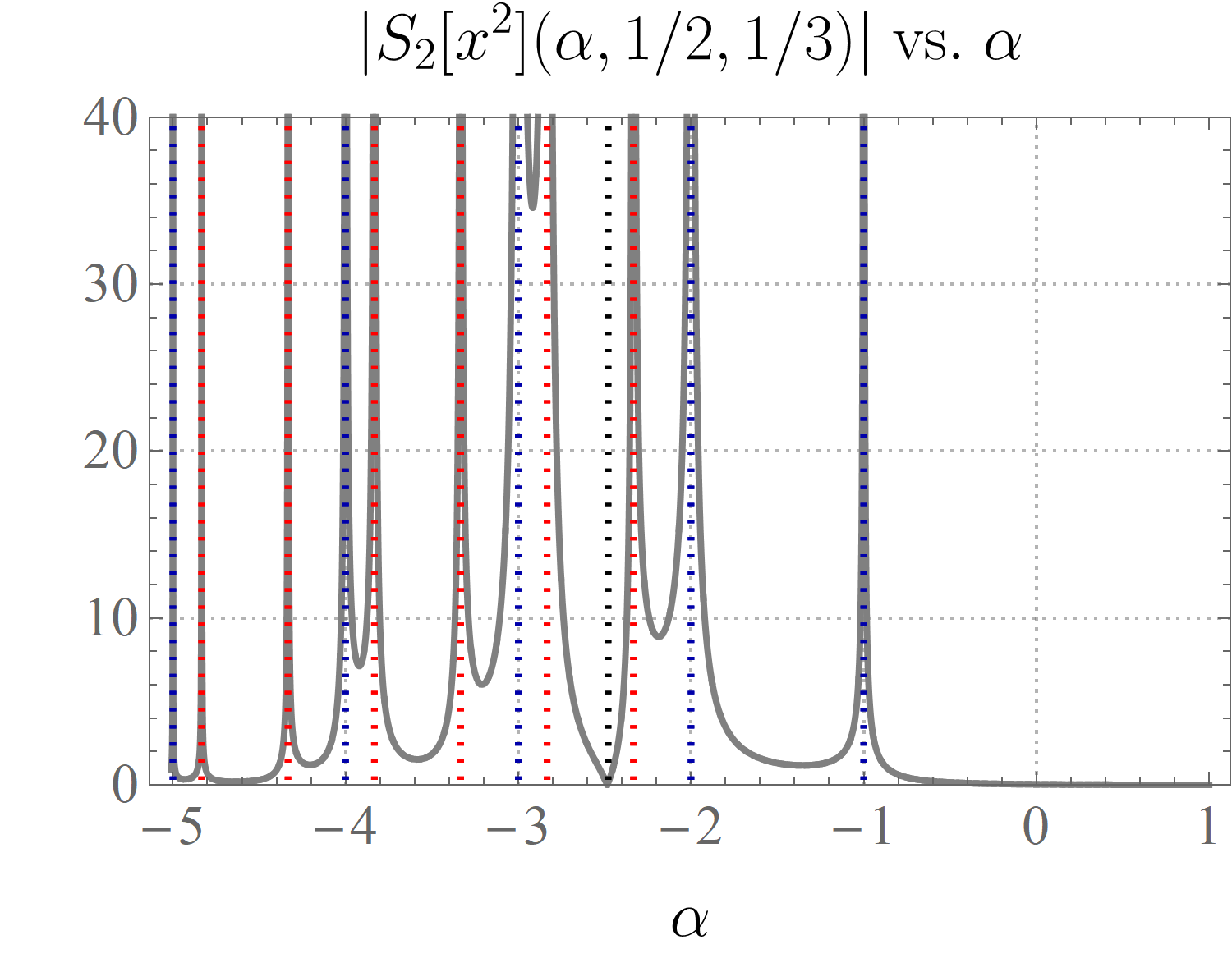}
		\includegraphics[scale=.54]{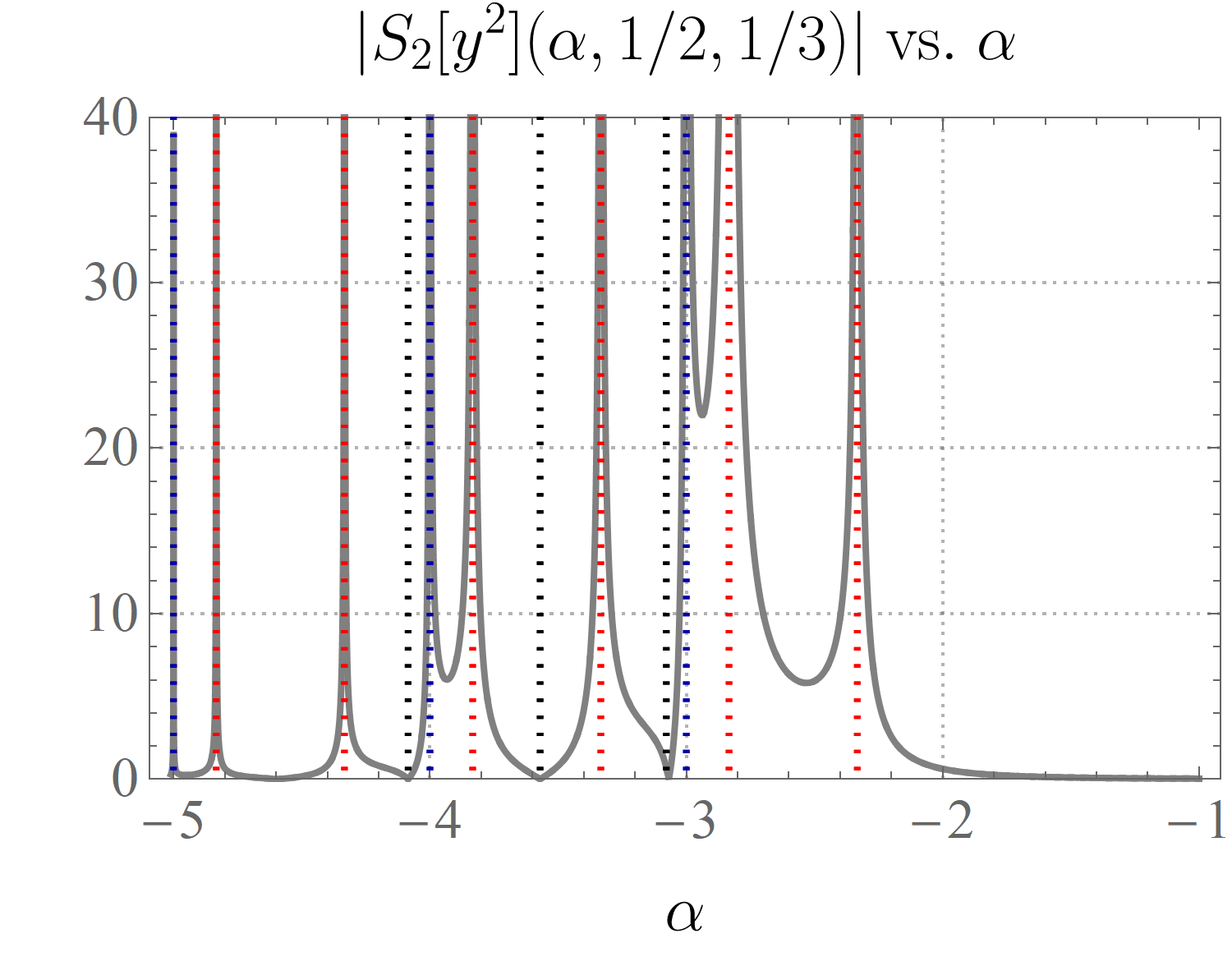}
	\end{center}
	\caption{The absolute values of the right-hand sides of \cref{eq:misc_aaa} (\emph{left}) and \cref{eq:misc_bbb} (\emph{right}) plotted against $\alpha$, for $\beta=1/2$ and $\gamma=1/3$ fixed. The singularities predicted in \cref{eq:misc_aa1}, \cref{eq:misc_bb1} have been drawn as dotted vertical lines, those associated with $\color{darkblue}\{ \alpha \in \bbZ\}$ in blue and those associated with $\color{red}\{\alpha+\gamma \in 2^{-1} \bbZ\}$ in red. It appears that all of the poles that could be present are present. The apparent zeroes of $S_2[F](\alpha,1/2,1/3)$ in the depicted range of $\alpha$ have been marked with dotted black lines and numerically computed to be $\approx -2.48503$ for $F=x^2$ and $\approx -3.06833$, $-3.57013$, and $-4.08562$ for $F=y^2$. }
	\label{fig:example}
\end{figure}

On the other hand, consider 
\begin{equation}
S_2[y^d](\alpha,\beta,\gamma) = S_2(\alpha+d,\alpha,\beta,\beta,\gamma) =   \int_0^1 \int_0^{x_2} x_1^{d+\alpha} x_2^{\alpha}(1-x_1)^\beta(1-x_2)^\beta (x_2-x_1)^{2\gamma} \dd x_1 \dd x_2. 
\end{equation}
By \Cref{eq:misc_vv1}, 
\begin{equation}
S_2[y^d](\alpha,\beta,\gamma) = \frac{\Gamma(1+\alpha+d)\Gamma(1+\beta)\Gamma(1+2\gamma)\Gamma(2+2\alpha+2\gamma+d)}{\Gamma(2+\alpha+2\gamma+d)\Gamma(3+2\alpha+\beta+2\gamma+d)} \cdot {}_3F_2(a',b';1),
\label{eq:misc_bbb}
\end{equation}
where $a'=(a_1',a_2,a_3) = (1+d+\alpha,-\beta,2+d+2\alpha+2\gamma)$ and $b'=(b'_1,b_2) = (2+d+\alpha+2\gamma,3+d+2\alpha+\beta+2\gamma)$. We again apply \Cref{thm:generic}, but now $\alpha_{1,*} = d +  \alpha$, $\alpha_{2,*} = d + 2 \alpha+2 \gamma$, in order to deduce that $S_2[y^d](\alpha,\beta,\gamma)$ extends analytically to 
\begin{equation}
\bbC^3_{\alpha,\beta,\gamma} \Big\backslash \Big[ \{ \alpha \in \bbZ^{\leq -1-d}\} \cup \{\alpha+\gamma \in 2^{-1}\bbZ^{\leq -2-d} \} \cup \{\beta \in \bbZ^{\leq -1}\} \cup \{\beta+\gamma \in 2^{-1} \bbZ^{\leq -2}\} \cup \{\gamma \in 2^{-1}\bbZ^{\leq -1}\} \Big].
\label{eq:misc_bb1}
\end{equation}
See \Cref{fig:example} for a numerically generated plot.

If we instead pick $F(x,y)=x^d+y^d$, which is in some sense the symmetrization of the previous two examples, the situation looks very different. Combining the formulas above yields 
\begin{multline}
S_2[x^d+y^d](\alpha,\beta,\gamma) = (S_2[x^d]+S_2[y^d])(\alpha,\beta,\gamma)= \frac{\Gamma(1+\alpha)\Gamma(1+\beta)\Gamma(1+2\gamma)\Gamma(2+2\alpha+2\gamma+d)}{\Gamma(2+\alpha+2\gamma)\Gamma(3+2\alpha+\beta+2\gamma+d)} \\ \times  \Big[ {}_3F_2(a,b;1)  +   \frac{\Gamma(1+\alpha+d)\Gamma(2+\alpha+2\gamma)}{\Gamma(1+\alpha)\Gamma(2+\alpha+2\gamma+d)}\cdot  {}_3F_2(a',b';1)\Big].
\label{eq:misc_ccc}
\end{multline}
On the other hand, by \Cref{thm:symmetric}, we know that $S_2[x^d+y^d]$ extends analytically to 
\begin{multline}
\bbC^{3}_{\alpha,\beta,\gamma} \Big\backslash \Big[ \{ \alpha \in \bbZ^{\leq -1-\bar{\delta}_1}\} \cup \{\alpha+\gamma \in \bbZ^{\leq -2-\bar{\delta}_2} \} \cup \{\beta \in \bbZ^{\leq -1-\bar{\atled}_1}\} \cup \{\beta+\gamma \in  \bbZ^{\leq -1-\bar{\atled}_2}\}  \\ \cup \{\gamma \in 2^{-1}\bbZ^{\leq -1}, \gamma \notin \bbZ\} \Big],
\label{eq:misc_aa2}
\end{multline}  
where $\bar{\delta}_1,\bar{\atled}_1,\bar{\delta}_2,\bar{\atled}_2$ are as in the theorem. 
In this example, $\bar{\delta}_1,\bar{\atled}_1,\bar{\atled}_2=0$, $\bar{\delta}_2= \lceil d/2 \rceil$, $\bar{d}_2=d$, and $\bar{d}_1 = \lfloor d/2 \rfloor$.
See \Cref{fig:example2}, in which $d=2$.

\begin{figure}[t!]
	\begin{center}
		\includegraphics[scale=.54]{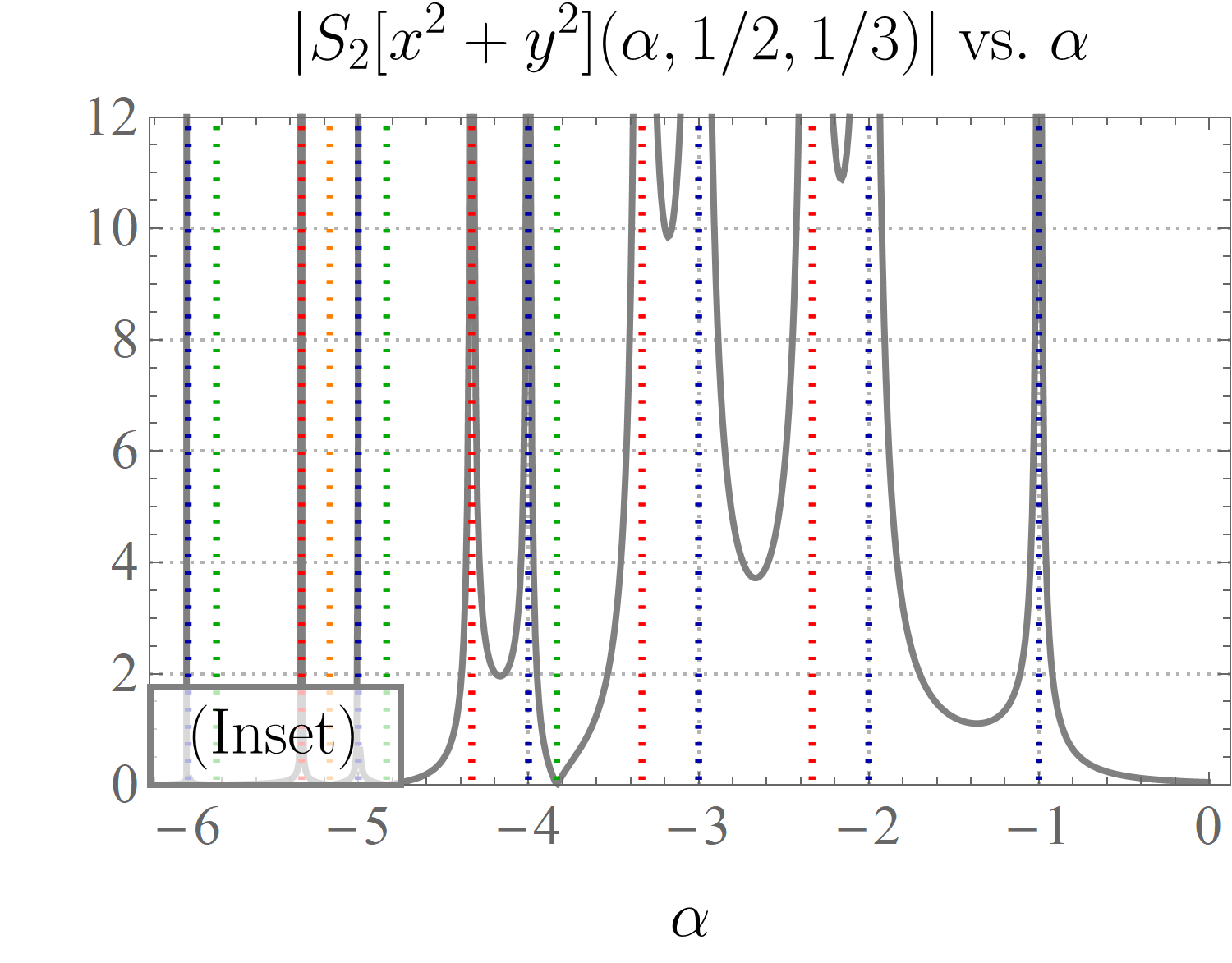}
		\includegraphics[scale=.54]{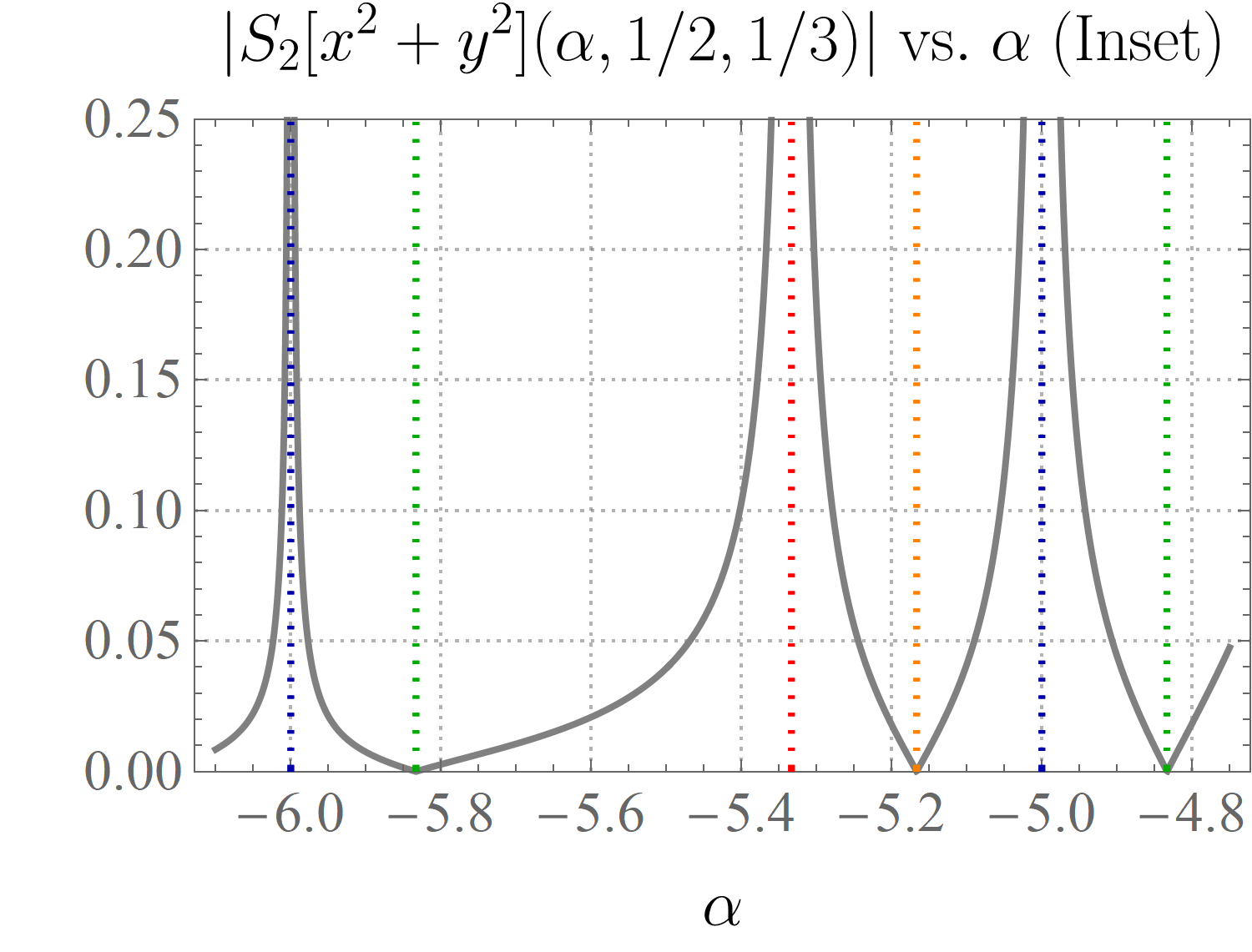}
	\end{center}
	\caption{The absolute value of the right side of \cref{eq:misc_ccc}, plotted against $\alpha \in (-6,0)$ (\emph{left}) and $\alpha \in (-6.1,-4.75)$ (\emph{right}), for $\beta=1/2$ and $\gamma=1/3$ fixed. Singularities associated with $\color{darkblue}\{ \alpha \in \bbZ\}$ are marked with blue lines and those with $\color{red}\{\alpha+\gamma \in \bbZ\}$ with red. The zeroes associated with $\color{mygreen}\{\alpha+\beta+\gamma \in \bbZ\}$ are marked in green and those with $\color{orange}\{\alpha+\beta+2\gamma \in \bbZ\}$ in orange. The location of the second plot is marked as an inset on the left plot (not to scale).}
	\label{fig:example2}
\end{figure}

In the sum \cref{eq:misc_ccc}, the poles of the individual summands at such $2\alpha + 2\gamma \in \bbZ^{\leq -2-d} \cap (2\bbZ+1) $ (which we can see from \Cref{fig:example} exist) must cancel. 
By \cref{eq:misc_nn3}, the residue of $S_2[x^d+y^d](\alpha,\beta,\gamma)$ at such a point is proportional to
\begin{equation}
\sum_{\kappa=0}^{k}  \binom{\beta}{\kappa} \binom{\beta}{k-\kappa}   \Big[ \frac{\Gamma(1+\alpha+\kappa)}{\Gamma(2+\alpha+2\gamma+\kappa)}+\frac{\Gamma(1+\alpha+\kappa+d)}{\Gamma(2+\alpha+2\gamma+\kappa+d)} \Big].
\label{eq:misc_z66}
\end{equation}
This therefore has to vanish whenever $-2-d-k$ is odd and $(\alpha,\gamma)\in \{2\alpha+2\gamma = -2 - d - k\} \subset \bbC^2_{\alpha,\gamma}$ is such that the functions in \cref{eq:misc_z66} are well-defined. 
A direct algebraic proof of this fact is not entirely trivial, but it is straightforward to check case-by-case. 

The function $S_{2;\mathrm{Reg}}[x^d+y^d](\alpha,\beta,\gamma)$ is plotted as a function of $\alpha\in \bbC$ in \Cref{fig:Delta}, still in the case $d=2$ -- for fixed $\beta,\gamma$. As expected, it appears to have no singularities, in accordance with \Cref{thm:symmetric}. Unlike in the case $F=1$, where Selberg's formula shows that $S_{2;\mathrm{Reg}}[1](\alpha,\beta,\gamma)$ is constant, $S_{2;\mathrm{Reg}}[x^d+y^d]$ is nonconstant. 

\begin{figure}[h!]
	\begin{center}
		\includegraphics[scale=.5]{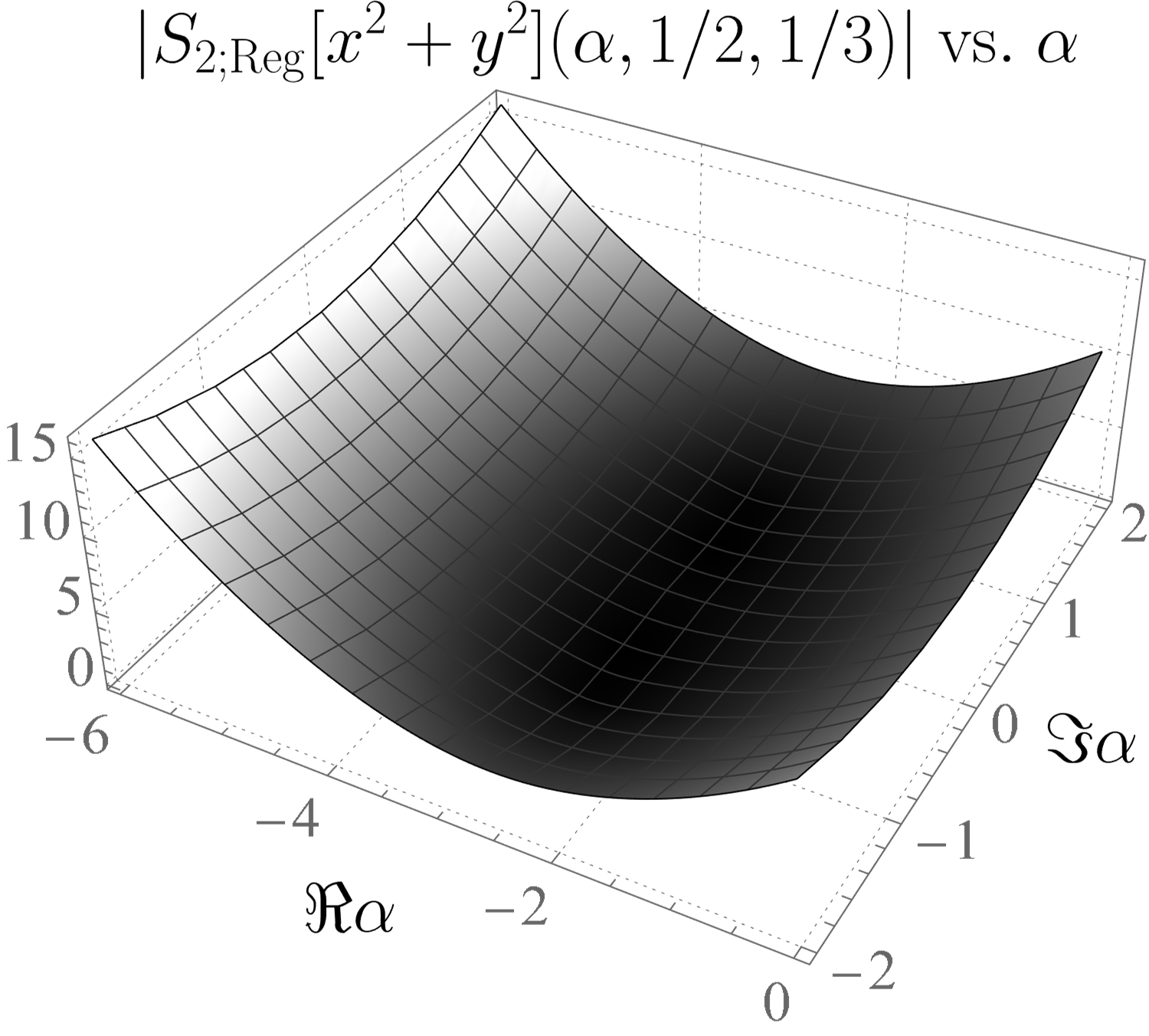}
	\end{center}
	\caption{
		The function $S_{2;\mathrm{Reg}}[x^2+y^2](\alpha,1/2,1/3)$ defined by \cref{eq:SNF_goal_symmetric}, plotted as a function of $\alpha\in \bbC$. 
	}
	\label{fig:Delta}
\end{figure}

Consider now $\dot{I}_2^{\mathrm{DF}}(\alpha_+,\beta_+)=\dot{I}_2^{\mathrm{DF}0}(\alpha_+,-\alpha_+,\beta_+,-\beta_+,1)$, which is a DF-symmetric integral with $\gamma_\pm = -1$. This is given concretely by 
\begin{multline}
\dot{I}_2^{\mathrm{DF}}(\alpha_+,\beta_+) = \frac{\Gamma(1+\alpha_+)\Gamma(1+\beta_+)}{\Gamma(2+\alpha_++\beta_+)}\int_\Gamma \Big[ z^{\alpha_++2\gamma} (1-z)^{\beta_+} \\ 
\times {}_2 F_1\Big(-2\gamma,1+\alpha_+,2+\alpha_++\beta_+; \frac{1}{z}\Big)\Big] \dd z
\end{multline}
when the real parts of $\alpha_+,\beta_+$ are sufficiently large. 
The Dotsenko--Fateev claim, \cref{eq:misc_x6x}, is, up to a sign, that
\begin{equation}
\dot{I}_2^{\mathrm{DF}}(\alpha_+,\beta_+) =-  \frac{\Gamma(1+\alpha_+)\Gamma(1+\beta_+) \Gamma(\alpha_+ )\Gamma(\beta_+)}{2\Gamma(1+\alpha_++\beta_+)\Gamma(\alpha_++\beta_+)}.
\label{eq:misc_nb6}
\end{equation}

\section{Explicit coordinates on $[0,1)^N_{\mathrm{tb}}$}
\label{sec:coordinates}

In this appendix, we discuss the total boundary (tb) blowup $[0,1)^N_{\mathrm{tb}}$, the mwc constructed by blowing up all of the facets of $[0,1)^N$, starting with those of the lowest dimension.

For each nonempty subset $S\subseteq \{1,\ldots,N\}$, let $\mathrm{F}_S$ denote the face of $\smash{[0,1)^N_{\mathrm{tb}}}$ corresponding to the facet $\{j\in S\Rightarrow x_j =0\}$ of $[0,1)^N_x$.
Tracing through the construction of the total boundary blowup, we have the following explicit choice of boundary-defining-functions (bdfs) of the various faces. If $N\geq 3$, these are different from the recursively defined boundary-defining-functions discussed in the introduction to \S\ref{sec:geometry}.

It is possible to prove:
\begin{propositionp}
	The function 
	\begin{equation}
	x_{\mathrm{F}_S} = x_{\mathrm{F}_S,N}  = \prod_{S_0\supseteq S} \Big[ \sum_{j\in S_0} x_j \Big]^{(-1)^{|S|-|S_0|}} 
	\end{equation}
	serves as a bdf of $\mathrm{F}_S$.
	
	Suppose that $\mathtt{I}$ is a (possibly empty) set of nested nonempty subsets of $\{1,\ldots,N\}$. Then, 
	\begin{equation} 
	\mathrm{f}_{\mathtt{I}}=\{x_{\mathrm{F}_S}=0\text{ for all }S\in \mathtt{I}\}
	\end{equation} 
	is a codimension $|\mathtt{I}|$ facet of $[0,1)_{\mathrm{tb}}^N$. This defines a bijective correspondence between the set of nested nonempty subsets of $\{1,\ldots,N\}$ and the set of facets of $[0,1)_{\mathrm{tb}}^N$.
	
	If $p$ lies in the interior of $\mathrm{f}_{\mathtt{I}}$, then, letting $\sigma \in \frakS_N$ denote any permutation consistent with $\mathtt{I}$, 
	\begin{equation}
	\varrho=x_{\sigma(1)},\hat{x}_{\sigma(2)}= x_{\sigma(2)}/x_{\sigma(1)},\cdots ,\hat{x}_{\sigma(N)}=x_{\sigma(N)}/x_{\sigma(N-1)} 
	\label{eq:misc_b77}
	\end{equation}
	give a local set of coordinates near $p$. 
	\label{prop:coordinates}
\end{propositionp}
Here, we say that $\sigma \in \frakS_N$ is \emph{consistent} with $\mathtt{I}$ if, whenever $j<k$, $\sigma(j)\in S\in \mathtt{I} \Rightarrow \sigma(k)\in S$.

We can cover $[0,1)_{\mathrm{tb}}^N$ with the $N!$ coordinate charts whose restrictions to the interior are of the form 
\begin{multline}
\{0 < x_{\sigma(N)} < 2 x_{\sigma(N-1)} < \cdots < 2^N x_{\sigma(1)} < 2^N \} \cap (0,1)^N \\ \to [0,1)_{x_{\sigma(1)}}\times [0,2)_{x_{\sigma(2)}/x_{\sigma(1)}}\times\cdots \times [0,2)_{x_{\sigma(N)}/x_{\sigma(N-1)}},
\end{multline}
for $\sigma \in \frakS_N$.

The preceding proposition is used to prove:
\begin{proposition}
	For any $M\in \{1,\ldots,N-1\}$ and nonempty $Q\subseteq \{1,\ldots,M\}$, 
	\begin{equation}
	x_{\mathrm{F}_Q,M} = \prod_{Q_0\subseteq \{M+1,\ldots,N\}} x_{\mathrm{F}_{Q\cup Q_0},N}
	\label{eq:forgetful_explicit}
	\end{equation}
	in $(0,1)^N_x$. 
	\label{prop:forgetful_explicit}
\end{proposition}
\begin{proof}
	A factor of $\sum_{j\in Q\cup Q_0} x_j$ appears on the right-hand side of \cref{eq:forgetful_explicit} to the power
	\begin{equation}
	\sum_{Q_1 \subseteq Q_0} (-1)^{|Q_0|},
	\end{equation}
	which is, by the binomial theorem, $+1$ if $Q_0=\varnothing$ and $0$ otherwise. Thus, $\prod_{Q_0\subseteq \{M+1,\ldots,N\}} x_{\mathrm{F}_{Q\cup Q_0},N} = \sum_{j\in Q} x_j$. 
\end{proof}

The full proof of \Cref{prop:coordinates} is somewhat incidental to the rest of the paper, so we merely illustrate the argument in the case $N=3$. This generalizes to the $N\geq 3$ case, and applies in an even simpler form to the $N=2$ case.

The total boundary blowup $[0,1)_{\mathrm{tb}}^N$ is defined as 
\begin{equation}
[[0,1)^3;\{x,y,z=0\};\{y,z=0\};\{x,z=0\},\{x,y=0\}],
\end{equation}
where the first blowup is that of $\{x,y,z=0\}$ and must be performed first. The other three blowups can be performed in any order, and each order yields a canonically diffeomorphic mwc. 
The input and output of the first blowup, yielding $[[0,1)^3;\{x,y,z=0\}]$, are
\begin{center}
	\begin{tikzpicture}
	\begin{scope}[shift={(-4,0,0)}, scale=1.65]
	\filldraw[lightgray!20] (0,0,0) -- (0,1.8,0) -- (1.8,1.8,0) -- (1.8,0,0) -- cycle;
	\filldraw[lightgray!20] (0,0,0) -- (0,1.8,0) -- (0,1.8,1.8) -- (0,0,1.8) -- cycle;
	\filldraw[lightgray!20] (0,0,0) -- (1.8,0,0) -- (1.8,0,1.8) -- (0,0,1.8) -- cycle;
	\node (1) at (1,1,0) {$y$};
	\node (2) at (1,0,1) {$z$};
	\node (3) at (0,1,1) {$x$};
	\draw[->] (0,0,0) -- (0,2,0);
	\draw[->] (0,0,0) -- (0,0,2);
	\draw[->] (0,0,0) -- (2,0,0);
	\end{scope}
	\begin{scope}[shift={(4,0,0)}, scale=1.65]
	\filldraw[lightgray!20] (0,0,0) -- (0,1.8,0) -- (1.8,1.8,0) -- (1.8,0,0) -- cycle;
	\filldraw[lightgray!20] (0,0,0) -- (0,1.8,0) -- (0,1.8,1.8) -- (0,0,1.8) -- cycle;
	\filldraw[lightgray!20] (0,0,0) -- (1.8,0,0) -- (1.8,0,1.8) -- (0,0,1.8) -- cycle;
	\draw (0,1,0) -- (0,0,1) -- (1,0,0) -- cycle;
	\node (1) at (1.5,1,0) {$y/(x+y+z)$};
	\node (2) at (1.5,0,1.5) {$z/(x+y+z)$};
	\node (3) at (0,2.25,2.5) {$x/(x+y+z)$};
	\node (4) at (.4,.2,.25) {$x+y+z$};
	\draw[->] (0,1,0) -- (0,2,0);
	\draw[->] (0,0,1) -- (0,0,2);
	\draw[->] (1,0,0) -- (2,0,0);
	\end{scope}
	\end{tikzpicture}
\end{center}
respectively, 
where we are marking the faces with boundary-defining-functions (using the Cartesian coordinates $x,y,z$ in place of $x_1,x_2,x_3$). The choice of bdfs on the blowup is in accordance with the prescription in the introduction of \S\ref{sec:geometry}. 

The next blowup, yielding 
\begin{equation} 
[[0,1)^3;\{x,y,z=0\};\{y,z=0\}],
\end{equation} 
has input and output
\begin{center}
	\begin{tikzpicture}
	\begin{scope}[shift={(-4,0,0)}, scale=1.65]
	\filldraw[lightgray!20] (0,0,0) -- (0,1.8,0) -- (1.8,1.8,0) -- (1.8,0,0) -- cycle;
	\filldraw[lightgray!20] (0,0,0) -- (0,1.8,0) -- (0,1.8,1.8) -- (0,0,1.8) -- cycle;
	\filldraw[lightgray!20] (0,0,0) -- (1.8,0,0) -- (1.8,0,1.8) -- (0,0,1.8) -- cycle;
	\draw (0,1,0) -- (0,0,1) -- (1,0,0) -- cycle;
	\node (1) at (1.5,1,0) {$y/(x+y+z)$};
	\node (2) at (1.5,0,1.5) {$z/(x+y+z)$};
	\node (3) at (0,2.25,2.5) {$x/(x+y+z)$};
	\node (4) at (.4,.2,.25) {$x+y+z$};
	\draw[->] (0,1,0) -- (0,2,0);
	\draw[->] (0,0,1) -- (0,0,2);
	\draw[->] (1,0,0) -- (2,0,0);
	\end{scope}
	\begin{scope}[shift={(4,0,0)}, scale=1.65]
	\filldraw[lightgray!20] (0,1,0) -- (0,0,1) -- (1,0,.35) -- (1,.35,0) -- cycle;
	\filldraw[lightgray!20] (1,.35,0) -- (0,1,0) -- (0,1.8,0) -- (2.4,1.8,0) -- (2.4,.35,0) -- cycle;
	\filldraw[lightgray!20] (1,0,.35) -- (1,.35,0) -- (2.4,.35,0) -- (2.4,0,.35) -- cycle;
	\filldraw[lightgray!20] (0,0,0) -- (0,1.8,0) -- (0,1.8,1.8) -- (0,0,1.8) -- cycle;
	\filldraw[lightgray!20] (1,0,.35) -- (0,0,1) -- (0,0,1.8) -- (2.4,0,1.8) -- (2.4,0,.35) -- cycle;
	\draw (0,1,0) -- (0,0,1) -- (1,0,.35) -- (1,.35,0) -- cycle;
	\node (1) at (1.5,1,0) {$y/(y+z)$};
	\node (2) at (1.5,0,1.5) {$z/(y+z)$};
	\node (3) at (0,2.25,2.5) {$x/(x+y+z)$};
	\node (4) at (.45,.2,.25) {$x+y+z$};
	\node (5) at (2,.1,0) {$(y+z)/(x+y+z)$};
	\draw[->] (0,1,0) -- (0,2,0);
	\draw[->] (0,0,1) -- (0,0,2);
	\draw[->] (1,0,.35) -- (2.5,0,.35);
	\draw[->] (1,.35,0) -- (2.5,.35,0);
	\end{scope}
	\end{tikzpicture}
\end{center}
Again, the choices of bdfs are in accordance with \S\ref{sec:geometry}.

Next, we blow up the facet of $[[0,1)^3;\{x,y,z=0\};\{y,z=0\}]$ corresponding to the $y$-axis. Because the previous blowup was located away from the facet being blown up now, we can use the sum 
\begin{equation} 
\frac{x}{x+y+z} + \frac{z}{x+y+z} = \frac{x+z}{x+y+z} 
\label{eq:misc_hu4}
\end{equation} 
of the bdfs of the adjacent faces in $[[0,1)^3;\{x,y,z=0\}]$ as a bdf of the front face of the current blowup rather than 
\begin{equation}
\frac{x}{x+y+z} + \frac{z}{y+z} = \frac{xy+2xz+yz+z^2}{(y+z)(x+y+z)}, 
\label{eq:misc_hu3}
\end{equation}
which would be the prescription in \S\ref{sec:geometry}. The choices in \cref{eq:misc_hu4}, \cref{eq:misc_hu3} are equivalent, in the sense that their quotient is a smooth, nonvanishing function on $[[0,1)^3;\{x,y,z=0\};\{y,z=0\};\{x,z=0\}]$. 
Given that \cref{eq:misc_hu4} serves as a bdf of the front face of the latest blowup, the quotient 
\begin{equation}
\frac{x/(x+y+z)}{(x+z)/(x+y+z)} = \frac{x}{x+z} 
\end{equation}
serves as a bdf in $[[0,1)^3;\{x,y,z=0\};\{y,z=0\};\{x,z=0\}]$ for the lift of the $yz$-plane, and 
\begin{equation}
\frac{z/(y+z)}{(x+z)/(x+y+z)} = \frac{z(x+y+z)}{(x+z)(y+z)}
\end{equation}
serves as a bdf for the lift of the $xy$-plane. In summary, the third blowup has input and output
\begin{center}
	\begin{tikzpicture} 
	\begin{scope}[shift={(-5,0,0)}, scale=1.65]
	\filldraw[lightgray!20] (0,0,0) -- (0,1.8,0) -- (2.4,1.8,0) -- (2.4,0,0) -- cycle;
	\filldraw[lightgray!20] (0,0,0) -- (0,1.8,0) -- (0,1.8,1.8) -- (0,0,1.8) -- cycle;
	\filldraw[lightgray!20] (0,0,0) -- (2.4,0,0) -- (2.4,0,1.8) -- (0,0,1.8) -- cycle;
	\draw (0,1,0) -- (0,0,1) -- (1,0,.35) -- (1,.35,0) -- cycle;
	\node (1) at (1.5,1,0) {$y/(y+z)$};
	\node (2) at (1.5,0,1.5) {$z/(y+z)$};
	\node (3) at (0,2.25,2.5) {$x/(x+y+z)$};
	\node (4) at (.45,.2,.25) {$x+y+z$};
	\node (5) at (2,.1,0) {$(y+z)/(x+y+z)$};
	\draw[->] (0,1,0) -- (0,2,0);
	\draw[->] (0,0,1) -- (0,0,2);
	\draw[->] (1,0,.35) -- (2.5,0,.35);
	\draw[->] (1,.35,0) -- (2.5,.35,0);
	\end{scope}
	\begin{scope}[shift={(3.5,0,0)}, scale=1.65]
	\filldraw[lightgray!20] (.3,0,2.4) -- (2.4,0,2.4) -- (2.4,0,.35) -- (1,0,.35) -- (.3,0,1) -- cycle;
	\filldraw[lightgray!20] (1,.35,0) -- (0,1,0) -- (0,1.8,0) -- (2.4,1.8,0) -- (2.4,.35,0) -- cycle;
	\filldraw[lightgray!20] (1,0,.35) -- (1,.35,0) -- (2.4,.35,0) -- (2.4,0,.35) -- cycle;
	\filldraw[lightgray!20] (.3,0,1) -- (0,.3,1) -- (0,.3,2.4) -- (.3,0,2.4) -- cycle;
	\filldraw[lightgray!20] (0,.3,1) -- (0,1,0) -- (0,1.8,0) -- (0,1.8,2.4) -- (0,.3,2.4) -- cycle;
	\draw[fill=lightgray!20] (0,1,0) -- (0,.3,1) -- (.3,0,1) -- (1,0,.35) -- (1,.35,0) -- cycle;
	\node (1) at (1.5,1,0) {$y/(y+z)$};
	\node (2) at (2.1,0,1.95) {$z(x+y+z)(x+z)^{-1}(y+z)^{-1}$};
	\node (3) at (0,2.25,2.25) {$x/(x+z)$};
	\node (4) at (.4,.25,.25) {$x+y+z$};
	\node (5) at (2,.1,0) {$(y+z)/(x+y+z)$};
	\draw[dotted] (0,0,1.5) -- (0,0,2.1) -- (-.2,0,2) node[left] {$(x+z)/(x+y+z)$};
	\draw[->] (0,1,0) -- (0,2,0);
	\draw[->] (.3,0,1) -- (.3,0,2.5);
	\draw[->] (0,.3,1) -- (0,.3,2.5);
	\draw[->] (1,0,.35) -- (2.5,0,.35);
	\draw[->] (1,.35,0) -- (2.5,.35,0);
	\end{scope}
	\end{tikzpicture}
\end{center}

The final blowup, yielding $[0,1)^3_{\mathrm{tb}}=[[0,1)^3;\{x,y,z=0\};\{y,z=0\};\{x,z=0\},\{x,y=0\}]$, is similar. We use $(x+y)/(x+y+z)$ as a bdf of the blowup of the face corresponding to the $z$-axis, and we can then use 
\begin{align}
\begin{split}
\frac{x/(x+z)}{(x+y)/(x+y+z)} &= \frac{x(x+y+z)}{(x+y)(x+z)} \\ 
\frac{y/(y+z)}{(x+y)/(x+y+z)} &= \frac{y(x+y+z)}{(x+y)(y+z)} 
\end{split}
\end{align}
as bdfs of the faces corresponding to the $yz$- and $xz$-planes, respectively.
Thus, we end up with
\begin{center}
	\begin{tikzpicture}
	\begin{scope}[scale=1.65]
	\filldraw[lightgray!20] (.3,0,2.4) -- (2.4,0,2.4) -- (2.4,0,.35) -- (1,0,.35) -- (.3,0,1) -- cycle;
	\filldraw[lightgray!20] (1,.35,0) -- (.3,1,0) -- (.3,1.8,0) -- (2.4,1.8,0) -- (2.4,.35,0) -- cycle;
	\filldraw[lightgray!20] (1,0,.35) -- (1,.35,0) -- (2.4,.35,0) -- (2.4,0,.35) -- cycle;
	\filldraw[lightgray!20] (.3,0,1) -- (0,.3,1) -- (0,.3,2.4) -- (.3,0,2.4) -- cycle;
	\filldraw[lightgray!20] (0,.3,1) -- (0,1,.3) -- (0,1.8,.3) -- (0,1.8,2.4) -- (0,.3,2.4) -- cycle;
	\filldraw[lightgray!20] (.3,1,0) -- (0,1,.3) -- (0,1.8,.3) -- (.3,1.8,0) -- cycle;
	\draw[fill=lightgray!20] (.3,1,0) -- (0,1,.3) -- (0,.3,1) -- (.3,0,1) -- (1,0,.35) -- (1,.35,0) -- cycle;
	\node (1) at (2,1.5,0) {$y(x+y+z)(x+y)^{-1}(y+z)^{-1}$};
	\node (2) at (2.1,0,1.95) {$z(x+y+z)(x+z)^{-1}(y+z)^{-1}$};
	\node (3) at (-1.25,2.25,2.25) {$x(x+y+z)(x+y)^{-1}(x+z)^{-1}$};
	\node (4) at (.4,.25,.25) {$x+y+z$};
	\node (5) at (2,.1,0) {$(y+z)/(x+y+z)$};
	\draw[dotted] (0,0,1.5) -- (0,0,2.1) -- (-.2,0,2) node[left] {$(x+z)/(x+y+z)$};
	\draw[dotted] (.1,1.5,.1) -- (.1,2.1,.1) node[above] {$(x+y)/(x+y+z)$};
	\draw[->] (.3,1,0) -- (.3,2,0);
	\draw[->] (0,1,.3) -- (0,2,.3);
	\draw[->] (.3,0,1) -- (.3,0,2.5);
	\draw[->] (0,.3,1) -- (0,.3,2.5);
	\draw[->] (1,0,.35) -- (2.5,0,.35);
	\draw[->] (1,.35,0) -- (2.5,.35,0);
	\end{scope} 
	\end{tikzpicture}
\end{center}
as our final result.

This establishes the first part of \Cref{prop:coordinates}, at least in the $N=3$ case. 

The rest can be deduced. For example, consider the upper-left corner of the hexagonal face $\mathrm{f}_{\{\{1,2,3\}\}}$ in $[0,1)^3_{\mathrm{tb}}$. This is $\mathrm{f}_{\{\{1\} ,\{1,2\},\{1,2,3\}\}}$.
Nearby, $z\gg y \gg x$, so, in some neighborhood $U$ of that corner, and
\begin{equation}
x+y+z  \in z C^\infty(U;\bbR^+), \qquad 
x+z \in z C^\infty(U;\bbR^+), \quad x+y \in y C^\infty(U;\bbR^+).
\end{equation}
Thus, the chosen bdfs depicted above are $x+y+z \in z C^\infty(U;\bbR^+)$, $(x+y)/(x+y+z) \in (y/z) C^\infty(U;\bbR^+)$, and 
\begin{equation}  
x(x+y+z)(x+y)^{-1} (x+z)^{-1} \in( x/y )C^\infty(U;\bbR^+). 
\end{equation}
This shows that $z,y/z,x/y$ serve as a valid coordinate system within $U$. The only permutation $\sigma\in \frakS_3$ consistent with $\mathtt{I} = \{\{1\} ,\{1,2\},\{1,2,3\}\}$ is $\sigma = (1,3)$, which reverses the order of $1,2,3$. That is, $\sigma(1) = 3$, $\sigma(2)=2$, and $\sigma(3) = 1$. The coordinates $\varrho,\hat{x}_j$ defined in \cref{eq:misc_b77} are 
\begin{equation} 
\varrho  = x_3 = z, \quad \hat{x}_2 = x_{2} / x_3 = y/z,
\end{equation} 
and $\hat{x}_1 = x_1/x_2 = x/y$. It can be seen that $U$ can be taken to be any open set not containing any of the other corners of $\mathrm{f}_{\{\{1,2,3\}\}}$. Each corner is analogous, so the final clause of \Cref{prop:coordinates} follows, at least in the considered $N=3$ case, from the computations above.

\printbibliography

\end{document}